\documentclass[11pt, a4paper]{article}
\usepackage[a4paper,total={6.2in, 9in}]{geometry}
\RequirePackage{amsthm,amsmath,amsfonts,amssymb}
\RequirePackage[numbers]{natbib}
\RequirePackage{graphicx}
\usepackage{bm}
\usepackage{bbm}
\usepackage{url}
\usepackage{comment}

\usepackage{xr-hyper}
\usepackage{hyperref} 
\usepackage{tikz}
\usepackage{xcolor}
\usepackage{graphicx}
\usepackage[title]{appendix}

\usetikzlibrary{bayesnet}
\usetikzlibrary{arrows}

\theoremstyle{plain}

\newtheorem{theorem}{Theorem}[section]
\newtheorem{lemma}[theorem]{Lemma}
\newtheorem{proposition}[theorem]{Proposition}
\newtheorem{corollary}[theorem]{Corollary}

\theoremstyle{remark}

\newtheorem*{remark}{Remark}

\newcommand{\rTV}{\right \|_{TV}}
\newcommand{\lTV}{\left \|}
\newcommand{\rW}{\right \|_{W}}
\newcommand{\lW}{\left \|}

\def\sN {\mathcal{N}}
\def\sP {\mathcal{P}}
\def\sO {\mathcal{O}}

\def\One {\bm{1}}

\def \L {\mathcal{L}}
\def \sX {\mathcal{X}}
\def \sY {\mathcal{Y}}
\def \E {\mathbb{E}}
\def \R {\mathbb{R}}

\def\d {\text{d}}
\def\x {\textbf{x}}
\def\betag {\bm{\beta}}
\def\y {\textbf{y}}

\def \I {\mathbb{I}}
\def \Fisher {\mathcal{I}}
\def \H {\mathbb{H}}

\def \bT {\bm{T}}

\def \det {\text{det}}
\def \Var {\text{Var}}

\def \simiid {\overset{\text{iid}}{\sim}}

\title{Dimension-free mixing times of Gibbs samplers for Bayesian hierarchical models}
\author{Filippo Ascolani\footnote{Department of Decision Sciences and BIDSA, Bocconi University, \url{filippo.ascolani@phd.unibocconi.it}}\;  and Giacomo Zanella\footnote{Department of Decision Sciences and BIDSA, Bocconi University,\url{giacomo.zanella@unibocconi.it}
}}

\begin{document}
\maketitle

\abstract{Gibbs samplers are popular algorithms to approximate posterior distributions arising from Bayesian hierarchical models. 
Despite their popularity and good empirical performances, however, there are still relatively few quantitative results on their convergence properties, 
e.g. much less than for gradient-based sampling methods. 
In this work we analyse the behaviour of total variation mixing times of Gibbs samplers targeting hierarchical models using tools from Bayesian asymptotics.
We obtain dimension-free convergence results under random data-generating assumptions, for a broad class of two-level models with generic likelihood function.
Specific examples with Gaussian, binomial and categorical likelihoods are discussed.}

\section{Introduction}\label{section_introduction}

Gibbs samplers \cite{C92} are a family of Markov Chain Monte Carlo (MCMC) algorithms \cite{B11} commonly used in various scientific fields. In the context of Bayesian Statistics, they are routinely employed to draw samples from posterior distributions of unknown parameters conditional to the observed data \cite{G15,MFR23}. 
Like most MCMC methods, they are guaranteed to converge to the correct posterior distribution as the number of iterations tends to infinity under mild assumptions \citep{R94}. However, understanding how quickly this convergence occurs, for example by quantifying the so-called mixing time of the Markov chain generated by the algorithm, is in general a hard task. 
In this paper we address this question for Gibbs samplers targeting certain classes of high-dimensional Bayesian hierarchical models. 
Analysing convergence properties, such as mixing times, is the key technical step needed to rigorously quantify the computational cost of MCMC algorithms.

\subsection{Hierarchical models}
Our motivating example is given by classical Bayesian hierarchical models of the form
\begin{equation}\label{eq:one_level_nested_intro}
\begin{aligned}
Y_j
\mid \theta_j & \sim f(\cdot \mid \theta_j) & j = 1, \dots, J,\\
 \theta_j\mid \psi &\overset{\text{iid}}{\sim} p(\cdot \mid \psi)& j = 1, \dots, J,\\
 \psi &\sim p_0(\cdot)\,.&
\end{aligned}
\end{equation}
Here the observed dataset $Y_{1:J}=(Y_j)_{j=1,\dots,J}$ is divided into $J$ groups, with data for each group typically containing multiple observations, e.g.\ $Y_j=(Y_{j1},\dots,Y_{jm})$.
Each group features some local (i.e.\ group-specific) parameters $\theta_j\in\R^\ell$, while $\psi\in\R^D$ are global (hyper)-parameters. 
Above $f(\cdot \mid \theta)$, $p(\cdot \mid \psi)$ and $p_0(\cdot)$ denote some likelihood function, local prior and global prior, respectively. See Section \ref{hierarchical} for the assumptions we require on each of those.
Given model \eqref{eq:one_level_nested_intro}, posterior inferences are based on the conditional distribution of $\psi$ and $\bm{\theta} = (\theta_1, \dots, \theta_J)$ given $Y_{1:J}$, which we denote as $\L(\d\bm{\theta}, \d \psi |Y_{1:J})$. 
Hierarchical models such as \eqref{eq:one_level_nested_intro} are the workhorse of Bayesian Statistics and are commonly employed in many applied contexts (see e.g. \cite{GH07,G13} and references therein). 
In this paper, we are mostly interested in the high-dimensional regime where $J\to\infty$, so that both the number of datapoints and parameters, i.e. $n=Jm$ and $p=J\ell+D$ respectively, diverge. 

One iteration of a Gibbs sampler targeting $\L(\d\bm{\theta}, \d \psi|Y_{1:J})$ sequentially samples each parameter from its full-conditional distribution, i.e.\ it performs the updates $\theta_j\sim\L(\d\theta_j|Y_{1:J},\psi)$ for $j=1,\dots,J$ and $\psi\sim\L(\d\psi|Y_{1:J},\bm{\theta})$.
Algorithms based on conditional updates are well-suited to model \eqref{eq:one_level_nested_intro}, since they naturally exploit the underlying sparse dependence structure. 
In particular, the conditional independence of $\theta_1,\dots,\theta_J$ given $Y_{1:J}$ and $\psi$ implies that the sequence of updates from the low-dimensional distributions $\L(\d\theta_j|Y_{1:J},\psi)$ for $j=1,\dots,J$ is equivalent to an exact  joint update from the high-dimensional distribution $\L(\d\bm{\theta}|Y_{1:J},\psi)$. 
Also, since local parameters interact only with local data conditional on $\psi$, i.e.\ $\L(\d\theta_j|Y_{1:J},\psi)=\L(\d\theta_j|Y_j,\psi)$, one iteration of the Gibbs sampler can typically be implemented with a computational cost that scales linearly with $J$.
For the sake of comparisons, a similar cost is required by a single likelihood evaluation or a single posterior gradient evaluation for model \eqref{eq:one_level_nested_intro}. 
See also Remark \ref{computational_cost} in Section \ref{section:main_result} for related discussion. 

The key question to properly assess the effectiveness of Gibbs samplers targeting model \eqref{eq:one_level_nested_intro} is how fast the resulting Markov chain converges to its stationary distribution $\L(\d\bm{\theta}, \d \psi|Y_{1:J})$.
Interestingly, such chain often enjoys dimension-free convergence speed, 
meaning that  the number of iterations required to converge does not grow (or grows only logarithmically) with $J$.
 \begin{figure}
\centering
\includegraphics[width=.4\textwidth]{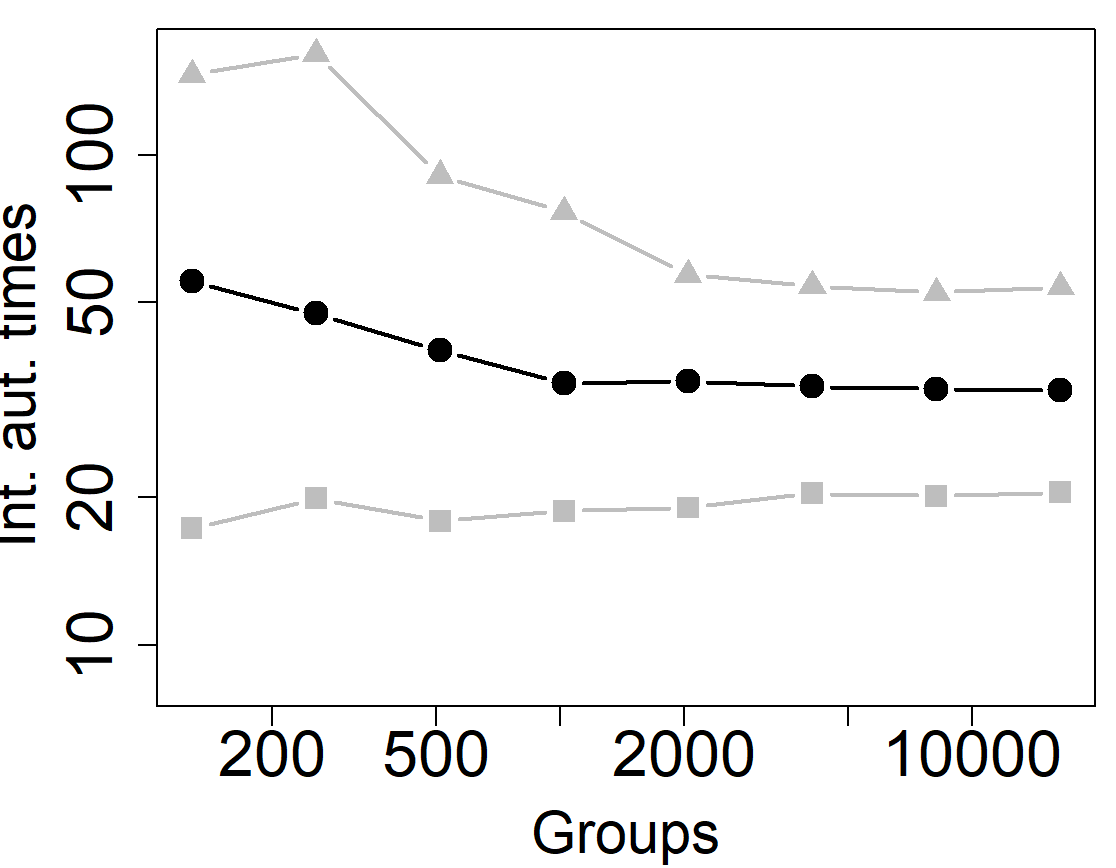} \quad
\includegraphics[width=.4\textwidth]{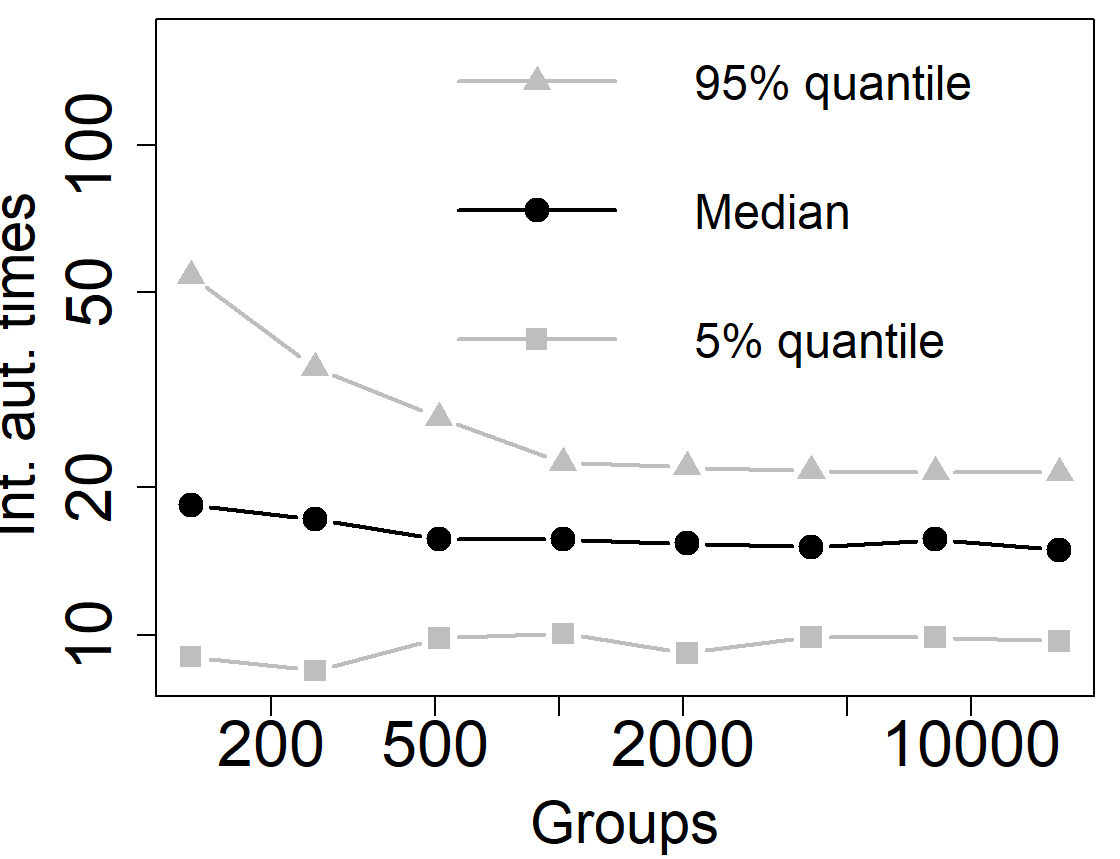}
 \caption{\small{
 Integrated autocorrelation times (on log-scale) of Gibbs samplers targeting the posterior distribution of model \eqref{eq:one_level_nested_intro} with specification \eqref{specification_figure1}.
Quantiles refer to repetitions over datasets randomly generated according to the model with true parameters $\mu^* = \tau^* = 1$. Left: $m = 3$. Right: $m = 5$. 
See Section \ref{binary_data} for more details.
  }}
 \label{fig:binom_simulations}
\end{figure}
Figure \ref{fig:binom_simulations} illustrates numerically this behaviour on a hierarchical logistic model, where the likelihood and prior in \eqref{eq:one_level_nested_intro} are specified as
\begin{equation}\label{specification_figure1}
f(y \mid \theta) = \binom{m}{y}\frac{e^{y\theta}}{(1+e^\theta)^m}, \quad p(\theta \mid \psi) = N(\theta \mid \mu, \tau^{-1}), \quad \psi = (\mu, \tau),
\end{equation} 
with $y \in \{0,\dots, m\}$ and $m$ being a positive integer.
The prior for $\psi = (\mu, \tau)$ is set to $\mu \mid \tau \sim N\left(0, 10^{3}/\tau\right)$ and $\tau \sim \text{Gamma}(0.1, 0.1)$. 
Full details on the simulation set-up of Figure \ref{fig:binom_simulations} are described in Section \ref{binary_data}. 
The results suggest that the number of iterations required by the Gibbs sampler to draw each sample from $\L(\d\bm{\theta}, \d \psi|Y_{1:J})$ remains bounded as 
$J$ grows and asymptotes to a finite value as $J\to\infty$. 
Combined with cost per iteration, this implies a computational complexity that grows linearly with $J$. 
Note that this complexity is smaller than the one of popular gradient-based MCMC methods when applied to these models (see Section \ref{sec:literature} for more details), supporting the idea that Gibbs samplers can achieve state-of-the-art performances for hierarchical models with sparse dependence structures.

In Section \ref{hierarchical} we provide rigorous support to the above empirical evidences. In particular, we study the asymptotic behavior of mixing times of Gibbs samplers targeting model \eqref{eq:one_level_nested_intro}. 
There we prove that mixing times remain bounded as $J\to\infty$ under mild assumptions on the likelihood $f$ and the global prior $p_0$. 
We instead require stronger assumptions on the local priors $p(\cdot \mid \psi)$, which we assume to be in the exponential family. 
Our results (see e.g.\ Theorem \ref{theorem_one_level_nested}) are average-case ones and hold with high probability with respect to the law of the data-generating process. 
To do so we assume the observed data $Y_{1:J}$ to be randomly generated. This allows to use tools of Bayesian asymptotics, such as Bernstein-von Mises type statements (see e.g.\ Chapter $10$ of \cite{V00}), to characterize the asymptotic posterior behaviour as $J\to\infty$ and then extract information about the limiting behaviour of the associated sequence of MCMC algorithms.

\subsection{Related literature}\label{sec:literature}
The literature on performances of MCMC methods is very broad.
The most well-studied classes of algorithm are probably gradient-based ones, such as Langevin \citep{R96} and Hamiltonian \citep{N11} Monte Carlo, see e.g.\ \cite{D17, DM17,D19} and related literature.
Available results suggest that the number of iterations (or target gradient evaluations) required by those algorithm to converge to stationarity increases with dimensionality, e.g.\ growing as $\mathcal{O}(J^{\alpha})$ with the dimensionality $J$, for some $\alpha>0$ that depends on the setup and type of algorithm \citep{R98,B13,WSC22}.
In the context of hierarchical models, given that each target gradient evaluation has a linear cost in $J$, this leads to a computational cost to sample from $\L(\d\bm{\theta}, \d \psi |Y_{1:J})$ that scales super-linearly with $J$, e.g.\ as $\mathcal{O}(J^{1+\alpha})$ with $\alpha>0$.
Comparing these results to the one we develop here for Gibbs samplers suggests that, while being state-of-the-art black-box schemes to sample from generic high-dimensional distributions with appropriate regularity conditions (e.g.\ log-concavity), default gradient-based MCMC schemes can be suboptimal for high-dimensional hierarchical models. See also \cite{SPZ23} for related numerical evidences.

Compared to gradient-based MCMC, results for Gibbs-type schemes are less abundant and more model-dependent. 
Notable recent examples include \cite{Y17,JH21,Q22}, which provide convergence bounds for hierarchical models, similar to \eqref{eq:one_level_nested_intro}, with Gaussian and Poisson likelihoods.
Another recent result is given by \cite{Q19}, which provides dimension-free convergence bounds for Gibbs samplers for high-dimensional probit regression models under appropriate regimes. 
Providing sharp non-asymptotic analyses like the ones above requires proof techniques, such as drift-and-minorization techniques \citep{R95} and random mappings \cite{Q19}, that are usually likelihood-specific and potentially hard to construct. 
For example, they may require to devise and study a suitable Lyapunov function that depends on the specific choices of both likelihood and priors in \eqref{eq:one_level_nested_intro} (see e.g.\ formulae (6) and (33) in \cite{JH21} and \cite{Y17}, respectively).
On the other hand, these approaches provide non-asymptotic bounds that apply to fixed sample size and dimensionality, thus being complimentary to the high-dimensional asymptotic analysis we develop here.

Interestingly, there are relatively few papers combining the tools of Bayesian asymptotics and MCMC theory in rigorous ways. The work in \cite{BC09} uses Bernstein-von Mises Theorem to provide polynomial bounds on the convergence of random walk Metropolis-Hastings schemes. After that, very recent papers use similar techniques to provide complexity analysis of MCMC schemes, see e.g.\ \cite{NW20,N22, T22} dealing with gradient-based methods, the first in the context of inverse problems. A brief discussion about the use of asymptotic posterior characterisations to study the convergence properties of Gibbs samplers is given in \cite{RS01}. 
A more in-depth use of Bayesian asymptotics to study data augmentation procedures is given in \cite{K14}, which also considers hierarchical models.  See Remark \ref{connection_localConsistency} in Section \ref{hierarchical} for more details on the results in \cite{K14}.
Finally, an interesting exception is given by Bayesian variable selection models, where multiple works have exploited the asymptotic behaviour of the posterior distribution to characterize the computational performances of Bayesian methods \cite{Y16,A21,Z21}.

\subsection{Sketch of the main arguments and structure of the paper}\label{sec:sketch}
The argument we employ to study Gibbs samplers targeting $\L(\d\bm{\theta}, \d \psi \mid Y_{1:J})$ can be decomposed in three main parts. First, if $p(\cdot \mid \psi)$ belongs to the exponential family, there exists a set of sufficient statistics $\bT = \bT(\bm{\theta})$, whose dimensionality does not depend on $J$, such that $\L\left(\d\psi \mid \bm{\theta}, Y_{1:J}\right)=\L\left(\d\psi \mid \bm{T}(\bm{\theta}), Y_{1:J}\right)$. 
Lemma \ref{sufficient_lemma} in Section \ref{sec:dim_red} shows that, as a result, the Gibbs sampler on $\L\left(\d\bm{\theta}, \d \psi\mid Y_{1:J}\right)$ has the same mixing times as the one on $\L\left(\d\bm{T}, \d \psi\mid Y_{1:J}\right)$. This allows to focus on the latter distribution which, unlike the former, is intractable but fixed dimensional. Note that this dimensionality reduction does not require the likelihood $f$ to admit sufficient statistics (see Remark \ref{rmk:lik_non_exp}) and is a peculiar property of Gibbs samplers, since it exploits the presence of exact updates. 
The second step consists in studying the asymptotic behaviour of $\L\left(\d\bm{T}, \d \psi \mid Y_{1:J}\right)$ as $J$ increases. In particular, Proposition \ref{limiting_sigma} shows that a suitable rescaling of $\left(\bT, \psi \right)$ converges to a multivariate Gaussian distribution in total variation distance. The proof combines a classical Bernstein-von Mises Theorem for $\psi$ (Lemma \ref{asymptotic_distribution_psi}) with a less standard Central Limit Theorem for $\bm{T}$ conditional on $\psi$ (Lemma \ref{asymptotic_distribution_T}). More details can be found in Section \ref{section_sketch_main_result}.
The final and key point is then to connect the convergence of the target distributions, in this case $\{\L\left(\d\bm{T}, \d \psi \mid Y_{1:J}\right)\}_{J\geq1}$, to the convergence of the associated Gibbs sampler operators. Theorem \ref{thm:mixing_limit} proves that the limiting behaviour of a sequence of Gibbs samplers is equivalent to the behaviour of the Gibbs sampler on the limiting distribution: this is shown in total variation distance and under warm start assumption. The fundamental link is given by Proposition \ref{prop:convergence_operators}, which provides an upper bound on the distance between Gibbs sampler operators in terms of the one between the target distributions. Since those results are of independent interest and are not specific to hierarchical models, we start by developing those in a general setup in Section \ref{section:gibbs}. 
Then, Section \ref{sec:parametric_model} recalls the Bernstein-von Mises Theorem and illustrates the results of Section \ref{section:gibbs} to the fixed-dimensional setting. 
Section \ref{hierarchical} develops the main results of the paper dealing with general hierarchical models (see e.g.\ Theorem \ref{theorem_one_level_nested}) and Section \ref{examples} verifies the general conditions for some specific likelihood families, e.g. Gaussian, binomial and categorical, together with providing numerical simulations and extension to different graphical model structures. 
Since a warm start initialization for the sampler is assumed throughout, the availability of feasible starts is discussed in Section \ref{Feasible_start}. 
Finally, Section \ref{sec:conclusions} discusses extensions and future work.

\section{Gibbs sampler and asymptotics}\label{section:gibbs}
In this section, after recalling basic definitions about Gibbs kernels and mixing times, we connect the convergence of a sequence of target distributions to the convergence of the associated Gibbs kernels. This leads to Theorem \ref{thm:mixing_limit}, which characterizes the limiting behaviour of the Gibbs samplers mixing times. 
Throughout this section, the target distributions are assumed to have fixed dimensionality.

\subsection{Setup and notation}
Let $(\pi_n)_{n \geq 1}=\left(\pi_n(\cdot\mid Y^{(n)})\right)_{n \geq 1}$ 
  be a sequence of probability distributions on a common product space $\sX = \sX_1 \times \dots \times \sX_K$, where each $\pi_n$ is allowed to depend on some observed data $Y^{(n)}\in\sY^{{(n)}}$.
In our applications, $\pi_n(\cdot \mid Y^{(n)})$ represents the posterior distribution 
of some unknown parameter $\x\in\sX$ conditioned on the data $Y^{(n)}$. 
For the sake of brevity, we will often omit the explicit dependence on $Y^{(n)}$.

Let $P_n$ be the Markov transition kernel of the deterministic-scan Gibbs sampler targeting $\pi_n$, defined as the product of $K$ kernels
\begin{equation}\label{initialGibbs}
P_n = P_{n,1} \cdot\cdot\cdot P_{n,K}\,.
\end{equation}
For each $i\in\{1,\dots,K\}$, $P_{n,i}$ is the transition kernel on $\sX$ that updates the $i$-th coordinate drawing it from its conditional distribution $\pi_n(\d x_i|\x^{(-i)})$, where $\x^{(-i)}=(x_j)_{j\neq i}$, while leaving the other components unchanged. Equivalently
\[
P_{n,i}\left(\x, S_{\x, i, A}\right) = \int_A \pi_n\left(\d y_i \mid \x^{(-i)}\right), \quad A \subset \sX_i, \quad i = 1, \dots, n,
\]
with
$
S_{\x, i, A} = \left\{\y \in \sX \, : \, y_j = x_j \, \forall  \, j \neq i \text{ and } y_i \in A \right\}
$. 
It is easy to show that $P_{n,i}$ is reversible with respect to $\pi_n$ for every $i$, 
so that $\pi_n$ is the invariant distribution of $P_n$ \citep{R04,H11,CLM23}.

Given $\epsilon\in (0,1)$, define the $\epsilon$-total variation mixing time of $P_n$ with starting distribution $\mu_n\in \mathcal{P}(\sX)$, where $\mathcal{P}(\sX)$ denotes the set of probability distribution on $\sX$, as 
\begin{align}
t^{(n)}_{mix}(\epsilon,\mu_n)&=\inf\left\{t\geq 0\,:\, \lTV \mu_nP_n^t-\pi_n \rTV<\epsilon \right\},
\end{align}
where $P^t$ denotes the $t$-th power of $P$, $\mu_n P^t_n(A) = \int_{\sX} P^t_n(\x,A)\mu_n(\d \x)$ for any $A\subseteq \sX$ and $\|\cdot\|_{TV}$ denotes the total variation norm.
By definition, mixing times quantify the number of Markov chain's iterations required to obtain a sample from the target distribution $\pi_n$ up to error $\epsilon$.
We will focus on worst-case mixing times with respect to $M$-warm starts.
The set of $M$-warm starts relative to a distribution $\pi$ is defined as
\begin{align}\label{N_class}
\sN\left(\pi, M \right)=&
\left\{\mu\in\mathcal{P}(\sX)\,:\,\mu(A)\leq M\pi(A) \hbox{ for all }A\subseteq \sX\right\}, &M \geq 1,\,\pi \in \mathcal{P}(\sX) \,,
\end{align}
and the associated worst-case mixing times for $P_n$ targeting $\pi_n$ are
\begin{align}\label{eq:t_mix1}
t^{(n)}_{mix}(\epsilon, M)&=\sup_{\mu_n\in \sN \left(\pi_n, M \right)}
t^{(n)}_{mix}(\epsilon,\mu_n)\,.
\end{align}
\begin{remark}
While being common in the literature, see e.g. \cite{D17, D19, T22} for gradient-based methods, the warm start assumption can be quite stringent and potentially unrealistic. In particular, assuming that the algorithm can be initialised by sampling the starting configuration from a warm start with relatively small $M$ (e.g.\ one that does not grow exponentially fast with dimensionality) may be unrealistic. In Section \ref{Feasible_start} we show that in the specific case of hierarchical models as in \eqref{eq:one_level_nested_intro} a feasible start, i.e. a starting distribution which can be implemented in practice and allows to control the value of $M$, is available under some assumptions.
\end{remark}

\subsection{Assumptions on the sequence of target distributions}
We consider settings where a rescaled version of the sequence $(\pi_n)_{n \geq 1}$ converges to a well defined limiting distribution as $n\to\infty$. This is often the case in a Bayesian context where some version of the Bernstein von-Mises theorem holds (see e.g.\ Theorem \ref{BvM} below).
The convergence of $(\pi_n)_{n \geq 1}$ occurs with high probability assuming the data $Y^{(n)} $ is randomly generated from some distribution. 
In particular, we assume for the rest of this section that $Y^{(n)}$ is random with distribution $Q^{(n)} \in \mathcal{P}\left(\sY^{{(n)}} \right)$.
The following assumption specifies the convergence we require for $(\pi_n)_{n \geq 1}$: 
\begin{enumerate}
\item [(A1)] There exists $\tilde{\pi} \in \mathcal{P}(\sX)$ and a sequence of transformations $\phi_n  \, : \, \sX \, \to \, \sX$ that act \emph{coordinate-wise}, i.e. where
\begin{align}\label{trans_phi}
\phi_n(\x) &= \left( \phi_{n,1}(x_1), \dots, \phi_{n, K}(x_K)\right)\,,&\x\in\sX
\end{align}
with $\phi_{n, j} \, : \, \sX_j \, \to \, \sX_j$  injective and measurable,
such that
\begin{align}\label{eq:conv_to_pi_inf}
\lTV \tilde{\pi}_n - \tilde{\pi}\rTV &\to 0 &\hbox{ as }n \to \infty\,,
\end{align}
in $Q^{(n)}$-probability, i.e. such that $\lim_{n \to \infty}Q^{(n)}(\lTV \tilde{\pi}_n - \tilde{\pi}\rTV>\epsilon)=0$ for every $\epsilon \in (0,1)$, where $\tilde{\pi}_n = \pi_n \circ \phi_n^{-1}$ is the law of $\tilde{\x} = \phi_n(\x)$ under $\x\sim\pi_n$. 
\end{enumerate}
\begin{remark}
The necessity of rescaling $\x$ by some transformation $\phi_n$ in \eqref{trans_phi} comes from the typical behaviour of posterior distributions in Bayesian models. Indeed, without rescaling, $\pi_n$ often converges to a random variable which is degenerate to a Dirac delta at a fixed value (e.g. the underlying data-generating parameter). Thus, in order to have a non-trivial limit and total variation convergence, which is essential for our purposes, a suitable rescaling is needed. In our context the specific form of this transformation is dictated by the theory of Bayesian asymptotics, see e.g. Theorem \ref{BvM} below. 
Moreover, we assume $\phi_n$ to act coordinate-wise because this class of transformations leaves 
Gibbs samplers invariant (see e.g.\ Lemma \ref{equivalent_statement} below), while general one-to-one transformations can alter the Gibbs sampler dynamics and change its convergence speed \citep{papaspiliopoulos2007general}.
\end{remark}
\begin{remark}
The results we develop below could be extended to more general versions of assumption (A1), including ones where the co-domain of $\phi_n$ is not equal to the domain, i.e.\ $\phi_n  \, : \, \sX \, \to \, \mathcal{Z}$ for some $\mathcal{Z}$, and where the limiting distribution $\tilde{\pi}$ is random, i.e.\ allowed to depend on the sequence $( Y^{(n)} )_n$. Since $(A1)$ is enough for our purposes and motivating applications, we do not consider such extensions here to keep notation simple.
\end{remark}

Let $\tilde{P}$ and $\tilde{P}_n$ be the kernels of the Gibbs samplers targeting $\tilde{\pi}$ and $\tilde{\pi}_n$, respectively.
The following lemma shows that studying total variation convergence from $M$-warm starts for the sequence of kernels $(P_n)_{n \geq 1}$ is equivalent to doing it for the sequence $(\tilde{P}_n)_{n \geq 1}$ .
The proof, which can be found in Appendix $C$, relies on the coordinate-wise and bijective requirements of (A1).
\begin{lemma}\label{equivalent_statement}
Under Assumption (A1) we have
\[
\sup_{\mu_n\in \sN \left(\pi_n, M \right)} \, \lTV \mu_n P_n^t-\pi_n\rTV  = \sup_{\tilde{\mu}_n\in \sN \left(\tilde{\pi}_n, M \right)} \, \lTV \tilde{\mu}_n \tilde{P}_n^t-\tilde{\pi}_n\rTV.
\]
\end{lemma}
 
 \subsection{Convergence of Gibbs samplers operators}
Since by $(A1)$ the stationary distribution of $\tilde{P}_n$, the Gibbs samplers targeting $\tilde{\pi}_n$, converges to the one of $\tilde{P}$, one may be tempted to translate such convergence at the level of the kernels, e.g. $\| \tilde{P}_n(\x, \cdot) - \tilde{P}(\x, \cdot) \|_{TV} \to 0$ for ($\tilde{\pi}$-almost) every $\x \in \sX$. However this is not only false for generic Markov operators, but even in the special class of Gibbs sampler operators: one can have $\lTV \tilde{\pi}_n - \tilde{\pi}\rTV \to 0$ as $n \to \infty$, while  $\| \tilde{P}_n(\x, \cdot) - \tilde{P}(\x, \cdot) \|_{TV}\nrightarrow 0$ for any $\x\in\sX$, see e.g. Example $A.1$ in Appendix A. 
The reason is that convergence of the joint distribution $\tilde{\pi}_n$ in total variation distance does not imply convergence of the associated conditional distributions, that are the building blocks of the Gibbs sampler operator. 
However, it turns out that a control on the total variation distance between two target distributions is in general sufficient to control the distance between the corresponding Gibbs sampler operators applied to warm starts. 
The following Proposition makes the connection precise. Interestingly, no assumptions on the target distribution and Gibbs samplers are required.
\begin{proposition}\label{prop:convergence_operators}
Let $P_1$ and $P_2$ be the transition kernels of Gibbs samplers targeting $\pi_1 \in \mathcal{P}(\sX)$ and $\pi_2 \in \mathcal{P}(\sX)$, respectively. 
Then we have
\begin{equation}\label{eq:tv_gibbs_kernel_bound}
\lTV \mu P_1 - \mu P_2  \rTV \leq 2MK\lTV \pi_1- \pi_2  \rTV,
\end{equation}
for every $\mu \in \sN (\pi_1, M) \cup \sN (\pi_2, M)$ and $M\geq 1$.
\end{proposition}
Proposition \ref{prop:convergence_operators} translates convergence of the stationary distributions, given by $(A1)$, into convergence of the Gibbs samplers operators when a warm start is considered. 
It is worth noting that a bound of this form cannot hold for generic Markov transition kernels. Indeed, consider transition kernels $P_1$ and $P_2$ with the same stationary distribution $\pi$: by basic properties of the total variation distance it holds $\lTV \mu P_1 - \mu P_2 \rTV \leq 2 \lTV \mu-\pi \rTV$. The latter bound cannot be improved in general, meaning that it is possible to find ergodic kernels $P_1$ and $P_2$ that get arbitrarily close to the above upper bound, see Example A.2 in Appendix A.

Proposition \ref{prop:convergence_operators} is used in the proof of Theorem \ref{thm:mixing_limit}, which shows that the limiting behaviour of $P_n$, in terms of distance to stationarity from $M$-warm starts, is completely characterized by the behaviour of the limiting operator $\tilde{P}$.
The proof of Theorem \ref{thm:mixing_limit} also relies on the fact that 
 the total variation distance between $\pi_1$ and $\pi_2$ provides a control on the distance between the two sets $\sN(\pi_1, M)$ and $\sN(\pi_2, M)$, as shown in the following Lemma. 
\begin{lemma}\label{constructive_lemma}
Let $\pi_1,\pi_2 \in \mathcal{P}(\sX)$. Then, for every $\mu_1 \in \sN(\pi_1, M)$, there exists $\mu_2 \in \sN(\pi_2, M)$ such that
$\lTV \mu_1-\mu_2 \rTV \leq M\lTV  \pi_1 - \pi_2 \rTV$.
\end{lemma}
Lemma \ref{constructive_lemma} implies that, under assumption $(A1)$, for every $\tilde{\mu} \in \sN(\tilde{\pi}, M)$ there exists a sequence $\{\tilde{\mu}_n\}_n$ such that $\tilde{\mu}_n \in \sN(\tilde{\pi}_n, M)$ and $\lTV \tilde{\mu}_n-\tilde{\mu} \rTV \to 0$ as $n \to \infty$ in $Q^{(n)}$-probability. We can now state Theorem \ref{thm:mixing_limit}.

\begin{theorem}\label{thm:mixing_limit}
Let assumption $(A1)$ holds. Then for every $t \in \mathbb{N}$ and $M\geq 1$ it holds
\[
\lim_{n \to \infty} \, \sup_{\mu_n\in \sN (\pi_n, M)} \, \lTV \mu_n P_n^t-\pi_n\rTV = \sup_{\tilde{\mu} \in \sN (\tilde{\pi}, M)} \, \lTV \tilde{\mu} \tilde{P}^t-\tilde{\pi} \rTV,
\]
in $Q^{(n)}$-probability.
\end{theorem}

\begin{remark}
An alternative approach to derive convergence statements on the sequence of Gibbs kernels 
would be to consider stronger forms of convergence for the sequence $(\tilde{\pi}_n)_{n \geq 1}$ than the one in total variation distance in \eqref{eq:conv_to_pi_inf}.
However, we prefer to derive results under weaker convergence requirements for $(\tilde{\pi}_n)_{n \geq 1}$ to allow for a more direct use of standard asymptotic results in the Bayesian literature (e.g.\ common formulations of the Bernstein-von Mises theorem), which are usually derived in terms of weaker metrics such as total variation one.
\end{remark}

\subsection{Implications for mixing times}
Denote the mixing times of $\tilde{P}$ as
\[
\tilde{t}_{mix}(\epsilon, M) = \underset{\tilde{\mu} \in \sN \left(\tilde{\pi}, M \right)}{\sup} \, 
\inf\left\{t\geq 1\,:\, \lTV \tilde{\mu}\tilde{P}^t-\tilde{\pi} \rTV<\epsilon \right\}.
\]
The following corollary of Theorem \ref{thm:mixing_limit} shows how to use $\tilde{t}_{mix}(\epsilon, M) $ to deduce statements on the behaviour of the sequence of mixing times of interest, $(t^{(n)}_{mix}(\epsilon, M))_{n \geq 1}$.
\begin{corollary}\label{mixingCorollary}
Let assumption $(A1)$ holds. If $(M, \epsilon)\in[1,\infty)\times (0,1)$ is such that $\tilde{t}_{mix}(\epsilon, M) < \infty$, then
\begin{align}\label{eq:mix_limiting_bounds}
Q^{(n)}\left(t^{(n)}_{mix}(\epsilon, M)\leq \tilde{t}_{mix}(\epsilon, M)\right) \to 1
\end{align}
as $n \to \infty$. Otherwise, if $(M, \epsilon)\in[1,\infty)\times (0,1)$ is such that $\tilde{t}_{mix}(\epsilon, M) = \infty$, then it holds
\[
Q^{(n)}\left(t^{(n)}_{mix}(\underline{\epsilon}, M)<T\right) \to 0
\]
as $n \to \infty$, for every $\underline{\epsilon} < \epsilon$ and $T > 0$.
\end{corollary}
\begin{remark}[Mixing times bounded in probability]
When $\tilde{t}_{mix}(\epsilon, M) < \infty$, the statement in \eqref{eq:mix_limiting_bounds} implies that $t^{(n)}_{mix}(\epsilon, M)=\sO_P(1)$ as $n\to\infty$, i.e.\ that the sequence of random variables $(t^{(n)}_{mix}(\epsilon, M))_{n\geq 1}$ is bounded in probability. The latter means that for every $\delta > 0$ there exist an integer $N_\delta$ and a real constant $B_\delta < \infty$ such that
$Q^{(n)}(t^{(n)}_{mix}(\epsilon, M) \leq B_\delta) \geq 1-\delta$ for every $n \geq N_\delta$, which holds by \eqref{eq:mix_limiting_bounds} taking $B_\delta = \tilde{t}_{mix}(\epsilon, M)$.
\end{remark}

By Corollary \ref{mixingCorollary}, establishing whether $\tilde{P}$ is ergodic (in the sense of yielding finite mixing times) or not is enough to discriminate between sequences of kernels $(P_n)_{n \geq 1}$ whose mixing times diverge as $n\to\infty$ as opposed to ones that do not (see e.g.\ Figure \ref{fig:normal_simulations} in Section \ref{examples} for an illustration). 
Since ergodicity of Gibbs samplers can be established 
under very mild assumptions \citep{R94}, 
in practice one can expect $\tilde{P}$ to be ergodic and thus $(t^{(n)}_{mix}(\epsilon, M))_{n \geq 1}$ to be bounded in probability whenever $(A1)$ holds for a well-behaved, non-singular limiting distribution $\tilde{\pi}$. 
Sections \ref{hierarchical} and \ref{examples} combine Corollary \ref{mixingCorollary} with dimensionality reduction techniques to provide results on Gibbs samplers targeting high-dimensional hierarchical models.

\begin{remark}[Alternative metrics]
It is natural to wonder whether the result of Corollary \ref{mixingCorollary} may hold for weaker metrics, like the one induced by the Wasserstein distance. However, it is possible to find examples where the convergence of the stationary distributions (in Wasserstein distance) does not imply convergence of the associated mixing times (neither the ones defined based on the TV distance nor the ones defined based on the Wasserstein one). 
The intuition is that the limiting distribution in weaker metrics (e.g.\ Wasserstein, weak convergence, etc) may ignore features of the joint distribution, such as full conditionals behaviours, that have a relevant impact on Gibbs sampler dynamics. 
For example, a sequence of increasingly correlated random variables (whose Gibbs samplers converge slower and slower) may converge to a single point mass, for which independence and immediate convergence automatically holds. See Example $A.3$ in Appendix A.
\end{remark}

\subsection{Explicit limiting bounds}
Corollary \ref{mixingCorollary} can also be used to derive quantitative bounds on the limiting behaviour of the mixing times $(t^{(n)}_{mix}(\epsilon, M))_{n \geq 1}$.
In particular, if one is able to establish explicit bounds on $\tilde{t}_{mix}(\epsilon, M)$,  then \eqref{eq:mix_limiting_bounds} implies a corresponding bound in high probability on $t^{(n)}_{mix}(\epsilon, M)$ for large $n$.
While deriving quantitative bounds on Gibbs samplers mixing times is in general hard, the limiting distribution $\tilde{\pi}$ is often more tractable than the original sequence $(\pi_n)_{n \geq 1}$, a common case being the one where $\tilde{\pi}$ is multivariate Gaussian while $(\pi_n)_{n \geq 1}$ is not.
In those scenarios explicit bounds on $\tilde{t}_{mix}(\epsilon, M)$ can be derived using available results on the convergence properties of Gibbs samplers targeting multivariate Gaussian distributions, see e.g. \cite{A91, K09, R97}. For example, Theorem $2$ in \cite{A91} provides an explicit bound for deterministic scan Gibbs samplers on Gaussian targets in $L^2$-distance (and therefore total variation \cite{AL22}).

In Sections \ref{hierarchical} and \ref{examples} we will apply this strategy mostly to cases where $K=2$, meaning that $\tilde{P}$ is a two-block Gibbs sampler. In this situation, one can use spectral gaps to bound Gibbs samplers mixing times, as shown in the Corollary \ref{mixing_gap}. 
Given a $\pi$-invariant kernel $P$ with $\pi \in \mathcal{P}(\sX)$ we define its spectral gap as
 \[
\begin{aligned}
\text{Gap}(P) = \inf_{f\,:\,\pi(f^2) < \infty,\, \text{Var}_\pi(f) > 0} \, \left\{\frac{\int_{\sX^2}\left[f(\y)-f(\x) \right]^2\pi(\d \x)P(\x, \d \y)}{2\text{Var}_\pi(f)}
\right\},
\end{aligned}
\]
where $f \, : \,  \sX \, \to \, \mathbb{R}$ are measurable functions, 
$\pi(f) = \int_{\sX}f(\x) \pi(\d \x)$ and $\text{Var}_\pi(f) = \int_{\sX} \left[f(\x)-\pi(f) \right]^2\pi(\d \x)$.
We refer to \cite{RR15} and the proof of Corollary \ref{mixing_gap} for discussion on why spectral gaps, which are commonly used for $\pi$-reversible chains, can be used to analyse two-block Gibbs samplers, which are technically not reversible. 
We also note that Corollary \ref{mixing_gap} is only one possible approach to bound $\tilde{t}_{mix}(\epsilon, M)$ and that any quantitative bound on the latter can be combined with Corollary \ref{mixingCorollary} to deduce limiting statements on $(t^{(n)}_{mix}(\epsilon, M))_{n \geq 1}$.

\begin{corollary}\label{mixing_gap}
Let $K=2$, assumption $(A1)$ be satisfied and Gap$(\tilde{P}) >0$. Then, for every $(M, \epsilon) \in [1, \infty) \times (0,1)$ 
 it holds
\begin{align*}
Q^{(n)}\left(t^{(n)}_{mix}({\epsilon}, M)\leq1+\frac{\log(M/2)-\log(\epsilon)}{-\log (1-\text{Gap}(\tilde{P}))}\right) \to 1
\quad\hbox{ as }n\to\infty\,.
\end{align*}
\end{corollary}

Given the result of Corollary \ref{mixing_gap}, it is natural to ask whether the convergence proved in Theorem \ref{thm:mixing_limit} could be rephrased in terms of spectral gaps, i.e. $\text{Gap}(P_n) \to \text{Gap}(\tilde{P})$. However, once again, convergence in total variation is too weak for this purpose: indeed it is not difficult to find examples where $(A1)$ holds and the associated Gibbs sampler spectral gaps do not converge, even under the stronger condition requiring $\| \tilde{P}_n(\x, \cdot) - \tilde{P}(\x, \cdot) \|_{TV}\to 0$ for any $\x\in\sX$, see Example $A.4$ in Appendix $A$. Controlling directly the spectral gaps would require extremely stringent conditions on the convergence of $\tilde\pi_n$ to $\tilde \pi$ that are rarely satisfied (e.g.\ uniform convergence of the associated densities on the log-scale, i.e.\ $\sup_{\x\in\sX}|\log\tilde{\pi}_n(\x) - \log\tilde{\pi}(\x)|\to 0$). An alternative approach to the direct warm-start mixing time analysis that we perform here, would be to consider asymptotic behaviours of \emph{approximate} spectral measures, such as approximate spectral gaps, see e.g.\ \cite{A21,T22}.

\section{Illustrative example: fixed-dimensional parametric models}\label{sec:parametric_model}
We first consider the fixed-dimensional case.
While this is not our main interest or motivating application, it allows to show the type of results we will derive and also introduce notation about classical Bayesian asymptotic results that we will use. 
In this setting $\pi_n(\d\psi) = p(\d\psi \mid Y^{(n)})$ is the posterior distribution of the Bayesian model defined as
\begin{equation}\label{parametric_model}
\begin{aligned}
Y_i \mid\psi  \overset{iid}{\sim} f(Y \mid \psi), \quad \psi \sim p_0(\psi),
\end{aligned}
\end{equation}
where $\psi = (\psi_1, \dots, \psi_K)$, with $\sX \in \R^K$, and $Y^{(n)} = (Y_1, \dots, Y_n)$, with $Y_i \in \sY$, $i = 1, \dots, n$, so that $\sY^{{(n)}} = \sY^n$. Moreover, if $Y_i \simiid Q$ for some $Q \in \mathcal{P}(\sY)$, we denote with $Q^{(n)}$ and $Q^{(\infty)}$ the associated product measures. We study the mixing times of the Gibbs sampler that updates one coordinate of $\psi$ at the time 
 as $n$ grows. In order to apply the results of Theorem \ref{thm:mixing_limit} we need a suitable transformation of $\psi$, that is given by the celebrated Bernstein-von Mises Theorem, which we now recall. 
The version we provide here, which makes stronger than needed assumptions, can be obtained combining Theorem $10.1$ in \cite{V00}, with other remarks in Chapter 10 therein, incuding Lemmas 10.4 and 10.6.
\begin{theorem}[Bernstein-von Mises]\label{BvM}
Consider model \eqref{parametric_model} and let the map $\psi \to f(\cdot \mid \psi)$ be one-to-one. Let the map $\psi \to \sqrt{f(y \mid \psi)}$ be continously differentiable for every $y\in\sY$, with non-singular and continuous Fisher Information $\mathcal{I}(\psi)$. Let the prior measure be absolutely continuous in a neighborhood of $\psi^* \in \sX$ with a continuous positive density at $\psi^*$. Finally, let $\Psi$ be a compact neighborhood of $\psi^*$ for which there exists a sequence of tests $u_n$ such that
\begin{equation}\label{tests}
\begin{aligned}
&\int_{\sY^{{(n)}}} u_n(y_1, \dots, y_n) \, \prod_{i = 1}^nf(\d y_i \mid \psi^*) \to 0, &\\
& \sup_{\psi \not\in \Psi}\int_{\sY^{{(n)}}} \left[1-u_n(y_1, \dots, y_n)\right] \, \prod_{i = 1}^nf(\d y_i \mid \psi) \to 0,&\hbox{as }n \to \infty\,.
\end{aligned}
\end{equation}
Then, if $Y_i \simiid Q_{\psi^*}$ for $i=1,2,\dots$ with $Q_{\psi^*}$ admitting density $f(y \mid \psi^*)$, it holds
\begin{align*}
\lTV \L\left(\d \tilde{\psi} \mid Y^{(n)}\right) - 
N\left(\mathcal{I}^{-1}(\psi^*)\Delta_{n,\psi^*}, \mathcal{I}^{-1}(\psi^*)\right) \rTV &\to 0,
&\hbox{as }n \to \infty
\end{align*}
in $Q_{\psi^*}^{(\infty)}$-probability, where $\tilde{\psi} = \sqrt{n}(\psi-\psi^*)$ and $\Delta_{n, \psi^*} = \frac{1}{\sqrt{n}}\sum_{i = 1}^n\nabla \log f(Y_i \mid \psi) \bigr\rvert_{\psi = \psi^*}$.
\end{theorem}
\begin{remark}
Differentiability of $\sqrt{f(y \mid \psi)}$ and continuity of $\Fisher(\psi)$ imply that the model is \emph{differentiable in quadratic mean}, which allows to prove local asymptotic normality of the log-likelihood function. See Theorem $7.2$ and Lemma $7.6$ in \cite{V00}.
\end{remark}
\begin{remark}
A \emph{test} is a measurable function $u \, : \, \sY^{{(n)}}\, \to \, [0,1]$. The integrals in \eqref{tests} represent probabilities of errors of first and second kind, respectively, when the null hypothesis $H_0 \, : \, \psi = \psi^*$ is rejected with probability $u(y_1, \dots, y_n)$.\end{remark}
   Loosely speaking, Theorem \ref{BvM} implies that, if the model is well-specified and $\psi$ is suitably rescaled, the posterior distribution converges to a multivariate normal. The result holds under some identifiability requirements: first of all, the true parameter $\psi^*$ must belong to the support of the prior; moreover, we must be able to separate $\psi^*$ from the complements of its neighborhood, given infinitely many data. Such assumption is mild in most interesting cases and it is implied by the existence of uniformly consistent estimators for $\psi$ (that is guaranteed if the support of $p_0$ is compact). See Chapter $10$ in \cite{V00} for more details. 
Finally, the Fisher Information matrix must be non singular. 
\begin{remark}
Notice that Theorem \ref{BvM} requires the model to be (perfectly) well-specified, which rarely happens in practice. However there exist extended versions for the case of misspecified likelihoods \citep{KV12}, where the limiting distribution is still Gaussian with a different covariance matrix. Indeed, we expect the results of this and the following sections to hold in a similar way under misspecification: of course the different limiting distribution will have an impact on the final result, especially in the application of Corollary \ref{mixing_gap}.
\end{remark}
We can now use Theorem \ref{thm:mixing_limit} and Corollary \ref{mixingCorollary} to bound the mixing times of the Gibbs sampler associated to model \eqref{parametric_model} as $n$ diverges.
\begin{proposition}\label{convergence_finite_dim}
Let model \eqref{parametric_model} satisfy the hypotheses of Theorem \ref{BvM} and  let $P_n$ be the Gibbs sampler kernel targeting $\pi_n(\d\psi) = p(\d\psi \mid Y^{(n)})$ by updating one coordinate of $\psi = (\psi_1, \dots, \psi_K)$ at a time. Then, for every $(M, \epsilon) \in [1, \infty) \times (0,1)$ there exists $T\left(\psi^*, \epsilon, M \right) < \infty$ such that
\[
\lim_{n\to\infty}
Q_{\psi^*}^{(n)}\left(t^{(n)}_{mix}(\epsilon, M)\leq T\left(\psi^*, \epsilon, M \right)\right) = 1\,.
\]
\end{proposition}
   Proposition \ref{convergence_finite_dim} shows that, under the conditions of Theorem \ref{BvM} and starting from an $M$-warm distribution, the number of iterations required to get $\epsilon$-close to the posterior distribution does not grow as $n \to \infty$. 
An application to the normal model with unknown mean and precision is given by Corollary $C.7$ in Section $C.10$ of Appendix C.

The main take-away of this Section is that, under relatively mild conditions, the Gibbs sampler behaves well with models of fixed dimensionality and growing number of observations. In the remaining of the paper we consider the more challenging setting of hierarchical models, where the number of parameters grows with the number of observations: in particular we will explore situations in which the number of required iterations remains fixed even with a growing dimensionality of the problem. 

\section{Hierarchical models with exponential family priors and generic likelihood}\label{hierarchical}
We consider a general class of hierarchical models, with data divided in $J$ groups, each having a set of group-specific parameters $\theta_j$. The latter share a common prior with hyper-parameters $\psi$.
Recalling \eqref{eq:one_level_nested_intro}, the model under consideration is
\begin{equation}\label{one_level_nested}
Y_j \mid \theta_j \sim f(\cdot \mid \theta_j)\,,\quad
 \theta_j\mid \psi\overset{\text{iid}}{\sim} p(\cdot \mid \psi)\,,\quad
 \psi \sim p_0(\cdot).
\end{equation}
We assume that the prior for $\theta_j \in \R^{\ell}$ belongs to the exponential family, that is
\begin{equation}\label{exponential_family}
p(\theta \mid \psi) = h(\theta)\text{exp}\left\{ \sum_{s = 1}^S\eta_s(\psi)T_s(\theta)-A(\psi)\right\},
\end{equation}
where $\psi \in \R^D$, $h \, : \, \R^{\ell} \, \to \, \R_+$ is a non-negative function and $\eta_s(\psi)$, $T_s(\theta)$ and $A(\psi)$ are known real-valued functions with domains $\R^D$, $\R^\ell$ and $\R^D$ respectively.
We will always assume the family to be minimal, that is both $(\eta_1(\psi),\dots,\eta_S(\psi))$ and $(T_1(\theta),\dots,T_S(\theta))$ are linearly independent.
On the other hand, we let $f(y \mid \theta)$ be an arbitrary likelihood function with data $y\in\R^m$ and parameters $\theta \in \R^{\ell}$, dominated by a suitable $\sigma$-finite measure (usually Lebesgue or counting one).

Denoting $\bm{\theta} = (\theta_1, \dots, \theta_J)$, $Y_{1:J} = \left(Y_1, \dots, Y_J\right)$ and $\pi_J(\d\bm{\theta}, \d\psi)=\L\left(\d\bm{\theta}, \d\psi \mid Y_{1:J}\right)$, we are interested in studying the two-block Gibbs sampler targeting $\pi_J(\d\bm{\theta}, \d\psi)$, i.e.\ the kernel defined as
\begin{equation}\label{two_blocks_gibbs_nested}
P_J\left(\left(\bm{\theta}^{(t-1)}, \psi^{(t-1)}\right), \left(\d \bm{\theta}^{(t)}, \d \psi^{(t)}\right)  \right) = \pi_J \left(\d \bm{\theta}^{(t)} \mid \psi^{(t-1)} \right)\pi_J\left(\d \psi^{(t)}\mid \bm{\theta}^{(t)} \right)\,.
\end{equation}
Throughout Section \ref{hierarchical} we denote by $\left(\bm{\theta}^{(t)}, \psi^{(t)} \right)_{t \geq 1}$ the Markov chain with operator $P_J$, and by $t^{(J)}_{mix}$ the associated mixing times, i.e.
\[
t^{(J)}_{mix}(\epsilon,\mu)=\inf\left\{t\geq 0\,:\, \lTV \mu P_J^t-\pi_J \rTV<\epsilon \right\}, \quad t^{(J)}_{mix}(\epsilon,M) = \sup_{\mu \in \sN \left(\pi_J, M \right)} t^{(J)}_{mix}(\epsilon, \mu).
\]

\subsection{Dimensionality reduction}\label{sec:dim_red}
In order to apply Corollary \ref{mixingCorollary} to characterize $t^{(J)}_{mix}$, we would need to study the asymptotic distribution of $\pi_J$ as $J \to \infty$. The latter is a distribution over $\ell J+D$ parameters, therefore its dimensionality grows with the size of the data. However, the next lemma shows that the convergence properties of $P_J$ can be described through a Gibbs sampler on an intractable, but fixed-dimensional target, namely $\hat{\pi}_J(\d\bT, \d\psi)=\L\left(\d\bT, \d\psi \mid Y_{1:J}\right)$ where $\bm{T} = \left(\sum_{j = 1}^JT_1(\theta_j), \dots, \sum_{j = 1}^JT_S(\theta_j) \right)$, with $T_s$ as in \eqref{exponential_family}.
Let $\left(\bT^{(t)}, \psi^{(t)} \right)_{t \geq 1}=\left(\bT(\bm{\theta}^{(t)}), \psi^{(t)} \right)_{t \geq 1}$ be the stochastic process obtained as a time-wise mapping of $\left(\bm{\theta}^{(t)}, \psi^{(t)} \right)_{t \geq 1}$ under $(\bm{\theta}, \psi)\mapsto (\bT(\bm{\theta}), \psi)$. The latter process contains all the information characterising the convergence of $\left(\bm{\theta}^{(t)}, \psi^{(t)} \right)_{t \geq 1}$, in the sense made precise in the following lemma. Below we denote by $\hat{P}_J$ the kernel of the two-block Gibbs sampler targeting $\hat{\pi}_J$.
\begin{lemma}\label{sufficient_lemma}
For each $J\geq 1$, the process $\left(\bT^{(t)}, \psi^{(t)} \right)_{t \geq 1}$ is a Markov chain, its transition kernel coincides with $\hat{P}_J$, and its mixing times $\hat{t}^{(J)}_{mix}$ satisfy
\begin{align*}
\sup_{\mu \in \sN \left(\pi_J, M \right)} t^{(J)}_{mix}(\epsilon, \mu)
&= \sup_{\nu \in \sN \left(\hat{\pi}_J, M \right)} \hat{t}^{(J)}_{mix}(\epsilon, \nu)
&(M, \epsilon) \in [1, \infty) \times (0,1)\,.
\end{align*}
\end{lemma}
\begin{remark}[Prior and likelihood assumptions]\label{rmk:lik_non_exp}
In order to reduce the dimensionality of the Markov chain under consideration, Lemma \ref{sufficient_lemma} requires the existence of sufficient statistics only for the prior density of the group-specific parameters. It does not require any condition on the likelihood function in model \eqref{one_level_nested}.
In particular, we have $\L\left(\d\psi \mid \bm{\theta}, Y_{1:J} \right) = \L\left(\d\psi \mid \bm{T}(\bm{\theta}), Y_{1:J} \right) $, while $\L\left(\d Y_{1:J} \mid \bm{\theta}, \psi  \right)\neq \L\left(\d Y_{1:J} \mid \bm{T}(\bm{\theta}), \psi \right)$ in general.
\end{remark}

 Lemma \ref{sufficient_lemma} allows to focus the analysis on the convergence speed of $\left(\bT^{(t)}, \psi^{(t)} \right)_{t \geq 1}$, 
 which is a chain whose dimensionality does not grow with the size of the data. 
 Note that its target distribution $\hat{\pi}_J$ is usually not available in closed form, and the corresponding two-block Gibbs sampler $\hat{P}_J$ cannot be implemented directly (unless by implementing the original algorithm $P_J$ and keeping track of $\left(\bT^{(t)}, \psi^{(t)} \right)_{t \geq 1}$). In this sense the latter chain is useful for convergence analysis purposes but less so as an algorithmic shortcut. 

The result of Lemma \ref{sufficient_lemma} is a peculiar property of the Gibbs sampler, which naturally ignores ancillary information about $\psi$ in $\bm{\theta}$. 
Indeed, the proof of Lemma \ref{sufficient_lemma} crucially relies on the fact that the algorithm is performing exact conditional updates and analogous reductions do not occur for most other MCMC schemes (e.g.\ Metropolis-Hastings based schemes, including gradient-based ones). 

This dimensionality reduction trick can be applied beyond hierarchical models and has already been employed in similar settings, mainly with the idea of obtaining suitable drift functions \citep{R95}: for example, in \cite{Q19} it is used to derive the convergence complexity of a data augmentation algorithm for the Bayesian probit regression model, while in \cite{RS15} a similar tecnique allows to study the geometric convergence rate of a Gibbs sampler for high dimensional Bayesian linear regression.

\subsection{Regularity assumptions and main result}\label{section:main_result}
In order to apply the techniques of Theorem \ref{thm:mixing_limit}, we need to provide an asymptotic characterization of $\hat{\pi}_J$. 
To do so we require the technical assumptions listed in this section. 
The assumptions will be verified in specific examples in Section \ref{normal_case} and \ref{binary_data}.

The approach we use to analyse $\hat{\pi}_J$, which is discussed after Theorem \ref{theorem_one_level_nested}, is based on the decomposition $\hat{\pi}_J(\d\bT, \d \psi)=\hat{\pi}_J(\d\psi)\hat{\pi}_J(\d\bT\mid\psi)$.
The first set of assumptions contains standard regularity and identifiability conditions to study the marginal distribution $\hat{\pi}_J(\d\psi)$. In particular, assumptions $(B1)-(B3)$ allow the application of Theorem \ref{BvM} to the posterior distribution of $\psi$. Their applicability has been discussed in Section \ref{sec:parametric_model}.
We denote the marginal likelihood of the model, obtained by integrating out the group specific parameter $\theta$, as
\begin{equation}\label{likelihood_data}
g(y \mid \psi) = \int_{\R^\ell}f(y \mid \theta)p(\theta \mid \psi) \, \d \theta\,,
\end{equation}
and its Fisher Information matrix as
\[
\left[\Fisher(\psi)\right]_{d, d'} = E \biggl[\left\{\partial_ {\psi_d} \log g(Y \mid \psi)\right\} \, \left\{\partial_{\psi_{d'}} \log g(Y \mid \psi)\right\} \biggr], \quad d, d' = 1, \dots, D.
\]
We will assume the following:
\begin{enumerate}
\item[$(B1)$] There exists $\psi^* \in \R^D$ such that
$
Y_j \simiid Q_{\psi^*}$ for $j = 1,2, \dots$,
where $Q_{\psi^*}$ admits density $g(y \mid \psi^*)$. Moreover the map $\psi \to g(\cdot \mid \psi)$ is one-to-one and the map $\psi \to \sqrt{g(x \mid \psi)}$ is continuously differentiable for every $x$.  Finally, the prior density $p_0$ is continuous and strictly positive in a neighborhood of $\psi^*$.
\item[$(B2)$] There exist a compact neighborhood $\Psi$ of $\psi^*$ and a sequence of tests $u_j \, : \, \R^{mJ} \, \to \, [0,1]$ such that
$\int_{\R^{mJ}}u_j\left(y_1, \dots, y_J\right) \prod_{j = 1}^Jg(y_j \mid \psi^*) \, \d y_{1:J} \to 0$
and\\
$
\sup_{\psi \not\in \Psi} \, \int_{\R^{mJ}}\left[ 1-u_j\left(y_1, \dots, y_J\right)\right] \prod_{j = 1}^Jg(y_j \mid \psi) \, \d y_{1:J} \to 0
$, as $J \to \infty$. 
\item[$(B3)$] The Fisher Information matrix $\Fisher(\psi)$ is non-singular and continuous w.r.t.\ $\psi$.
\end{enumerate}

The second set of regularity assumptions (B4)-(B6) are described and discussed in Appendix B. They deal with smoothness and regularity of the conditional distribution $\hat{\pi}_J(\bT|\psi)$ and they allow to derive a suitable conditional Central Limit Theorem in total variation for $\hat{\pi}_J(\bT|\psi)$ as $J\to\infty$.

We can now state the main result of this section. Below we denote the product measures associated to $Q_{\psi^*}$ by $Q^{(J)}_{\psi^*}$ and $Q^{(\infty)}_{\psi^*}$.
\begin{theorem}\label{theorem_one_level_nested}
Consider model \eqref{one_level_nested} and the Gibbs sampler defined as in \eqref{two_blocks_gibbs_nested}, with mixing times $t^{(J)}_{mix}(\epsilon,M)$. Then, under assumptions $(B1)$-$(B6)$, for every $(M, \epsilon) \in [1, \infty) \times (0,1)$ there exists $T\left(\psi^*, \epsilon, M \right) < \infty$ such that
\[
Q_{\psi^*}^{(J)}\left(t^{(J)}_{mix}(\epsilon, M)\leq T\left(\psi^*, \epsilon, M \right)\right) \to 1,
\]
as $J \to \infty$. 
It follows that $t^{(J)}_{mix}(\epsilon, M)=\sO_P(1)$ as $J\to\infty$.
\end{theorem}

\begin{remark}\label{computational_cost}
Theorem \ref{theorem_one_level_nested} provides a formal proof of the linear in $J$ cost for Gibbs samplers on hierarchical models. Indeed, it proves that a bounded (in $J$) number of iterations suffices to get a good mixing: assuming that the cost of a single iteration scales linearly with $J$, which is typically the case, this implies an overall computational cost of order $\sO_P(J)$. Note that a single evaluation of the likelihood of $(\bm{\theta}, \psi)$, or the associated gradients, which is required at every iteration of usual gradient-based methods, yields a cost of the same order.
\end{remark}
\begin{remark}\label{connection_localConsistency}
The conclusions of Theorem \ref{theorem_one_level_nested} are similar in spirit to those of \cite[Thm.1]{K14}. Also there the convergence of Gibbs Samplers targeting two-level hierarchical models is studied using tools from Bayesian asymptotics. The results therein, which deal with convergence of ergodic averages 
when the algorithm is started in stationarity, are quite different from ours, which deal with mixing times. Nonetheless they also support the idea that Gibbs samplers targeting two-level hierarchical models can exhibit $\sO_P(1)$ convergence as $J\to\infty$.
\end{remark}

\subsection{Posterior convergence lemmas for Theorem \ref{theorem_one_level_nested}}\label{section_sketch_main_result}
The proof of Theorem \ref{theorem_one_level_nested} can be found in Appendix C.
It relies on Lemma \ref{sufficient_lemma}, which allows to focus on the two-blocks Gibbs sampler targeting $\hat{\pi}_J(\d\bT,\d \psi)$, and on Lemmas \ref{asymptotic_distribution_psi} and \ref{asymptotic_distribution_T} below. 
These two lemmas imply that $\hat{\pi}_J(\d\bT, \d\psi)$ satisfies assumption $(A1)$ as $J\to\infty$ and that the associated limiting kernel is ergodic, thus allowing to apply Corollary \ref{mixingCorollary}. 

In order to prove $(A1)$ for $\hat{\pi}_J(\d\bT, \d\psi)=\L\left(\d\bm{T}, \d\psi \mid Y_{1:J}\right)$, we need to identify a suitable transformation of $\left(\bm{T}, \psi\right)$, denoted by $\left( \tilde{\bm{T}}, \tilde{\psi}\right)$. 
We define a one-to-one transformation of $\psi$ as
\begin{equation}\label{tilde_psi}
\tilde{\psi} = \sqrt{J}\left(\psi-\psi^* \right)-\Delta_J, \quad \Delta_J = \frac{1}{\sqrt{J}}\sum_{j = 1}^J\Fisher^{-1}(\psi^*) \nabla \log g(Y_j \mid \psi^*).
\end{equation}
The asymptotic distribution of $\tilde{\psi}$ follows directly through Theorem \ref{BvM}, as summarized in the next lemma.
\begin{lemma}\label{asymptotic_distribution_psi}
Define $\tilde{\psi}$ as in \eqref{tilde_psi}. Under assumptions $(B1)-(B3)$ it holds
\[
\lTV \L(\d\tilde{\psi} \mid Y_{1:J})-N\left(\bm{0}, \Fisher^{-1}(\psi^*) \right) \rTV \to 0,
\]
as $J \to \infty$, in $Q_{\psi^*}^{(\infty)}$-probability.
\end{lemma}
 Let $M^{(1)}(\psi \mid y) = \left(M_1^{(1)}(\psi \mid y), \dots, M_S^{(1)}(\psi \mid y) \right) \in \R^S$ with $M^{(1)}_{s}(\psi \mid y) = E \left[T_s(\theta_j) \mid Y_j = y, \psi \right]$ and
\begin{equation}\label{def_C_V}
\left[C(\psi)\right]_{s,d} = E_{Y_j} \left[\partial_{\psi_d}M_s^{(1)}\left(\psi \mid Y_j \right) \right], \quad \left[V(\psi) \right]_{s,s'} = E_{Y_j} \left[\text{Cov}\left(T_s(\theta_j), T_{s'}(\theta_j) \mid Y_j, \psi \right) \right],
\end{equation}
with $s, s' = 1, \dots S$ and $d = 1, \dots, D$. We use the notation $E_{Y_j}[\cdot]$ for expectations with respect to the law of $Y_j$ as defined in $(B1)$. Then we define a one-to-one transformation of $\bm{T}$ as
\begin{equation}\label{tilde_T}
\tilde{\bm{T}} = \frac{1}{\sqrt{J}}\sum_{j = 1}^J\left[T(\theta_j)-M^{(1)}\left(\psi^* \mid Y_j \right) \right]-C(\psi^*) \Delta_J,
\end{equation}
with $C(\psi^*)$ defined in \eqref{def_C_V}. The next lemma proves the required asymptotic normality of $\tilde{\bm{T}}$, conditional to $\tilde{\psi}$ .
\begin{lemma}\label{asymptotic_distribution_T}
Let $\tilde{\bm{T}}$ be as in \eqref{tilde_T}. Under assumptions $(B1)$-$(B6)$ for every $\tilde{\psi}$ it holds
\[
\lTV \L(\d\tilde{\bm{T}} \mid Y_{1:J}, \tilde{\psi})-N\left(C(\psi^*)\tilde{\psi}, V(\psi^*) \right) \rTV \to 0,
\]
as $J \to \infty$, for $Q_{\psi^*}^{(\infty)}$-almost every $(Y_1,Y_2,\dots)$.
\end{lemma}
Lemma $C.18$ in Section $C.14$ of Appendix C combines Lemmas \ref{asymptotic_distribution_psi} and \ref{asymptotic_distribution_T} to prove that $\L(\d\tilde{\bm{T}}, \tilde{\psi} \mid Y_{1:J})$ converges in total variation to a multivariate  Gaussian vector with non singular covariance matrix, which allows to apply  Corollary \ref{mixingCorollary} as desired.
\begin{remark}
The definition of $\tilde{\bm{T}}$ and Lemma \ref{asymptotic_distribution_T} are an important part of the proof of Theorem \ref{theorem_one_level_nested}. Lemma \ref{asymptotic_distribution_T} relies on the fact that, conditional to $\tilde{\psi}$ and $Y_{1:J}$,  $\bm{T}$ is a sum of independent (but not identically distributed) terms. The proof of convergence in total variation requires more than the usual tools from Lindeberg-Feller Central Limit Theorem, as discussed in Appendix B after assumptions $(B5)$ and $(B6)$.
\end{remark}
\subsection{Analysis of the limiting chain}
 As a byproduct of the proof of Theorem \ref{theorem_one_level_nested}, it is possible to characterize the limiting distribution of the rescaled vector $\left(\tilde{\bm{T}}, \tilde{\psi} \right)$, as the next proposition shows.
\begin{proposition}\label{limiting_sigma}
Consider the same assumptions of Theorem \ref{theorem_one_level_nested}. Then
\[
\lTV \L(\d\tilde{\bm{T}}, \d\tilde{\psi} \mid Y_{1:J})-N\left(\bm{0}, \Sigma \right) \rTV \to 0,
\]
as $J \to \infty$, in $Q_{\psi^*}^{(\infty)}$-probability, where
\begin{equation}\label{eq:Sigma_def}
\Sigma =
\begin{bmatrix}
V(\psi^*)+C(\psi^*)\mathcal{I}^{-1}(\psi^*)C^\top(\psi^*) & \quad C(\psi^*)\mathcal{I}^{-1}(\psi^*)\\
&\\
\mathcal{I}^{-1}(\psi^*)C^\top(\psi^*) & \mathcal{I}^{-1}(\psi^*)
\end{bmatrix}
\end{equation}
with $C(\psi^*)$ and $V(\psi^*)$ defined in \eqref{def_C_V}.
\end{proposition}
The expression for the limiting covariance in \eqref{eq:Sigma_def} can be used to investigate the convergence properties of the limiting Gibbs sampler, since the spectral gap is explicitly computable from that. 
  We can then  apply Corollary \ref{mixing_gap} and obtain the following result.
\begin{corollary}\label{spectral_radius}
Under the assumptions of Theorem \ref{theorem_one_level_nested}, for every $(M, \epsilon) \in [1, \infty) \times (0,1)$, we have $
Q_{\psi^*}^{(J)}\left(t^{(J)}_{mix}(\epsilon, M)\leq T\left(\psi^*, \epsilon, M \right)\right) \to 1
$ 
 as $J \to \infty$, with
\[
\begin{aligned}
&T\left(\psi^*, \epsilon, M \right)= 1 + \frac{\log(M/2)-\log(\epsilon)}{-\log \left(1-\gamma(\psi^*) \right)},\\
&\gamma(\psi^*)= \min \left\{  \frac{1}{1+\lambda_i} \, : \, \lambda_i \text{ eigenvalue of } V^{-1}(\psi^*)C(\psi^*)\mathcal{I}^{-1}(\psi^*)C^\top(\psi^*) \right\}.
\end{aligned}
\]
\end{corollary}
 Thus, once the limiting distribution is obtained, an upper bound on the mixing times can be derived by computing the eigenvalues of a $S \times S$ matrix. As an application, the next corollary provides the value of $\gamma$ when $S = D = 1$.
\begin{corollary}\label{spectral_single}
Consider the same setting of Corollary \ref{spectral_radius}, with $S = D = 1$. Then we have
\begin{equation}\label{gamma_single}
\gamma(\psi^*) = \frac{\text{Var}_{Y_j}\left(E \left[T(\theta_j) \mid \psi^*,Y_j \right] \right)}{\text{Var}\left(T(\theta_j) \mid \psi^* \right)}.
\end{equation}
\end{corollary}
By the law of total variance, we have that $\gamma(\psi^*) \to 0$
if and only if
\begin{equation*}
\frac{\text{Var}_{Y_j}\left(E \left[T(\theta_j) \mid \psi^*,Y_j \right] \right)
}{
E \left[\text{Var}_{Y_j}\left(T(\theta_j) \mid \psi^*,Y_j \right)\right]} \to 0,
\end{equation*}
i.e., loosely speaking, when the data $Y_j$ yield little information about $T(\theta_j)$ and therefore about $\psi$. This phenomenon arises since model \eqref{one_level_nested} is an example of centered parametrization, see e.g.\  \cite{G95,PRS03, PRS07}. The formula in \eqref{gamma_single} resembles the definition of the so-called Bayesian fraction of missing information \citep{L94}, with the notable difference of not involving an infimum over a set of test functions.

\section{Examples}\label{examples}
In this section various examples, which differ by the choice of likelihoods and priors, are discussed.

\subsection{Hierarchical normal model}\label{normal_case}
Consider the following hierarchical specification:
\begin{equation}\label{normal_model}
\begin{aligned}
Y_{j,i} \mid \theta_j
&\sim N\left(\theta_j,\tau_0^{-1} \right), 
& i = 1, \dots, m,\; j = 1, \dots, J\\
 \theta_j \mid\mu, \tau_1 &\simiid N( \mu, \tau_1^{-1}) \, ,
&j=1,\dots,J \\
(\mu,  \tau_1) & \sim p_0(\cdot)\,.
\end{aligned}
\end{equation}
where $(\mu, \tau_1)$ are unknown hyperparameters. 
In this section we assume $\tau_0$ to be fixed and known, see Section \ref{sec:lik_pars} for the case with $\tau_0$ unknown. 
The prior $p_0$ can be any distribution satisfying the assumptions stated in Proposition \ref{prop_normal_model} below.
It can be seen that \eqref{normal_model} is a particular case of model \eqref{one_level_nested}, with
$f(Y_j \mid \theta_j) = \prod_{i = 1}^mN(Y_{j,i} \mid \theta_j, \tau_0^{-1})$, 
$ p(\cdot \mid \mu, \tau_1) = N(\mu, \tau_1^{-1})$.
The marginal likelihood of $Y_j$ conditional to $(\mu, \tau_1, \tau_0)$ is given by
\begin{align}\label{marginal_normal}
g(y \mid \mu, \tau_1, \tau_0) &= N \left(y\mid\mu, \tau_0^{-1}I + \tau_1^{-1}\H \right)&y\in\R^m,
\end{align}
where $I$ is the $m \times m$ identity matrix and $\H$ is the $m \times m$ matrix of ones. 

We consider three Gibbs sampler specifications, which vary depending on which parameters are unknown and treated as random and which blocking rules are used. 
First, when $\tau_1$ is fixed, we define $P_1$ as the transition kernel of the Gibbs sampler that targets 
$\L\left(\d \bm{\theta}, \d\mu\mid Y_{1:J} \right)$ by alternating updates from 
$\L\left(\d\bm{\theta}\mid \mu, Y_{1:J} \right)$ and 
$\L\left(\d\mu\mid \bm{\theta}, Y_{1:J} \right)$.
If instead $\mu$ and $\tau_1$ are unknown, we define $P_2$ and $P_3$ as the transition kernels of the two Gibbs samplers targeting 
$\L\left( \d\bm{\theta}, \d\mu,\d\tau_1\mid Y_{1:J} \right)$ 
by alternating updates from 
$\L\left(\d\bm{\theta}, \d\mu\mid \tau_1, Y_{1:J} \right)$
and 
$\L\left(\d\tau_1\mid \bm{\theta}, \mu,Y_{1:J} \right)$
for $P_2$; and 
$\L\left(\d\bm{\theta}\mid \tau_1, Y_{1:J} \right)$,
$\L\left(\d\mu\mid \bm{\theta}, \tau_1,Y_{1:J} \right)$,
$\L\left(\d\tau_1\mid \bm{\theta}, \mu,Y_{1:J} \right)$
 for $P_3$.  In the following we will show that the asymptotic behaviour of $P_2$ and $P_3$ is essentially the same.

It is possible to prove that $P_1$ falls directly in the setting of Theorem \ref{theorem_one_level_nested}, with $T(\theta_j) = \theta_j$ for $P_1$. Even if $P_2$ and $P_3$ are not exactly particular cases of the general theorem, since different update schemes are considered, it turns out that they can be studied with the same tools introduced in the previous section, with $T(\theta_j) = \left(\theta_j, (\theta_j-\mu^*)^2 \right)$.

The next proposition shows that the settings introduced above lead to well-behaved asymptotic regimes. Here $t^{(J)}_{mix, l}(\epsilon, M)$ denotes the mixing times of the Gibbs sampler defined by $P_l$ with $l \in\{1, 2,3\}$.

\begin{proposition}\label{prop_normal_model}
 Let $Y_j \simiid Q_{\psi^*}$, with $Q_{\psi^*}$ admitting density $g(y \mid \psi^*)$ as in \eqref{marginal_normal}, where $\psi^* = (\mu^*, \tau_1^*, \tau_0^*)$, and consider model \eqref{normal_model} with $\tau_0 = \tau_0^*$.
Consider the Gibbs sampler with operator $P_l$, with $l\in\{1,2,3\}$, and let the prior density $p_0$ be continuous and strictly positive in a neighborhood of $\mu^*$ when $l=1$ and $(\mu^*, \tau_1^*)$ when $l\in\{2,3\}$. Finally, when $l = 1$ let $\tau_1 = \tau_1^*$.
Then for every $(M, \epsilon) \in [1, \infty) \times (0,1)$ there exists $T_l\left(\psi^*, \epsilon, M \right) < \infty$ such that
\begin{align}\label{eq:bound_i}
Q_{\psi^*}^{(J)}\left(t^{(J)}_{mix, l}(\epsilon, M)\leq T_l\left(\psi^*, \epsilon, M \right)\right) &\to 1
&\hbox{ as }J \to \infty,\;
l=1,2,3\,.
\end{align}
\end{proposition}

Under model \eqref{normal_model}, the matrices in Corollary \ref{spectral_radius} can be explicitly computed, leading to the following result.
\begin{corollary}\label{spectral_normal}
Under the same assumptions and notation of Proposition \ref{prop_normal_model}, for every $(M, \epsilon) \in [1, \infty) \times (0,1)$, \eqref{eq:bound_i} holds with
\begin{align*}
T_l\left(\psi^*, \epsilon, M \right) &= 1 + \frac{\log(M/2)-\log(\epsilon)}{-\log \left(1-\gamma_l(\psi^*) \right)},
&l = 1, 2, 3\,,
\end{align*}
where
\begin{equation}\label{eq:gaps_normal}
\gamma_1(\psi^*) = \left(1+\frac{\tau_1^*}{m\tau_0^*}\right)^{-1}
\quad\hbox{and}\quad \gamma_2(\psi^*) = \gamma_3(\psi^*) =\gamma_1(\psi^*) ^2\,.
\end{equation}
\end{corollary}

 The expressions for the asymptotic gaps in \eqref{eq:gaps_normal} are insightful in many ways.
First, $\mu^*$ does not appear in any of the spectral gaps, meaning that the limiting value of the mean parameter seems not to play a role in the asymptotic behaviour of the Gibbs sampler. 
Moreover, the gaps are a function of the ratio $(m\tau_0^*)^{-1}\tau_1^*$, that is the ratio of the prior and likelihood precisions, respectively. In particular the gaps converge to $0$, i.e. the upper bound on the mixing times diverges, if and only if $(m\tau_0^*)^{-1}\tau_1^*\to \infty$, which happens when the prior is increasingly more informative than the data. As discussed after Corollary \ref{spectral_single}, such phenomenon arises since all the three formulations are an example of centered parametrization \citep{G95,PRS03}. 
On the contrary, the gaps converge to $1$, i.e. asymptotically a single iteration suffices, if and only if $(m\tau_0^*)^{-1}\tau_1^*\to 0$. 

When $\tau_1$ is fixed and $p_0(\mu)$ is Gaussian, then 
$\L\left( \d\bm{\theta}, \d\mu\mid Y_{1:J} \right)$ is a multivariate Gaussian and  $P_1$ is amenable to finite-sample analysis. In fact, the expression for $\gamma_1(\psi^*)$ appeared previously in the literature, see e.g.\ \cite{PRS03}. The result in Corollary \ref{spectral_normal} is, however, different since it is asymptotic and it applies also to general priors.

On the contrary, a finite-sample analysis of $P_2$ are $P_3$ is hard even when $p_0(\mu)$ is Gaussian (see e.g. \cite{JH21, Q22, Y17}) and $\gamma_2(\psi^*)$ and $\gamma_3(\psi^*)$ did not appear previously in the literature, to the best of our knowledge.
It is interesting that, regardless of the value of  $(m,\mu^*,\tau^*_1,\tau_0^*)$, including the random precision parameter, when moving from $P_1$ to either $P_2$ or $P_3$, always slows down the sampler (asymptotically), since $\gamma_1(\psi^*)>\gamma_i(\psi^*)$ for $i=2,3$, and that the two blocking rules of $P_2$ and $P_3$ are asymptotically equivalent in terms of mixing times, since $\gamma_2(\psi^*)=\gamma_3(\psi^*)$.

\subsection{Models with binary and categorical data}\label{binary_data}
Let now $f(y \mid \theta)$ be a probability mass function, whose point masses are denoted by $y_0, \dots, y_m$, with $m < \infty$, such that for every $\theta \in \mathbb{R}^K$ we have
\begin{equation}\label{eq:ass_discrete}
\sum_{r = 0}^mf(y_r \mid \theta) = 1, \quad f(y_r \mid \theta) > 0, \quad  r = 0, \dots, m.
\end{equation}
The assumption in \eqref{eq:ass_discrete} is mild and holds for most likelihoods usually employed with categorical data, e.g. multinomial logit and probit. We focus on hierarchical models with normal priors, i.e.\,
\begin{equation}\label{finite_model}
Y_{j} \mid\theta_j \sim f(Y_{j} \mid \theta_j)\,,\quad
\theta_1, \dots, \theta_J \mid\mu, \tau \simiid N(\mu, \tau^{-1})\,,\quad
 (\mu,  \tau) \sim p_0(\cdot)\,.
\end{equation}
For example the case $f(y \mid \theta) = \binom{m}{y}\frac{e^{y\theta}}{(1+e^\theta)^m}$, with $y = 0, \dots, m$, corresponds to the logistic hierarchical model with Gaussian random effects. 
The prior $p_0$ can be any distribution satisfying the assumptions stated in Proposition \ref{prop_logit_model} below.
 We define $P$ as the transition kernel of the Gibbs sampler that targets 
$\L\left( \d\bm{\theta}, \d\mu, \d\tau \mid Y_{1:J} \right)$ by alternating updates from 
$\L\left(\d\bm{\theta}\mid \mu, \tau, Y_{1:J} \right)$ and 
$\L\left(\d\mu, \d\tau \mid \bm{\theta}, Y_{1:J} \right)$. 
This is a particular case of the setting of Theorem \ref{theorem_one_level_nested}, with $\psi = (\mu, \tau)$ and $T(\theta_j) = (\theta_j, \theta_j^2)$. 
Notice that usually $\L\left(\d\bm{\theta}\mid \mu, \tau, Y_{1:J} \right)$ is not known in closed form (with the notable exception of the probit case, see \cite{DD19}), but nonetheless exact sampling is often feasible through adaptive rejection sampling (see e.g. \cite{GW92}) since each $\theta_j$ is one dimensional.
The marginal likelihood is given by
\begin{equation}\label{marginal_finite}
g(y \mid \psi) = \int_{\mathbb{R}}f(y \mid \theta) N\left(\theta \mid \mu, \tau^{-1} \right) \, \d \theta.
\end{equation}
  The next lemma shows that assumptions $(B4)$-$(B6)$ follow directly from \eqref{finite_model}.
\begin{lemma}\label{finite_B4B6}
Consider model \eqref{finite_model} and let $Y_j \simiid Q_{\psi^*}$, with $Q_{\psi^*}$ admitting density $g(y \mid \psi^*)$ as in \eqref{marginal_finite}, with $\psi^* = (\mu^*, \tau^*)$. Then assumptions $(B4)$-$(B6)$ are satisfied.
\end{lemma}
 Thus, in order to apply Theorem \ref{theorem_one_level_nested}, it suffices to prove assumptions $(B2)$ and $(B3)$, i.e. that the parameters $\psi$ are identifiable with non singular Fisher Information matrix. Therefore, as formalized in the next proposition, standard identifiability conditions (which are also necessary to consistently estimate $\psi$) are sufficient to prove boundedness of the mixing times.
\begin{proposition}\label{prop_logit_model}
Consider model \eqref{finite_model} and let $Y_j \simiid Q_{\psi^*}$, with $Q_{\psi^*}$ admitting density $g(y \mid \psi^*)$ as in \eqref{marginal_finite}, where $\psi^* = (\mu^*, \tau^*)$. Consider the Gibbs sampler with operator $P$ and let $p_0$ be continuous and strictly positive in a neighborhood of $\psi^*$. Let the map $\psi \to g(\cdot \mid \psi)$ be one-to-one, with non singular and continuous $\Fisher(\psi)$. Finally, assume tests as in $(B2)$ exist. Then for every $(M, \epsilon) \in [1, \infty) \times (0,1)$ there exists $T\left(\psi^*, \epsilon, M \right) < \infty$ such that
\[
Q_{\psi^*}^{(J)}\left(t^{(J)}_{mix}(\epsilon, M)\leq T\left(\psi^*, \epsilon, M \right)\right) \to 1\qquad\hbox{ as }J \to \infty\,.
\]
\end{proposition}
\begin{remark}\label{rmk:logit_ident_m}
In most cases $m \geq 2$ is required to avoid the pair $(\mu, \tau)$ being not identifiable and the associated Fisher Information matrix being singular. For example Lemma $C.35$ in Section $C.23$ of Appendix C shows that with the logit link $\Fisher(\psi)$ is singular if and only if $m = 1$.
\end{remark}
  As already discussed in the Section \ref{section_introduction}, the results of Proposition \ref{prop_logit_model} are illustrated on simulated data in Figure \ref{fig:binom_simulations}. Since mixing times are very hard to approximate numerically in high-dimensions, we employ the Integrated Autocorrelation Times (IATs) as an empirical measure of convergence time. The IAT associated to a $\pi$-invariant Markov chain $X = \{X^{(t)}\}_{t\geq 1}$ and a test function $f\in L^2 (\pi)$ is defined as
 \begin{equation}\label{IAT}
\textsc{IAT}(f) = 1+2\sum_{t = 2}^\infty\text{Corr}\left(f(X^{(1)}), f(X^{(t)}) \right)\,.
\end{equation}
Loosely speaking, IAT($f$) is the number of MCMC samples that is equivalent to a single independent sample in terms of estimation of $\int f(x)\pi(dx)$, thus the higher IAT the slower the convergence. When dealing with hierarchical models as in \eqref{finite_model}, we compute the maximum IAT over all the parameters (both global and group specific). We estimate the IAT with the ratio of the number of iterations and the effective sample size, as described in \cite{GF15}, with the effective sample size computed with the R package \emph{mcmcse} \citep{F21}. For a review of different methods to estimate the IATs, see \cite{T10}.
In Figure \ref{fig:binom_simulations} we plot the quantiles of the IATs as a function of the number of groups for the Gibbs sampler, implemented using adaptive rejection sampling \citep{GW92} for the exact updates of local parameters with full conditionals $\L\left(\d\theta_j\mid \mu, \tau, Y_{1:J} \right)$. As expected by Proposition \ref{prop_logit_model}, the IATs do not diverge as $J$ increases for both values of $m$ under consideration. Note that variability decreases as $J$ increases and the posterior gets closer to its asymptotic limit.
\begin{corollary}\label{spectral_logit}
Consider the same setting of Proposition \ref{prop_logit_model}. For every $(M, \epsilon) \in [1, \infty) \times (0,1)$ define
\[
T\left(\psi^*, \epsilon, M \right) = 1 + \frac{\log(M/2)-\log(\epsilon)}{-\log \left(1-\gamma(\psi^*) \right)},
\]
for $\gamma(\psi^*) \in (0,1)$ as in Corollary \ref{spectral_radius}. Then
\[
Q_{\psi^*}^{(J)}\left(t^{(J)}_{mix}(\epsilon, M)\leq T\left(\psi^*, \epsilon, M \right)\right) \to 1\qquad\hbox{ as }J \to \infty\,.
\]
\end{corollary}
The study of the limiting spectral properties, i.e.\ of $\gamma(\psi^*)$, can be useful to predict under which scenarios the Gibbs sampler will perform well or not for large $J$.
 We illustrate this by  
 considering model \eqref{finite_model} with logit link and known $\tau$ set to $1$. In this setting, where $\mu$ is the only global parameter, the value 
 of $\gamma(\psi^*)$ 
 can be computed as in \eqref{gamma_single} through simple one-dimensional numerical integration. In Figure \ref{fig:binom_gap} we compare the resulting mixing time upper bound, $T\left(\psi^*, \epsilon, M \right) $,  
  with the numerical estimates of IATs defined in \eqref{IAT}, obtained by running a long MCMC chain with a moderately large value of $J$.
We compare such quantities for different values of the true success probability induced by $\mu^*$, i.e.\ $\int_{\mathbb{R}}f(1 \mid \theta) N\left(\theta \mid \mu^*, 1 \right) \, \d \theta$. 
Both theoretical and empirical measures of convergence highlight that the performances of the Gibbs sampler deteriorate when the problem is not balanced: such conclusion is coherent with the findings in \cite{J19}, that considers an asymptotic regime with increasing imbalancedness. 
  \begin{figure}
\centering
\includegraphics[width=.4\textwidth]{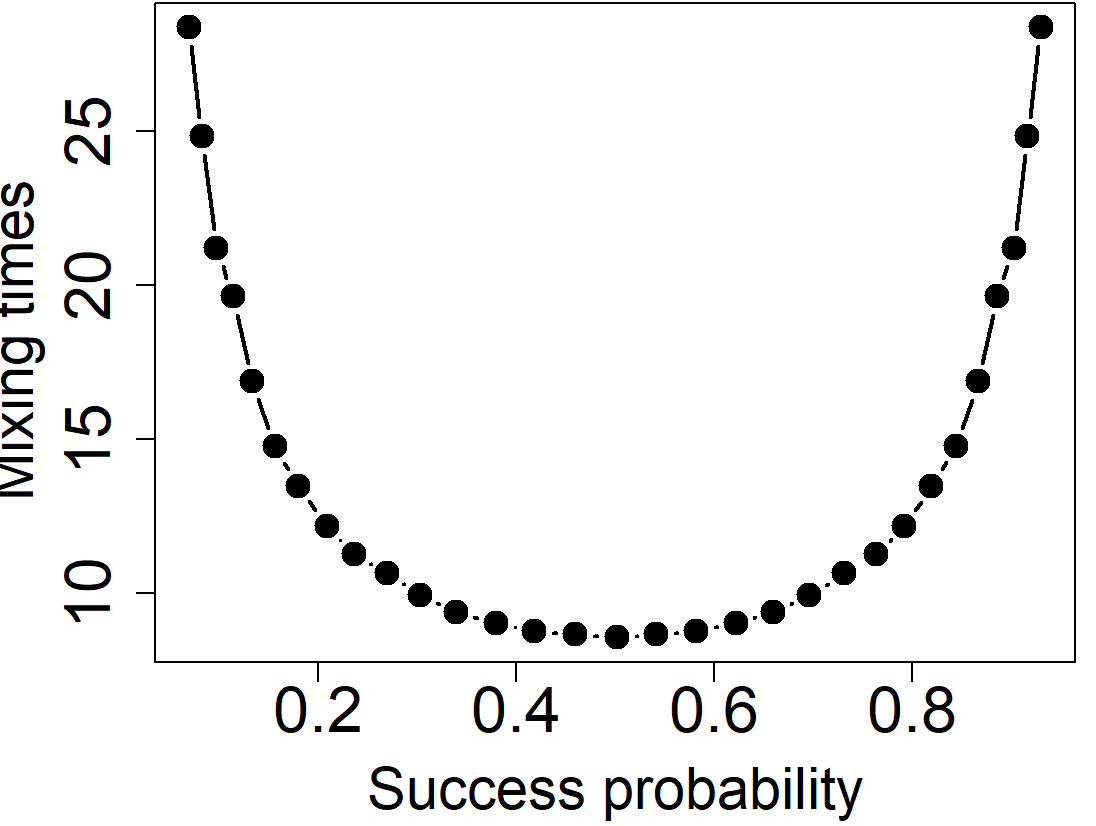} \quad
\includegraphics[width=.4\textwidth]{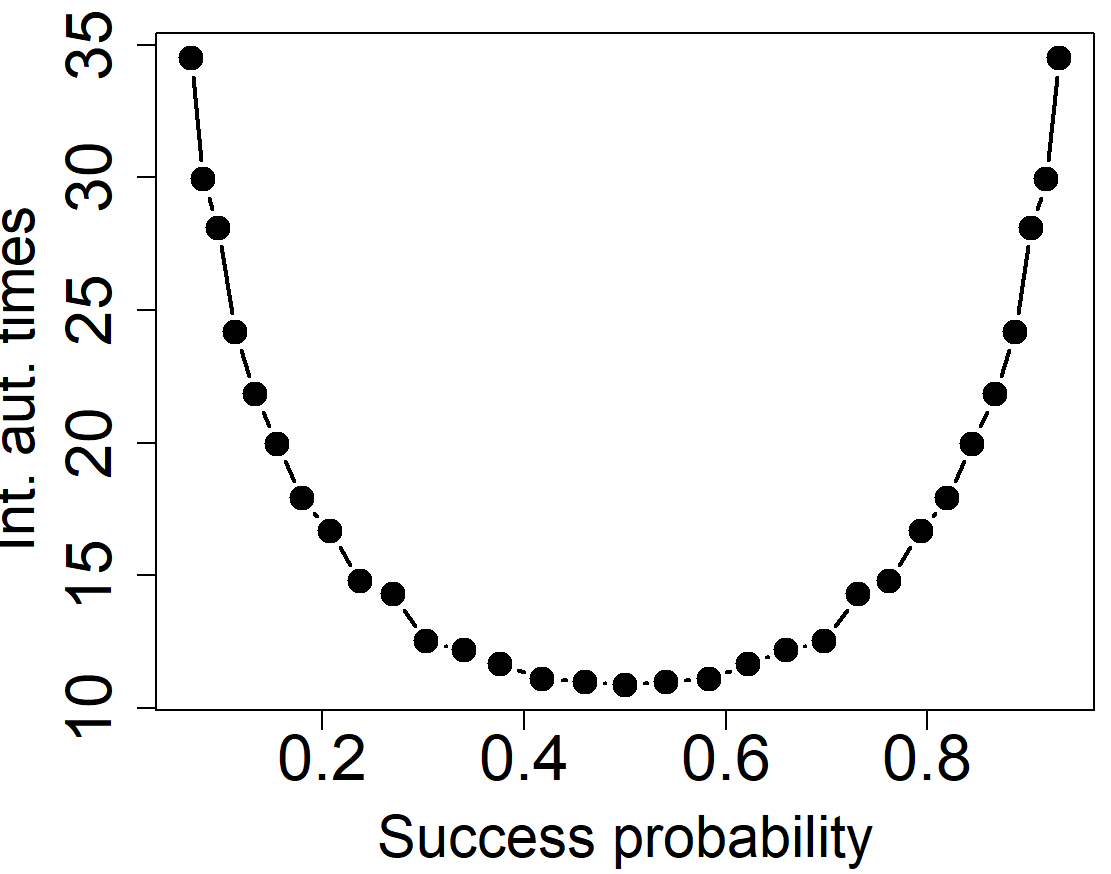}
 \caption{\small{Left: upper bounds on mixing times for model \eqref{finite_model} with $\tau$ known, where $\tau^* = 1$, $\mu^* \in (-3,3)$, $m = 1$, $M = 2$ and $\epsilon = 0.2$. A priori $\mu \sim N\left(0, 10^{3}\right)$. Right: median IATs with $J = 2000$.}} \label{fig:binom_gap}
\end{figure}

\subsection{Different graphical models structure}
In the previous subsections we have studied applications of Theorem \ref{theorem_one_level_nested} for some specification of the hierarchical model in \eqref{one_level_nested}. 
These correspond to the graphical models in the leftmost panel of Figure \ref{fig:graph_models}.
While this structure is very common in Bayesian modeling and it constitutes our main motivating application, the techniques we developed - and in particular the dimensionality reduction and posterior asymptotic approach - can be applied to different classes of models, including other widely used ones.
\begin{figure}
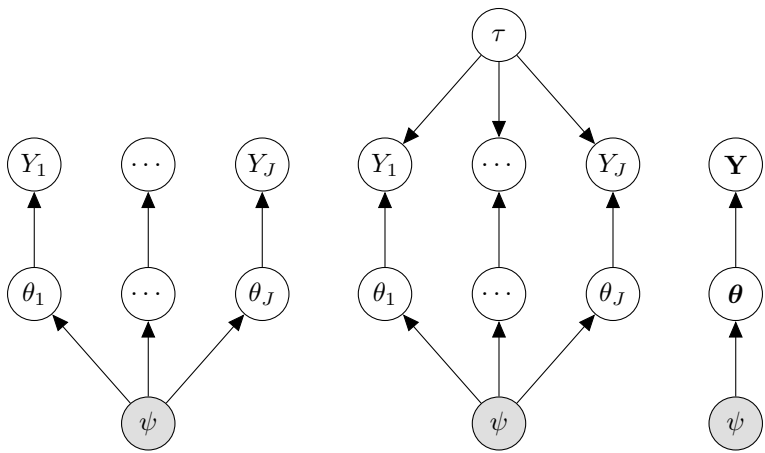

\centering
  \tikz{
 \node[obs] (x) {$\psi$};%
  \node[latent,above=of x,xshift=0cm,fill] (y_generic) {$\dots$}; %
   \node[latent,above=of y_generic,xshift=0cm,fill] (Y_generic) {$\dots$}; %
 \node[latent,above=of x,xshift=-1.5cm,fill] (y) {$\theta_1$}; %
 \node[latent,above=of y,xshift=0cm,fill] (y1) {$Y_1$}; %
 \node[latent,above=of x,xshift=1.5cm] (z) {$\theta_J$}; %
  \node[latent,above=of z,xshift=0cm, fill] (y2) {$Y_J$}; %
 \edge {x} {y,z, y_generic}   
 \edge {y} {y1}
 \edge {z} {y2}
 \edge {y_generic}{Y_generic}
 }
 \quad \quad
 \tikz{
 \node[obs] (x) {$\psi$};%
  \node[latent,above=of x,xshift=0cm,fill] (y_generic) {$\dots$}; %
   \node[latent,above=of y_generic,xshift=0cm,fill] (Y_generic) {$\dots$}; %
 \node[latent,above=of x,xshift=-1.5cm,fill] (y) {$\theta_1$}; %
 \node[latent,above=of y,xshift=0cm,fill] (y1) {$Y_1$}; %
 \node[latent,above=of x,xshift=1.5cm] (z) {$\theta_J$}; %
  \node[latent,above=of z,xshift=0cm, fill] (y2) {$Y_J$}; %
  \node[latent,above=of Y_generic,xshift=0cm, fill] (tau) {$\tau$};

 \edge {x} {y,z, y_generic}   
 \edge {y} {y1}
 \edge {z} {y2}
 \edge {y_generic}{Y_generic}
 \edge {tau} {Y_generic, y1, y2}
 }
 \quad \quad
 \tikz{
 \node[obs] (x) {$\psi$};%
 \node[latent,above=of x,xshift=0cm,fill] (y) {$\bm{\theta}$}; %
 \node[latent,above=of y,xshift=0cm,fill] (y1) {$\textbf{Y}$}; %
 \edge {x} {y}   
 \edge {y} {y1}
 }
 \caption{\small{Graphical models of different hierarchical structures. Left: one level nested model as in Theorem \ref{theorem_one_level_nested}. Center: hyperparameters specifying the likelihood. Right: dependent latent parameters.}}\label{fig:graph_models}
\end{figure}
Here we provide two examples, the first is a relatively direct extension of the model in \eqref{one_level_nested} with the addition of parameters in the likelihood, the second is a more different setting of Gaussian Process regression where the latent parameters are not independent. See respectively the center and rightmost panels in Figure \ref{fig:graph_models} for the resulting graphical models.
More generally, we expect our methodology to be potentially useful to analyse samplers for models that feature 
a fixed set of hyperparameters $\psi$, conditional to which a growing set of parameters or latent variables is tractable enough for posterior sampling.

\subsubsection{Likelihood parameters}\label{sec:lik_pars}
Consider again the hierarchical normal model
\begin{equation}\label{extended_normal_model}
\begin{aligned}
Y_{j,i} \mid \theta_j, \tau_0 \sim N\left(\theta_j,\tau_0^{-1} \right),\quad
 \theta_j \mid\mu, \tau_1 &\simiid N( \mu, \tau_1^{-1})\,,\quad
(\mu,  \tau_1, \tau_0)\sim p_0(\cdot)\,,
\end{aligned}
\end{equation}
with $i=1,\dots,m$ and $j=1,\dots,J$.
The unknown parameters are now given by the triplet $\psi = (\mu, \tau_1, \tau_0)$. 
We denote with $P$ the transition kernel of the Gibbs sampler targeting 
$\L\left( \d\bm{\theta}, \d\mu,\d\tau_1, \d\tau_0\mid Y_{1:J} \right)$ 
by alternating updates from 
$\L\left(\d\bm{\theta}, \d\mu\mid \tau_1, \tau_0, Y_{1:J} \right)$ 
and 
$\L\left(\d\tau_1, \d\tau_0\mid \bm{\theta}, \mu,Y_{1:J} \right)$.
This cannot be seen as a specific case of Theorem \ref{theorem_one_level_nested} with $\psi = (\mu, \tau_1, \tau_0)$, since $\tau_0$ is a parameter of the likelihood $f$ and therefore there is no conditional independence between $Y_j$ and $\psi$, given $\theta_j$. 
However, an approach similar to the one of the previous section can be employed. 
In particular, a result analogous to Lemma \ref{sufficient_lemma} can be derived, with $T(\theta_j) = \left(\left( \theta_j -\bar{Y_j}\right)^2, \left( \theta_j-\mu \right)^2 \right)$ playing the role of the sufficient statistics and $\bar{Y}_j = \frac{1}{m}\sum_{i = 1}^mY_{j,i}$. It is interesting to notice that $T$ in this case depends also on the data $Y_{1:J}$, exactly because the group specific parameters $\bm{\theta}$ do not contain all the information regarding $\psi$. The next proposition shows that also this specification leads to a well-behaved asymptotic regime.
\begin{proposition}\label{prop_extended_normal_model}
Consider model \eqref{extended_normal_model} with $m \geq 2$ and let $Y_j \simiid Q_{\psi^*}$, with $Q_{\psi^*}$ admitting density $g(y \mid \psi^*)$ as in \eqref{marginal_normal}, where $\psi^* = (\mu^*, \tau_1^*, \tau_0^*)$.
Consider the Gibbs sampler with operator $P$ and let the prior density $p_0$ be a continuous and strictly positive in a neighborhood of $\psi^*$.
Then for every $(M, \epsilon) \in [1, \infty) \times (0,1)$ there exists $T\left(\psi^*, \epsilon, M \right) < \infty$ such that
\begin{align}\label{eq:bound_normal}
Q_{\psi^*}^{(J)}\left(t^{(J)}_{mix}(\epsilon, M)\leq T\left(\psi^*, \epsilon, M \right)\right) &\to 1
&\hbox{ as }J \to \infty\,.
\end{align}
\end{proposition}
An explicit value for $T\left(\psi^*, \epsilon, M \right)$ can be found through Corollary \ref{mixing_gap}, as shown in the next corollary.
\begin{corollary}\label{extended_spectral_normal}
Consider the same setting of Proposition \ref{prop_extended_normal_model}. 
Then, for every $(M, \epsilon) \in [1, \infty) \times (0,1)$, \eqref{eq:bound_normal} holds with
\begin{align*}
T\left(\psi^*, \epsilon, M \right) &= 1 + \frac{\log(M/2)-\log(\epsilon)}{-\log \left(1-\gamma(\psi^*) \right)}\,,
\end{align*}
where
\[
\gamma(\psi^*) = \left(1+\frac{1}{m-1}\left(1-\frac{\tau_1^*}{m\tau_0^*} \right)^2+\left(\frac{\tau_1^*}{m\tau_0^*} \right)^2\right)^{-1}.
\]
\end{corollary}
\begin{remark}
The assumption $m \geq 2$ cannot be relaxed: indeed, if a single observation per group is available, the pair $(\tau_1, \tau_0)$ is not identifiable and the Fisher Information matrix is singular. For an empirical illustration of the issues arising in this context, see the top left panel in Figure \ref{fig:normal_simulations} or Section $6.2$ of \cite{RS15}.
\end{remark}
 Unlike the case of Corollary \ref{spectral_normal}, in this setting the limiting gap does not depend on $m$ only through the ratio of prior and likelihood precisions, but also directly on its value. 
 Loosely speaking, a higher value of $m$ allows to better recover the relation between $\tau_0$ and $\tau_1$.

The results of Proposition \ref{prop_extended_normal_model} and Corollary \ref{extended_spectral_normal} are illustrated on simulated data in Figure \ref{fig:normal_simulations}, which depicts the Integrated Autocorrelations Times (IATs) as defined in \eqref{IAT}. When the model is not identifiable, i.e. $m = 1$ (top left panel), the IATs diverge with the number of groups, while with $m = 3$ and $m = 5$ they stabilize as $J$ increases. Differently from the binomial setting of Figure \ref{fig:normal_simulations}, the IATs grow for small values of $J$ before the asymptotic regime kicks in.
\begin{figure}
\centering
\includegraphics[width=\textwidth]{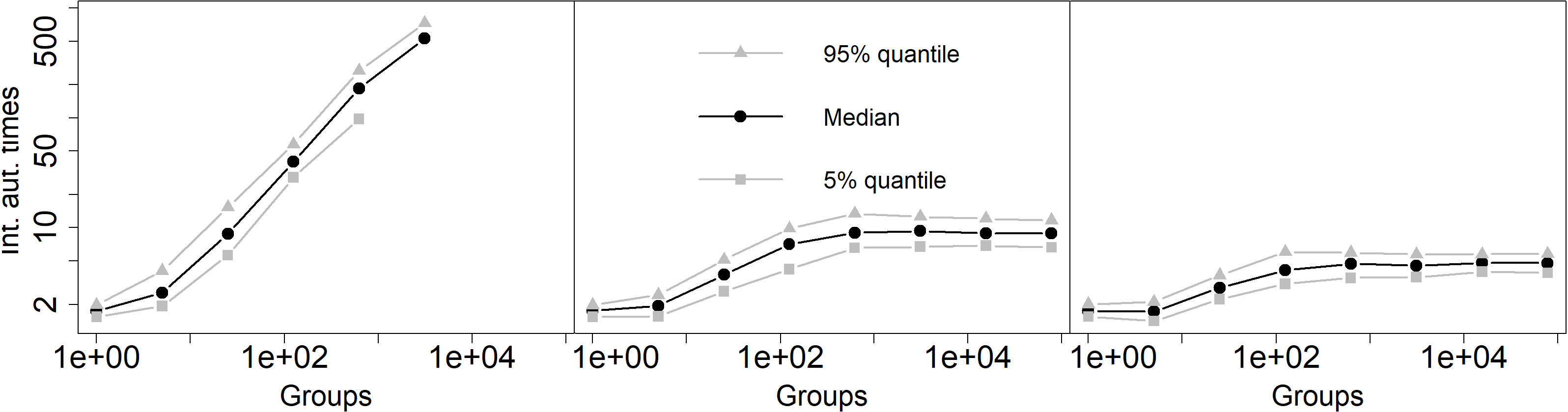} \quad
 \caption{\small{Quantiles of the integrated autocorrelations times (on log-scale) for model \eqref{extended_normal_model} with $\mu^* = 4$, $\tau^*_0 = 1$ and $\tau^*_1 = 3$. A priori $(\tau_0, \tau_1) \overset{\text{i.i.d.}}{\sim} \text{Gamma}(1,1)$ and $p_0(\mu) \propto 1$.  Top left: $m = 1$ (last points not plotted due to numerical instability). Center: $m = 3$. Top right: $m = 5$. }}\label{fig:normal_simulations}
\end{figure}

\subsubsection{Gaussian processes}
We now consider the popular setting where the groups are identified by a continuous covariate (e.g. location) and group specific parameters are modeled through a Gaussian process. It turns out that the main arguments of the paper, namely dimensionality reduction and impact of posterior asymptotic characterization, can be applied also in this context. This section, compared to the previous ones, aims to provide a proof of concept rather than a detailed analysis, 
e.g.\ we directly assume limiting statements on the posterior distributions of interest. 
Nonetheless we find it useful to show how widely our methodology could be applied and illustrate interesting directions of ongoing work.

 Assume to observe $n$ data points $Y(s_i)$ with $i = 1, \dots, n$, at a set of locations $\left(s_1, \dots, s_n\right)$, together with input variables or covariates $x(s_i)\in\R$. 
We consider Gaussian Process regression models of the form 
\begin{equation}\label{eq:gauss_proc_structure}
\begin{aligned}
Y(s_i) \mid \betag & \sim f(\cdot \mid \beta(s_i),x(s_i)), & i = 1, \dots, n\\
 \betag^{(n)}\mid \psi &\sim N(\theta\One, \tau_\beta^{-1}R^{(n)})
&\\
 \psi &\sim p_0(\cdot).&
\end{aligned}
\end{equation}
where $\betag = \left(\beta(s_1), \dots, \beta(s_n) \right)^\top$ is a Gaussian Process (GP) observed at $\left(s_1, \dots, s_n\right)$ and $f$ is a density function with respect to a suitable dominating measure. 
Here $\One_n = (1, \dots, 1)^\top$ is an $n$-dimensional vector and $R^{(n)}=\left(R_{ij} \right)_{i,j=1,\dots,n}$ is a $n \times n$ correlation matrix, with $R_{ij} = \text{Corr}\left(\beta(s_i), \beta(s_j) \right)$,
 defined through a suitable kernel function, that we assume to be fixed and known. 
Typically, strength of correlation among coefficients at different locations depends on their distance, with $R_{ij}$ defined e.g.\ through a kernel of the Mat\'ern family (see e.g.\ Section 4.2.1 in \cite{WR06}).
In this Section we focus on a single real covariate for notational convenience, but everything could be restated on a general $p$-dimensional space with little effort: direct analogues of the next lemma and corollaries similarly follow.
We first consider cases where the likelihood function has no specific hyper-parameters, such as in the common binary case where $Y(s_j) \mid \betag \sim \text{Bernoulli}(\sigma(\beta(s_j)x(s_j)))$, with $\sigma$ logistic link function and $Y(s_j)\in\{0,1\}$.

Let $P_n$ be the kernel of the Gibbs sampler which targets $\pi_n(\d\betag, \d\theta,\d\tau_\beta)=\L\left(\d\betag, \d\theta,\d\tau_\beta \mid  Y^{(n)} \right)$, by sequentially performing updates from the full conditionals of $\betag$, $\theta$ and $\tau_\beta$.
Despite the different graphical model structure, the analysis of mixing times of $P_n$ as $n\to\infty$ can be approached with the techniques we developed above, regardless of the specific likelihood used in \eqref{eq:gauss_proc_structure}.
The first step is to perform a dimensionality reduction analogous to the one in Section \ref{sec:dim_red}.
Define $\psi = (\theta, \tau_\beta)$ and $\bm{T}(\betag) = \left(T_{\theta}, T_{\tau_\beta}\right)$, where $T_{\theta} = \One^\top R^{-1} \betag$, $T_{\tau_\beta} = \betag^\top R^{-1} \betag$, which play the same role of global parameters and sufficient statistics in Lemma \ref{sufficient_lemma}. Indeed it 
holds $\L\left(\d\psi \mid \betag,  Y^{(n)}) = \L(\d\psi \mid \bm{T}(\betag),  Y^{(n)}\right)$ 
and we can provide an analogue of Lemma \ref{sufficient_lemma} for model \eqref{eq:gauss_proc_structure}.
\begin{lemma}\label{sufficient_lemma_GP}
Let $\pi_n$ and $P_n$ be defined as above for model \eqref{eq:gauss_proc_structure}.
Let $\hat{P}_n$ be the transition kernel of Gibbs sampler targeting $\hat{\pi}_n(\d\bT, \d\theta,\d\tau_\beta)=\L\left(\d\bT, \d\theta,\d\tau_\beta \mid  Y^{(n)} \right)$ which sequentially performs updates from the full conditionals of $\bT$, $\theta$ and $\tau_\beta$. Let $(\bT^{(t)}, \d\theta^{(t)},\d\tau_\beta^{(t)} )_{t \geq 1}$ be the stochastic process obtained as a time-wise transformation of $(\bm{\beta}^{(t)}, \d\theta^{(t)},\d\tau_\beta^{(t)} )_{t \geq 1}$. Then $(\bT^{(t)}, \d\theta^{(t)},\d\tau_\beta^{(t)} )_{t \geq 1}$ is a Markov chain, its transition kernel coincides with $\hat{P}_n$, and its mixing times $\hat{t}^{(n)}_{mix}$ satisfy
\begin{align*}
\sup_{\mu \in \sN \left(\pi_n, M \right)} t^{(n)}_{mix}(\epsilon, \mu)
&= \sup_{\nu \in \sN \left(\hat{\pi}_n, M \right)} \hat{t}^{(n)}_{mix}(\epsilon, \nu)
&M\geq 1\,.
\end{align*}
\end{lemma}
 Also, provided a rescaled version of $(\bT, \theta,\tau_\beta)$ converges to a suitable limit conditional on the data, the mixing times are bounded with respect to the number of observations.
\begin{corollary}\label{corollary_GP}
Under model \eqref{eq:gauss_proc_structure}, let $\hat{\pi}_n$ satisfy assumption $(A1)$ for a given data generating process $ Y^{(n)} \sim Q^{(n)}$, with limiting distribution $\tilde{\pi}$. 
If $(M, \epsilon) \in [1,\infty) \times (0,1)$ is such that $\tilde{t}_{mix}(\epsilon, M) < \infty$, then it holds
\begin{equation}\label{conclusion_GP}
Q^{(n)}\left(t^{(n)}_{mix}(\epsilon, M)\leq \tilde{t}_{mix}(\epsilon, M)\right) \to 1
\qquad\hbox{as }n \to \infty\,.
\end{equation}
\end{corollary}

 In some cases the likelihood contains some unknown parameters that are also included in the Bayesian model.
A common example is the likelihood precision $\tau_\epsilon$ in normal linear models with spatially varying regression coefficients (see e.g.\ \cite{G03} or Section $2$ in \cite{WR06}), where
\begin{align}\label{eq:norm_lik_GP}
Y(s_i) \mid \betag & \sim N(\beta(s_i)x(s_i),\tau_\epsilon^{-1}), & i = 1, \dots, n.
\end{align}
Let $P_n$ be the Gibbs sampler kernel targeting $\pi_n(\d\betag, \d\theta,\d\tau_\beta,\d\tau_\epsilon)=\L\left(\d\betag, \d\theta,\d\tau_\beta,\d\tau_\epsilon \mid  Y^{(n)} \right)$, by sequentially performing updates from the full conditionals of $\betag$, $\theta$, $\tau_\beta$ and $\tau_\epsilon$.
Analogously to Section \ref{sec:lik_pars}, the results of Lemma \ref{sufficient_lemma_GP} and Corollary \ref{corollary_GP} extend to this context with $\psi = (\theta, \tau_\beta, \tau_\epsilon)$ and $\bm{T}$ defined as $\bm{T} = \left(T_{\theta}, T_{\tau_\beta}, T_{\tau_\epsilon}  \right)$, where $T_{\tau_\epsilon} = \left( Y^{(n)}-D\betag \right)^\top \left( Y^{(n)}-D\betag \right)$ and $D$ is the $n\times n$ diagonal matrix with values $\left(x(s_1), \dots, x(s_n) \right)$.
This is summarized in the next corollary.
\begin{corollary}\label{corollary_GP_ext}
Under model \eqref{eq:gauss_proc_structure} with likelihood as in \eqref{eq:norm_lik_GP}, assume the conditions of Corollary \ref{corollary_GP} are satisfied with $\psi = (\theta, \tau_\beta, \tau_\epsilon)$ and $\bm{T} = \left(T_{\theta}, T_{\tau_\beta}, T_{\tau_\epsilon}  \right)$. Then \eqref{conclusion_GP} holds.
\end{corollary}

 Similarly to the hierarchical normal case, studied in Section \ref{normal_case}, if the precisions $(\tau_\beta, \tau_\epsilon)$ are fixed in specification \eqref{eq:norm_lik_GP}, then the spectral gap of $P_n$ can be explicitly studied to deduce limiting bounds on mixing times (see e.g.\ \cite{BS16}); while if the precisions are unknown, as it is mostly the case in applications, the performances of $P_n$ have only been empirically studied through simulations. 
The methodology we introduce here can be used to formally analyze the behaviour of these samplers as $n \to \infty$.

To conclude this section, it is important to note that in this context the kernel $P_n$ may or may not be directly implementable, depending on the specific model formulation.   
In the commonly used linear case, the full conditional distribution $\pi_n(\d\betag\mid \psi)$ is normal, so that sampling becomes accessible and $P_n$ is directly the algorithm used to sample from $\pi_n$. 
See e.g.\ Appendix $2$ of \cite{BS16} for details on the implementation, including expressions for the full conditionals. 
In other cases, e.g.\ for log-concave likelihoods such as the binary regression ones, adaptive rejection sampling techniques (e.g. \cite{GW92}) can be used in low dimensions. 
In the more general case the exact update from $\pi_n(\d\betag\mid \psi)$ is commonly replaced with a Metropolis update from $\pi_n(\d\betag\mid \psi)$ (using e.g.\ a gradient-based kernel such as MALA or HMC). In the latter case, the Gibbs kernel $P_n$ we analyse here is an idealized version of the practically used Metropolis-within-Gibbs kernel. Under suitable (mild) assumptions, we expect the convergence properties of this idealized scheme to provide a lower bound to the Metropolis-within-Gibbs schemes used in practice.
Also, we expect the convergence of the two kernels to be of the same order when the kernel used for the Metropolis updates on the full conditional mixes fast. 
Providing quantitative results in this direction is an interesting area for future work, which we are currently pursuing. This would extend the applicability of the proof techniques developed in this work to broad classes of non conditionally-conjugate models, such as Gaussian Processes with non-Gaussian likelihood discussed above. See Section \ref{sec:conclusions} for more details.

\section{Feasible start}\label{Feasible_start}

All the previous results are stated in terms of mixing times from worst case $M$-warm start, as defined in \eqref{N_class}. Since starting from $\mu \in \mathcal{N}(\pi_J, M)$ with small $M$ (e.g. not increasing with $J$) may be in principle infeasible, it is of interest to provide an explicit example of a starting distribution that can be implemented in practice, a so-called feasible start, where the associated value of $M$ can be controlled. 
In the setting of Theorem \ref{theorem_one_level_nested}, the properties of the Gibbs samplers combined with 
the probabilistic structure of hierarchical models allow to translate the problem of feasible starts into the one of having a good initialisation for the hyper-parameters $\psi$, as we now show. 
Indeed, assume that the maximum marginal likelihood estimator $\hat{\psi}_J = \text{arg max} \prod_{j = 1}^Jg(Y_j \mid \psi)$, with $g$ as in \eqref{likelihood_data}, is well-defined. 
Let $\mu_J \in \mathcal{P}\left(\R^{lJ+D} \right)$ be given by
\begin{align}\label{def_feasible}
\mu_J\left(B  \right) &= \int_B\text{Unif}\left(  \hat{\psi}_J, c/\sqrt{J} \right)(\d \psi)\prod_{j = 1}^Jp(\theta_j \mid Y_j, \psi)\, \d \bm{\theta}
&B \subset \R^{lJ+D}
\end{align}
where $c > 0$ is a fixed constant and $\text{Unif}\left(\psi, r \right)$ denotes the uniform distribution over the closed ball of center $\psi$ and radius $r > 0$. Therefore, the initial point is obtained by sampling from the uniform distribution around the maximum likelihood estimator for $\psi$ and, conditional on this value, from the posterior distribution of the groups specific parameters. The next theorem shows that this choice leads to a good asymptotic behaviour of the mixing times.
\begin{theorem}\label{theorem_feasible}
Consider the same setting of Theorem \ref{theorem_one_level_nested} and let $\mu_J \in \mathcal{P}\left(\R^{lJ+D} \right)$ as in \eqref{def_feasible}. Then, for every $\epsilon\in(0,1)$ there exists $T\left(\psi^*, \epsilon, c \right) < \infty$ such that
\[
\lim \inf_{J \to \infty} \, Q_{\psi^*}^{(J)}\left(t^{(J)}_{mix}(\epsilon, \mu_J)\leq T\left(\psi^*, \epsilon, c \right)\right)  \to  1\qquad\hbox{ as }J \to \infty\,.
\]
\end{theorem}
 The difference with Theorem \ref{theorem_one_level_nested} is in the specification of the starting distribution, that is now made explicit. Note that whether or not $\mu_J$ is a feasible start in practice depends on whether the maximum likelihood estimate $\hat{\psi}_J$ can be computed, using e.g.\ an Expectation-Maximization algorithm, up to a $\sO(1/\sqrt{J})$ error. 
\begin{remark}
By its definition in \eqref{initialGibbs}, the Gibbs sampler does not depend on the starting point of the first block. Therefore Theorem \ref{theorem_feasible} extends to any $\mu_J \in \mathcal{P}\left(\R^{lJ+D} \right)$ such that
\begin{align*}
\mu_J\left(\R^{lJ}\times A \right) &= \text{Unif}\left(  \hat{\psi}_J, c/\sqrt{J} \right)(A)
&A \subset \R^D\,.
\end{align*}
\end{remark}

\section{Future works}\label{sec:conclusions}
A first natural extension in this context would be the case where no fixed dimensional sufficient statistic is available, i.e. $p(\cdot\mid\psi)$ in \eqref{eq:one_level_nested_intro} does not belong to the exponential family. Since the above dimensionality reduction does not apply there, a possibility is to study the marginal chain induced on $\psi$; indeed the latter has the same properties of the Gibbs sampler on $(\bm{\theta}, \psi)$, see e.g. \cite{R01}.
Also, in this work we have focused on the case with well-specified likelihoods but, as discussed after Theorem \ref{BvM}, we expect the misspecified setting to behave in qualitatively similar ways.

Secondly, when dealing with Gibbs samplers, it is often the case that some of the conditional updates cannot be performed exactly. A natural solution is to employ more general coordinate-wise schemes, where exact sampling is replaced by Markov updates with stationary measure given by the conditional distribution. For example in hierarchical models for categorical data (see Section \ref{binary_data}), while in principle exact conditional sampling is feasible, the parameters $\theta_j$ are often sampled in a Metropolis-within-Gibbs fashion, for reasons of computational efficiency and easiness of implementation. 
While algorithmically convenient, the modification makes theoretical analysis significantly more involved: in particular Proposition \ref{prop:convergence_operators} ceases to hold and the dimensionality reduction given by Lemma \ref{sufficient_lemma} is not available without exact sampling. In ongoing work we are considering a different strategy, by providing lower bounds on the approximate conductance \citep{LS93}: our preliminary results suggest that, provided the conditional Markov updates have good spectral properties, general coordinate-wise schemes can enjoy the same dimension-free convergence of the Gibbs sampler.
Another interesting direction would be to derive results analogous to the ones in Section \ref{section:gibbs} for other MCMC kernels (e.g.\ gradient-based ones) under appropriate regularity assumptions on the sequence of target distribution, potentially exploiting tools from the recent work in \cite{CJ23}.

Finally, we expect (at least parts of) our methodology to be applicable much beyond hierarchical models as in \eqref{eq:one_level_nested_intro}. For example, when fitting (finite or infinite) Bayesian mixture models, it is customary to use a Gibbs sampler over a properly augmented space by introducing latent allocation variables (see e.g. \cite{DR94}): this leads to a problem of increasing dimensionality, since the number of latent variables grows linearly with $n$. An asymptotic analysis, as performed in this paper, seems accessible: indeed, posterior concentration results are available \citep{N13} and a dimensionality reduction similar to Lemma \ref{sufficient_lemma} can be exploited. However there are still significant challenges to perform a rigorous analysis in this setting: for example posterior contraction is often proved using Wasserstein distance, that is in general too weak for our purposes. We leave the discussion of such issues to a future work. 

\medskip
\medskip

\textbf{Funding. }
GZ acknowledges support from the European Research Council (ERC), through StG ``PrSc-HDBayLe'' grant ID 101076564.

\bibliographystyle{chicago}
\bibliography{Bib_gibbs}

\begin{thebibliography}{}

\bibitem[\protect\citeauthoryear{Amit}{Amit}{1991}]{A91}
Amit, Y. (1991).
\newblock {On Rates of Convergence of Stochastic Relaxation for Gaussian and
  Non-Gaussian Distributions}.
\newblock {\em J. Multivar. Anal.\/}~{\em 38}, 82--99.

\bibitem[\protect\citeauthoryear{Andrieu, Lee, Power, and Wang}{Andrieu
  et~al.}{2022}]{AL22}
Andrieu, C., A.~Lee, S.~Power, and A.~Q. Wang (2022).
\newblock {Explicit convergence bounds for Metropolis Markov chains:
  isoperimetry, spectral gaps and profiles}.
\newblock {\em arXiv preprint arXiv:2211.08959\/}.

\bibitem[\protect\citeauthoryear{Atchad\'e}{Atchad\'e}{2021}]{A21}
Atchad\'e, Y.~F. (2021).
\newblock {Approximate Spectral Gaps for Markov Chain Mixing Times in High
  Dimensions}.
\newblock {\em SIAM. J. MATH. DATA SCI.\/}~{\em 3}, 854--872.

\bibitem[\protect\citeauthoryear{Bally and Caramellino}{Bally and
  Caramellino}{2015}]{BC16}
Bally, V. and L.~Caramellino (2015).
\newblock {Asymptotic development for the CLT in total variation distance}.
\newblock {\em Bernoulli\/}~{\em 22}, 2442--2485.

\bibitem[\protect\citeauthoryear{Bass and Sahu}{Bass and Sahu}{2016}]{BS16}
Bass, M.~R. and S.~K. Sahu (2016).
\newblock {A comparison of centring parameterisations of Gaussian process-based
  models for Bayesian computation using MCMC}.
\newblock {\em Stat. Comput.\/}~{\em 27}, 1491--1512.

\bibitem[\protect\citeauthoryear{Belloni and Chernozhukov}{Belloni and
  Chernozhukov}{2009}]{BC09}
Belloni, A. and V.~Chernozhukov (2009).
\newblock {On the computational complexity of MCMC-based estimators in large
  samples}.
\newblock {\em Ann. Statist.\/}~{\em 37}, 2011--2055.

\bibitem[\protect\citeauthoryear{Beskos, Pillai, Roberts, Sanz-Serna, and
  Stuart}{Beskos et~al.}{2013}]{B13}
Beskos, A., N.~Pillai, G.~Roberts, J.~Sanz-Serna, and A.~Stuart (2013).
\newblock {Optimal tuning of the hybrid Monte Carlo algorithm}.
\newblock {\em Bernoulli\/}~{\em 19}, 1501--1534.

\bibitem[\protect\citeauthoryear{Bhattacharya and Rao}{Bhattacharya and
  Rao}{2010}]{BR10}
Bhattacharya, R.~N. and R.~R. Rao (2010).
\newblock {\em {Normal Approximations and Asymptotic Expansions}}.
\newblock Society for Industrial and Applied Mathematics.

\bibitem[\protect\citeauthoryear{Bobkov, Chistyakov, and G\"otze}{Bobkov
  et~al.}{2014}]{B14}
Bobkov, S.~G., G.~P. Chistyakov, and F.~G\"otze (2014).
\newblock {Berry-Essen bounds in the entropic central limit theorem}.
\newblock {\em Probab. Theory Relat. Fields\/}~{\em 159}, 435--478.

\bibitem[\protect\citeauthoryear{Brooks, Gelman, Jones, and Meng}{Brooks
  et~al.}{2011}]{B11}
Brooks, S., A.~Gelman, G.~L. Jones, and X.~Meng (2011).
\newblock {\em {Handbook of Markov Chain Monte Carlo}}.
\newblock Chapman and Hall.

\bibitem[\protect\citeauthoryear{Caprio and Johansen}{Caprio and
  Johansen}{2023}]{CJ23}
Caprio, R. and A.~Johansen (2023).
\newblock {A calculus for Markov chain Monte Carlo: studying approximations in
  algorithms}.
\newblock {\em arXiv preprint arXiv:2310.03853\/}.

\bibitem[\protect\citeauthoryear{Casella and george}{Casella and
  george}{1992}]{C92}
Casella, G. and E.~I. george (1992).
\newblock {Explaining the Gibbs Sampler}.
\newblock {\em Am. Stat.\/}~{\em 46}, 167--174.

\bibitem[\protect\citeauthoryear{Chlebicka, Latuszynski, and
  Miasojedow}{Chlebicka et~al.}{2023}]{CLM23}
Chlebicka, I., K.~Latuszynski, and B.~Miasojedow (2023).
\newblock {Solidarity of Gibbs Samplers: the spectral gap}.
\newblock {\em arXiv preprint arXiv:2304.02109\/}.

\bibitem[\protect\citeauthoryear{Dalalyan}{Dalalyan}{2017}]{D17}
Dalalyan, A.~S. (2017).
\newblock {Theoretical Guarantees for Approximate Sampling from Smooth and
  Log-Concave Densities}.
\newblock {\em J. R. Stat. Soc. Ser. B.\/}~{\em 79}, 651--676.

\bibitem[\protect\citeauthoryear{Diaconis, Khare, and Saloff-Coste}{Diaconis
  et~al.}{2008}]{D08}
Diaconis, P., K.~Khare, and L.~Saloff-Coste (2008).
\newblock {Gibbs Sampling, Exponential Families and Orthogonal Polynomials}.
\newblock {\em Stat. Sci.\/}~{\em 23}, 151--178.

\bibitem[\protect\citeauthoryear{Diebolt and Robert}{Diebolt and
  Robert}{1994}]{DR94}
Diebolt, J. and C.~P. Robert (1994).
\newblock {Estimation of Finite Mixture Distributions through Bayesian
  Sampling}.
\newblock {\em J. R. Stat. Soc. Ser. B.\/}~{\em 56}, 363--375.

\bibitem[\protect\citeauthoryear{Durante}{Durante}{2019}]{DD19}
Durante, D. (2019).
\newblock {Conjugate Bayes for probit regression via unified skew-normal
  distributions}.
\newblock {\em Biometrika\/}~{\em 106}, 765--779.

\bibitem[\protect\citeauthoryear{Durmus and Moulines}{Durmus and
  Moulines}{2017}]{DM17}
Durmus, A. and E.~Moulines (2017).
\newblock {Nonasymptotic convergence analysis for the unadjusted Langevin
  algorithm}.
\newblock {\em Ann. Appl. Probab.\/}~{\em 27}, 1551--1587.

\bibitem[\protect\citeauthoryear{Dwivedi, Chen, Wainwright, and Yu}{Dwivedi
  et~al.}{2019}]{D19}
Dwivedi, R., Y.~Chen, M.~J. Wainwright, and B.~Yu (2019).
\newblock {Log--concave sampling: Metropolis--Hastings algorithms are fast!}
\newblock {\em J. Mach. Learn. Res.\/}~{\em 20}, 1--42.

\bibitem[\protect\citeauthoryear{Feller}{Feller}{1970}]{F70}
Feller, W. (1970).
\newblock {\em {An Introduction to Probability Theory and Its Applications}}.
\newblock John Wiley \& Sons.

\bibitem[\protect\citeauthoryear{Flegal, Hughes, Vats, Gupta, and Maji}{Flegal
  et~al.}{2021}]{F21}
Flegal, J.~M., J.~Hughes, D.~Vats, K.~Gupta, and U.~Maji (2021).
\newblock {mcmcse: Monte Carlo Standard Errors for MCMC}.
\newblock {\em R package\/}.

\bibitem[\protect\citeauthoryear{Gelfand, Kim, Sirmans, and Banerjee}{Gelfand
  et~al.}{2003}]{G03}
Gelfand, A.~E., H.~J. Kim, C.~Sirmans, and S.~Banerjee (2003).
\newblock {Spatial Modelling With Spatially Varying Coefficient Processes}.
\newblock {\em J. Am. Stat. Assoc.\/}~{\em 98}, 387--396.

\bibitem[\protect\citeauthoryear{Gelfand, Sahu, and Carlin}{Gelfand
  et~al.}{1995}]{G95}
Gelfand, A.~E., S.~K. Sahu, and B.~P. Carlin (1995).
\newblock {Efficient Parametrisations for Normal Linear Mixed Models}.
\newblock {\em Biometrika\/}~{\em 82}, 479--488.

\bibitem[\protect\citeauthoryear{Gelman, Carlin, Stern, Dunson, Vehtari, and
  Rubin}{Gelman et~al.}{2013}]{G13}
Gelman, A., J.~B. Carlin, H.~S. Stern, D.~B. Dunson, A.~Vehtari, and D.~B.
  Rubin (2013).
\newblock {\em {Bayesian Data Analysis}}.
\newblock CRC press.

\bibitem[\protect\citeauthoryear{Gelman and Hill}{Gelman and Hill}{2007}]{GH07}
Gelman, A. and J.~L. Hill (2007).
\newblock {\em {Data Analysis Using Regression and Multilevel/Hierarchical
  Models}}.
\newblock Cambridge University Press.

\bibitem[\protect\citeauthoryear{Gilks and Wild}{Gilks and Wild}{1992}]{GW92}
Gilks, W.~R. and P.~Wild (1992).
\newblock {Adaptive Rejection Sampling for Gibbs Sampling}.
\newblock {\em J. R. Stat. Soc. Ser. C\/}~{\em 41}, 337--348.

\bibitem[\protect\citeauthoryear{Gong and Flegal}{Gong and Flegal}{2015}]{GF15}
Gong, L. and J.~M. Flegal (2015).
\newblock {A Practical Sequential Stopping Rule for High-Dimensional Markov
  Chain Monte Carlo}.
\newblock {\em J. Comput. Graph. Stat.\/}~{\em 25}, 684--700.

\bibitem[\protect\citeauthoryear{Green, Latuszynski, Pereyra, and Robert}{Green
  et~al.}{2015}]{G15}
Green, P.~J., K.~Latuszynski, M.~Pereyra, and C.~P. Robert (2015).
\newblock {Bayesian computation: a summary of the current state, and samples
  backwards and forwards}.
\newblock {\em Stat. Comput.\/}~{\em 25}, 835--862.

\bibitem[\protect\citeauthoryear{Hobert}{Hobert}{2011}]{H11}
Hobert, J.~P. (2011).
\newblock {The data augmentation algorithm: Theory and methodology}.
\newblock {\em Handbook of Markov chain Monte Carlo\/}, 253--293.

\bibitem[\protect\citeauthoryear{Jin and Hobert}{Jin and Hobert}{2022}]{JH21}
Jin, Z. and J.~P. Hobert (2022).
\newblock {Dimension free convergence rates for Gibbs samplers for Bayesian
  linear mixed models}.
\newblock {\em Stoch. Process. Their Appl.\/}~{\em 148}, 25--67.

\bibitem[\protect\citeauthoryear{Johndrow, Smith, Pillai, and Dunson}{Johndrow
  et~al.}{2019}]{J19}
Johndrow, J.~E., A.~Smith, N.~Pillai, and D.~B. Dunson (2019).
\newblock {MCMC for Imbalanced Categorical Data}.
\newblock {\em J. Am. Stat. Assoc.\/}~{\em 114}, 1394--1403.

\bibitem[\protect\citeauthoryear{Kamatani}{Kamatani}{2014}]{K14}
Kamatani, K. (2014).
\newblock {Local consistency of Markov chain Monte Carlo methods}.
\newblock {\em Ann. Inst. Stat. Math.\/}~{\em 66}, 63--74.

\bibitem[\protect\citeauthoryear{Khare and Zhou}{Khare and Zhou}{2009}]{K09}
Khare, K. and H.~Zhou (2009).
\newblock {Rates of convergence of some multivariate Markov chains with
  polynomial eigenfunctions}.
\newblock {\em Ann. Appl. Probab.\/}~{\em 2}, 737--777.

\bibitem[\protect\citeauthoryear{Kleijn and van~der Vaart}{Kleijn and van~der
  Vaart}{2012}]{KV12}
Kleijn, B. J.~K. and A.~W. van~der Vaart (2012).
\newblock {The Bernstein-Von-Mises theorem under misspecification}.
\newblock {\em Electron. J. Stat.\/}~{\em 6}, 353--381.

\bibitem[\protect\citeauthoryear{Liu}{Liu}{1994}]{L94}
Liu, J.~S. (1994).
\newblock {Fraction of Missing Information and Convergence Rate for Data
  Augmentation}.
\newblock In {\em Computationally Intensive Statistical Methods: Proceedings of
  the 26th Symposium Interface}.

\bibitem[\protect\citeauthoryear{Lov\'asz and Simonovits}{Lov\'asz and
  Simonovits}{1993}]{LS93}
Lov\'asz, L. and M.~Simonovits (1993).
\newblock {Random Walks in a Convex Body and an Improved Volume Algorithm}.
\newblock {\em Random Struct. and Alg.\/}~{\em 4}, 359--412.

\bibitem[\protect\citeauthoryear{Martin, Frazier, and Robert}{Martin
  et~al.}{2023}]{MFR23}
Martin, G.~M., D.~T. Frazier, and C.~P. Robert (2023).
\newblock {Computing Bayes: From Then `Til Now}.
\newblock {\em Stat. Sci.\/}~{\em In press}.

\bibitem[\protect\citeauthoryear{Neal}{Neal}{2011}]{N11}
Neal, R.~M. (2011).
\newblock {MCMC using Hamiltonian dynamics}.
\newblock {\em Handbook of Markov Chain Monte Carlo\/}, 113--162.

\bibitem[\protect\citeauthoryear{Negrea, Yang, Feng, Roy, and Huggins}{Negrea
  et~al.}{2022}]{N22}
Negrea, J., J.~Yang, H.~Feng, D.~M. Roy, and J.~H. Huggins (2022).
\newblock {Statistical Inference with Stochastic Gradient Algorithms}.
\newblock {\em arXiv preprint arXiv:2207.12395\/}.

\bibitem[\protect\citeauthoryear{Nguyen}{Nguyen}{2013}]{N13}
Nguyen, X. (2013).
\newblock {Convergence of latent mixing measures in finite and infinite mixture
  models}.
\newblock {\em Ann. Statist.\/}~{\em 41}, 370--400.

\bibitem[\protect\citeauthoryear{Nickl and Wang}{Nickl and Wang}{2022}]{NW20}
Nickl, R. and S.~Wang (2022).
\newblock {On polynomial-time computation of high-dimensional posterior
  measures by Langevin-type algorithms}.
\newblock {\em J. Eur. Math. Soc.\/}.

\bibitem[\protect\citeauthoryear{Papaspiliopoulos, Roberts, and
  Zanella}{Papaspiliopoulos et~al.}{2020}]{PZ20}
Papaspiliopoulos, O., G.~Roberts, and G.~Zanella (2020).
\newblock {Scalable inference for crossed random effects models}.
\newblock {\em Biometrika\/}~{\em 107}, 25--40.

\bibitem[\protect\citeauthoryear{Papaspiliopoulos, Roberts, and
  Sk\"old}{Papaspiliopoulos et~al.}{2003}]{PRS03}
Papaspiliopoulos, O., G.~O. Roberts, and M.~Sk\"old (2003).
\newblock {Non-Centered Parameterizations for Hierarchical Models and Data
  Augmentation (with discussion)}.
\newblock In {\em Bayesian Statistics (J. M. Bernardo, M. J. Bayarri, J. O.
  Berger, A. P. Dawid, D. Heckerman, A. F. M. Smith and M. West, eds.)}, pp.\
  307--326.

\bibitem[\protect\citeauthoryear{Papaspiliopoulos, Roberts, and
  Sk\"old}{Papaspiliopoulos et~al.}{2007}]{PRS07}
Papaspiliopoulos, O., G.~O. Roberts, and M.~Sk\"old (2007).
\newblock {A General Framework for the Parametrization of Hierarchical Models}.
\newblock {\em Stat. Sci.\/}, 59--73.

\bibitem[\protect\citeauthoryear{Papaspiliopoulos, Roberts, and
  Sk{\"o}ld}{Papaspiliopoulos et~al.}{2007}]{papaspiliopoulos2007general}
Papaspiliopoulos, O., G.~O.~R. Roberts, and M.~Sk{\"o}ld (2007).
\newblock {A General Framework for the Parametrization of Hierarchical Models}.
\newblock {\em Statistical Science\/}, 59--73.

\bibitem[\protect\citeauthoryear{Papaspiliopoulos, Stumpf-F\'etizon, and
  Zanella}{Papaspiliopoulos et~al.}{2023}]{SPZ23}
Papaspiliopoulos, O., T.~Stumpf-F\'etizon, and G.~Zanella (2023).
\newblock {Scalable computation for Bayesian hierarchical models}.
\newblock {\em arXiv preprint arXiv:2103.10875\/}.

\bibitem[\protect\citeauthoryear{Petrov}{Petrov}{1956}]{P56}
Petrov, V.~V. (1956).
\newblock {A local theorem for densities of sums of independent random
  variables}.
\newblock {\em Theory Probab. Appl.\/}~{\em 84}, 316--322.

\bibitem[\protect\citeauthoryear{Qin and Hobert}{Qin and Hobert}{2019}]{Q19}
Qin, Q. and J.~P. Hobert (2019).
\newblock {Convergence complexity analysis of Albert and Chib's algorithm for
  Bayesian probit regression}.
\newblock {\em Ann. Statist.\/}~{\em 47}, 2320--2347.

\bibitem[\protect\citeauthoryear{Qin and Hobert}{Qin and Hobert}{2022}]{Q22}
Qin, Q. and J.~P. Hobert (2022).
\newblock {Wasserstein-based methods for convergence complexity analysis of
  MCMC with applications}.
\newblock {\em Ann, Appl. Prob.\/}~{\em 32}, 124--166.

\bibitem[\protect\citeauthoryear{Rajaratnam and Sparks}{Rajaratnam and
  Sparks}{2015}]{RS15}
Rajaratnam, B. and D.~Sparks (2015).
\newblock {MCMC-Based Inference in the Era of Big Data: A Fundamental Analysis
  of the Convergence Complexity of High-Dimensional Chains}.
\newblock {\em arXiv preprint arXiv:1508.00947\/}.

\bibitem[\protect\citeauthoryear{Roberts and Rosenthal}{Roberts and
  Rosenthal}{1998}]{R98}
Roberts, G.~O. and J.~S. Rosenthal (1998).
\newblock {Optimal scaling of discrete approximations to Langevin diffusions}.
\newblock {\em J. R. Stat. Soc. Ser. B\/}~{\em 60}, 255--268.

\bibitem[\protect\citeauthoryear{Roberts and Rosenthal}{Roberts and
  Rosenthal}{2001}]{R01}
Roberts, G.~O. and J.~S. Rosenthal (2001).
\newblock {Markov Chains and De-Initializing Processes}.
\newblock {\em Scand. J. Stat.\/}~{\em 28}, 489--504.

\bibitem[\protect\citeauthoryear{Roberts and Rosenthal}{Roberts and
  Rosenthal}{2004}]{R04}
Roberts, G.~O. and J.~S. Rosenthal (2004).
\newblock {General state space Markov chains and MCMC algorithms}.
\newblock {\em Probab. Surv.\/}~{\em 60}, 255--268.

\bibitem[\protect\citeauthoryear{Roberts and Sahu}{Roberts and
  Sahu}{1994}]{R94}
Roberts, G.~O. and S.~H. Sahu (1994).
\newblock {Simple conditions for the convergence of the Gibbs sampler and
  Metropolis-Hastings algorithms}.
\newblock {\em Stoch. Process. Their Appl.\/}~{\em 49}, 207--216.

\bibitem[\protect\citeauthoryear{Roberts and Sahu}{Roberts and
  Sahu}{1997}]{R97}
Roberts, G.~O. and S.~H. Sahu (1997).
\newblock {Updating Schemes, Correlation Structure, Blocking and
  Parameterization for the Gibbs Sampler}.
\newblock {\em J. R. Stat. Soc. Ser. B\/}~{\em 59}, 291--317.

\bibitem[\protect\citeauthoryear{Roberts and Sahu}{Roberts and
  Sahu}{2001}]{RS01}
Roberts, G.~O. and S.~H. Sahu (2001).
\newblock {Approximate Predetermined Convergence Properties of the Gibbs
  Sampler}.
\newblock {\em J. Comput. Graph. Statist.\/}~{\em 10}, 216--229.

\bibitem[\protect\citeauthoryear{Roberts and Tweedie}{Roberts and
  Tweedie}{1996}]{R96}
Roberts, G.~O. and R.~L. Tweedie (1996).
\newblock {Exponential convergence of Langevin distributions and their discrete
  approximations}.
\newblock {\em Bernoulli\/}~{\em 2}, 341--363.

\bibitem[\protect\citeauthoryear{Rosenthal}{Rosenthal}{1995}]{R95}
Rosenthal, J.~S. (1995).
\newblock {Minorization Conditions and Convergence Rates for Markov Chain Monte
  Carlo}.
\newblock {\em J. Am. Stat. Assoc\/}~{\em 90}, 558--566.

\bibitem[\protect\citeauthoryear{Rosenthal and Rosenthal}{Rosenthal and
  Rosenthal}{2015}]{RR15}
Rosenthal, J.~S. and P.~Rosenthal (2015).
\newblock {Spectral bounds for certain two-factor non-reversible MCMC
  algorithms}.
\newblock {\em Electron. Commun. Probab..\/}~{\em 20}, 1--10.

\bibitem[\protect\citeauthoryear{Ross}{Ross}{2011}]{R11}
Ross, N. (2011).
\newblock {Fundamentals of Stein's method}.
\newblock {\em Probab. Surv.\/}~{\em 8}, 210--293.

\bibitem[\protect\citeauthoryear{Smith}{Smith}{1953}]{S53}
Smith, W.~L. (1953).
\newblock {A frequency-function form of the central limit theorem}.
\newblock {\em Math. Proc. Camb. Philos. Soc.\/}~{\em 49}, 462--472.

\bibitem[\protect\citeauthoryear{Tang and Yang}{Tang and Yang}{2022}]{T22}
Tang, R. and Y.~Yang (2022).
\newblock {Computational Complexity of Metropolis-Adjusted Langevin Algorithms
  for Bayesian Posterior Sampling}.
\newblock {\em arXiv preprint arXiv:2206.06491\/}.

\bibitem[\protect\citeauthoryear{Thompson}{Thompson}{}]{T10}
Thompson, M.
\newblock {A Comparison of Methods for Computing Autocorrelation Time}.
\newblock {\em Technical Report No. 1007, Department of Statistics, University
  of Toronto\/}.

\bibitem[\protect\citeauthoryear{Van~der Vaart}{Van~der Vaart}{2000}]{V00}
Van~der Vaart, A.~W. (2000).
\newblock {\em {Asymptotic Statistics}}.
\newblock Cambridge University Press.

\bibitem[\protect\citeauthoryear{Williams and Rasmussen}{Williams and
  Rasmussen}{2006}]{WR06}
Williams, C.~K. and C.~E. Rasmussen (2006).
\newblock {\em {Gaussian Processes for Machine Learning}}.
\newblock Cambridge MA: MIT press.

\bibitem[\protect\citeauthoryear{Wu, Schmidler, and Chen}{Wu
  et~al.}{2022}]{WSC22}
Wu, K., S.~Schmidler, and Y.~Chen (2022).
\newblock {Minimax Mixing Time of the Metropolis-Adjusted Langevin Algorithm
  for Log-Concave Sampling}.
\newblock {\em J. Mach. Learn. Res.\/}~{\em 23}, 1--63.

\bibitem[\protect\citeauthoryear{Yang and Rosenthal}{Yang and
  Rosenthal}{2022}]{Y17}
Yang, J. and J.~S. Rosenthal (2022).
\newblock {Complexity results for MCMC derived from quantitative bounds}.
\newblock {\em Ann. Appl. Prob.\/}~{\em 33}, 1459--1500.

\bibitem[\protect\citeauthoryear{Yang, Wainwright, and Jordan}{Yang
  et~al.}{2016}]{Y16}
Yang, J., M.~J. Wainwright, and M.~I. Jordan (2016).
\newblock {On the computational complexity of high-dimensional Bayesian
  variable selection}.
\newblock {\em Ann. Statist.\/}~{\em 44}, 2497--2532.

\bibitem[\protect\citeauthoryear{Zhou, Yang, Vats, Roberts, and Rosenthal}{Zhou
  et~al.}{2022}]{Z21}
Zhou, Q., J.~Yang, D.~Vats, G.~O. Roberts, and J.~S. Rosenthal (2022).
\newblock {Dimension-free mixing for high-dimensional Bayesian variable
  selection}.
\newblock {\em J. R. Stat. Soc. Ser. B\/}~{\em 84}, 1751--1784.

\end{thebibliography}


\newpage
\begin{appendices}

\section{Simple counter-examples for Section \ref{section:gibbs}}

\subsection{Convergence of the stationary distribution does not imply pointwise convergence of Gibbs operators}

Let $\sX = [0,1]^2$ and define $A_n = \left[\frac{r_n}{l_n}, \frac{r_n+1}{l_n}\right]$, where
\[
r_n = n-2^{k_n}, \quad l_n = 2^{k_n}, \quad k_n = \lfloor \log_2 n \rfloor,
\]
with $\lfloor a \rfloor$ denoting the integer part of $a$ and $n \geq 2$. Therefore $\{A_n\}_n$ is a collection of intervals with decreasing length, such that $x \in A_n$ infinitely often, for every $x \in [0,1]$. We define a sequence $\{\pi_n\}_n \subset \mathcal{P}\left(\sX \right)$ as
\[
\pi_n(\d x_1 \mid x_2) =
\begin{cases}
\mathbbm{1}_{[0,1]}(x_1) \, \d x_1, \quad x_2 \not \in A_n\\
\delta_0(\d x_1), \quad x_2 \in A_n
\end{cases},
\quad
\pi_n(\d x_2) = \mathbbm{1}_{[0,1]}(x_2) \, \d x_2,
\]
where $\mathbbm{1}_{A}(x) \, \d x$ denotes the uniform measure on $A$. Define now
\[
\pi(\d x_1, \d x_2) = \mathbbm{1}_{[0,1]}(x_1)\mathbbm{1}_{[0,1]}(x_2)\d x_1 \d x_2
\]
and denote $C = \{0\} \times A_n$. For every $B \subset \sX$ we have
\[
\begin{aligned}
|\pi_n(B) -\pi(B) | &\leq |\pi_n\left(B \cap C\right) -\pi\left(B\cap C\right) | + |\pi_n\left(B \cap C^c\right) -\pi\left(B\cap C^c\right) | \\
& = \pi_n\left(B \cap C\right) \leq \pi_n\left( C\right).
\end{aligned}
\]
Therefore we conclude
\[
\lTV \pi_n - \pi \rTV \leq \pi_n\left( C\right) \to 0,
\]
as $n \to \infty$. However, if $P_n$ and $P$ are the operators of the associated Gibbs samplers, for every $\x \in \sX$ it holds
\[
\lTV P_n(\x, \cdot) -P(\x, \cdot) \rTV \geq |P_n(\x, C) -P(\x, C) |,
\]
so that, since $x_2 \in A_n$ infinitely often, we get 
\[
\lTV P_n(\x, \cdot) -P(\x, \cdot) \rTV = 1
\]
infinitely often. Incidentally, it is not difficult to show that Gap$(P_n) = 0$ for every $n$, while Gap$(P) = 1$. Example $1.4$ shows that this mismatch may hold under significantly less pathological scenarios.

\subsection{Equality of the stationary distributions does not imply closeness of the transition operators}
Let $\pi_1 = \pi_2 = \pi$, with $\pi$ the standard Gaussian distribution. Moreover, let
\[
P_1(x, \cdot) = \epsilon\pi(\cdot) + (1-\epsilon) \delta_x(\cdot) \quad \text{and} \quad P_2(x, \cdot) = \epsilon\pi(\cdot) + (1-\epsilon) \delta_{-x}(\cdot),
\]
with $\epsilon \in [0,1)$.  $P_1$ and $P_2$ are uniformly ergodic transition operators with invariant distribution $\pi$. Let $\mu$ be the truncation of $\pi$ on the positive real numbers: it is easy to show that $\mu \in \sN(\pi, 2)$. However
\[
\lTV \mu P_1 - \mu P_2 \rTV \geq (1-\epsilon)\left[\mu( (0, \infty))-\mu( (-\infty, 0]) \right] = 1-\epsilon.
\]
Moreover, it holds that $\lTV \mu-\pi \rTV = 1/2$, so that we conclude
\[
2\lTV \mu-\pi \rTV -\epsilon \leq \lTV \mu P_1 - \mu P_2 \rTV \leq 2\lTV \mu-\pi \rTV.
\]

\subsection{Convergence of the stationary distribution in Wasserstein distance does not imply convergence of the mixing times for Gibbs sampler operators}

Let $\sX = \mathbb{R}^2$ and $\bar{\pi}_n(\d \x) = N(x_1 \mid 0,1/n)N(x_2 \mid 0,1/n) \d x_1 \d x_2$. Define $\pi_n$ to be the truncation of $\bar{\pi}_n$ on the set
\[
A = \left\{(-\infty, 0] \times (-\infty, 0]\right\} \, \bigcup \,\left\{ [0, + \infty) \times [0, +\infty)\right\}.
\]
Let $f \, : \sX \, \to \, \mathbb{R}$ be a Lipschitz function with constant $1$. Then it holds
\[
\int_\sX\left[f(x_1, x_2)-f(0,0) \right]\pi_n(\d \x) \leq \int_\sX\sqrt{x_1^2+x_2^2} \, \pi_n(\d \x) \to 0,
\]
as $n \to \infty$, so that $\lW \pi_n-\pi \rW \to 0$, where $\pi(\d \x) = \delta_{(0,0)}(\x)$ and $\lW \cdot \rW$ denotes the Wasserstein distance.

If $P$ is the kernel of the Gibbs sampler targeting $\pi$, then it is immediate to show that
\[
\sup_{\mu \in \sN(\pi, M)} \,\lW \mu P - \pi \rW = 0
\]
for every $M \geq 1$, so that the mixing times in Wasserstein distance are equal to $1$ for every $\epsilon > 0$.

Instead, denote with $\mu_n$ the truncation of $\pi_n$ on $A_1 = (-\infty, 0] \times (-\infty, 0]$. It is easy to show that $\mu_n \in \sN(\pi_n, 2)$, but
\[
\mu_nP_n^t(A_1)-\pi_n(A_1) = \frac{1}{2}
\]
for every $n$ and $t$, where $P_n$ is the kernel of the Gibbs sampler targeting $\pi_n$. 
Since the Wasserstein distance is stronger than the weak one, there exists an absolute constant $c$ such that $\lW \mu_nP_n^t-\pi_n \rW \geq c$ for every $n$ and $t$. Therefore, with $\epsilon$ small enough and $M \geq 2$, the mixing times of $P_n$ in Wasserstein distance are equal to infinity for every $n$.

\subsection{Convergence of the stationary distribution does not imply convergence of the spectral gaps for Gibbs operators}

Let $\sX = \mathbb{R}^2$ and 
\[
\pi(\d \x) = N(x_1 \mid 0,1)N(x_2 \mid 0,1) \d x_1 \d x_2,
\]
where $N(x \mid \mu, \sigma^2)$ is the density function of a gaussian distribution with mean $\mu$ and variance $\sigma^2$. Define $\pi_n$ to be the truncation of $\pi$ on the set $A_n$, where
\[
A_n = \left\{(-\infty, n] \times (-\infty, n]\right\} \, \bigcup \,\left\{ [n, + \infty) \times [n, +\infty)\right\}.
\]
If $P_n$ and $P$ are the operators of the associated Gibbs samplers, it is not difficult to show that
\[
\lTV \pi_n - \pi \rTV \to 0 \quad \text{and} \quad \lTV P_n(\x, \cdot) -P(\x, \cdot) \rTV \to 0
\]
as $n \to \infty$, for every $\x \in \sX$. However, if $B_n = (-\infty, n] \times (-\infty, n]$ we have
\[
\pi_n(B_n) > 0 \quad \text{and} \quad \int_{B_n}P_n\left(\x, B_n^c \right) \pi_n(\d \x) = 0, 
\]
so that Gap$(P_n) = 0$ for every $n$, while Gap$(P) = 1$.

\section{Regularity assumptions (B4)-(B6) for Theorem \ref{theorem_one_level_nested}}
Let
\begin{align}\label{posterior_moments}
M^{(p)}_{s}(\psi \mid y) &= E \left[T^p_s(\theta_j) \mid Y_j = y, \psi \right]\,,\\
M^{(p)}_{s,s'}(\psi \mid y) &= E \left[T^p_s(\theta_j)T^p_{s'}(\theta_j) \mid Y_j = y, \psi \right],
\end{align}
be the posterior moments of $\bT$ given $\psi$, denote $M^{(p)}(\psi \mid y) = \left(M_1^{(p)}(\psi \mid y), \dots, M_S^{(p)}(\psi \mid y) \right) \in \R^S$  and
\begin{equation}\label{def_C_V}
\left[C(\psi)\right]_{s,d} = E_{Y_j} \left[\partial_{\psi_d}M_s^{(1)}\left(\psi \mid Y_j \right) \right], \quad \left[V(\psi) \right]_{s,s'} = E_{Y_j} \left[\text{Cov}\left(T_s(\theta_j), T_{s'}(\theta_j) \mid Y_j, \psi \right) \right],
\end{equation}
with $s, s' = 1, \dots S$ and $d = 1, \dots, D$. Moreover we write $B_\delta$ for the ball of center $\psi^*$ and radius $\delta$, and denote expectations with respect to the law of $Y_j$ as defined in $(B1)$ by $E_{Y_j}[\cdot]$.
\begin{enumerate}
\item[$(B4)$] The expectation $M^{(p)}_{s}(\psi \mid y)$ is well defined for every $y$ and $p = 1, \dots, 6$. Moreover, there exist $\delta_4 > 0$ and $C$ finite constant such that for every $\psi \in B_{\delta_4}$ it holds
$E_{Y_j}\left[\left \lvert  \partial_{\psi_d}M^{(6)}_s(\psi \mid Y_j) \right\rvert \right] < C$, 
$E_{Y_j}\left[\left \lvert \partial_{\psi_d}\partial_{\psi_{d'}}M^{(1)}_s(\psi \mid Y_j) \right\rvert \right] < C$, \newline $
E_{Y_j}\left[\left \lvert  \partial_{\psi_d}M^{(1)}_{s, s'}(\psi \mid Y_j) \right\rvert \right] < C$ and $E_{Y_j}\left[\left \lvert  \partial_{\psi_d}\left\{M^{(1)}_{s}(\psi \mid Y_j)M^{(1)}_{s'}(\psi \mid Y_j)\right\} \right\rvert \right] < C$
for $s, s' = 1, \dots, S$ and  $d, d' = 1, \dots, D$.
 Finally, the matrix $V(\psi^*)$ defined in \eqref{def_C_V} is non singular.
\end{enumerate}
Assumption $(B4)$ can be understood as a smoothness condition. The posterior distribution of $\bm{T}$ should not change considerably, if we move from $\psi^*$ to a sufficiently close $\psi$: this is measured in terms of the derivative of the posterior moments, that must be finite in average. Thanks to $(B4)$ we can prove a suitable conditional Central Limit Theorem to show convergence of a rescaled version of $\bm{T}$, conditional to $\psi$ and $Y_{1:J}$.

We define the posterior characteristic function of $T(\theta_j) = \left(T_1(\theta_j), \dots, T_S(\theta_j)   \right)$ and $\sum_{j = 1}^kT(\theta_j)$, given $\psi$, as
$
\varphi\left(t \mid Y_j, \psi \right)= E \left[e^{it^\top T(\theta_j)} \mid Y_j, \psi \right]$ for $t\in\R^S$.
and
$
\varphi^{(k)}\left(t \mid Y_{1:k}, \psi \right) = \prod_{j = 1}^k\varphi\left(t \mid Y_j, \psi \right)
$, 
respectively.
We will assume:
\begin{enumerate}
\item[$(B5)$] There exist $k \geq 1$ and $\delta_5 >0$ such that
\[
\sup_{\psi \in B_{\delta_5}} \, \int_{\R^S} \left \lvert\varphi^{(k)}\left(t \mid Y_{1:k}, \psi \right) \right\rvert^2 \, \d t < \infty,
\]
for almost every $Y_1,\dots,Y_k \simiid Q_{\psi^*}$.
\item[$(B6)$] There exist $k' \geq 1$ and $\delta_6 >0$ such that
\[
\sup_{\psi \in B_{\delta_6}} \, \sup_{|t| > \epsilon} \left \lvert \varphi^{(k')}\left(t \mid Y_{1:k'}, \psi \right) \right \rvert < \phi(\epsilon),
\]
for almost every $Y_1,\dots,Y_k \simiid Q_{\psi^*}$, with $\phi(\epsilon) < 1$ for every $\epsilon > 0$.
\end{enumerate}
Assumptions $(B5)$ and $(B6)$ allow the convergence of $\bm{T}$ to hold for the total variation distance, that is stronger than the weak one, proved through $(B4)$. Loosely speaking, integrability of the characteristic function and its strictly positive distance from $1$ guarantee that the distribution is far from being discrete: the latter is exactly the case where weak convergence does not translate to stronger metrics. The problem of proving Central Limit theorems in total variation distance has received considerable attention over the decades: it can be tackled with Fourier-based techniques \citep{P56, S53}, as we do here, but also with Stein's method (see \cite{R11} for a survey), Malliavin calculus (e.g. \cite{BC16}) or through bounds based on entropy (e.g. \cite{B14}). Conditions $(B5)$ and $(B6)$ are somewhat reminiscent of the ones in Theorem $19.3$ in \cite{BR10}.

\newpage
\section{Proofs}

\subsection{Statement and proof of Lemma \ref{equivalent_representation}}
\begin{lemma}\label{equivalent_representation}
Let $\sN  \subset \mathcal{P}(\sX)$ and $\pi \in \mathcal{P}(\sX)$. Then
\[
\sup_{\mu \in \sN } \, \inf\left\{t\geq 1\,:\, \lTV \mu P^t-\pi \rTV<\epsilon \right\} =  \inf\left\{t\geq 1\,:\,  \sup_{\mu \in \sN } \,\lTV \mu P^t-\pi \rTV<\epsilon \right\},
\]
for every Markov transition kernel $P$.
\end{lemma}
\begin{proof}
Let
\[
t^{(1)} 
=
 \sup_{\mu \in \mathcal{N}} \, \inf\left\{t\geq 1\,:\, \lTV \mu P^t-\pi \rTV<\epsilon \right\}, \quad t^{(2)} 
=
\inf\left\{t\geq 1\,:\,  \sup_{\mu \in \mathcal{N}} \,\lTV \mu P^t-\pi \rTV<\epsilon \right\}.
\]
Assume $t^{(1)} < \infty$. Then $\lTV \mu P^{t^{(1)}}-\pi \rTV<\epsilon$ for every $\mu \in \mathcal{N}$. This implies 
\[
\sup_{\mu \in \mathcal{N}} \,\lTV \mu P^{t^{(1)}}-\pi \rTV<\epsilon,
\]
i.e. $t^{(2)} \leq t^{(1)}$. With a similar reasoning, if $t^{(2)} < \infty$ we have $t^{(1)} \leq t^{(2)}$. Therefore $t^{(1)} = t^{(1)}$ if either $t^{(1)} < \infty$ or $t^{(2)} < \infty$.

Assume now $t^{(1)} = \infty$ and fix $t^* > 0$. By definition of $t^{(1)}$ there exists $\mu \in \mathcal{N}$ such that
\[
\lTV \mu P^{t^*}-\pi \rTV \geq \epsilon,
\]
that implies
\[
\sup_{\mu \in \mathcal{N}} \,\lTV \mu P^{t^*}-\pi \rTV \geq \epsilon,
\]
i.e. $t^{(2)} > t^*$. Since $t^*$ is arbitrary, we have $t^{(2)} = \infty$. With a similar reasoning, if $t^{(2)} = \infty$ it holds $t^{(1)} = \infty$.
\end{proof}

\subsection{Statement and proof of Lemma \ref{lemma_closure}}
\begin{lemma}\label{lemma_closure}
Let $M\geq 1$, $\pi \in \mathcal{P}(\sX)$, $\mu \in \sN(\pi, M)$ and $P$ be a $\pi$-invariant Markov transition kernel.
Then $\mu P^t \in \sN(\pi, M)$, for every $t \in \mathbb{N}$.
\end{lemma}
\begin{proof}
Let $A\subseteq \sX$. 
Since $\mu \in \sN(\pi, M)$ and $P$ is $\pi$-invariant, we have
$(\mu P) (A) \leq M (\pi P) (A) =  M \pi(A)$. 
Thus $\mu P \in \sN(\pi, M)$ and the result follows by induction on $t$.
\end{proof}

\subsection{Proof of Lemma \ref{equivalent_statement}}
\begin{proof}
Let $\hat{P}_n = P_n  \circ \phi_n^{-1} $ be the push-forward operator of $P_n$ under $\phi_n$, defined as 
\begin{equation}\label{def_pushforward}
\hat{P}_n(\x, B) = P_n\left(\phi_n^{-1}(\x), \phi_n^{-1}(B) \right)
\end{equation}
for every $\x\in \phi_n(\sX)$ and $B\subseteq \sX$. Since $\phi_n$ is an injective transformation, $\hat{P}_n$ is a well-defined Markov transition kernel (see e.g.\ Lemma $1$ in \cite{PZ20}). Moreover, since $\phi_n$ is coordinate-wise as in \eqref{trans_phi} we have $\hat{P}_n= \hat{P}_{n,1} \dots \hat{P}_{n,K}$, where
\[
\begin{aligned}
\hat{P}_{n,i}\left(\x, S_{\x, i, A}\right) &= P_{n,i}\left(\phi^{-1}_{n}(\x), S_{\phi^{-1}_{n}(\x), i, \phi_{n,i}^{-1}(A)}\right) = \int_{\phi_{n,i}^{-1}(A)} \pi_n\left(\d y_i \mid \phi_n^{-1}(\x)^{(-i)}\right) \\
&= \int_A \tilde{\pi}_n\left(\d y_i \mid \x^{(-i)}\right), \quad A \subset \sX_i,
\end{aligned}
\]
so that $\hat{P}_n$ is exactly the operator of the Gibbs sampler targeting $\tilde{\pi}_n$, i.e. $\tilde{P}_n = \hat{P}_n$.

Therefore, since $\phi_n$ is an injective transformation, by Corollary $2$ in \cite{R01} we have
\[
\lTV \mu_n P_n^t-\pi_n\rTV =\lTV \tilde{\mu}_n \tilde{P}_n^t-\tilde{\pi}_n\rTV,
\]
with $\tilde{\mu}_n = \mu_n \circ \phi_n^{-1}$. 
To conclude the proof, we show that $\tilde{\mu}_n \in \sN\left(\tilde{\pi}_n, M \right)$ if and only if $\mu_n \in \sN\left(\pi_n, M \right)$. Indeed, to prove the implication from right to left, by definition of push-forward measure we have
\[
\begin{aligned}
\tilde{\mu}_n(A) =& \mu_n \left(\phi_n^{-1}(A) \right) = \int_{\phi_n^{-1}(A)}\frac{\d \mu_n }{\d \pi_n}(\x) \, \pi_n(\d \x) \leq M \pi_n\left(\phi_n^{-1}(A) \right) = M\tilde{\pi}_n(A),
\end{aligned}
\]
for every set $A \subset \sX$. 
Equivalently we obtain the other implication.
\end{proof}

\subsection{Proof of Proposition \ref{prop:convergence_operators}}
For any $\pi \in \mathcal{P}(\sX)$ and $Q$ Markov transition kernel with state space $\sX$, we define $\left(\pi \otimes Q\right) \in \mathcal{P}\left(\sX \times \sX\right)$ as
\[
\left(\pi \otimes Q\right)(B) = \int_B Q(\x, \d \y)\pi(\d \x)
\]
for every $B \subseteq \sX \times \sX$. 
\begin{lemma}\label{inequality_composition}
Let $\pi_1 , \pi_2 \in \mathcal{P}(\sX)$ and $Q$ be a Markov transition kernel with state space $\sX$. Then
\[
\| \pi_1 \otimes Q - \pi_2 \otimes Q \|_{TV} = \|  \pi_1-\pi_2 \|_{TV}.
\]
\end{lemma}
\begin{proof}
By definition of total variation distance we have
\[
\begin{aligned}
&\| \pi_1 \otimes Q - \pi_2 \otimes Q \|_{TV}\\
&\begin{aligned}
 = \underset{f \, : \, \sX \times \sX \, \to \, [0,1]}{\sup} \biggl \vert \int_{\sX \times \sX} f(\x, \y)Q(\x, \d \y) \pi_1(\d \x) - \int_{\sX \times \sX} f(\x, \y)Q(\x, \d \y) \pi_2(\d \x) \biggr \vert
 \end{aligned} \\
 &\begin{aligned}
 = \underset{f \, : \, \sX \times \sX \, \to \, [0,1]}{\sup} \biggl \vert \int_{\sX}\left(\int_{\sX} f(\x, \y)Q(\x, \d \y)\right) \pi_1(\d \x) - \int_{\sX}\left(\int_{\sX} f(\x, \y)Q(\x, \d \y)\right) \pi_2(\d \x) \biggr \vert
 \end{aligned} \\
 &\begin{aligned}
 \leq \underset{g \, : \, \sX  \, \to \, [0,1]}{\sup} \biggl \vert \int_{\sX } g(\x) \pi_1(\d \x) - \int_{\sX} g(\x)\pi_2(\d \x) \biggr \vert = \|  \pi_1-\pi_2 \|_{TV}.
 \end{aligned} \\
\end{aligned}
\]
Also, taking $f(\x,\y)=g(\x)$ for every $(\x,\y)\in \sX \times \sX$ we have
\begin{align*}
\|  \pi_1-\pi_2 \|_{TV}
=&
\underset{g \, : \, \sX  \, \to \, [0,1]}{\sup} \biggl \vert \int_{\sX } g(\x) \pi_1(\d \x) - \int_{\sX} g(\x)\pi_2(\d \x) \biggr \vert
\\\leq&
\underset{f \, : \, \sX \times \sX \, \to \, [0,1]}{\sup} \biggl \vert \int_{\sX \times \sX} f(\x, \y)Q(\x, \d \y) \pi_1(\d \x) - \int_{\sX \times \sX} f(\x, \y)Q(\x, \d \y) \pi_2(\d \x) \biggr \vert
\\ =&
\| \pi_1 \otimes Q - \pi_2 \otimes Q \|_{TV}\,.
\end{align*}
\end{proof}
For $j=1,2$, denote the kernel of the Gibbs sampler targeting $\pi_j$ as 
$P_j = P_{j,1} \dots P_{j,K}$, where
\[
P_{j,i}\left(\x, S_{\x, i, A}\right) = \int_A \pi_j\left(\d y_i \mid \x^{(-i)}\right), \quad A \subset \sX_i,
\]
with $S_{\x, i, A} = \left\{\y \in \sX \, : \, y_j = x_j \, \forall  \, j \neq i \text{ and } y_i \in A \right\}$ as in the main. 
By definition, $P_i(\x, \d \y)$ depends only on $\x^{(-i)}$. 
Thus we can define $\left(\pi^{(-i)} \otimes Q\right) \in \mathcal{P}\left(\sX^{(-i)} \times \sX\right)$ as
\[
\left(\pi^{(-i)} \otimes P_i\right)(B) = \int_B P_i\left(\x^{(-i)}, \d \y\right)\pi\left(\d \x^{(-i)}\right),
\]
for every $B \subset \sX^{(-i)} \times \sX$ and similarly for 
\[
\left(\pi^{(-1)} \otimes P\right)\in \mathcal{P}\left(\sX^{(-1)} \times \sX\right) \quad \text{and} \quad \left(\pi^{(-i)} \otimes \prod_{j \geq i}P_j\right)\in \mathcal{P}\left(\sX^{(-i)} \times \sX\right),
\]
with $i= 1\, \dots, K$. Given this notation we have the following Lemmas.
\begin{lemma}\label{inequality_product}
We have
\[
\lTV \mu P_1-\mu P_2 \rTV \leq M\lTV \pi_2^{(-1)} \otimes P_1-\pi_2^{(-1)} \otimes P_2 \rTV
\]
for every $\mu \in \sN(\pi_2, M)$ and $M \geq 1$.
\end{lemma}
\begin{proof}
By definition of total variation distance 
\[
\lTV \mu P_1 - \mu P_2 \rTV = \underset{f \, : \, \sX \, \to \, [0,1]}{\sup} \left \lvert \int_{\sX} f(\y) \mu P_1(\d \y) - \int_{\sX} f(\y) \mu P_2(\d \y)  \right \rvert.
\]
Then, by definition of $\sN(\pi_2, M)$, it holds
\[
\begin{aligned}
\| & \mu P_1 - \mu P_2 \|_{TV}\\
&\begin{aligned}
 = M\underset{f \, : \, \sX \, \to \, [0,1]}{\sup} \biggl \vert \int_{\sX^K} \frac{f(\y)}{M}\int_{\sX^{(-1)}}&\frac{\d \mu^{(-1)} }{\d \pi_2^{(-1)}}(\x^{(-1)}) P_1(\x^{(-1)}, \d \y) \pi_2\left(\d \x^{(-1)}\right) \\
 &- \int_{\sX} \frac{f(\y)}{M}\int_{\sX^{(-1)}}\frac{\d \mu^{(-1)} }{\d \pi_2^{(-1)}}(\x^{(-1)}) P_2(\x^{(-1)}, \d \y) \pi_2\left(\d \x^{(-1)}\right) \biggr \vert
 \end{aligned} \\
&
\begin{aligned}
 \leq M\underset{g \, : \, \sX^{(-1)} \times \sX \, \to \, [0,1]}{\sup} \biggl \lvert \int_{\sX^{(-1)} \times \sX} g(\x^{(-1)}, \y) &P_1\left( \x^{(-1)}, \d \y \right) \pi_2\left(\d \x^{(-1)}\right) \\- &\int_{\sX^{(-1)} \times \sX}  g(\x^{(-1)}, \y) P_2\left( \x^{-1}, \d \y\right) \pi_2\left(\d \x^{(-1)}\right)  \biggr \rvert
 \end{aligned}\\
& = M\lTV \pi_2^{(-1)} \times P_1-\pi_2^{(-1)} \otimes P_2 \rTV.
\end{aligned}
\]
\end{proof}
\begin{lemma}\label{inequality_chain}
We have
\begin{equation}
\begin{aligned}\label{eq:inequality_chain}
\lTV \pi_1^{(-i)} \otimes \prod_{j \geq i}P_{1,j}- \pi_2^{(-i)} \otimes \prod_{j \geq i}P_{2,j} \rTV& \leq 2\lTV \pi_1-\pi_2 \rTV\\
&+\lTV \pi_1^{(-(i+1))} \otimes \prod_{j \geq i+1}P_{1,j}-\pi_2^{(-(i+1))} \otimes \prod_{j \geq i+1}P_{2,j} \rTV
\end{aligned}
\end{equation}
for every $i = 1, \dots, K-1$ and
\[
\lTV \pi_1^{(-K)} \otimes P_{1,K}- \pi_2^{(-K)} \otimes P_{2,K} \rTV = \lTV \pi_1-\pi_2 \rTV.
\]
\end{lemma}
\begin{proof}
We start by proving \eqref{eq:inequality_chain}. Notice that, by definition of $P_{1,i}$ and $P_{2,i}$, we have
\[
\begin{aligned}
\int_{\sX^{(-i)}\times \sX}&g\left(\x^{(-i)}, \y \right)\prod_{j \geq i}P_{1,j}\left(\x^{(-i)}, \d \y\right)\pi_1^{(-i)}\left(\d \x^{(-i)} \right) \\
&= \int_{\sX\times \sX^{(-i)}}h\left(\x, \y^{(-i)} \right)\prod_{j \geq i+1}P_{1,j}\left(\x^{(-i-1)}, \d \y\right)\pi_1\left(\d \x \right)
\end{aligned}
\]
and
\[
\begin{aligned}
\int_{\sX^{(-i)}\times \sX}&g\left(\x^{(-i)}, \y \right)\prod_{j \geq i}P_{2,j}\left(\x^{(-i)}, \d \y\right)\pi_2^{(-i)}\left(\d \x^{(-i)} \right) \\
&= \int_{\sX\times \sX^{(-i)}}h\left(\x, \y^{(-i)} \right)\prod_{j \geq i+1}P_{2,j}\left(\x^{(-i-1)}, \d \y\right)\pi_2\left(\d \x \right),
\end{aligned}
\]
where $g \, : \, \sX^{(-i)}\times \sX \, \to \, \R$ is any measurable function and $h$ is the composition of $g$ and the function $c \, : \, \sX^{(-i)}\times \sX \, \to \, \sX \times \sX^{(-i)}$ that relocates the $(K-1+i)$-th element of a vector after the $(i-1)$-th element. Since there is a one-to-one relationship between functions $g$ and $h$, we have
\begin{equation}\label{first_chain}
\begin{aligned}
&\lTV \pi_1^{(-i)} \otimes \prod_{j \geq i}P_{1,j}- \pi_2^{(-i)} \otimes \prod_{j \geq i}P_{2,j} \rTV = \lTV \pi_1 \otimes \prod_{j \geq i+1}P_{1,j}- \pi_2 \otimes \prod_{j \geq i+1}P_{2,j} \rTV.
\end{aligned}
\end{equation}
Then by triangular inequality and Lemma \ref{inequality_composition} we have
\begin{equation}\label{second_chain}
\begin{aligned}
\lTV \pi_1 \otimes \prod_{j \geq i+1}P_{1,j}- \pi_2 \otimes \prod_{j \geq i+1}P_{2,j} \rTV
\leq&  \lTV \pi_1 \otimes \prod_{j \geq i+1}P_{1,j}- \pi_2 \otimes \prod_{j \geq i+1}P_{1,j} \rTV\\
&+\lTV \pi_2 \otimes \prod_{j \geq i+1}P_{1,j}- \pi_2 \otimes \prod_{j \geq i+1}P_{2,j} \rTV\\
& \leq
\lTV\pi_1-\pi_2\rTV+
\lTV \pi_2 \otimes \prod_{j \geq i+1}P_{1,j}- \pi_2 \otimes \prod_{j \geq i+1}P_{2,j} \rTV\,.
\end{aligned}
\end{equation}
Notice that $\prod_{j \geq i+1}P_{1,j}$ and $\prod_{j \geq i+1}P_{2,j}$ do not depend on $x_{i+1}$ by construction, that implies
\begin{equation*}
\begin{aligned}
& \lTV \pi_2 \otimes \prod_{j \geq i+1}P_{1,j}- \pi_2 \otimes \prod_{j \geq i+1}P_{2,j} \rTV\\
&\\
&
\begin{aligned}
\quad = \underset{h \, : \, \sX \times \sX  \, \to \, [0,1]}{\sup} \biggl \lvert \int_{\sX\times \sX}h\left(\x, \y \right)&\prod_{j \geq i+1}P_{1,j}\left(\x^{(-(i+1))}, \d \y\right)\pi_2\left(\d \x \right)\\
&-\int_{\sX\times \sX}h\left(\x, \y \right)\prod_{j \geq i+1}P_{2,j}\left(\x^{(-(i+1))}, \d \y\right)\pi_2\left(\d \x \right) \biggr \rvert,
\end{aligned}
\end{aligned}
\end{equation*}
so that we have
\begin{equation*}
\begin{aligned}
& \lTV \pi_2 \otimes \prod_{j \geq i+1}P_{1,j}- \pi_2 \otimes \prod_{j \geq i+1}P_{2,j} \rTV\\
&\\
& \begin{aligned}
\quad = \underset{h \, : \, \sX \times \sX  \, \to \, [0,1]}{\sup} \biggl \lvert& \int_{\sX^{(-(i+1))}\times \sX}\int_{\sX_{i+1}}h\left(\x, \y \right)\pi_2\left(\d x_{i+1}\mid x^{(-(i+1))}\right)\prod_{j \geq i+1}P_{1,j}\left(\x^{(-(i+1))}, \d \y\right)\pi_2\left(\d \x^{(-(i+1))} \right)\\
&-\int_{\sX^{(-(i+1))}\times \sX}\int_{\sX_{i+1}}h\left(\x, \y \right)\pi_2\left(\d x_{i+1}\mid x^{(-(i+1))}\right)\prod_{j \geq i+1}P_{2,j}\left(\x^{(-(i+1))}, \d \y\right)\pi_2\left(\d \x^{(-(i+1))} \right) \biggr \rvert
\end{aligned}\\
& \quad \leq  \lTV \pi_2^{(-(i+1))} \otimes \prod_{j \geq i+1}P_{1,j}- \pi_2^{(-(i+1))} \otimes \prod_{j \geq i+1}P_{2,j} \rTV.
\end{aligned}
\end{equation*}
Moreover, it is clear that 
\[
\lTV \pi_2^{(-(i+1))} \otimes \prod_{j \geq i+1}P_{1,j}- \pi_2^{(-(i+1))} \otimes \prod_{j \geq i+1}P_{2,j} \rTV
\leq
\lTV \pi_2 \otimes \prod_{j \geq i+1}P_{1,j}- \pi_2 \otimes \prod_{j \geq i+1}P_{2,j} \rTV,
\]
thus combining the two above inequalities we get
\begin{equation}\label{third_chain}
\lTV \pi_2 \otimes \prod_{j \geq i+1}P_{1,j}- \pi_2 \otimes \prod_{j \geq i+1}P_{2,j} \rTV
=
\lTV \pi_2^{(-(i+1))} \otimes \prod_{j \geq i+1}P_{1,j}- \pi_2^{(-(i+1))} \otimes \prod_{j \geq i+1}P_{2,j} \rTV\,.
\end{equation}
Combining \eqref{first_chain}, \eqref{second_chain} and \eqref{third_chain} with the fact that
\begin{equation*}
\begin{aligned}
\lTV \pi_2^{(-(i+1))} \otimes \prod_{j \geq i+1}P_{1,j}- \pi_2^{(-(i+1))} \otimes \prod_{j \geq i+1}P_{2,j} \rTV
\leq&
\lTV\pi_1-\pi_2\rTV\\
&+
\lTV \pi_1^{(-(i+1))} \otimes \prod_{j \geq i+1}P_{1,j}- \pi_2^{(-(i+1))} \otimes \prod_{j \geq i+1}P_{2,j} \rTV\,
\end{aligned}
\end{equation*}
we finally obtain \eqref{eq:inequality_chain}. When $i = K$ the result follows by noticing that
\[
\pi_1^{(-K)} \otimes P_{1,K} = \pi_1 \quad \text{and} \quad \pi_2^{(-K)} \otimes P_{2,K} = \pi_2
\]
by definition.
\end{proof}
\begin{proof}[Proof of Proposition \ref{prop:convergence_operators}]
Without loss of generality, let $\mu \in \sN(\pi_2, M)$. By Lemma \ref{inequality_product} and the triangle inequality we have
\[
\begin{aligned}
\lTV \mu P_1 - \mu P_2 \rTV &\leq M\lTV \pi_2^{(-1)} \otimes P_1 - \pi_2^{(-1)} \otimes P_2 \rTV\\
& \leq M\lTV\pi_1-\pi_2\rTV+M\lTV \pi_1^{(-1)} \otimes P_1 - \pi_2^{(-1)} \otimes P_2 \rTV
\end{aligned}
\]
and the result follows by applying $K$ times Lemma \ref{inequality_chain}.
\end{proof}

\subsection{Proof of Lemma \ref{constructive_lemma}}

\begin{proof}
With an abuse of notation, let $\pi_1(x)$, $\pi_2(x)$ and $\mu_1(x)$ be densities of $\pi_1$, $\pi_2$ and $\mu_1$ with respect to a common dominating measure, such as $\tau=\pi_1+\pi_2$. 
Let $\bar{\mu}$ be the measure on $\sX$ with density $\bar{\mu}(x) = \min \left\{\mu_1(x), M\pi_2(x) \right\} $ for $x\in\sX$. 
By construction $\bar{\mu}$ is a sub-probability since
\[
\bar{\mu}(\sX)= \int_\sX \bar{\mu}(x) \tau(\d x) \leq \int_\sX \mu_1(x)\tau(\d x) = 1.
\]
Therefore, we can define a probability distribution $\mu_2\in\sP(\sX)$ with density
\begin{align*}
\mu_2(x) &= \bar{\mu}(x) + \alpha \, \max \left\{M \pi_2(x)-\mu_1(x), 0 \right\},
&x\in\sX
\end{align*}
where
\[ \alpha 
= 
\frac{1-\int \bar{\mu}(x)\tau(\d x)
}{
\int_\sX \max \left\{M \pi_2(x)-\mu_1(x), 0 \right\}\tau( \d x)} \in (0,1).
\]
Notice that $\mu_2(x)\leq M\pi_2(x)$ for every $x \in \sX$ since
\[
\mu_2(x) = 
\begin{cases}
M\pi_2(x), \quad \text{if } \mu_1(x) > M\pi_2(x)\,,\\
(1-\alpha)\mu_1(x)+\alpha M\pi_2(x), \quad \text{if } \mu_1(x) \leq M\pi_2(x)\,.
\end{cases}
\]
Thus $\mu_2 \in \sN(\pi_2, M)$. By definition of total variation distance and of $\tilde{\mu}$, we have
\[
\begin{aligned}
\lTV \mu_1- \mu_2 \rTV 
=& 
\int_\sX \max \left\{\mu_1(x) - \mu_2(x), 0 \right\} \tau(\d x) 
=
\int_\sX \max \left\{\mu_1(x)- M\pi_2(x), 0 \right\} \tau(\d x) 
\\\leq&
 M \, \int_\sX \max \left\{\pi_1(x) - \pi_2(x) ,0\right\}\tau(\d x) 
 = M\lTV \pi_1- \pi_2 \rTV.
\end{aligned}
\]
\end{proof}

\subsection{Proof of Theorem \ref{thm:mixing_limit}}
\begin{proof}
  By Lemma \ref{equivalent_statement} the statement is equivalent to
\begin{equation}\label{eq:sups_tilde}
\lim_{n \to \infty} \, \sup_{\tilde{\mu}_n\in \sN \left(\tilde{\pi}_n, M \right)} \, 
\lTV \tilde{\mu}_n \tilde{P}_n^t-\tilde{\pi}_n\rTV 
=
\sup_{\tilde{\mu} \in \sN \left(\tilde{\pi}, M \right)} \, \lTV \tilde{\mu} \tilde{P}^t-\tilde{\pi} \rTV
\end{equation}
in $Q^{(n)}$-probability, where $\tilde{P}_n$ is the kernel of the Gibbs sampler targeting $\tilde{\pi}$.

Consider $\| \tilde{\mu}_n \tilde{P}_n^t-\tilde{\pi}_n\|_{TV}$ with 
$\tilde{\mu}_n \in \sN(\tilde{\pi}_n, M)$. 
By Lemma \ref{constructive_lemma}, there exists $\tilde{\mu} \in \sN (\tilde{\pi}, M)$ such that 
\begin{equation}\label{eq:mu_n_close_mu}
\lTV \tilde{\mu}_n - \tilde{\mu} \rTV 
\leq
 M\lTV \tilde{\pi}_n - \tilde{\pi} \rTV.
\end{equation}
By the triangular inequality we can decompose $\| \tilde{\mu}_n \tilde{P}_n^t-\tilde{\pi}_n\|_{TV}$ as follows
\begin{equation}\label{eq:first_step_A2}
\lTV \tilde{\mu}_n\tilde{P}^{t}_n - \tilde{\pi}_n \rTV \leq
\lTV \tilde{\mu}_n\tilde{P}^{t}_n - \tilde{\mu}\tilde{P}_n^{t} \rTV+\lTV \tilde{\mu}\tilde{P}^{t}_n - \tilde{\mu}\tilde{P}^{t} \rTV
+\lTV \tilde{\mu}\tilde{P}^{t} - \tilde{\pi} \rTV
+\lTV \tilde{\pi}_n - \tilde{\pi} \rTV
.
\end{equation}
Combining \eqref{eq:mu_n_close_mu} with 
the monotonicity of the total variation distance with respect to the application of transition kernels, we obtain
\begin{equation}\label{computations}
\lTV \tilde{\mu}_n\tilde{P}^{t}_n - \tilde{\mu}\tilde{P}_n^{t} \rTV
\leq \lTV \tilde{\mu}_n-\tilde{\mu} \rTV
\leq M\lTV \tilde{\pi}_n-\tilde{\pi} \rTV\,.
\end{equation}
For the second term in \eqref{eq:first_step_A2}, we want to prove that if $\tilde{\mu} \in \sN (\tilde{\pi}, M)$ we have
\begin{equation}\label{main_bound}
\lTV \tilde{\mu} \tilde{P}_n^t - \tilde{\mu} \tilde{P}^t  \rTV \leq 2MKt\lTV \tilde{\pi}_n- \tilde{\pi}  \rTV
\end{equation}
for every $t \geq 1$. Indeed, the case $t = 1$ holds by Proposition \ref{prop:convergence_operators}. Assume now \eqref{main_bound} holds for $t-1$, with $t \geq 2$. Then by the triangular inequality we have
\[
\begin{aligned}
\lTV \tilde{\mu} \tilde{P}_n^t - \tilde{\mu} \tilde{P}^t  \rTV &\leq \lTV \tilde{\mu} \tilde{P_n}^t - \mu \tilde{P}^{t-1}\tilde{P}_n \rTV + \lTV \mu \tilde{P}^t - \mu \tilde{P}^{t-1}\tilde{P}_n \rTV\\
& \leq \lTV \tilde{\mu} \tilde{P}_n^{t-1} - \tilde{\mu} \tilde{P}^{t-1}  \rTV + \lTV \mu \tilde{P}^{t-1}\tilde{P} - \mu \tilde{P}^{t-1}\tilde{P}_n \rTV.
\end{aligned}
\]
By induction hypothesis we have
\begin{equation}\label{first_induction}
\lTV \tilde{\mu} \tilde{P}_n^{t-1} - \tilde{\mu} \tilde{P}^{t-1}  \rTV \leq 2MK(t-1)\lTV \tilde{\pi}_n - \tilde{\pi} \rTV.
\end{equation}
Moreover, by Lemma \ref{lemma_closure} we have that $\tilde{\mu} \tilde{P}^{t-1} \in \sN(\tilde{\pi}, M)$, so that from the case $t = 1$ we obtain
\begin{equation}\label{second_induction}
\lTV \mu \tilde{P}^{t-1}\tilde{P} - \mu \tilde{P}^{t-1}\tilde{P}_n \rTV \leq 2MK\lTV \tilde{\pi}_n - \tilde{\pi} \rTV.
\end{equation}
Then \eqref{main_bound} follows by \eqref{first_induction} and \eqref{second_induction}.
Combining \eqref{eq:first_step_A2}, \eqref{computations} and \eqref{main_bound}, for every $\tilde{\mu}_n \in \sN (\tilde{\pi}_n, M)$ there exists $\tilde{\mu} \in \sN (\tilde{\pi}, M)$ such that 
\begin{equation*}
\lTV \tilde{\mu}_n\tilde{P}^{t}_n - \tilde{\pi}_n \rTV \leq
(2MKt+M+1)
\lTV \tilde{\pi}_n- \tilde{\pi}  \rTV
  +\lTV \tilde{\mu}\tilde{P}^{t} - \tilde{\pi}_n \rTV.
\end{equation*}
Thus
\[
\begin{aligned}
\underset{\tilde{\mu}_n \in \sN (\tilde{\pi}_n, M)}{\sup} \lTV \tilde{\mu}_n\tilde{P}^{t}_n - \tilde{\pi}_n \rTV 
& \leq (2MKt+M+1)\lTV \tilde{\pi}_n- \tilde{\pi}  \rTV+\underset{\tilde{\mu} \in \sN (\tilde{\pi}, M)}{\sup} \lTV \tilde{\mu}\tilde{P}^{t} - \tilde{\pi} \rTV.
\end{aligned}
\]
It follows that, for any $\epsilon > 0$, we have
\begin{equation}\label{limsup}
\begin{aligned}
Q^{(n)}&\left(
\underset{\tilde{\mu}_n \in \sN (\tilde{\pi}_n, M)}{\sup} \lTV \tilde{\mu}_n\tilde{P}^{t}_n - \tilde{\pi}_n \rTV 
-
 \underset{\tilde{\mu} \in \sN (\tilde{\pi}, M)}{\sup} \lTV \tilde{\mu}\tilde{P}^{t} - \tilde{\pi} \rTV \geq \epsilon \right)\\
& \leq Q^{(n)}\left(\lTV \tilde{\pi}_n- \tilde{\pi}  \rTV \geq (2MKt+M+1)^{-1}\epsilon \right) \to 0,
\end{aligned}
\end{equation}
as $n \to \infty$ by $(A1)$ and $(2MKt+M+1)^{-1}\epsilon>0$.

We now prove the reverse inequality of \eqref{limsup} to establish \eqref{eq:sups_tilde}.
Given $\tilde{\mu} \in \sN (\tilde{\pi}, M)$, by Lemma \ref{constructive_lemma}, there exists $\tilde{\mu}_n \in \sN (\tilde{\pi}_n, M)$ such that 
$\lTV \tilde{\mu} - \tilde{\mu}_n \rTV\leq M\lTV \tilde{\pi}_n- \tilde{\pi} \rTV$.
Then we proceed analogously to above, first decomposing $\lTV \tilde{\mu}\tilde{P}^{t} - \tilde{\pi} \rTV$ 
as
\begin{equation}\label{second1}
\begin{aligned}
\lTV \tilde{\mu}\tilde{P}^{t} - \tilde{\pi} \rTV& 
\leq
\lTV \tilde{\mu}\tilde{P}^{t} - \tilde{\mu}\tilde{P}_n^{t} \rTV 
+
\lTV \tilde{\mu}\tilde{P}_n^{t} - \tilde{\mu}_n\tilde{P}_n^{t} \rTV 
+ 
\lTV \tilde{\mu}_n\tilde{P}_n^{t} -\tilde{\pi}_n \rTV
+
\lTV \tilde{\pi}_n -\tilde{\pi} \rTV
\end{aligned}
\end{equation}
and then applying Proposition \ref{prop:convergence_operators} using an argument analogous to above to get
\[
\lTV \tilde{\mu}\tilde{P}^{t} - \tilde{\pi} \rTV \leq  \lTV \tilde{\mu}_n\tilde{P}_n^{t} - \tilde{\pi}_n \rTV + (2MKt+M+1)\lTV \tilde{\pi}_n - \tilde{\pi} \rTV\,.
\]
It follows
\[
\sup_{\mu_n\in \sN (\tilde{\pi}_n, M)} \,\lTV \tilde{\mu}_n\tilde{P}_n^{t} - \tilde{\pi}_n \rTV 
\geq 
\sup_{\tilde{\mu} \in \sN (\tilde{\pi}, M)} \,\lTV \tilde{\mu}\tilde{P}^{t} - \tilde{\pi} \rTV -(2MKt+M+1)\lTV \tilde{\pi}_n - \tilde{\pi} \rTV.
\]
Fixing $\epsilon > 0$ arbitrary constant we have
\begin{equation}\label{liminf}
\begin{aligned}
Q^{(n)}&\left(
\underset{\tilde{\mu}_n \in \sN (\tilde{\pi}_n, M)}{\sup} \lTV \tilde{\mu}_n\tilde{P}^{t}_n - \tilde{\pi}_n \rTV 
-
 \underset{\tilde{\mu} \in \sN (\tilde{\pi}, M)}{\sup} \lTV \tilde{\mu}\tilde{P}^{t} - \tilde{\pi} \rTV \leq -\epsilon \right)\\
& \leq Q^{(n)}\left(\lTV \tilde{\pi}_n- \tilde{\pi}  \rTV \geq \frac{\epsilon}{2MKt+M+1} \right) \to 0,
\end{aligned}
\end{equation}
as $n \to \infty$ by $(A1)$ and $(2MKt+M+1)^{-1}\epsilon>0$. The result follows by combining \eqref{limsup} and \eqref{liminf}.
\end{proof}

\subsection{Proof of Corollary \ref{mixingCorollary}}
\begin{proof}
Thanks to Lemma \ref{equivalent_representation} we can write
\[
t^{(n)}_{mix}(\epsilon, M) = \inf\left\{t\geq 1\,:\,  \sup_{\mu_n \in \sN (\pi_n, M)} \,\lTV \mu_n P_n^t-\pi_n \rTV<\epsilon \right\}
\]
and
\[
\tilde{t}_{mix}(\epsilon, M) = \inf\left\{t\geq 1\,:\,  \sup_{\tilde{\mu} \in \sN (\tilde{\pi}, M)} \,\lTV \tilde{\mu} \tilde{P}^t-\tilde{\pi} \rTV<\epsilon \right\}.
\]
Assume $(A1)$ and denote $t^* = \tilde{t}_{mix}(\epsilon, M) < \infty$ for brevity. By definition of $t^*$ we have $\delta = \sup_{\tilde{\mu} \in \sN (\tilde{\pi}, M)} \,\lTV \tilde{\mu} \tilde{P}^{t^*}-\tilde{\pi} \rTV < \epsilon$. Thus
\[
\begin{aligned}
Q^{(n)}&\left(t^{(n)}_{mix}(\epsilon, M)\leq t^*\right)= Q^{(n)}\left(\sup_{\mu_n \in \sN (\pi_n, M)} \,\lTV \mu_n P_n^{t^*}-\pi_n \rTV<\epsilon \right)\\
&= Q^{(n)}\left(\sup_{\mu_n \in \sN (\pi_n, M)} \,\lTV \mu_n P_n^{t^*}-\pi_n \rTV-\sup_{\tilde{\mu} \in \sN (\tilde{\pi}, M)} \,\lTV \tilde{\mu} \tilde{P}^{t^*}-\tilde{\pi} \rTV<\epsilon-\delta \right)\\
& \to 1,
\end{aligned}
\]
as $n \to \infty$ by Theorem \ref{thm:mixing_limit}.

As regards the second part of the statement, let $(A1)$ hold and fix $T > 0$. Denote $\delta = \sup_{\tilde{\mu} \in \sN (\tilde{\pi}, M)} \,\lTV \tilde{\mu} \tilde{P}^{T}-\tilde{\pi} \rTV$ and notice that by assumption $\delta \geq \epsilon > \underline{\epsilon}$. Thus
\[
\begin{aligned}
\lim \inf_{n \to \infty} Q^{(n)}\left(t^{(n)}_{mix}(\underline{\epsilon}, M)<T\right) &= \lim \inf_{n \to \infty} Q^{(n)}\left(\sup_{\mu_n \in \sN (\pi_n, M)} \,\lTV \mu_n P_n^{T}-\pi_n \rTV < \underline{\epsilon}\right)\\
&= \lim \inf_{n \to \infty} Q^{(n)}\left(\delta-\sup_{\mu_n \in \sN (\pi_n, M)} \,\lTV \mu_n P_n^{T}-\pi_n \rTV \geq \delta-\underline{\epsilon} \right)\\
&\to 0,
\end{aligned}
\]
as $n \to \infty$ by Theorem \ref{thm:mixing_limit}.
\end{proof}

\subsection{Proof of Corollary \ref{mixing_gap}}
We need a preliminary well known lemma, whose proof we include for self-containedness.
\begin{lemma}\label{positiveDefinite_Gibbs}
Let $P$ be a Gibbs sampler kernel with $K=2$ and target $\pi\in\mathcal{P}(\sX_1\times\sX_2)$. Then
\[
\lTV \mu P^t - \pi \rTV \leq \frac{M}{2}\left(1-\text{Gap}(P)\right)^{t},
\]
for every $\mu \in \mathcal{N}(\pi, M)$ and $t \geq 1$.
\end{lemma}
\begin{proof}
Let $\mu \in \mathcal{N}(\pi, M)$ and $t\geq 1$. By Corollary $1$ in \cite{R01} we have
\begin{equation}\label{equality_deinitializing}
\lTV \mu P^t - \pi \rTV = \lTV \mu^{(-1)}\hat{P}^t - \pi^{(-1)} \rTV,
\end{equation}
where $\hat{P}$ is the Markov transition kernel on $\sX_2$ defined as
\begin{align*}
\hat{P}(x_2, \d y_2)& = \int_{\sX_1}\pi(\d y_2 \mid y_1)\pi(\d y_1 \mid x_2)
&x_2\in\sX_2\,.
\end{align*}
Note that $\hat{P}$ is $\pi^{(-1)}$-reversible.
Also, for every $f \in L^2(\pi^{(-1)})$, i.e.\ $f \, : \, \sX_2 \, \to \, \R$ such that $ \| f  \|_2^2 =\pi^{(-1)}(f^2)$ is finite, we have
\[
\begin{aligned}
\int_{\sX_2^2}&f(x_2)f(y_2)\hat{P}(x_2 ,\d y_2) \pi(\d x_2) \\
&= \int_{\sX_2^2}f(x_2)f(y_2)\int_{\sX_1}\pi(\d y_2 \mid y_1)\pi(\d y_1 \mid x_2)\pi(\d x_2)\\
& = \int_{\sX_1} \left[\int_{\sX_2}f(y_2)\pi(\d y_2 \mid y_1) \right]\left[\int_{\sX_2}f(x_2)\pi(\d x_2 \mid y_1) \right]\pi(\d y_1)\\
& = \int_{\sX_1} \left[\int_{\sX_2}f(y_2)\pi(\d y_2 \mid y_1) \right]^2\pi(\d y_1) \geq 0,
\end{aligned}
\]
so that $\hat{P}$ is also positive semi-definite. 
Since $\hat{P}$ is reversible and positive semi-definite, we have (see e.g.\ equation (5) in \cite{AL22}) that
\begin{equation}\label{bound_gap_l2}
\left \lvert \left \lvert \hat{P}^t(f) \right \rvert \right \rvert_2 \leq \left \lvert \left \lvert f \right \rvert \right \rvert_2\left(1-\text{Gap}(\hat{P}) \right)^{t},
\end{equation}
for every $f $ such that $\pi(f) = 0$. 
Choosing $f = \frac{\d \mu^{(-1)}}{\d \pi^{(-1)}}-1$ and using the reversibility of $\hat{P}$ (see e.g.\ Section $2.1$ in \cite{K09}) we also have
\begin{equation}\label{chisquare_bound}
\lTV \mu^{(-1)}\hat{P}^t - \pi^{(-1)} \rTV \leq \frac{1}{2}\left \lvert \left \lvert \mu^{(-1)} \hat{P}^t(f)  \right \rvert \right \rvert_2,
\end{equation}
where $\mu^{(-1)} \hat{P}^t(f) = \int f(x_2) \mu^{(-1)} \hat{P}^t(\d x_2)$.
With the same choice of $f$, we have
\[
\left \lvert \left \lvert f \right \rvert \right \rvert_2^2 = \int \left(\frac{\d \mu^{(-1)}}{\d \pi^{(-1)}}(x_2)-1 \right)^2\pi^{-1}(\d x_2) \leq M^2
\]
since $\mu^{(-1)} \in \mathcal{N}(\pi^{(-1)}, M)$. Thus, combining \eqref{bound_gap_l2} with \eqref{chisquare_bound} we obtain
\[
\lTV \mu P^t - \pi \rTV \leq \frac{M}{2}\left(1-\text{Gap}(\hat{P}) \right)^{t}.
\]
Finally, for every $f \, : \, \sX_2 \, \to \, \R$ with $\left \lvert \left \lvert f \right \rvert \right \rvert_2 < \infty$ it holds
\[
\frac{\int_{\sX_2^2}\left[f(y_2)-f(x_2) \right]^2\pi(\d x_2)\hat{P}(x_2, \d y_2)}{2\text{Var}_\pi^{(-1)}(f)} = \frac{\int_{\sX^2}\left[g(\y)-g(\x) \right]^2\pi(\d \x)P(\x, \d \y)}{2\text{Var}_\pi(f)},
\]
where $g(\x) = f(x_2)$. Therefore $\text{Gap}(\hat{P}) \geq \text{Gap}(P)$ and we get
\[
\lTV \mu P^t - \pi \rTV \leq \frac{M}{2}\left(1-\text{Gap}(P)\right)^{t},
\]
as desired.
\end{proof}

\begin{proof}[Proof of Corollary \ref{mixing_gap}]
By Lemma \ref{positiveDefinite_Gibbs} we obtain
\[
 \tilde{t}_{mix}(\epsilon, M) \leq 1+\frac{\log(M/2)-\log(\epsilon)}{-\log \left(1-\text{Gap}(\tilde{P}) \right)},
\]
and the result follows by the first part of Corollary \ref{mixingCorollary}.
\end{proof}

\subsection{Proof of Proposition \ref{convergence_finite_dim}}
\begin{proof}
By Theorem \ref{BvM}, assumption $(A1)$ is satisfied with
\[
\phi_{n} (\psi) = \sqrt{n}(\psi -\psi^*) - \mathcal{I}^{-1}(\psi^*)\Delta_{n, \psi^*},
\]
and $\tilde{\pi} = N\left( \bm{0}, \mathcal{I}^{-1}(\psi^*)\right)$. Since $\tilde{\pi}$ is the distribution of a multivariate normal with non singular covariance matrix, then it is easy to show $\tilde{t}_{mix}(\epsilon, M) < \infty$ for every $(M, \epsilon) \in [1, \infty)\times(0,1)$, see e.g. Theorem $2$ in \cite{A91}. 
\end{proof}

\subsection{Statement and proof of Corollary \ref{convergence_normal_likelihood}}
We illustrate the result of Proposition \ref{convergence_finite_dim} on a simple example of model \eqref{parametric_model} with normal likelihood and unknown mean and precision, that is
\begin{equation}\label{eq:convergence_normal_likelihood}
f(y \mid \mu, \tau) = N\left(y \mid \mu, \tau^{-1}\right),
\end{equation}
where $K = 2$ and $\psi = (\mu, \tau)$. 
Notice that, even if a conjugate prior exists, it is common to place independent priors on $\mu$ and $\tau$, for which the Gibbs sampler defined in \eqref{initialGibbs} becomes a reasonable option.
\begin{corollary}\label{convergence_normal_likelihood}
Consider model \eqref{parametric_model} with likelihood as in \eqref{eq:convergence_normal_likelihood}. Let $Y_i \simiid Q_{\psi^*}$, with $Q_{\psi^*}$ admitting density $f(y \mid \psi^*)$ and $\psi^* = (\mu^*, \tau^*)\in\R\times \R_+$. Moreover let $p_0$ be absolutely continuous in a neighborhood of $\psi^*$ with a continuous positive density at $\psi^*$. Consider the Gibbs sampler defined in \eqref{initialGibbs}. Then, for every $M \geq 1$ and $\epsilon > 0$ we have
\[
Q_{\psi^*}^{(n)}\left(t^{(n)}_{mix}(\epsilon, M)\leq 1\right) \to 1,
\]
as $n \to \infty$.
\end{corollary}
For the proof we need a preliminary Lemma, whose proof we include for self-containedness and because it will be useful to refer to later on. 
\begin{lemma}\label{tests_normal}
Consider the same setting of Corollary \ref{convergence_normal_likelihood}. Then conditions \eqref{tests} are satisfied.
\end{lemma}
\begin{proof}[Proof of Lemma \ref{tests_normal}] 
Define
\[
\Psi = \Psi_1 \times \Psi_2 = \left[\mu^* - 1, \mu^* + 1\right] \times \left[\frac{\tau^*}{2}, 2\tau^* \right]
\]
compact neighborhood of $\psi^*$ and
\[
u_n(Y_1, \dots, Y_n) = 1- \mathbbm{1}_{g_1(Y_{1:n}) \leq c_1} \, \mathbbm{1}_{g_2(Y_{1:n}) \leq c_2},
\]
where $c_1 = 1/2$, $c_2 = (2\tau^*)^{-1}$ and
\[
g_1(Y_{1:n}) = \left\lvert \bar{Y} - \mu^*\right \rvert, \quad \text{and} \quad g_2(Y_{1:n}) = \left \lvert \frac{1}{n}\sum_{i = 1}^n\left(Y_i-\bar{Y}\right)^2-\frac{1}{\tau^*} \right \rvert,
\]
with $\bar{Y} = \frac{1}{n}\sum_{i = 1}^nY_i$. Since $Y_i \simiid N(\mu, \tau^{-1})$, then $g_1(Y_{1:n})$ and $g_2(Y_{1:n})$ are equal in distribution, respectively, to
\[
h_1(Z_{1:n}, \mu, \tau) = \left \lvert \frac{1}{\sqrt{\tau}}\bar{Z}+\mu-\mu^* \right \rvert, \quad h_2(Z_{1:n}, \mu, \tau) = \left \lvert\frac{1}{\tau}\frac{1}{n}\sum_{i = 1}^n\left(Z_i-\bar{Z}\right)^2-\frac{1}{\tau^*} \right \rvert,
\]
where $Z_i \simiid N(0, 1)$. By the Law of Large numbers we have
\[
\bar{Z} \to 0, \quad \text{and} \quad \frac{1}{n}\sum_{i = 1}^n\left(Z_i-\bar{Z}\right)^2 \to 1
\]
almost surely as $n \to \infty$. This implies
\[
\begin{aligned}
\int u_n(y_1, \dots, y_n) \, \prod_{i = 1}^nf(\d y_i \mid \psi^*) \leq &P\left(h_1(Z_{1:n}, \mu^*, \tau^*)  > c_1 \right)\\
&+P\left(h_2(Z_{1:n}, \mu^*, \tau^*)  > c_2 \right) \to 0,
\end{aligned}
\]
as $n \to \infty$. Also, we have
\[
\begin{aligned}
\sup_{\psi \not \in \Psi} \,\int \left[1-u_n(y_1, \dots, y_n)\right] \, \prod_{i = 1}^nf(\d y_i \mid \psi)  \leq& \sup_{\tau \not \in \Psi_2} \, P\left(h_2(Z_{1:n}, \mu, \tau)  \leq c_2 \right)\\
&+ \sup_{\mu \not \in \Psi_1, \, \tau \in \Psi_2} \, P\left(h_1(Z_{1:n}, \mu, \tau)  \leq c_1 \right).
\end{aligned}
\]
Now notice that by the reverse triangle inequality we have
\[
\begin{aligned}
\sup_{\tau \not \in \Psi_2}& \, P\left(h_2(Z_{1:n}, \mu, \tau)  \leq c_2 \right)  = \sup_{\tau \not \in \Psi_2} \, P\left(\left \lvert\frac{1}{\tau}\frac{1}{n}\sum_{i = 1}^n\left(Z_i-\bar{Z}\right)^2-\frac{1}{\tau^*}\right\rvert  \leq c_2 \right)\\
&\leq \sup_{\tau \not \in \Psi_2} \,  P \left(\left \lvert\frac{1}{n}\sum_{i = 1}^n\left(Z_i-\bar{Z}\right)^2-1\right\rvert \geq \left \lvert 1-\frac{\tau}{\tau^*}\right \rvert-c_2\tau \right) \to 0,
\end{aligned}
\]
by definition of $\Psi_2$, as $n \to \infty$. Finally, again by reverse triangle inequality, we have
\[
\sup_{\mu \not \in \Psi_1, \, \tau \in \Psi_2} \, P\left(h_1(Z_{1:n}, \mu, \tau)  \leq c_1 \right) \leq \sup_{\mu \not \in \Psi_1, \, \tau \in \Psi_2} \, P\left(|\bar{Z} | \geq \sqrt{\tau}\left(|\mu - \mu^*| - c_1 \right) \right) \to 0,
\]
as $n \to \infty$.
\end{proof}
\begin{proof}[Proof of Corollary \ref{convergence_normal_likelihood}]
In this case  $\psi = (\mu, \tau)$ and
\[
f(y \mid \psi) = \sqrt{\frac{\tau}{2\pi}}e^{-\frac{\tau}{2}(y-\mu)^2}.
\]
By Lemma \ref{tests_normal} conditions \eqref{tests} are satisfied.
Also, the map $\psi \to f(y \mid \psi)$ is one-to-one, the map $\psi \to \sqrt{f(y \mid \psi)}$ is continuously differentiable, and 
  the Fisher information matrix is
\[
\Fisher(\psi) =
\begin{bmatrix}
\frac{\tau}{2} & 0\\
0 & \frac{1}{2\tau}
\end{bmatrix},
\]
which is non singular and continuous as a function of $\psi$. Thus the conditions of Theorem \ref{BvM} and Proposition \ref{convergence_finite_dim} are satisfied. Finally, since we are considering a two-blocks Gibbs sampler, by Corollary \ref{mixing_gap} we have
\[
T\left(\psi^*, \epsilon, M \right) = 1+\frac{\log(M/2)-\log(\epsilon)}{-\log \left(1-\text{Gap}(\tilde{P}) \right)},
\]
where $\tilde{P}$ is the Gibbs sampler targeting a bivariate normal distribution with covariance matrix given by $\Fisher^{-1}(\psi^*)$. Since the latter is diagonal, the Gibbs sampler coincides with independent sampling, so that $\text{Gap}(\tilde{P}) = 1$.
\end{proof}

\subsection{Proof of Lemma \ref{sufficient_lemma}}
\begin{proof}
Denote by $\left(\bm{\theta}^{(t)}, \psi^{(t)} \right)_{t \geq 1}$ the Markov chain with kernel $P_J$ defined in \eqref{two_blocks_gibbs_nested}. The Markovianity of the induced sequence $\left(\bT^{(t)}, \psi^{(t)} \right)_{t \geq 1}$ follows by the one of $\left(\psi^{(t)}\right)_{t \geq 1}$, which is well known \citep{D08, R01}. We now show that $\left(\bm{T}^{(t)}, \psi^{(t)} \right)_{t \geq 1}$ admits $\hat{P}_J$ as kernel. 
The conditional distribution of $\left(\bT^{(t)}, \psi^{(t)} \right)$ given $\left(\bT^{(t-1)}, \psi^{(t-1)} \right)$ is given by
\begin{align*}
\L\left(\d\bm{T}^{(t)}, \d\psi^{(t)} \mid \bm{T}^{(t-1)}, \psi^{(t-1)} \right) 
&= 
\L\left(\d\bm{T}^{(t)} \mid \bm{T}^{(t-1)},\psi^{(t-1)} \right)
\L\left(\d\psi^{(t)} \mid \bm{T}^{(t)}, \psi^{(t-1)},\bm{T}^{(t-1)} \right)
\\&= 
\hat{\pi}_J\left(\d\bm{T}^{(t)} \mid \psi^{(t-1)} \right)
\L\left(\d\psi^{(t)} \mid \bm{T}^{(t)}, \psi^{(t-1)} \right),
\end{align*}
where the last equality 
 follows by \eqref{two_blocks_gibbs_nested} and the definition of $\hat{\pi}_J$. 
By the exponential family assumption in \eqref{exponential_family}, $\bm{T}$ is a set of sufficient statistics for $\psi$, so that
\begin{equation}\label{deinitializing}
\pi_J\left(\d\psi \mid \bm{\theta} \right) 
=\L\left(\d\psi \mid \bm{\theta},Y_{1:J}\right)
=\L\left(\d\psi \mid \bT(\bm{\theta}),  Y_{1:J}\right)
=\hat{\pi}_J\left(\d\psi \mid \bm{T}(\bm{\theta}) \right).
\end{equation}
Combining \eqref{two_blocks_gibbs_nested} and \eqref{deinitializing} we have 
\begin{equation}\label{key_equality}
\begin{aligned}
\L\left(\d \psi^{(t)} \mid \bm{T}^{(t)}, \psi^{(t-1)} \right) &= \int
\pi_J\left(\d \psi^{(t)} \mid \bm{\theta} \right)\pi_J\left(\d\bm{\theta} \mid \bm{T}^{(t)}, \psi^{(t-1)} \right)\\
& = \int
\hat{\pi}_J\left(\d \psi^{(t)} \mid \bm{T}(\bm{\theta}) \right)\pi_J\left(\d\bm{\theta} \mid \bm{T}^{(t)}, \psi^{(t-1)} \right)= \hat{\pi}_J\left(\d \psi^{(t)} \mid \bm{T}^{(t)} \right)
\end{aligned}
\end{equation}
since $\bm{T}(\bm{\theta})=\bm{T}^{(t)}$ almost surely under $\pi_J\left(\d\bm{\theta} \mid \bm{T}^{(t)}, \psi^{(t-1)} \right)$. 
Thus we can conclude
\begin{align*}
\L\left(\d\bm{T}^{(t)}, \d\psi^{(t)} \mid \bm{T}^{(t-1)}, \psi^{(t-1)} \right) 
&= \hat{\pi}_J\left(\d\bm{T}^{(t)} \mid \psi^{(t-1)} \right)\hat{\pi}_J\left(\d\psi^{(t)}\mid \bm{T}^{(t)} \right)\\
&=\hat{P}_J\left(\left(\bm{T}^{(t-1)}, \psi^{(t-1)} \right),\left(\d\bm{T}^{(t)}, \d\psi^{(t)} \right)\right) \,,
\end{align*}
as desired. From the above one can easily deduce that $\left(\bm{\theta}^{(t)}, \psi^{(t)} \right)_{t \geq 1}$ and $\left(\bm{T}^{(t)}, \psi^{(t)} \right)_{t \geq 1}$ are \emph{co-deinitializing} as in \cite{R01} and thus, by 
Corollary 2 therein, for every $\mu \in\sP \left(\R^{\ell J}\times \R^D \right)$ we have
\begin{equation}\label{equality_suff}
\lTV \mu P_J^t-\pi_J \rTV = \lTV \nu \hat{P}_J^t-\hat{\pi}_J\rTV,
\end{equation}
where $\nu\in \sP \left(\R^S\times \R^D \right)$ is the push forward of $\mu$ under $(\bm{\theta}, \psi)\mapsto (\bT(\bm{\theta}), \psi)$.  Moreover, by \eqref{N_class}  we have that $\nu \in \sN \left(\hat{\pi}_J, M \right)$ whenever $\mu \in \sN \left(\pi_J, M \right)$. 
It follows that $\sup_{\mu \in \sN \left(\pi_J, M \right)} t^{(J)}_{mix}(\epsilon, \mu)\leq \sup_{\nu \in \sN \left(\hat{\pi}_J, M \right)} \hat{t}^{(J)}_{mix}(\epsilon, \nu)$. For the reverse inequality, fix $\nu \in \sN \left(\hat{\pi}_J, M \right)$ and take $\mu(\d\bm{\theta}, \d \psi) = \int \pi_J\left(\d\bm{\theta} \mid \bm{T}, \psi\right)\nu(\d \bm{T}, \d \psi)$.  
By \eqref{N_class} we have $\mu \in \sN \left(\pi_J, M \right)$ and thus \eqref{equality_suff}. It follows $\sup_{\nu \in \sN \left(\hat{\pi}_J, M \right)} \hat{t}^{(J)}_{mix}(\epsilon, \nu) \leq \sup_{\mu \in \sN \left(\pi_J, M \right)} t^{(J)}_{mix}(\epsilon, \mu)$ as desired.
\end{proof}

\subsection{Proof of Lemma \ref{asymptotic_distribution_psi}}
\begin{proof}
The result follows immediately from Theorem \ref{BvM}, whose assumptions are given exactly by assumption $(B1)-(B3)$, with likelihood $g(y \mid \psi)$.
\end{proof}

\subsection{Proof of Lemma \ref{asymptotic_distribution_T}}
The proof is divided in two main steps: in Section \ref{subs_weak_convergence} the result is proved under the weak metric (Lemma \ref{weak_asymptotic_distribution_T}) and it is extended to the total variation distance in Section \ref{subs_tv_convergence}.

First of all we need two technical lemmas, that we prove for completeness.
\begin{lemma}\label{lemma_modulus}
Let $S$ and $p$ be two positive integers. Then there exists a constant $C = C(S,p)$ such that
\[
|\x|^{p} \leq 1+C\sum_{s = 1}^Sx_s^{2p}
\]
for every $\x \in \R^S$.
\end{lemma}
\begin{proof}
Since $(1-|\x|^{p})^2 \geq 0$, we have $|\x|^{p} \leq 1+|\x|^{2p}$. Moreover, by the Multinomial Theorem, we get
\[
|\x|^{2p} = \left(\sum_{s = 1}^Sx_s^2 \right)^p = \sum_{\bm{k} \in \mathbb{P}}\binom{p}{k_1\,\dots\, k_S}\prod_{s = 1}^Sx_s^{2k_s},
\]
where $\mathbb{P} = \left\{\bm{k} = (k_1, \dots, k_S) \, : \, k_s \text{ positive integer, }\sum_{s = 1}^Sk_s = p \right\}$. Since
\[
\prod_{s = 1}^Sx_s^{2k_s} \leq \left(\max_{s}|x_s| \right)^{2p} \leq \sum_{s = 1}^Sx_s^{2p},
\]
the result follows by choosing $C = \sum_{\bm{k} \in \mathbb{P}}\binom{p}{k_1 \,\dots\, k_S}$.
\end{proof}
\begin{lemma}\label{convergence_delta}
Under assumption $(B3)$, the random variables $\Delta_J = \left(\Delta_{J,1}, \dots, \Delta_{J, D}  \right)$ defined in \eqref{tilde_psi} are such that for every $\beta > 0$ we have
\[
\frac{1}{J^\beta}\Delta_{J, d} \quad \to \quad 0,
\]
$Q_{\psi^*}^{(\infty)}$-almost surely as $J \to \infty$ for every $d=1,\dots,D$.
 \end{lemma}
\begin{proof}
Recall that
\[
\Delta_{J, d} = \frac{1}{\sqrt{J}}\sum_{j = 1}^J\left[\Fisher^{-1}(\psi^*) \nabla \log g(Y_j \mid \psi^*)\right]_d =: \frac{1}{\sqrt{J}}\sum_{j = 1}^JX_{j,d}
\]
and $\Fisher^{-1}(\psi^*) \partial_{\psi_d} \log g(Y_j \mid \psi^*)$ has zero mean and finite variance, by $(B3)$. Therefore, by Chebychev inequality
\[
P \left(\left \lvert \frac{1}{J^\beta}\Delta_{J, d} \right\rvert > \epsilon \right) \leq \frac{\Var \left(X_{1,d} \right)}{\epsilon^2J^{1+2\beta}},
\]
for every $\epsilon > 0$. This implies
\[
\sum_{J = 1}^\infty P \left(\left \lvert \frac{1}{J^\beta}\Delta_{J, d} \right\rvert > \epsilon \right) \leq \sum_{J = 1}^\infty \frac{\Var \left(X_{1,d} \right)}{\epsilon^2J^{1+2\beta}} < \infty,
\]
and the result follows by Borel-Cantelli Lemma.
\end{proof}
\subsubsection{Weak convergence}\label{subs_weak_convergence}
In order to ease the following exposition, denote
\begin{equation}\label{psi_J}
\psi^{(J)}:=\psi^*+\frac{\tilde{\psi}+\Delta_J}{\sqrt{J}}, \quad J\geq 1\,.
\end{equation}
The next lemma proves convergence of $\tilde{\bm{T}}$ using the weak metric, denoted by $\lTV \cdot \right \|_W$.
\begin{lemma}\label{weak_asymptotic_distribution_T}
Define $\tilde{\psi}$ and $\tilde{\bm{T}}$ as in \eqref{tilde_psi} and \eqref{tilde_T}, respectively. Under assumptions $(B1)-(B4)$, for every $\tilde{\psi}\in\R^D$ it holds
\begin{equation}\label{eq:weak_asymptotic_distribution_T}
\left\| \L(\d\tilde{\bm{T}} \mid Y_{1:J}, \tilde{\psi})-N\left(C(\psi^*)\tilde{\psi}, V(\psi^*) \right) \right\|_{W} \to 0,
\end{equation}
$Q_{\psi^*}^{(\infty)}$-almost surely as $J \to \infty$.
\end{lemma}
\begin{proof}
For ease of notation, denote
\[
\mu = C(\psi^*)\tilde{\psi} \quad \text{and} \quad \Xi := V(\psi^*).
\]
By definition of $M^{(p)}_s$, we have
\[
E\left[T^p_s(\theta_j) \mid Y_j, \psi^{(J)} \right]=M_s^{(p)}\left(\psi^{(J)} \mid Y_j \right).
\]
Conditional on $\tilde{\psi}$, the group specific statistics $T_s(\theta_j)$ are independent across $j = 1, \dots, J$. Thus, by Lyapunov version of Central Limit Theorem, in order to obtain \eqref{eq:weak_asymptotic_distribution_T} it suffices to show
\begin{align}
\frac{1}{\sqrt{J}}\sum_{j = 1}^J\left[M^{(1)}\left(\psi^{(J)} \mid Y_j \right)-M^{(1)}\left(\psi^* \mid Y_j \right) \right]-C(\psi^*)\Delta_J \quad &\to \quad \mu
\label{first_to_prove}\\
\frac{1}{J}\sum_{j = 1}^J\text{Cov}\left(T_s(\theta_j), T_{s'}(\theta_j) \mid Y_j, \psi^{(J)} \right) \quad &\to \quad \Xi_{s,s'}\label{second_to_prove}\\
\frac{1}{J^{3/2}}\sum_{j = 1}^JE_{Y_j}\left[\left \lvert T(\theta_j)-M^{(1)}\left(\psi^* \mid Y_j \right) \right \rvert^3 \mid Y_j, \psi^{(J)} \right]\quad &\to \quad 0,\label{third_to_prove}
\end{align}
$Q_{\psi^*}^{(\infty)}$-almost surely as $J \to \infty$, with $s, s' = 1, \dots, S$.
We prove the three above results sequentially below, which concludes the proof of \eqref{eq:weak_asymptotic_distribution_T}.
\end{proof}
 \begin{proof}[Proof of \eqref{first_to_prove}]
For any $s = 1, \dots, S$, by \eqref{psi_J} and the multivariate Taylor formula it holds
\[
M_s^{(1)}\left(\psi^{(J)} \mid Y_j \right)-M_s^{(1)}\left(\psi^* \mid Y_j \right) = \sum_{d = 1}^D\frac{\tilde{\psi}_d+\Delta_{J,d}}{\sqrt{J}} \partial_{\psi_d}M_s^{(1)}\left(\psi^* \mid Y_j \right) + R_2(Y_j),
\]
where
\[
R_2(Y_j) = \sum_{d, d' = 1}^D\frac{(\tilde{\psi}_d+\Delta_{J,d})(\tilde{\psi}_{d'}+\Delta_{J,d'})}{J}\int_0^1(1-t)\partial_{\psi_d}\partial_{\psi_{d'}}M_s^{(1)}\left(\psi^*+t\frac{\tilde{\psi}+\Delta_J}{\sqrt{J}} \mid Y_j \right) \, \d t.
\]
Therefore
\begin{equation}\label{main_first_to_prove}
\begin{aligned}
\frac{1}{\sqrt{J}}\sum_{j = 1}^J&\left[M^{(1)}_s\left(\psi^{(J)} \mid Y_j \right)-M^{(1)}_s\left(\psi^* \mid Y_j \right) \right] =\\
&= \sum_{d = 1}^D(\tilde{\psi}_d+\Delta_{J,d})\frac{1}{J}\sum_{j = 1}^J \partial_{\psi_d}M_s^{(1)}\left(\psi^* \mid Y_j \right)+\frac{1}{\sqrt{J}}\sum_{j = 1}^JR_2(Y_j)\,,
\end{aligned}
\end{equation}
where
\begin{align}
&\frac{1}{\sqrt{J}}\sum_{j = 1}^JR_2(Y_j) 
=\nonumber\\& 
\sum_{d, d' = 1}^D\frac{(\tilde{\psi}_d+\Delta_{J,d})(\tilde{\psi}_{d'}+\Delta_{J,d'})}{J^{1/4}}\frac{1}{J^{5/4}}\sum_{j = 1}^J\int_0^1(1-t)\partial_{\psi_d}\partial_{\psi_{d'}}M_s^{(1)}\left(\psi^*+t\frac{\tilde{\psi}+\Delta_J}{\sqrt{J}} \mid Y_j \right) \, \d t.\label{remainder_first_to_prove}
\end{align}
As regards \eqref{remainder_first_to_prove}, for every $d,d' = 1, \dots, D$ by Lemma \ref{convergence_delta} it holds
\begin{equation}\label{first_to_prove_to_zero}
\frac{(\tilde{\psi}_d+\Delta_{J,d})(\tilde{\psi}_{d'}+\Delta_{J,d'})}{J^{1/4}} 
=
\frac{\tilde{\psi}_d\tilde{\psi}_{d'}}{J^{1/4}}
+
\tilde{\psi}_d
\frac{\Delta_{J,d'}}{J^{1/4}} 
+
\tilde{\psi}_{d'}
\frac{\Delta_{J,d}}{J^{1/4}} 
+
\frac{\Delta_{J,d}}{J^{1/8}} 
\frac{\Delta_{J,d'}}{J^{1/8}} 
\quad \to \quad 0,
\end{equation}
 $Q_{\psi^*}^{(\infty)}$-almost surely as $J \to \infty$. Moreover, with the change of variables $x = t/J^{1/4}$ we have
\[
\begin{aligned}
&\biggl\lvert \frac{1}{J^{5/4}}\sum_{j = 1}^J\int_0^1(1-t)\partial_{\psi_d}\partial_{\psi_{d'}}M_s^{(1)}\left(\psi^*+t\frac{\tilde{\psi}+\Delta_J}{\sqrt{J}} \mid Y_j \right) \, \d t \biggr\rvert \\
&\leq \int_0^{J^{1/4}}\frac{1}{J}\sum_{j = 1}^J\left \lvert \partial_{\psi_d}\partial_{\psi_{d'}}M_s^{(1)}\left(\psi^*+x\frac{\tilde{\psi}+\Delta_J}{J^{1/4}} \mid Y_j \right) \right \rvert \, \d x\\
&\leq \int_{-J^{1/4}}^{J^{1/4}}\frac{1}{J}\sum_{j = 1}^J\left \lvert \partial_{\psi_d}\partial_{\psi_{d'}}M_s^{(1)}\left(\psi^*+x \mid Y_j \right) \right \rvert \, \d x,
\end{aligned}
\]
where the last inequality follows from $\left \lvert \frac{\tilde{\psi}+\Delta_J}{J^{1/4}} \right \rvert \leq 1$ for $J$ high enough, thanks to Lemma \ref{convergence_delta}. Moreover, $\frac{1}{J^{1/4}} < \delta_4$ for $J$ high enough, so that 
\[
\begin{aligned}
&\biggl\lvert \frac{1}{J^{5/4}}\sum_{j = 1}^J\int_0^1(1-t)\partial_{\psi_d}\partial_{\psi_{d'}}M_s^{(1)}\left(\psi^*+t\frac{\tilde{\psi}+\Delta_J}{\sqrt{J}} \mid Y_j \right) \, \d t \biggr\rvert \\
&\leq \int_{\delta_4}^{\delta_4}\frac{1}{J}\sum_{j = 1}^J\left \lvert \partial_{\psi_d}\partial_{\psi_{d'}}M_s^{(1)}\left(\psi^*+x \mid Y_j \right) \right \rvert \, \d x\\
&= \frac{1}{J}\sum_{j = 1}^J\int_{\delta_4}^{\delta_4}\left \lvert \partial_{\psi_d}\partial_{\psi_{d'}}M_s^{(1)}\left(\psi^*+x \mid Y_j \right) \right \rvert \, \d x.
\end{aligned}
\]
By the Law of Large Numbers and $(B4)$ it holds
\begin{align}
\frac{1}{J}\sum_{j = 1}^J&\int_{\delta_4}^{\delta_4}\left \lvert \partial_{\psi_d}\partial_{\psi_{d'}}M_s^{(1)}\left(\psi^*+x \mid Y_j \right) \right \rvert \, \d x \nonumber\\
&\to  \int_{-\delta_4}^{\delta_4}E\left[\left \lvert \partial_{\psi_d}\partial_{\psi_{d'}}M_s^{(1)}\left(\psi^*+x \mid Y_j \right) \right \rvert \right] \, \d x < 2C\delta_4.\label{first_to_prove_bounded}
\end{align}
By combining \eqref{first_to_prove_to_zero} and \eqref{first_to_prove_bounded}, we can conclude
\[
\left \lvert \frac{1}{\sqrt{J}}\sum_{j = 1}^JR_2(Y_j) \right\rvert  \quad \to \quad 0,
\]
$Q_{\psi^*}^{(\infty)}$-almost surely as $J \to \infty$. As regards \eqref{main_first_to_prove}, by the Law of Large Numbers we have
\[
\frac{1}{J}\sum_{j = 1}^J \partial_{\psi_d}M_s^{(1)}\left(\psi^* \mid Y_j \right) \quad \to \quad E \left[\partial_{\psi_d}M_s^{(1)}\left(\psi^* \mid Y_j \right) \right] = C_{s,d}(\psi^*),
\]
that is finite thanks to $(B4)$. Therefore, we can conclude that for any $s = 1, \dots, S$ we have
\[
M_s^{(1)}\left(\psi^{(J)} \mid Y_j \right)-M_s^{(1)}\left(\psi^* \mid Y_j \right) -\sum_{d = 1}^DC_{s,d}(\psi^*)\Delta_{J,d} \quad \to \quad \sum_{d = 1}^DC_{s,d}(\psi^*)\tilde{\psi_d},
\]
$Q_{\psi^*}^{(\infty)}$-almost surely as $J \to \infty$ and thus \eqref{first_to_prove} holds. 
\end{proof}

\begin{proof}[Proof of \eqref{second_to_prove}]
For every $s, s' = 1, \dots, S$ by multivariate Taylor formula it holds
\[
\text{Cov}\left(T_s(\theta_j), T_{s'}(\theta_j) \mid Y_j, \psi^{(J)} \right) = \text{Cov}\left(T_s(\theta_j), T_{s'}(\theta_j) \mid Y_j, \psi^* \right) + R_{1,cov}(Y_j),
\]
where
\[
R_{1,cov}(Y_j) = \sum_{d = 1}^D\frac{\tilde{\psi}_d+\Delta_{J,d}}{\sqrt{J}}\int_0^1(1-t)\partial_{\psi_d}\text{Cov}\left(T_s(\theta_j), T_{s'}(\theta_j) \mid Y_j, \psi^*+t\frac{\tilde{\psi}+\Delta_J}{\sqrt{J}} \right)\, \d t.
\]
Notice that
\[
\frac{1}{J}\sum_{j = 1}^JR_{1,cov}(Y_j) = \sum_{d = 1}^D\frac{\tilde{\psi}_d+\Delta_{J,d}}{J^{1/4}}\int_0^1(1-t)\frac{1}{J^{5/4}}\sum_{j = 1}^J\partial_{\psi_d}\text{Cov}\left(T_s(\theta_j), T_{s'}(\theta_j) \mid Y_j, \psi^*+t\frac{\tilde{\psi}+\Delta_J}{\sqrt{J}} \right) \, \d t.
\]
With the same arguments of before we have $\frac{\tilde{\psi}_d+\Delta_{J,d}}{J^{1/4}} \to 0$ and
\[
\begin{aligned}
\biggl\lvert\int_0^1(1-t)&\frac{1}{J^{5/4}}\sum_{j = 1}^J\partial_{\psi_d}\text{Cov}\left(T_s(\theta_j), T_{s'}(\theta_j) \mid Y_j, \psi^*+t\frac{\tilde{\psi}+\Delta_J}{\sqrt{J}} \right) \, \d t \biggr\rvert\\
& \leq \frac{1}{J}\sum_{j = 1}^J \int_{-\delta_4}^{\delta_4} \left \lvert \partial_{\psi_d}\text{Cov}\left(T_s(\theta_j), T_{s'}(\theta_j) \mid Y_j, \psi^*+x \right) \right \rvert \, \d x \\
& \to \quad \int_{-\delta_4}^{\delta_4}E \left[\left \lvert \partial_{\psi_d}\text{Cov}\left(T_s(\theta_j), T_{s'}(\theta_j) \mid Y_j, \psi^*+x \right) \right \rvert \right] \, \d x
\end{aligned}
\]
$Q_{\psi^*}^{(\infty)}$-almost surely as $J \to \infty$. Notice that by $(B4)$ we have
\[
\begin{aligned}
E& \left[\left \lvert \partial_{\psi_d}\text{Cov}\left(T_s(\theta_j), T_{s'}(\theta_j) \mid Y_j, \psi^*+x \right) \right \rvert \right] \\
&\leq E \left[\left \lvert \partial_{\psi_d}M^{(1)}_{s,s'}\left(\psi^*+x \mid Y_j \right) \right \rvert \right]+E \left[\left \lvert \partial_{\psi_d}\left\{M^{(1)}_{s}\left(\psi^*+x \mid Y_j \right)M^{(1)}_{s'}\left(\psi^*+x \mid Y_j \right)\right\} \right \rvert \right] \\
&\leq 2C,
\end{aligned}
\] 
for every $x \in (-\delta_4, \delta_4)$ . Therefore, we can conclude
\[
\left \lvert\frac{1}{J}\sum_{j = 1}^JR_{1,cov}(Y_j) \right \rvert \to 0,
\]
$Q_{\psi^*}^{(\infty)}$-almost surely as $J \to \infty$. Thus, by the Law of Large Numbers we have
\[
\frac{1}{J}\sum_{j = 1}^J \text{Cov}\left(T_s(\theta_j), T_{s'}(\theta_j) \mid Y_j, \psi^* \right) \quad \to \quad E\left[ \text{Cov}\left(T_s(\theta_j), T_{s'}(\theta_j) \mid Y_j, \psi^* \right)\right],
\]
$Q_{\psi^*}^{(\infty)}$-almost surely as $J \to \infty$. 
\end{proof}
\begin{proof}[Proof of \eqref{third_to_prove}]
By Lemma \ref{lemma_modulus} we have
\[
\begin{aligned}
\frac{1}{J^{3/2}}\sum_{j = 1}^J&E_{Y_j}\left[\left \lvert T(\theta_j)-M^{(1)}\left(\psi^* \mid Y_j \right) \right \rvert^3 \mid Y_j, \psi^{(J)} \right] \\
&\leq \frac{1}{\sqrt{J}}+C\frac{1}{J^{3/2}}\sum_{s = 1}^S\sum_{j = 1}^JM^{(6)}\left(\psi^{(J)} \mid Y_j \right)+ C\frac{1}{J^{3/2}}\sum_{s = 1}^S\sum_{j = 1}^J\left[M^{(1)}\left(\psi^* \mid Y_j \right) \right]^6.
\end{aligned}
\]
By Jensen inequality $\left[M^{(1)}\left(\psi^* \mid Y_j \right) \right]^6 \leq M^{(6)}\left(\psi^* \mid Y_j \right)$ and by the Law of Large Numbers
\[
\frac{1}{J}\sum_{s = 1}^S\sum_{j = 1}^J M^{(6)}\left(\psi^* \mid Y_j \right) \quad \to \quad \sum_{s = 1}^SE\left[ T^6_s(\theta_j) \mid \psi^*\right] < \infty
\]
$Q_{\psi^*}^{(\infty)}$-almost surely as $J \to \infty$. Thus to prove \eqref{third_to_prove} it suffices to show
\[
\frac{1}{J^{3/2}}\sum_{s = 1}^S\sum_{j = 1}^JM^{(6)}\left(\psi^{(J)} \mid Y_j \right) \quad \to \quad 0
\]
$Q_{\psi^*}^{(\infty)}$-almost surely as $J \to \infty$. For every $s = 1, \dots, S$ by multivariate Taylor formula it holds
\[
M_s^{(6)}\left(\psi^{(J)} \mid Y_j \right) = M_s^{(6)}\left(\psi^* \mid Y_j \right) + R_{1,6}(Y_j),
\]
where
\[
R_{1,6}(Y_j) = \sum_{d = 1}^D\frac{\tilde{\psi}_d+\Delta_{J,d}}{\sqrt{J}}\int_0^1(1-t)\partial_{\psi_d}M_s^{(6)}\left(\psi^*+t\frac{\tilde{\psi}+\Delta_J}{\sqrt{J}} \mid Y_j \right) \, \d t.
\]
Notice that
\[
\frac{1}{J}\sum_{j = 1}^JR_{1,6}(Y_j) = \sum_{d = 1}^D\frac{\tilde{\psi}_d+\Delta_{J,d}}{J^{1/4}}\int_0^1(1-t)\frac{1}{J^{5/4}}\sum_{j = 1}^J\partial_{\psi_d}M_s^{(6)}\left(\psi^*+t\frac{\tilde{\psi}+\Delta_J}{\sqrt{J}} \mid Y_j \right) \, \d t,
\]
and with the same arguments of before we have $\frac{\tilde{\psi}_d+\Delta_{J,d}}{J^{1/4}} \to 0$ $Q_{\psi^*}^{(\infty)}$-almost surely as $J \to \infty$ and
\[
\begin{aligned}
\biggl\lvert\int_0^1(1-t)&\frac{1}{J^{5/4}}\sum_{j = 1}^J\partial_{\psi_d}M_s^{(6)}\left(\psi^*+t\frac{\tilde{\psi}+\Delta_J}{\sqrt{J}} \mid Y_j \right) \, \d t \biggr\rvert\\
& \leq \frac{1}{J}\sum_{j = 1}^J \int_{-\delta_4}^{\delta_4} \left \lvert \partial_{\psi_d}M_s^{(6)}(\psi^*+x \mid Y_j) \right \rvert \, \d x \\
& \to \int_{-\delta_4}^{\delta_4}E \left[\left \lvert \partial_{\psi_d}M_s^{(6)}(\psi^*+x \mid Y_j) \right \rvert \right] \, \d x < 2\delta_4C,
\end{aligned}
\]
by $(B4)$. Therefore, we can conclude
\[
\left \lvert\frac{1}{J}\sum_{j = 1}^JR_{1,6}(Y_j) \right \rvert \to 0,
\]
$Q_{\psi^*}^{(\infty)}$-almost surely as $J \to \infty$. Moreover, by the Law of Large Numbers we have
\[
\frac{1}{J}\sum_{j = 1}^JM_s^{(6)}\left(\psi^* \mid Y_j \right) \quad \to \quad E\left[M_s^{(6)}\left(\psi^* \mid Y_j \right) \right] = E\left[T_s^6(\theta_j) \mid \psi^*\right],
\]
by $(B1)$ and the definition of conditional expectation. Therefore
\[
\frac{1}{J^{3/2}}\sum_{j = 1}^JM_s^{(6)}\left(\psi^*+\frac{\tilde{\psi}_d+\Delta_{J,d}}{\sqrt{J}} \mid Y_j \right) \to 0,
\]
from which \eqref{third_to_prove} follows.
\end{proof}
\subsubsection{Total variation convergence}\label{subs_tv_convergence}
We extend the weak convergence to total variation using characteristic functions, in particular exploiting the conditions in Lemma \ref{lemma:conditions_TV_conv}.
Here we first state some other technical lemmas that will be required later on. 
\begin{lemma}\label{Taylor}
Let $X$ be a $\R^S$-valued random vector with zero mean and characteristic function $\varphi_X(u)$. Then for every $u\in\R^S$
\begin{align*}
\varphi_X(u)=& 1-\frac{1}{2}E\left[(u^\top X)^2\right]+\frac{\theta}{6}E\left[|u^\top X|^3\right],
\end{align*}
for some $\theta=\theta(u)\in\mathbb{C}$ such that $|\theta| \leq 1$.
\end{lemma}
\begin{proof}
Taylor formula for the complex exponential reads
\[
e^{ix} = 1+ix-\frac{x^2}{2}+\frac{x^3}{6}e^{iz},
\]
where $z\in\mathbb{C}$ is such that $0 \leq |z| \leq |x|$. By $x = u^\top X$, we have
\[
\varphi_X(u) = 1+iE \left[u^\top X \right]-\frac{1}{2}E \left[\left(u^\top X\right)^2 \right]+\frac{\theta}{6}E \left[\left \lvert u^\top X\right \rvert^3 \right],
\]
with $\theta = e^{iz}$, recalling that $|e^{iz}| \leq 1$ for any $z$.
The result follows from $E \left[u^\top X \right]=0$.
\end{proof}
\begin{lemma}\label{difference}
Let $X \in \R^S$ and $Y \in \R^S$ be independent random vectors with the same distribution. Then
\[
\varphi_{X-Y}(u) = \left \lvert \varphi_X(u) \right \rvert^2.
\]
\end{lemma}
\begin{proof}
By independence we can write
\[
\varphi_{X-Y}(u) = E \left[e^{iu^\top X} \right]E \left[e^{-iu^\top X} \right],
\]
where
\[
E \left[e^{iu^\top X} \right] = E \left[\cos u^\top X \right]+iE \left[\sin u^\top X \right] = a+ib,
\]
for suitable $a$ and $b$. Since $\cos x$ is even and $\sin x$ is odd, we can write
\[
\left \lvert \varphi_{X-Y}(u) \right \rvert = \left \lvert (a+ib)(a-ib)\right \rvert = a^2+b^2 = \left \lvert \varphi_X(u) \right \rvert^2
\]
Since $X-Y$ has a symmetric density by construction $\left \lvert \varphi_{X-Y}(u) \right \rvert = \varphi_{X-Y}(u)$ and the result follows.
\end{proof}

\begin{corollary}\label{corollary_modulus_ cf}
Let $X $ be a $\R^S$-valued random vector with characteristic function $\varphi_X(u)$. Then
\[
\left \lvert \varphi_X(u) \right \rvert^2 \leq e^{-u^\top \Var(X) u+\frac{2|u|^3}{3}\left[1+C\sum_{s = 1}^SE\left[X_i^6\right]\right]},
\]
for $u \in \R^S$, where $C$ is a finite constant independent of $u$. 
\end{corollary}
\begin{proof}
Let $Y$ be an independent copy of $X$. By Lemma \ref{difference}, it holds
\[
\left \lvert \varphi_X(u) \right \rvert^2 = \varphi_{X-Y}(u),
\]
where $\varphi_{X-Y}(u)$ is a real function, since it is the characteristic function of a random variable with symmetric density. Therefore, by Lemma \ref{Taylor} it holds
\[
\varphi_{X-Y}(u) = 1-\frac{1}{2}E\left[(u^\top Z)^2\right]+\frac{\theta}{6}E\left[|u^\top Z|^3\right],
\]
where $Z = X-Y$ and $\theta=\theta(u)\in\R$.
Recalling that $e^x \geq 1+x$ for every $x$, we have
\[
\varphi_{X-Y}(u) \leq e^{-\frac{1}{2}E\left[(u^\top Z)^2\right]+\frac{\theta}{6}E\left[|u^\top Z|^3\right]}.
\] 
By Lemma $8.8$ in \cite{BR10} it holds
\[
E\left[(u^\top Z)^2\right] = 2E\left[(u^\top X)^2\right] = 2u^\top \Var(X) u
\]
and
\[
E\left[(u^\top Z)^3\right] \leq 4E\left[(u^\top X)^3\right] \leq 4|u|^3E \left[|X|^3 \right].
\]
Moreover by Lemma \ref{lemma_modulus} we have
\[
E \left[|X|^3 \right] \leq 1 + C\sum_{s = 1}^SE\left[X_i^6\right].
\]
Therefore
\[
\varphi_{X-Y}(u) \leq e^{-u^\top \Var(X) u+\frac{2|u|^3\theta}{3}\left[1+C\sum_{s = 1}^SE\left[X_i^6\right]\right]}
\] 
and the result follows from $|\theta| \leq 1$.
\end{proof}
 The following lemma is a minor variation of commonly used techniques to prove total variation Central Limit Theorems.
\begin{lemma}\label{lemma:conditions_TV_conv}
Let $(X_J)_{J\geq 1}$ and $X$ be $\R^S$-valued random variables with characteristic functions $(\varphi_J)_{J\geq 1}$ and $\varphi$, respectively. 
Denote by $L^1(\R^S)$ the space of complex-valued integrable functions with domain $\R^S$. 
If
\begin{enumerate}
\item[(a)] $X_J$ converges weakly to $X$ as $J\to\infty$
\item[(b)] $\varphi$ belongs to $L^1(\R^S)$, i.e.\ $\int_{\R^S} \left \lvert \varphi\left(t\right)\right \rvert \, \d t <\infty$
\item[(c)] $\lim_{A\to\infty}\limsup_{J\to\infty}\int_{|t|\geq A} \left \lvert \varphi_J\left(t\right)\right \rvert \, \d t= 0$.
\end{enumerate}
then $X_J$ converges to $X$ in total variation as $J\to\infty$.
\end{lemma}
\begin{proof}
First we prove that $\lim_{J\to\infty} \|\varphi_J-\varphi\|_{L^1}
=0$. 
By the triangle inequality, for every $A>0$ we have
\begin{align}\label{eq:split_above_A}
\|\varphi_J-\varphi\|_{L^1}
\leq
\int_{|t|< A}|\varphi_J(t) -\varphi(t) |\, \d t
+
\int_{|t|\geq A}|\varphi_J(t)|\, \d t+\int_{|t|\geq A}|\varphi(t) |\, \d t\,.
\end{align}
Since weak convergence implies pointwise convergence of characteristic functions, assumption (a) implies that  $\varphi_J(t)\to\varphi(t)$ as $J\to\infty$ for every $t\in\R^S$. Thus by the Dominated Convergence Theorem and $|\varphi_J(t) -\varphi(t) |\leq |\varphi_J(t)|+|\varphi(t) |=2$ , we have $\int_{|t|< A}|\varphi_J(t) -\varphi(t) |\, \d t\to 0$ as $J\to\infty$ for every $A>0$.
It follows by \eqref{eq:split_above_A} that 
\begin{align}\label{eq:limsup_above_A}
0\leq \limsup_{J\to\infty}
\|\varphi_J-\varphi\|_{L^1}
\leq
\int_{|t|\geq A}|\varphi(t) |\, \d t
+
\limsup_{J\to\infty}\int_{|t|\geq A}|\varphi_J(t)|\, \d t
\,,
\end{align}
for every $A>0$.
By assumption (b) $\lim_{A\to\infty}\int_{|t|\geq A}|\varphi(t) |\, \d t= 0$. Combining with assumption (c), taking the limit $A\to\infty$ we obtain $\limsup_{J\to\infty} \|\varphi_J-\varphi\|_{L^1}\leq 0$ and thus $\lim_{J\to\infty} \|\varphi_J-\varphi\|_{L^1}=0$.

Then, note that $\varphi\in L^1(\R^S)$ and $\|\varphi_J-\varphi\|_{L^1}\to 0$ as $J\to\infty$ imply $\varphi_J\in L^1(\R^S)$ eventually as $J\to\infty$, since by the triangle inequality
$$
\|\varphi_J\|_{L^1}
\leq 
\|\varphi_J-\varphi\|_{L^1}
+
\|\varphi\|_{L^1}
<\infty
$$
for $J$ large enough.
Thus, by the Inversion formula, for $J$ large enough $X_J$ and $X$ admit density functions w.r.t.\ the Lebesgue measure, which can be written as
$f_{X_J}(\bm{t}) = \frac{1}{(2\pi)^{S}}\int_{\R^S}e^{-i\bm{t}^\top t}\varphi_J(t) \, \d t$ and $f_{X}(\bm{t}) = \frac{1}{(2\pi)^{S}}\int_{\R^S}e^{-i\bm{t}^\top t}\varphi(t) \, \d t$.
Thus
\begin{align*}
|f_{X_J}(\bm{t})-f_{X}(\bm{t})|&=
\left|\frac{1}{(2\pi)^{S}}\int_{\R^S}e^{-i\bm{t}^\top t}\varphi_J(t) \, \d t-
\frac{1}{(2\pi)^{S}}\int_{\R^S}e^{-i\bm{t}^\top t}\varphi(t) \, \d t
\right|\\
&\leq
\int_{\R^S}\left|e^{-i\bm{t}^\top t}(\varphi_J(t) -\varphi(t) )\right|\, \d t
\leq
\|\varphi_J-\varphi\|_{L^1}\to 0\,
\end{align*}
as $J\to\infty$ for every $\bm{t}\in\R^S$.  
By Scheff\'e Theorem, total variation convergence is implied by pointwise convergence of the densities. 
\end{proof}

\begin{proof}[Proof of Lemma \ref{asymptotic_distribution_T}]
Fix $\tilde{\psi}\in\R^D$ and denote $\mu = C(\psi^*)\tilde{\psi}$
and
$\Xi = V(\psi^*)$. 
We will prove conditions (a), (b) and (c) of Lemma \ref{lemma:conditions_TV_conv} to show that $\L(\d\tilde{\bm{T}}\mid Y_{1:J},\tilde{\psi})\stackrel{TV}\to N\left(\mu, \Xi \right)$ for $Q_{\psi^*}^{(\infty)}$-almost every $Y$ as $J \to \infty$.

Condition (a) is shown in Proposition \ref{weak_asymptotic_distribution_T}. Regarding condition (b), the characteristic function of the limiting distribution $N\left(\mu, \Xi \right)$ is $ \varphi(t)=e^{i\mu^\top t-\frac{1}{2}t^\top \Xi t}$, which is integrable since $\Xi$ is positive definite by  (B4).

We now turn to condition (c). Let
\begin{align*}
\tilde{\varphi}(t \mid Y_{1:J}, \psi)
&=
\E\left[e^{it^\top\tilde{\bT}}\mid Y_{1:J},\psi\right]
&t \in \R^S
\end{align*}
 be the characteristic function of $\L\left(\d\tilde{\bm{T}}\mid Y_{1:J}, \psi\right)$. 
Using the definition of $\tilde{\bT}$ in \eqref{tilde_T}, and the fact that $T_s(\theta_j)$ are conditionally independent given $\tilde{\psi}$, we can write $\tilde{\varphi}$ as
\[
\tilde{\varphi}(t \mid Y_{1:J}, \tilde{\psi}) =
e^{-it^\top \alpha_J}
\prod_{j = 1}^J\varphi\left(\frac{t}{\sqrt{J}} \mid Y_j, \psi^{(J)}\right),
\]
where $\alpha_J=C(\psi^*)\Delta_J
+\frac{1}{\sqrt{J}}\sum_{j = 1}^JM^{(1)}(\psi^* \mid Y_j)$, $
\varphi\left(t \mid Y_j, \psi \right)= E \left[e^{it^\top T(\theta_j)} \mid Y_j, \psi \right]$ as in the definition of (B5) and $\psi^{(J)}$ as in \eqref{psi_J}. Since $\alpha_J\in\R^S$ we have $|e^{-it^\top \alpha_J}|=1$ and thus
\begin{equation}\label{eq:varphi_tilda_and_not}
\left \lvert \tilde{\varphi}(t \mid Y_{1:J}, \psi) \right \rvert = \left \lvert \prod_{j = 1}^J\varphi\left(\frac{t}{\sqrt{J}} \mid Y_j, \psi\right) \right \rvert.
\end{equation}
For every 
 $\epsilon > 0$, by \eqref{eq:varphi_tilda_and_not} and the subadditivity of $\limsup$ we have
\[
\begin{aligned}
&\lim_{A\to\infty}\limsup_{J\to\infty}\int_{|t| > A}\left \lvert \tilde{\varphi}(t \mid Y_{1:J}, \tilde{\psi})\right \rvert \, \d t 
\leq
\\&
\lim_{A\to\infty}\limsup_{J\to\infty}\int_{A < |t| < \epsilon\sqrt{J}}\left \lvert \prod_{j = 1}^J\varphi\left(\frac{t}{\sqrt{J}} \mid Y_j, \psi^{(J)}\right)\right \rvert \, \d t
+
\limsup_{J\to\infty}\int_{|t| > \epsilon\sqrt{J}}\left \lvert \prod_{j = 1}^J\varphi\left(\frac{t}{\sqrt{J}} \mid Y_j, \psi^{(J)}\right)\right \rvert \, \d t.
\end{aligned}
\]
Lemma \ref{second_part_cf} shows that the second $\limsup$ in the last line is equal to $0$ for every $\epsilon>0$, while Lemma \ref{first_part_cf} shows that the $\lim_{A\to\infty}\limsup_{J\to\infty}$ term goes to $0$ when $\epsilon$ is chosen as in \eqref{eq:choice_lambda}. Thus condition (c) follows by taking $\epsilon$ as in \eqref{eq:choice_lambda} in the above inequality.
\end{proof}

\begin{lemma}\label{second_part_cf}
Under the same setting and notation as in the proof of Lemma \ref{asymptotic_distribution_T}, for every $\epsilon > 0$ we have
\[
\limsup_{J \to \infty} 
\int_{|t| > \epsilon\sqrt{J}}\left \lvert \prod_{j = 1}^J\varphi\left(\frac{t}{\sqrt{J}} \mid Y_j, \psi^{(J)}\right)\right \rvert \, \d t = 0
\]
$Q_{\psi^*}^{(\infty)}$-almost surely.
\end{lemma}

\begin{proof}
Consider the change of variables $x = t/\sqrt{J}$. Then
\begin{align*}
&\int_{|t| > \epsilon\sqrt{J}}\left \lvert \prod_{j = 1}^J\varphi\left(\frac{t}{\sqrt{J}} \mid Y_j, \psi^{(J)}\right)\right \rvert \, \d t= J^{S/2}\int_{|x| > \epsilon}\left \lvert \prod_{j = 1}^J\varphi\left(x \mid Y_j, \psi^{(J)}\right)\right \rvert \, \d x.
\end{align*}
Let $k$ and $B_{\delta_5}$ be as in $(B5)$ and $k'$ and $B_{\delta_6}$ be as in $(B6)$. 
Take $J$ high enough so that $J\geq 2k$ as well as $\psi^{(J)}\in B:=B_{\delta_5}\cap B_{\delta_6}$, so that
\begin{align*}
&\int_{|x| > \epsilon}\left \lvert \prod_{j = 1}^J\varphi\left(x \mid Y_j, \psi^{(J)}\right)\right \rvert \, \d x
\leq
\sup_{\psi\in B}
\int_{|x| > \epsilon}
\left \lvert \prod_{j = 1}^{2k}\varphi\left(x \mid Y_j, \psi\right)\right \rvert 
\left \lvert \prod_{j = 2k+1}^J\varphi\left(x \mid Y_j, \psi\right)\right \rvert 
\, \d x
\,.
\end{align*}
For every $a\in\R_+$ denote its integer part as $\lfloor a\rfloor$. 
By (B6), for every $\psi\in B$ we have
\begin{align*}
\left \lvert \prod_{j = 2k+1}^J\varphi\left(x \mid Y_j, \psi\right)\right \rvert
\leq
\prod_{s = 1}^{\lfloor \frac{J-2k}{k'}\rfloor}
A_s
\leq \phi(\epsilon)^{\lfloor \frac{J-2k}{k'}\rfloor},
\qquad \hbox{with }
A_s =\left \lvert \prod_{j = 2k+1+(s-1)k'}^{2k+1+sk'}\varphi\left(x \mid Y_j, \psi\right)\right \rvert 
\end{align*}
almost surely, where we exploited the fact that each $A_s$ is distributed as $\varphi^{(k')}\left(t \mid Y_{1:k'}, \psi \right)$ in (B6). 
Therefore
\[
\begin{aligned}
&\int_{|x| > \epsilon}\left \lvert \prod_{j = 1}^J\varphi\left(x \mid Y_j, \psi^{(J)}\right)\right \rvert \, \d x
 \leq 
\phi(\epsilon)^{\lfloor \frac{J-2k}{k'}\rfloor}
\sup_{\psi\in B}
\int_{|x| > \epsilon}\left \lvert \prod_{j = 1}^{2k}\varphi\left(x \mid Y_j, \psi\right)\right \rvert \, \d x.
\end{aligned}
\]
almost surely.
By H\"older Inequality and $(B5)$, we have
\[
\begin{aligned}
c =
\sup_{\psi\in B}
\int_{|x| > \epsilon}&\left \lvert \prod_{j = 1}^{2k}\varphi\left(x \mid Y_j, \psi\right)\right \rvert \, \d x 
\leq
\sup_{\psi\in B}
\int_{\R^S}\left \lvert \prod_{j = 1}^{2k}\varphi\left(x \mid Y_j, \psi\right)\right \rvert \, \d x \leq\\
& \left\{\sqrt{\sup_{\psi \in B} \, \int_{\R^S}\left \lvert \prod_{j = 1}^{k}\varphi\left(x \mid Y_j, \psi\right)\right \rvert^2 \, \d x}\right\}
\left\{\sqrt{\sup_{\psi \in B} \, \int_{\R^S}\left \lvert \prod_{j = k+1}^{2k}\varphi\left(x \mid Y_j, \psi\right)\right \rvert^2 \, \d x}\right\} < \infty,
\end{aligned}
\]
almost surely. Therefore it holds
\[
\begin{aligned}
\int_{|t| > \epsilon\sqrt{J}}&\left \lvert \prod_{j = 1}^J\varphi\left(\frac{t}{\sqrt{J}} \mid Y_j, \psi^{(J)}\right)\right \rvert \, \d t 
\leq
J^{S/2}\phi(\epsilon)^{\lfloor \frac{J-2k}{k'}\rfloor}
c,
\end{aligned}
\]
that goes to $0$ as $J \to \infty$, since $\phi(\epsilon) < 1$ by $(B6)$.
\end{proof}

\begin{lemma}\label{first_part_cf}
Under the same setting and notation as in the proof of
Lemma \ref{asymptotic_distribution_T}, 
let $\lambda>0$ be such that the matrix $V(\psi^*) - \lambda I$ is positive definite. Such $\lambda$ can be found, since $V(\psi^*)$ is positive definite by (B4).
Then, given 
\begin{equation}\label{eq:choice_lambda}
\epsilon = \frac{\lambda}{1+C\sum_{s = 1}^SE \left[T_s(\theta_1)^6 \mid \psi^* \right]}
\end{equation}
 we have
\[
\lim_{A\to\infty}
\limsup_{J \to \infty} \, \int_{A < |t| < \epsilon\sqrt{J}}\left \lvert \prod_{j = 1}^J\varphi\left(\frac{t}{\sqrt{J}} \mid Y_j, \psi^{(J)}\right)\right \rvert \, \d t=0
\]
$Q_{\psi^*}^{(\infty)}$-almost surely.
\end{lemma}
\begin{proof}
By Corollary \ref{corollary_modulus_ cf}, we have
\[
\left \lvert \varphi(u \mid Y_j, \psi) \right \rvert^2 \leq e^{-u^\top \Var\left(T(\theta_j) \mid Y_j, \psi \right) u +\frac{2|u|^3}{3}\left[1+C\sum_{s = 1}^SE \left[T_s(\theta_j)^6 \mid Y_j, \psi \right]\right]},
\]
for every $u \in \R^S$ and $\psi\in\R^D$. Therefore
\begin{align}\label{eq:prod_varphis_bound}
\left \lvert \prod_{j = 1}^J\varphi\left(\frac{t}{\sqrt{J}} \mid Y_j, \psi\right) \right \rvert^2 \leq e^{-t^\top \frac{1}{J}\sum_{j = 1}^J\Var\left(T(\theta_j) \mid Y_j, \psi \right) t +\frac{2|t|^3}{3\sqrt{J}}\left[1+C\frac{1}{J}\sum_{j = 1}^J\sum_{s = 1}^SE \left[T_s(\theta_j)^6 \mid Y_j, \psi \right]\right]}.
\end{align}
Notice that in the proof of \eqref{third_to_prove} we have shown through (B4) that
\begin{equation}\label{eq:lim_6_mom}
\frac{1}{J}\sum_{j = 1}^JE \left[T_s(\theta_j)^6 \mid Y_j, \psi^{(J)}\right] \to E \left[T_s(\theta_1)^6 \mid \psi^* \right]
\end{equation}
$Q_{\psi^*}^{(\infty)}$-almost surely as $J \to \infty$, for every $s = 1, \dots, S$. Thus,
combining \eqref{eq:choice_lambda} 
and \eqref{eq:lim_6_mom}, for every $|t| \leq \epsilon \sqrt{J}$ we have
\begin{equation}\label{eq:bound_moemnt_3_}
\left \lvert e^{\frac{2|t|^3}{3\sqrt{J}}\left[1+C\frac{1}{J}\sum_{j = 1}^J\sum_{s = 1}^SE \left[T_s(\theta_j)^6 \mid Y_j, \psi \right]\right]} \right \rvert^2
\leq
e^{\lambda t^\top t},
\end{equation}
almost surely for $J$ high enough. Finally by \eqref{eq:prod_varphis_bound} and \eqref{eq:bound_moemnt_3_}
\begin{equation}\label{final_bound}
\begin{aligned}
\int_{A < |t| < \epsilon\sqrt{J}}&\left \lvert \prod_{j = 1}^J\varphi\left(\frac{t}{\sqrt{J}} \mid Y_j, \psi^{(J)}\right)\right \rvert \, \d t
\leq
\int_{|t| >A}e^{-t^\top\Xi^{(J)}
t} \, \d t\,,
\end{aligned}
\end{equation}
with 
\[
\Xi^{(J)} =
\frac{1}{J}\sum_{j = 1}^J\Var\left(T(\theta_j) \mid Y_j, \psi^{(J)}\right)-\lambda I\,.
\]
Since $\Xi^{(J)}\to V(\psi^*)-\lambda I$ by \eqref{second_to_prove}, and $V(\psi^*)-\lambda I$ is positive definite by definition of $\lambda$,  by Dominated Convergence Theorem
\begin{equation}\label{final_bound}
\begin{aligned}
\limsup_J\int_{A < |t| < \epsilon\sqrt{J}}&\left \lvert \prod_{j = 1}^J\varphi\left(\frac{t}{\sqrt{J}} \mid Y_j, \psi^{(J)}\right)\right \rvert \, \d t
\leq
\int_{|t| >A}e^{-t^\top(V(\psi^*)-\lambda I)
t} \, \d t\,,
\end{aligned}
\end{equation}
Since the right hand side of \eqref{final_bound} is integrable the conclusion follows by taking $A\to\infty$.
\end{proof}

\subsection{Proof of Theorem \ref{theorem_one_level_nested}}
We first need a technical lemma.
\begin{lemma}\label{lemma_marginal_conditional}
Let $\left\{ Y^{(n)} \right\}_n$ be a sequence of random elements with state space $\sY^{(n)}$, such that $Y^{(n)} \sim Q^{(n)}$ with $Q^{(n)} \in \mathcal{P}\left(\sY^{(n)} \right)$. Let $\{\pi_n\}_n$ be a sequence of Markov kernels from $\sY^{(n)}$ to $\sX = \sX_1 \times  \sX_2$ and let $\pi \in \mathcal{P}(\sX)$. If
\[
\lTV \pi_{n,1}(\cdot) -\pi_1(\cdot) \rTV \to 0 \quad \text{and} \quad \lTV \pi_n(\cdot \mid x) -\pi(\cdot \mid x) \rTV \to 0, \text{ for $\pi_1$-almost every } x \in \sX_1,
\]
as $n \to \infty$ in $Q^{(n)}$-probability, where $\pi_{n,1}$ and $\pi_1$ are the marginal distributions on $\sX_1$ of $\pi_n$ and $\pi$ respectively, then
\[
\lTV \pi_n(\cdot) -\pi(\cdot) \rTV \to 0,
\]
as $n \to \infty$ in $Q^{(n)}$-probability
\end{lemma}
\begin{proof}
Let $f \, : \, \sX \, \to \, [0,1]$ be a measurable function. By the triangular inequality we have
\[
\begin{aligned}
\biggl \lvert  \int_{\sX} f(x_1, x_2) & \pi_n(\d x_1, \d x_2) - \int_{\sX} f(x_1, x_2) \pi(\d x_1, \d x_2) \biggr \rvert \leq \\
& \left \lvert  \int_\sX f(x_1, x_2) \pi_n(\d x_2 \mid x_1) \pi_{n,1}(\d x_1) - \int_\sX f(x_1, x_2) \pi_n(\d x_2 \mid x_1) \pi_1(\d x_1) \right \rvert +\\
& \left \lvert  \int_\sX f(x_1, x_2) \pi_n(\d x_2 \mid x_1) \pi_1(\d x_1) - \int_\sX f(x_1, x_2) \pi(\d x_2 \mid x_1) \pi_1(\d x_1) \right \rvert.
\end{aligned}
\]
Notice that
\[
\begin{aligned}
 \underset{f}{\sup}& \, \left \lvert  \int_\sX f(x_1, x_2) \pi_n(\d x_2 \mid x_1) \pi_{n,1}(\d x_1) - \int_\sX f(x_1, x_2) \pi_n(\d x_2 \mid x_1) \pi_1(\d x_1) \right \rvert \\
 &\leq \lTV \pi_{n,1}(\cdot) -\pi_1(\cdot) \rTV \to 0,
 \end{aligned}
\]
as $n \to \infty$ in $Q^{(n)}$-probability, by assumption. Moreover we have
\[
\begin{aligned}
 \underset{f}{\sup} & \, \left \lvert  \int_\sX f(x_1, x_2) \pi_n(\d x_2 \mid x_1) \pi_1(\d x_1) - \int_\sX f(x_1, x_2) \pi(\d x_2 \mid x_1) \pi_1(\d x_1) \right \rvert \leq \\
 &  \int_{\sX_1} \underset{f}{\sup} \,\left \lvert \int_{\sX_2}f(x_1, x_2) \pi_n(\d x_2 \mid x_1) - \int_{\sX_2} f(x_1, x_2) \pi (\d x_2 \mid x_1) \right \rvert \, \pi_1(\d x_1).
\end{aligned}
\]
The integrand on the right hand side goes to $0$ as $n \to \infty$ in $Q^{(n)}$-probability, by assumption. Therefore, by Dominated Convergence Theorem, we have
\[
 \underset{f}{\sup} \, \left \lvert  \int_\sX f(x_1, x_2) \pi_n(\d x_2 \mid x_1) \pi_1(\d x_1) - \int_\sX f(x_1, x_2) \pi(\d x_2 \mid x_1) \pi_1(\d x_1) \right \rvert \to 0,
\]
 as $n \to \infty$ in $Q^{(n)}$-probability, as desired.
\end{proof}
\begin{proof}[Proof of Theorem \ref{theorem_one_level_nested}]
Lemma \ref{asymptotic_distribution_psi} shows that $\tilde{\psi}$ converges to a Normal distribution with zero mean and non-singular covariance matrix $\mathcal{I}^{-1}(\psi^*)$. Similarly, Lemma \ref{asymptotic_distribution_T} shows that, conditional to every $\tilde{\psi}$, $\tilde{\bm{T}}$ converges to a Normal distribution with mean and variance (denoted by $E_\infty[\cdot]$ and $\Var_\infty(\cdot) \,$) given by
\[
E_\infty[\tilde{\bm{T}} \mid \tilde{\psi}] = C(\psi^*)\tilde{\psi}, \quad \Var_\infty\left(\tilde{\bm{T}} \mid \tilde{\psi}\right) = V(\psi^*).
\]
Therefore, by Lemma \ref{lemma_marginal_conditional}, we conclude that $\left(\tilde{\bm{T}}, \tilde{\psi}\right)$ converges in total variation to a $(S+D)$-dimensional Gaussian distribution $\tilde{\pi}$ with zero mean and covariance matrix $\Sigma$ given by 
\[
\Sigma = 
\begin{bmatrix}
\Sigma_{\tilde{\bm{T}}} & \Sigma_{\tilde{\psi}\tilde{\bm{T}}}^\top\\
 \Sigma_{\tilde{\psi}\tilde{\bm{T}}} & \Sigma_{\tilde{\psi}}
\end{bmatrix},
\]
where $\Sigma_{\tilde{\psi}} = \mathcal{I}^{-1}(\psi^*) \in \R^{D \times D}$ and $\Sigma_{\tilde{\bm{T}}}\in \R^{S \times S}$ are the limiting variances of $\tilde{\psi}$ and $\tilde{\bm{T}}$, while $\Sigma_{\tilde{\psi}\tilde{\bm{T}}} \in \R^{D \times S}$ is the limiting covariance. Thus, thanks to standard properties of the multivariate Gaussian distribution, the determinant of $\Sigma$ can be computed as
\[
\begin{aligned}
\det (\Sigma) &= \det(\Sigma_{\tilde{\psi}})\det \left(\Sigma_{\tilde{\bm{T}}}- \Sigma_{\tilde{\psi}\tilde{\bm{T}}}^\top\Sigma_{\tilde{\psi}}^{-1} \Sigma_{\tilde{\psi}\tilde{\bm{T}}}\right) = \det(\Sigma_{\tilde{\psi}})\det \left(\Var_\infty\left(\tilde{\bm{T}} \mid \tilde{\psi}\right) \right) \\
&= \det \left(\mathcal{I}^{-1}(\psi^*) \right)\det \left( V(\psi^*) \right),
\end{aligned}
\]
which implies that $\Sigma$ is non singular. Indeed, $\det \left(\mathcal{I}^{-1}(\psi^*)\right) > 0$ by $(B3)$, while $\det \left(V(\psi^*) \right) > 0$ by (B4). Therefore, by Theorem $1$ in \cite{R97}, the Gibbs sampler on the limit Gaussian target has a strictly positive spectral gap. Moreover, since the Gibbs sampler in \eqref{two_blocks_gibbs_nested} has two blocks, by Lemma \ref{positiveDefinite_Gibbs} we have $\tilde{t}_{mix}(\epsilon, M) < \infty$ for every $M$ and $\epsilon$: thus the result follows by Corollary \ref{mixingCorollary}.
\end{proof}
\subsection{Proof of Proposition \ref{limiting_sigma}}

\begin{proof}
Using the notation $E_\infty[\cdot]$,$\Var_\infty(\cdot)$ and $\text{Cov}_\infty(\cdot, \cdot)$ for the limiting mean, variance and covariance, by Propositions \ref{asymptotic_distribution_psi} and \ref{asymptotic_distribution_T} we have
\[
E_\infty[\tilde{\psi}] = \bm{0}_D, \quad \Var_\infty(\tilde{\psi}) =\Fisher^{-1}(\psi^*)
\]
and
\[
E_\infty[\tilde{\bm{T}} \mid \tilde{\psi}] = C(\psi^*)\tilde{\psi}, \quad \Var_\infty\left(\tilde{\bm{T}} \mid \tilde{\psi}\right) = V(\psi^*).
\]
By standard properties of the multivariate Gaussian distribution we have
\[
E_\infty[\tilde{\bm{T}}] = \bm{0}_S, \quad \text{Cov}_\infty\left(\bm{T}, \tilde{\psi} \right) = C(\psi^*)\Var_\infty(\tilde{\psi}) = C(\psi^*)\Fisher^{-1}(\psi^*)
\]
and
\[
\begin{aligned}
\Var_\infty(\bm{T}) &= \Var_\infty\left(\tilde{\bm{T}} \mid \tilde{\psi}\right)+\text{Cov}_\infty\left(\bm{T}, \tilde{\psi} \right)\Var_\infty^{-1}(\tilde{\psi})\text{Cov}^\top_\infty\left(\bm{T}, \tilde{\psi} \right)\\
& = V(\psi^*)+C(\psi^*)\mathcal{I}^{-1}(\psi^*)C^\top(\psi^*),
\end{aligned}
\]
as desired.
\end{proof}

\subsection{Proof of Corollary \ref{spectral_radius}}
We need three preliminary lemmas. The first one is a special version of well-known results (e.g. \cite{R97}).
\begin{lemma}\label{autoregressive}
The Gibbs sampler targeting the distribution in Proposition \ref{limiting_sigma} can be written as
\[
\begin{bmatrix}
\tilde{\bm{T}}^{(t)}\\
\tilde{\psi}^{(t)}
\end{bmatrix}
= B
\begin{bmatrix}
\tilde{\bm{T}}^{(t-1)}\\
\tilde{\psi}^{(t-1)}
\end{bmatrix}
+
\begin{bmatrix}
U_1\\
U_2
\end{bmatrix},
\]
where
\[
B = 
\begin{bmatrix}
\textbf{O}_{S\times S} &\quad & C(\psi^*) \\
\\
\textbf{O}_{D \times S} &\quad \quad & \mathcal{I}^{-1}(\psi^*)C^\top(\psi^*)\left\{V(\psi^*)+C(\psi^*)\mathcal{I}^{-1}(\psi^*)C^\top(\psi^*)\right\}^{-1}C(\psi^*)
\end{bmatrix}
\]
and
\[
\begin{bmatrix}
U_1\\
U_2
\end{bmatrix} \sim N \left(\textbf{0}_{S+D}, \Sigma-B\Sigma B^\top \right)
\]
\end{lemma}
\begin{proof}
By Proposition \ref{asymptotic_distribution_T} we have
\[
E \left[\tilde{\bm{T}}^{(t)} \mid \tilde{\bm{T}}^{(t-1)},\tilde{\psi}^{(t-1)} \right] = C(\psi^*)\tilde{\psi}^{(t-1)}.
\]
Moreover, by Proposition \ref{limiting_sigma} and standard properties of the multivariate Gaussian distribution, we have
\[
\begin{aligned}
E \bigl[\tilde{\psi}^{t} &\mid \tilde{\bm{T}}^{(t-1)},\tilde{\psi}^{(t-1)} \bigr] \\
&= E\left[\mathcal{I}^{-1}(\psi^*)C^\top(\psi^*)\left\{V(\psi^*)+C(\psi^*)\mathcal{I}^{-1}(\psi^*)C^\top(\psi^*)\right\}^{-1}\tilde{\bm{T}}^{(t)} \mid \tilde{\bm{T}}^{(t-1)},\tilde{\psi}^{(t-1)} \right]\\
&= \mathcal{I}^{-1}(\psi^*)C^\top(\psi^*)\left\{V(\psi^*)+C(\psi^*)\mathcal{I}^{-1}(\psi^*)C^\top(\psi^*)\right\}^{-1}C(\psi^*)\tilde{\psi}^{(t-1)},
\end{aligned}
\]
as desired.
\end{proof}
\begin{lemma}\label{auxiliary_spectral_radius}
Let
\[
M = 
\begin{bmatrix}
\textbf{O}_{S\times S} & A \\
\\
\textbf{O}_{D \times S} & W
\end{bmatrix},
\]
with $A \in \mathbb{R}^{S \times D}$ and $W \in \mathbb{R}^{D \times D}$. Then $M$ and $W$ have the same non null eigenvalues.
\end{lemma}
\begin{proof}
Let $\mu \neq 0$ be an eigenvalue of $M$, with eigenvector $x = [x_S^\top, x_D^\top]^\top$. We have
\[
Mx = \mu x \quad \Leftrightarrow \quad
\begin{bmatrix}
Ax_D\\
Wx_D
\end{bmatrix}
=
\begin{bmatrix}
\mu x_S\\
\mu x_D
\end{bmatrix},
\]
so that $\mu$ is an eigenvalue of $W$ with eigenvector $x_D$. Indeed, $x_D$ is different from the null vector, since $\mu \neq 0$.

Let $\lambda \neq 0$ be an eigenvalue of $W$ with eigenvector $x_D$. Then
\[
M
\begin{bmatrix}
\frac{Ax_D}{\lambda}\\
x_D
\end{bmatrix}
 = 
\begin{bmatrix}
Ax_D\\
Wx_D
\end{bmatrix}
=
\lambda
\begin{bmatrix}
\frac{Ax_D}{\lambda}\\
x_D
\end{bmatrix},
\]
so that $\lambda$ is an eigenvalue of $M$, with eigenvector
\[
\begin{bmatrix}
\frac{Ax_D}{\lambda}\\
x_D
\end{bmatrix},
\]
as desired.
\end{proof}
\begin{lemma}\label{eig_product}
Let $A \in \mathbb{R}^{D \times S}$ and $B \in \mathbb{R}^{S \times D}$. Then the matrices $AB$ and $BA$ have the same non-null eigenvalues. 
\end{lemma}
\begin{proof}
Let $\lambda \neq 0$ be an eigenvalue of $AB$, with eigenvector $v \in \mathbb{R}^D$. Then
\[
\lambda Bv = B(AB)v = (BA)Bv.
\]
Since $Bv \neq \textbf{0}$ we conclude that $\lambda$ is an eigenvalue of $BA$ with eigenvector $Bv$.
\end{proof}
\begin{proof}[Proof of Corollary \ref{spectral_radius}]
With $B$ as in Lemma \ref{autoregressive}, by Theorem $1$ in \cite{R97} the spectral gap of the Gibbs sampler with operator $\tilde{P}$ is given by 
\[
\text{Gap}(\tilde{P}) = \min \left\{1-|\lambda_i| \, : \, \lambda_i \text{ eigenvalue of } B \right\}
\]
Thus, by Lemma \ref{auxiliary_spectral_radius}, with $M := B$ and 
\[
W = \mathcal{I}^{-1}(\psi^*)C^\top(\psi^*)\left\{V(\psi^*)+C(\psi^*)\mathcal{I}^{-1}(\psi^*)C^\top(\psi^*)\right\}^{-1}C(\psi^*),
\]
we have
\[
\text{Gap}(\tilde{P}) = \min \left\{1-|\lambda_i| \, : \, \lambda_i \text{ eigenvalue of } \mathcal{I}^{-1}(\psi^*)C^\top(\psi^*)\left\{V(\psi^*)+C(\psi^*)\mathcal{I}^{-1}(\psi^*)C^\top(\psi^*)\right\}^{-1}C(\psi^*) \right\}.
\]
By Lemma \ref{eig_product} with
\[
A = \mathcal{I}^{-1}(\psi^*)C^\top(\psi^*), \quad B = \left\{V(\psi^*)+C(\psi^*)\mathcal{I}^{-1}(\psi^*)C^\top(\psi^*)\right\}^{-1}C(\psi^*)
\]
we deduce
\[
\text{Gap}(\tilde{P}) = \min \left\{1-|\lambda_i| \, : \, \lambda_i \text{ eigenvalue of } \left\{V(\psi^*)+C(\psi^*)\mathcal{I}^{-1}(\psi^*)C^\top(\psi^*)\right\}^{-1}C(\psi^*)\mathcal{I}^{-1}(\psi^*)C^\top(\psi^*) \right\}.
\]
Notice that
\[
\begin{aligned}
&\left\{V(\psi^*)+C(\psi^*)\mathcal{I}^{-1}(\psi^*)C^\top(\psi^*)\right\}^{-1}C(\psi^*)\mathcal{I}^{-1}(\psi^*)C^\top(\psi^*)\\
& = I-\left\{V(\psi^*)+C(\psi^*)\mathcal{I}^{-1}(\psi^*)C^\top(\psi^*)\right\}^{-1}V(\psi^*).
\end{aligned}
\]
Since $\lambda$ is an eigenvalue of $A$ if and only if $1-\lambda$ is an eigenvalue of $I - A$, it follows that
\[
\text{Gap}(\tilde{P}) = \min \left\{1-|1-\lambda_i| \, ; \, \lambda_i \text{ eigenvalue of } \left\{V(\psi^*)+C(\psi^*)\mathcal{I}^{-1}(\psi^*)C^\top(\psi^*)\right\}^{-1}V(\psi^*) \right\}.
\]
Moreover the eigenvalues of the inverse are the inverse of the eigenvalues, so that the rate of convergence is equal to
\[
\text{Gap}(\tilde{P}) = \min \left\{1-\left \lvert 1-\frac{1}{\lambda_i}\right \rvert \, ; \, \lambda_i \text{ eigenvalue of } V^{-1}(\psi^*)\left\{V(\psi^*)+C(\psi^*)\mathcal{I}^{-1}(\psi^*)C^\top(\psi^*)\right\} \right\}.
\]
Since 
\[
\begin{aligned}
V^{-1}(\psi^*)&\left\{V(\psi^*)+C(\psi^*)\mathcal{I}^{-1}(\psi^*)C^\top(\psi^*)\right\}= I+V^{-1}(\psi^*)C(\psi^*)\mathcal{I}^{-1}(\psi^*)C^\top(\psi^*),
\end{aligned}
\]
we have
\[
\text{Gap}(\tilde{P}) = \min \left\{1-\left \lvert 1-\frac{1}{1+\lambda_i}\right \rvert \, ; \, \lambda_i \text{ eigenvalue of } V^{-1}(\psi^*)C(\psi^*)\mathcal{I}^{-1}(\psi^*)C^\top(\psi^*) \right\}.
\]
Moreover both $V^{-1}(\psi^*)$ and $C(\psi^*)\mathcal{I}^{-1}(\psi^*)C^\top(\psi^*)$ are positive semi-definite, so that also their product is positive semi-definite and has positive eigenvalues. Therefore we conclude
\[
\text{Gap}(\tilde{P}) = \min \left\{\frac{1}{1+\lambda_i} \, ; \, \lambda_i \text{ eigenvalue of } V^{-1}(\psi^*)C(\psi^*)\mathcal{I}^{-1}(\psi^*)C^\top(\psi^*) \right\}
\]
and the result follows by Corollary \ref{mixing_gap}.
\end{proof}

\subsection{Proof of Corollary \ref{spectral_single}}
We need a preliminary lemma, that we prove for self-containedness.
\begin{lemma}\label{moments_exp}
Let $p(\theta \mid \psi)$ be as in \eqref{exponential_family}. Then it holds
\[
E[T(\theta) \mid \psi] = \frac{\partial_\psi A(\psi)}{\partial_\psi\eta(\psi)}, \quad \text{Var}(T(\theta) \mid \psi) = \left\{\partial^2_\psi A(\psi)-\frac{\partial^2_\psi \eta(\psi)\partial_\psi A(\psi)}{\partial_\psi \eta(\psi)} \right\}\left[\partial_\psi \eta(\psi) \right]^{-2}.
\]
\end{lemma}
\begin{proof}
Differentiating the following equality
\begin{equation}\label{basic_equality}
 1 = \int p(\theta \mid \psi) \, \d \theta,
\end{equation}
by the regularity properties of the exponential family we get
\[
0 = \int \partial_\psi p(\theta \mid \psi) \, \d \theta = \partial_\psi \eta(\psi)E[T(\theta) \mid \psi]+\partial_\psi A(\psi),
\]
and the formula for the expected value follows. As regards the variance, differentiating \eqref{basic_equality} twice, we obtain
\[
0 = \partial^2_\psi \eta(\psi)E[T(\theta) \mid \psi]-\partial^2_\psi A(\psi)+\left[\partial_\psi \eta(\psi) \right]^{2}E[T^2(\theta) \mid \psi]-2\left[\partial_\psi \eta(\psi) \right]^{2}E^2[T(\theta) \mid \psi]+\left[\partial_\psi A(\psi) \right]^{2}.
\]
Noticing that
\[
\left[\partial_\psi \eta(\psi) \right]^{2}E^2[T(\theta) \mid \psi] = \left[\partial_\psi A(\psi) \right]^{2}
\]
and rearranging, we get
\[
\partial^2_\psi A(\psi)-\partial^2_\psi \eta(\psi)E[T(\theta) \mid \psi] = \left[\partial_\psi \eta(\psi) \right]^{2}\text{Var}(T(\theta) \mid \psi),
\]
from which the result follows.
\end{proof}
\begin{proof}[Proof of Corollary \ref{spectral_single}]
By Corollary \ref{spectral_radius}, we have
\[
\gamma(\psi^*) = \frac{1}{1+\lambda} \quad\hbox{with } \lambda = \frac{C^2(\psi^*)}{V(\psi^*)\Fisher(\psi^*)},
\]
where
\[
\begin{aligned}
&C(\psi) = E_{Y_j}\left[\partial_\psi E[T(\theta_j) \mid Y_j, \psi] \right], \\
& V(\psi) = E_{Y_j}\left[\text{Var}(T(\theta_j) \mid Y_j, \psi) \right],\\
&\Fisher(\psi) = -E_{Y_j}\left[ \partial^2_{\psi}\log g(Y_j \mid \psi)\right],
\end{aligned}
\]
with $g(y \mid \psi)$ as in \eqref{likelihood_data}. As regards $C(\psi)$, notice that
\[
\begin{aligned}
\partial_\psi E[T(\theta) \mid Y, \psi] 
 =&
 \frac{\int T(\theta) f(Y \mid \theta) \partial_\psi p(\theta \mid \psi) \, \d \theta}{g(Y \mid \psi)}
 -\\&
 \frac{\left[\int T(\theta) f(Y \mid \theta) p(\theta \mid \psi) \, \d \theta\right]\left[\int f(Y \mid \theta) \partial_\psi p(\theta \mid \psi) \, \d \theta\right]}{g^2(Y \mid \psi)} 
\\=&
 \partial_\psi \eta(\psi) E \left[ T^2(\theta) \mid Y, \psi\right]-\partial_\psi \eta(\psi) E^2 \left[ T(\theta) \mid Y, \psi\right] 
\\=&
 \partial_\psi \eta(\psi) \text{Var}\left(T(\theta) \mid Y, \psi \right).
\end{aligned}
\]
Therefore
\begin{equation}\label{C_computation}
C^2(\psi^*) = \left[\partial_\psi \eta(\psi^*)\right]^2 E ^2_{Y_j}\left[\text{Var}\left(T(\theta_j) \mid Y_j, \psi^* \right)\right].
\end{equation}
As regards $\Fisher(\psi)$, notice that
\[
\begin{aligned}
\partial_{\psi}\log g(Y_j \mid \psi) &= \frac{\int f(Y \mid \theta) \partial_\psi p(\theta \mid \psi) \, \d \theta}{g(Y \mid \psi)} = \partial_\psi \eta(\psi)\frac{\int T(\theta) f(Y \mid \theta)p(\theta \mid \psi) \, \d \theta}{g(Y \mid \psi)}-\partial_\psi A(\psi)
\end{aligned}
\]
and
\[
\begin{aligned}
\partial^2_{\psi}\log g(Y_j \mid \psi) &= \partial^2_\psi \eta (\psi) E \left[T(\theta) \mid Y, \psi \right] -\partial^2_\psi A(\psi)+\partial_\psi \eta(\psi)\frac{\int T(\theta) f(Y \mid \theta)\partial_\psi p(\theta \mid \psi) \, \d \theta}{g(Y \mid \psi)}\\
& -\partial_\psi \eta(\psi)\frac{\left[\int T(\theta) f(Y \mid \theta) p(\theta \mid \psi) \, \d \theta\right]\left[\int f(Y \mid \theta) \partial_\psi p(\theta \mid \psi) \, \d \theta\right]}{g^2(Y \mid \psi)}\\
& = \partial^2_\psi \eta (\psi) E \left[T(\theta) \mid Y, \psi \right] -\partial^2_\psi A(\psi)+\left[\partial_\psi \eta(\psi)\right]^2\text{Var} \left(T(\theta) \mid Y, \psi \right).
\end{aligned}
\]
Noticing that, by Lemma \ref{moments_exp}, we have
\[
\begin{aligned}
\partial^2_\psi \eta (\psi) E \left[T(\theta) \mid Y, \psi \right] -\partial^2_\psi A(\psi) &= \left\{\partial^2_\psi A(\psi)-\frac{\partial^2_\psi \eta(\psi)\partial_\psi A(\psi)}{\partial_\psi \eta(\psi)} \right\} \\
&= \left[\partial_\psi \eta(\psi)\right]^2\text{Var} \left(T(\theta) \mid \psi \right),
\end{aligned}
\]
we get
\begin{equation}\label{fisher_computation}
\begin{aligned}
\Fisher(\psi^*) &= \left[\partial_\psi \eta(\psi^*)\right]^2\text{Var} \left(T(\theta_j) \mid \psi^* \right)- \left[\partial_\psi \eta(\psi^*)\right]^2E_{Y_j}\left[\text{Var} \left(T(\theta_j) \mid Y_j,\psi^* \right)\right] \\
&= \left[\partial_\psi \eta(\psi^*)\right]^2\text{Var}_{Y_j}\left(E \left[T(\theta_j) \mid Y_j,\psi^* \right]\right),
\end{aligned}
\end{equation}
by the Law of Total Variance. Combining \eqref{C_computation} and \eqref{fisher_computation}, it holds
\[
\lambda = \frac{E ^2_{Y_j}\left[\text{Var}\left(T(\theta_j) \mid Y_j, \psi^* \right)\right]}{V(\psi^*)\text{Var}_{Y_j}\left(E \left[T(\theta_j) \mid Y_j,\psi^* \right]\right)} = \frac{E_{Y_j}\left[\text{Var}\left(T(\theta_j) \mid Y_j, \psi^* \right)\right]}{\text{Var}_{Y_j}\left(E \left[T(\theta_j) \mid Y_j,\psi^* \right]\right)}.
\]
The expression for $\gamma(\psi^*)$ follows by rearranging and applying the Law of Total Variance.
\end{proof}

\subsection{Proof of Proposition \ref{prop_normal_model}}
First of all notice that, by Bayes' Theorem, we have
\begin{equation}\label{posterior_theta}
\theta_j \mid Y_j, \mu, \tau_1 \overset{\text{ind.}}{\sim} N\left(m_j, (m\tau_0+\tau_1)^{-1}\right),
\end{equation}
where 
\[
m_j = \frac{m\tau_0}{m\tau_0+\tau_1}\bar{Y}_j+\frac{\tau_1}{m\tau_0+\tau_1}\mu.
\]
Recall that by (B1) we have
\[
Y_j \simiid g(\cdot \mid \psi^*) = N \left( \mu^*, (\tau_0^*)^{-1}I + (\tau_1^*)^{-1}\H \right),
\]
so that
\begin{equation}\label{distribution_mean}
\bar{Y}_j = \frac{1}{m}\sum_{i = 1}^mY_{j,i} \simiid N\left(\mu^*,  \frac{1}{\tau_1^*}+\frac{1}{m\tau_0^*}\right).
\end{equation}
Moreover we need some preliminary lemmas.
\begin{lemma}\label{normal_moments}
Let $X \sim N(\nu, \sigma^2)$. Then
\[
E[X^p] = \sum_{i = 0}^p\binom{p}{i}\nu^i\sigma^{p-i}E[Z^{p-i}],
\]
where $Z \sim N(0, 1)$ and 
\[
E[Z^{s}] = 
\begin{cases} 
0 \quad \text{if $s$ is odd}\\
2^{-s/2}\frac{s!}{(s/2)!} \quad \text{if $s$ is even}
\end{cases}
\]
\end{lemma}
\begin{proof}
The result follows by noticing $X = \nu+\sigma Z$ and applying Netwon's Binomial Theorem.
\end{proof}
\begin{lemma}\label{determinant}
Let $A$ be $m \times m$ matrix such that $A = aI+b\H$, with $a \neq b$ and $a \neq (1-m)b$. Then det$(A) = [a+mb]a^{m-1}$ and $A^{-1} = \frac{1}{a}\I-\frac{b}{a(a+mb)}\H$.
\end{lemma}
\begin{proof}
We start by the determinant
\[
\begin{aligned}
\det
\begin{pmatrix}
c & d & \cdots & d \\
d & c & \cdots & d \\
\vdots  & \vdots  & \ddots & \vdots  \\
d & d & \cdots & c 
\end{pmatrix} &= [c+(m-1)d] \det
\begin{pmatrix}
1 & 1 & \cdots & 1 \\
d & c & \cdots & d \\
\vdots  & \vdots  & \ddots & \vdots  \\
d & d & \cdots & c 
\end{pmatrix}\\
& \medskip \\
&= [c+(m-1)d]\begin{pmatrix}
1 & 1 & \cdots & 1 \\
0 & c-d & \cdots & 0 \\
\vdots  & \vdots  & \ddots & \vdots  \\
0 & 0 & \cdots & c-d 
\end{pmatrix} =  [c+(m-1)d](c-d)^{m-1},
\end{aligned}
\]
where the first equality comes by adding to the first row all the others, while the second comes by subtracting the first row (scaled by $d$) from all the others. In our case $c = a+b$ and $d = b$, that is $\det(A) = [a+mb]a^{m-1}$, as desired. With our assumptions we get that the determinant is different from zero.

As regards the inverse we prove $A^{-1} = x I + y\H$ for suitable $x$ and $y$. Indeed
\[
\left(aI + b\H \right)\left(xI+y\H \right) = axI+ay\H+bx\H+by\H^2 = axI+(ay+bx+mby)\H.
\]
Setting the above equal to $I$, we obtain $x=1/a$ and 
\[
ay+bx+mby = 0
\quad \Rightarrow \quad 
y(a+mb) = -\frac{b}{a}
\quad \Rightarrow \quad 
y = -\frac{b}{a(a+mb)}
\]
as desired.
\end{proof}
\begin{lemma}\label{fisher_normal_hier}
Consider the marginal likelihood as in \eqref{marginal_normal}, with $\psi^* = (\mu^*, \tau_1^*, \tau_0^*)$. Then we have
\begin{equation}\label{eq:fisher_normal_hier}
\mathcal{I}(\psi^*) =
\begin{pmatrix}
\frac{m\tau_0^*\tau_1^*}{\tau_1^*+m\tau_0^*} & 0 & 0 \\
0 
&  \frac{m^2(\tau_0^*)^2}{2(\tau_1^*)^2(\tau_1^*+m\tau_0^*)^2} 
& \frac{m}{2(\tau_1^*+m\tau_0^*)^2}  
\\
0 
& \frac{m}{2(\tau_1^*+m\tau_0^*)^2} 
& \frac{m-1}{2(\tau_0^*)^2}+\frac{(\tau_1^*)^2}{2(\tau_0^*)^2(\tau_1^*+m\tau_0^*)^2} 
\end{pmatrix}
\end{equation}
\end{lemma}
\begin{proof}
The log--likelihood $l(\psi) = \log g(y \mid \psi)$ is given by
\[
l(\mu, \tau_0, \tau_1) = -\frac{1}{2}\log 2\pi - \frac{1}{2} \log \left(\det(\Sigma) \right)-\frac{1}{2}(Y_1-\mu I)^t\Sigma^{-1}(Y_1-\mu I),
\]
with $\Sigma = \tau_0^{-1}I+\tau_1^{-1}\H$. By Lemma \ref{determinant} with $a = \tau_0^{-1}$ and $b = \tau_1^{-1}$ we have
\[
\det(\Sigma) = [\tau_0^{-1}+m\tau_1^{-1}](\tau_0^{-1})^{m-1}, \quad \Sigma^{-1} = \tau_0I-\frac{\tau_0^2}{\tau_1+m\tau_0}\H.
\]
Thus, the log--likelihood becomes
\[
\begin{aligned}
l(\mu, \tau_0, \tau_1) =& -\frac{1}{2}\log 2\pi +\frac{m-1}{2}\log \tau_0-\frac{1}{2}\log(\tau_0^{-1}+m\tau_1^{-1})-\frac{\tau_0}{2}\sum_{i = 1}^m(Y_{1,i}-\mu)^2\\
&+\frac{\tau_0^2}{2(\tau_1+m\tau_0)}(Y_1-\mu I)^t\H(Y_1-\mu I).
\end{aligned}
\]
Rewriting the last expression we get
\[
\begin{aligned}
l(\mu, \tau_0, \tau_1) =& -\frac{1}{2}\log 2\pi +\frac{m-1}{2}\log \tau_0-\frac{1}{2}\log(\tau_0^{-1}+m\tau_1^{-1})-\frac{\tau_0}{2}\sum_{i = 1}^m(Y_{1,i}-\mu)^2\\
&+\frac{\tau_0^2}{2(\tau_1+m\tau_0)}\left(\sum_{i = 1}^m(Y_{1,i}-\mu) \right)^2.
\end{aligned}
\]
The required derivatives are given by
\[
\begin{aligned}
&\frac{\partial^2 l}{\partial \mu^2} = 
-\frac{m\tau_0\tau_1}{\tau_1+m\tau_0},
\quad
\frac{\partial^2 l}{\partial \tau_1^2} = -\frac{m\tau_0(2\tau_1+m\tau_0)}{2\tau_1^2(\tau_1+m\tau_0)^2}+\frac{\tau_0^2}{(\tau_1+m\tau_0)^3}\left(\sum_{i = 1}^m(Y_{1,i}-\mu) \right)^2,\\
&\frac{\partial^2 l}{\partial \tau_0^2} = - \frac{m-1}{2 \tau_0^2}-\frac{\tau_1(\tau_1+2m\tau_0)}{2\tau_0^2(\tau_1+m\tau_0)^2}+\frac{(\tau_1+m\tau_0)^2-2m\tau_0\tau_1-m^2\tau_0^2}{(\tau_1+m\tau_0)^3}\left(\sum_{i = 1}^m(Y_{1,i}-\mu) \right)^2, \\
& \frac{\partial^2 l}{\partial \mu \partial \tau_0} = \sum_{i = 1}^m(Y_{1,i}-\mu)-\frac{2m\tau_0\tau_1+m^2\tau_0^2}{(\tau_1+m\tau_0)^2}\sum_{i = 1}^m(Y_{1,i}-\mu),\\
& \frac{\partial^2 l}{\partial \mu \partial \tau_1} = \frac{\tau_0^2}{(\tau_1+m\tau_0)^2}\sum_{i = 1}^m(Y_{1,i}-\mu),\quad
\frac{\partial^2 l}{\partial \tau_0 \partial \tau_1} = \frac{m}{2(\tau_1+m\tau_0)^2}-\frac{\tau_0\tau_1}{(\tau_1+m\tau_0)^3}\left(\sum_{i = 1}^m(Y_{1,i}-\mu) \right)^2.
\end{aligned}
\]
The entries of the Fisher Information matrix reported in \eqref{eq:fisher_normal_hier} can then be computed from the above expressions by taking expectations with respect to $Y_1$ and exploiting that
\[
\begin{aligned}
&\E[Y_{1,i}-\mu] = 0, \quad 
\E\left[(Y_{1,i}-\mu)^2\right] = Var(Y_{1,i}-\mu) = 
 \frac{\tau_0+\tau_1}{\tau_0\tau_1},\\
&\begin{aligned}
\E \left[\left(\sum_{i = 1}^m(Y_{1,i}-\mu) \right)^2 \right] &= Var \left(\sum_{i = 1}^m(Y_{1,i}-\mu) \right) = [1, \dots, 1]Var(Y_1)[1, \dots, 1]^t
\\
& = [1, \dots, 1]\left(\tau_0^{-1}I+\tau_1^{-1}\H \right)[1, \dots, 1]^t 
\\
&= m\left(\frac{m\tau_0+\tau_1}{\tau_0\tau_1} \right).
\end{aligned}
\end{aligned}
\]
Thus we can compute the entries of the Fisher Information matrix as
\[
\begin{aligned}
&\begin{aligned}
\E\left[\frac{\partial^2 l}{\partial \tau_0^2}\right] &= - \frac{m-1}{2 \tau_0^2}-\frac{\tau_1(\tau_1+2m\tau_0)}{2\tau_0^2(\tau_1+m\tau_0)^2}+\frac{m(\tau_1+m\tau_0)^2-2m^2\tau_0\tau_1-m^3\tau_0^2}{\tau_0\tau_1(\tau_1+m\tau_0)^2}\\
&= -\frac{m-1}{2\tau_0^2}-\frac{\tau_1^2}{2\tau_0^2(\tau_1+m\tau_0)^2},
\end{aligned}\\ 
&\E \left[\frac{\partial^2 l}{\partial \tau_1^2}\right] = -\frac{m\tau_0(2\tau_1+m\tau_0)}{2\tau_1^2(\tau_1+m\tau_0)^2}+\frac{m\tau_0}{\tau_1(\tau_1+m\tau_0)^2} = -\frac{m^2\tau_0^2}{2\tau_1^2(\tau_1+m\tau_0)^2},\\
&\E \left[ \frac{\partial^2 l}{\partial \mu \partial \tau_0}\right] = 0, \quad \E \left[ \frac{\partial^2 l}{\partial \mu \partial \tau_1}\right] = 0,\\
&\E \left[\frac{\partial^2 l}{\partial \tau_0 \partial \tau_1}\right] = \frac{m}{2(\tau_1+m\tau_0)^2}-\frac{m}{(\tau_1+m\tau_0)^2} = -\frac{m}{2(\tau_1+m\tau_0)^2},
\end{aligned}
\]
as desired.
\end{proof}
\begin{lemma}\label{auxiliary_characteristic}
Let $X \sim N(\nu, \sigma^2)$. Then
\[
\left \lvert E \left[e^{i(aX^2+bX)} \right] \right \rvert \leq  \frac{e^{-\frac{\sigma^2}{2}\frac{(2\nu a + b)^2}{1+4a^2\sigma^4}}}{\left(1+4a^2\sigma^4\right)^{1/4}},
\]
for every $(a, b) \in \R_2$.
\end{lemma}
\begin{proof}
By definition of expectation we have
\[
\begin{aligned}
E \left[e^{i(aX^2+bX)} \right] &= \int_{\R} e^{i(az^2+bz)}\frac{1}{\sqrt{2\pi \sigma^2}}e^{-\frac{(z-\nu)^2}{2\sigma^2}} \, \d z = \frac{e^{-\frac{\nu^2}{2\sigma^2}}}{\sqrt{2\pi\sigma^2}}\int_\R e^{-\frac{1}{2}\left[z^2\left(\frac{1}{\sigma^2}-2ia \right)-2z\left(\frac{\nu}{\sigma^2}+ib \right) \right]}\, \d z
\end{aligned}
\]
Notice that
\[
\begin{aligned}
z^2\left(\frac{1}{\sigma^2}-2ia \right)-2z\left(\frac{\nu}{\sigma^2}+ib \right) &= \left(\frac{1-2ia\sigma^2}{\sigma^2}\right)\left[z^2-2z \frac{\nu+ib\sigma^2}{1-2ia\sigma^2} +\left(\frac{\nu+ib\sigma^2}{1-2ia\sigma^2} \right)-\left(\frac{\nu+ib\sigma^2}{1-2ia\sigma^2}\right)^2 \right]\\
&= \left(\frac{1-2ia\sigma^2}{\sigma^2}\right)\left(z-\frac{\nu+i\sigma^2b}{1-2ia\sigma^2} \right)^2-\frac{(\nu+ib\sigma^2)^2}{\sigma^2(1-2ia\sigma^2)},
\end{aligned}
\]
so that
\[
\frac{1}{\sqrt{2\pi \sigma^2}}\int_\R e^{-\frac{1}{2}\left[z^2\left(\frac{1}{\sigma^2}-2ia \right)-2z\left(\frac{\nu}{\sigma^2}+ib \right) \right]}\, \d z = \frac{e^{\frac{(\nu+ib\sigma^2)^2}{2\sigma^2(1-2ia\sigma^2)}}}{\sqrt{1-2ia\sigma^2}}.
\]
Finally, we get
\begin{equation}\label{auxiliary_car}
E \left[e^{i(aX^2+bX)} \right] = e^{-\frac{\nu^2}{2\sigma^2}}\frac{e^{\frac{(\nu+ib\sigma^2)^2}{2\sigma^2(1-2ia\sigma^2)}}}{\sqrt{1-2ia\sigma^2}}.
\end{equation}
With simple computations we obtain
\[
\begin{aligned}
\frac{(\nu+ib\sigma^2)^2}{2\sigma^2(1-2ia\sigma^2)} &= \frac{(\nu^2+2i\nu b\sigma^2-b^2\sigma^4)(1+2ia\sigma^2)}{2\sigma^2(1+4a^2\sigma^4)}\\
&= \frac{\nu^2+2i\nu b\sigma^2-b^2\sigma^4+2i\nu^2a\sigma^2-4\nu ab\sigma^2-2i\sigma^6ab^2}{2\sigma^2(1+4a^2\sigma^4)}\\
&= \frac{\nu^2+2i(\nu b\sigma^2+\nu^2a\sigma^2-\sigma^6ab^2)-4\nu ab\sigma^4-\sigma^4b^2}{2\sigma^2(1+4a^2\sigma^4)}.
\end{aligned}
\]
Thus, by \eqref{auxiliary_car} we can write
\[
E \left[e^{i(aX^2+bX)} \right] = e^{-\frac{\nu^2}{2\sigma^2}}\frac{e^{\frac{(\nu+ib\sigma^2)^2}{2\sigma^2(1-2ia\sigma^2)}}}{\sqrt{1-2ia\sigma^2}},
\]
that implies
\[
\left \lvert E \left[e^{i(aX^2+bX)} \right] \right \rvert \leq \frac{e^{-\frac{4\nu^2 a^2\sigma^4+4\nu ab\sigma^4+b^2\sigma^4}{2\sigma^2(1+4a^2\sigma^4)}}}{| \sqrt{1-2ia\sigma^2}|} = \frac{e^{-\frac{\sigma^2}{2}\frac{(2\nu a + b)^2}{1+4a^2\sigma^4}}}{\left(1+4a^2\sigma^4\right)^{1/4}},
\]
as desired.
\end{proof}
Define
\begin{equation}\label{eq:psi_T_normal}
\psi = (\mu, \tau_1) \quad \text{and} \quad \bm{T} = \bm{T}(\bm{\theta}) = \left( \sum_{j = 1}^J\theta_j, \sum_{j = 1}^J(\theta_j-\mu^*)^2\right).
\end{equation}
Next three lemmas show that assumptions $(B1)-(B6)$ are satisfied for $(\bm{T}, \psi)$ as defined above.
\begin{lemma}\label{lemma_normal_B1B3}
Consider the setting of Proposition \ref{prop_normal_model}. Then assumptions $(B1)-(B3)$ are satisfied for $(\bm{T}, \psi)$ as in \eqref{eq:psi_T_normal}.
\end{lemma}
\begin{proof}
It is easy to show that assumption $(B1)$ is satisfied, with $g(\cdot)$ as in \eqref{marginal_normal}. As regards $(B2)$, suitable tests can be defined analogously to Lemma \ref{tests_normal}.

Finally, by Lemma \ref{fisher_normal_hier}, the Fisher Information is given by
\[
\frac{m\tau_0^*\tau_1^*}{m\tau^*_0+\tau_1^*}
\]
for $l=1$ and by
\[
\begin{bmatrix}
\frac{m\tau_0^*\tau_1^*}{m\tau^*_0+\tau_1^*} & 0\\
0 & \frac{m^2(\tau_0^*)^2}{2(\tau_1^*)^2(\tau_1^*+m\tau_0^*)^2}
\end{bmatrix},
\]
for $l = 2,3$. Therefore $(B3)$ is satisfied for any $\psi^*$.
\end{proof}
\begin{lemma}\label{lemma_normal_B4}
Consider the setting of Proposition \ref{prop_normal_model}. Then assumption $(B4)$ is satisfied for $(\bm{T}, \psi)$ as in \eqref{eq:psi_T_normal}.
\end{lemma}
\begin{proof}
Since $T(\theta_j) = (\theta_j, (\theta_j-\mu^*)^2)$ it holds
\[
M_s^{(p)}(\mu, \tau_1 \mid Y_j) = E\left[\theta_j^{sp} \mid \mu, \tau_1 \right], \quad M_{1,2}^{(1)}(\mu, \tau_1 \mid Y_j) = E\left[\theta_j(\theta_j^*-\mu^*)^2 \mid \mu, \tau_1 \right].
\]
By Lemma \ref{normal_moments} and \eqref{posterior_theta}, we obtain
\[
E\left[\theta_j^{k} \mid \mu, \tau_1 \right] = \sum_{i = 0}^k\binom{k}{i}\left(\frac{m\tau_0}{m\tau_0+\tau_1}\bar{Y}_j+\frac{\tau_1}{m\tau_0+\tau_1}\mu \right)^i\left(\frac{1}{m\tau_0+\tau_1} \right)^{(k-i)/2}E[Z^{k-i}].
\]
It is a finite sum of infinitely times differentiable terms (with respect to $\mu$ and $\tau_1$). Moreover, for every $k \geq 1$, thanks to Lemma \ref{normal_moments} and \eqref{distribution_mean}, $E_{Y_j} \left[|\bar{Y}_j|^k \mid \mu, \tau_1 \right]$ is uniformly bounded over $(\mu, \tau_1)$ belonging to a bounded set.

Therefore, choosing $\delta_4 < \tau_1^*$, it is easy to find $C<\infty$ that satisfies assumption $(B4)$.
\end{proof}

\begin{lemma}\label{lemma_normal_B5B6}
Consider the setting of Proposition \ref{prop_normal_model}. Then assumptions $(B5)$ and $(B6)$ are satisfied for $(\bm{T}, \psi)$ as in \eqref{eq:psi_T_normal}.
\end{lemma}
\begin{proof}
Assume $\mu^* = 0$, the general case follows by similar calculations. Recall that the posterior distribution of $\theta_j$ is given by $N(m_j, \sigma^2)$, with $m_j$ as  in \eqref{posterior_theta} and
\[
\sigma^2 = \frac{1}{m\tau_0+\tau_1}.
\]
By Lemma \ref{auxiliary_characteristic} we have
\begin{equation}\label{ineq_charact}
\left \lvert E \left[e^{i(t_1\theta_j+t_2\theta_j^2)} \mid Y_j, \mu, \tau_1 \right] \right \rvert^2 \leq \frac{e^{-\sigma^2\frac{(2m_j t_2+t_1)^2}{1+4t_2^2\sigma^4}}}{\left(1+4t_2^2\sigma^4 \right)^{1/2}}.
\end{equation}
Moreover, notice that
\[
\begin{aligned}
\int_{\R} e^{-c\sigma^2\frac{(2m_j t_2+t_1)^2}{1+4t_2^2\sigma^4}} \, \d t_1 = \sqrt{\frac{\pi}{c\sigma^2}}\sqrt{1+4t_2^2\sigma^4},
\end{aligned}
\]
for any $c > 0$. Since $\theta_j$ are independent, given $\mu$ and $\tau_1$, by H\"older inequality we write
\[
\begin{aligned}
\int_{\R_2} &\left \lvert E \left[e^{i(t_1\sum_{j = 1}^3\theta_j+t_2\sum_{j = 1}^3\theta_j^2)} \mid Y, \mu, \tau_1 \right] \right \rvert^2 \, \d t_1 \d t_2 = \int_{\R_2} \prod_{j = 1}^3\left \lvert E \left[e^{i(t_1\theta_j+t_2\theta_j^2)} \mid Y_j, \mu, \tau_1 \right] \right \rvert^2 \, \d t_1 \d t_2 \\
&\leq \int_{\R_2} \prod_{j = 1}^3\frac{e^{-\sigma^2\frac{(2m_j t_2+t_1)^2}{1+4t_2^2\sigma^4}}}{\left(1+4t_2^2\sigma^4 \right)^{1/2}} \, \d t_1 \d t_2 = \int_{\R} \frac{1}{\left(1+4t_2^2\sigma^4 \right)^{3/2}} \left(\int_{\R}\prod_{j = 1}^3e^{-\sigma^2\frac{(2\nu_j t_2+t_1)^2}{1+4t_2^2\sigma^4}} \, \d t_1 \right) \, \d t_2\\
& \leq  \int_{\R} \frac{1}{\left(1+4t_2^2\sigma^4 \right)^{3/2}} \prod_{j = 1}^3 \left(\int_{\R}e^{-3\sigma^2\frac{(2\nu_j t_2+t_1)^2}{1+4t_2^2\sigma^4}} \, \d t_1 \right)^{1/3} \, \d t_2\\
& = \sqrt{\frac{\pi}{3\sigma^2}} \int_\R\frac{1}{1+4t_2^2\sigma^4} \, \d t_2.
\end{aligned}
\]
Therefore
\[
\int_{\R^2} \left \lvert\varphi^{(3)}\left(t \mid Y, \psi \right) \right\rvert^2 \, \d t \leq \sqrt{\frac{\pi}{3\sigma^2}} \int_\R\frac{1}{1+4t_2^2\sigma^4} \, \d t_2 < \infty,
\]
where the right hand side does not depend on the data and it is a continuous function of $\mu$ and $\tau_1$. This implies $(B5)$ is satisfied with $k = 3$.

As regards $(B6)$, by Lemma \ref{auxiliary_characteristic} if $t_2 \neq 0$ we have
\[
|\varphi^{(1)}(t \mid Y_j, \mu, \tau_1)| \leq \frac{1}{\left(1+4t_2^2\sigma^4 \right)^{1/4}},
\]
while if $t_2 = 0$ then
\[
|\varphi^{(1)}(t \mid Y_j, \mu, \tau_1)| \leq e^{-\frac{\sigma^2}{2}t_1^2}.
\]
Therefore
\[
|\varphi^{(1)}(t \mid Y_j, \mu, \tau_1)| \leq \max \left\{ \frac{1}{\left(1+4t_2^2\sigma^4 \right)^{1/4}},e^{-\frac{\sigma^2}{2}t_1^2} \right\},
\]
so that
\[
\underset{|t| > \epsilon}{\sup} \, |\varphi^{(1)}(t \mid Y_j, \mu, \tau_1)| \leq \max \left\{ \frac{1}{\left(1+\epsilon^2\sigma^4 \right)^{1/4}},e^{-\frac{\sigma^2}{8}\epsilon^2} \right\},
\]
since at least one between $t_1$ and $t_2$ must be larger than $\epsilon/2$. Notice that the right hand side does not depend on $Y_j$ and is strictly smaller than $1$ for every triplet $(\mu, \tau_1, \tau_0)$. Since $\sigma^2$ is a continuous function of $\mu$ and $\tau_1$, assumption $(B6)$ is satisfied by choosing $\delta_6 < \tau_1^*$ and $k'=1$.

\end{proof}
\begin{proof}[Proof of Proposition \ref{prop_normal_model}]
The result for $P_1$ follows directly by Theorem \ref{theorem_one_level_nested}, whose assumptions are satisfied by Lemmas \ref{lemma_normal_B1B3}, \ref{lemma_normal_B4} and \ref{lemma_normal_B5B6}. As regards $P_2$ and $P_3$, they are not particular cases of Theorem \ref{theorem_one_level_nested}, since the two operators are different by the one in \eqref{two_blocks_gibbs_nested}. However, the result follows by very similar arguments, that we briefly summarize. Since by construction
\[
\L\left(\d\psi \mid \bm{\theta}, Y_{1:J} \right) = \L\left(\d\psi \mid \bm{T}(\bm{\theta}), Y_{1:J} \right)
\]
a direct analogue of Lemma \ref{sufficient_lemma} holds. Moreover, following the proof of Theorem \ref{theorem_one_level_nested}, Lemmas \ref{asymptotic_distribution_psi}, \ref{asymptotic_distribution_T} and \ref{lemma_marginal_conditional} hold for $\bm{T}$ in \eqref{eq:psi_T_normal}. Finally, Corollary \ref{spectral_normal} proves that the limiting spectral gaps associated to $P_2$ and $P_3$ are strictly positive: by Lemma \ref{positiveDefinite_Gibbs} this implies $\tilde{t}_{mix}(\epsilon, M) < \infty$ for $P_2$, being a two-block Gibbs sampler. The same holds for $P_3$, since in the limit it can be reduced to a two-block Gibbs sampler, as it will be clear by the proof of Corollary \ref{spectral_normal}.
\end{proof}

\subsection{Proof of Corollary \ref{spectral_normal}}
We split the proof in two different cases.
\subsubsection{Proof of Corollary \ref{spectral_normal} for $\gamma_1(\psi^*)$}
\begin{proof}
By Corollary \ref{spectral_single}, the spectral gap is equal to
\[
\gamma_1(\psi) = \frac{\text{Var}_{Y_j}\left(E \left[\theta_j \mid \psi,Y_j \right] \right)}{\text{Var}\left(\theta_j \mid \psi \right)}.
\]
By \eqref{posterior_theta} and \eqref{distribution_mean} we have
\[
\begin{aligned}
&\text{Var}_{Y_j}\left(E \left[\theta_j \mid \psi,Y_j \right]\right) = \left(\frac{m\tau_0}{m\tau_0+\tau_1}\right)^2,\quad
\text{Var}\left(\bar{Y}_j\right) = \frac{m\tau_0}{\tau_1(m\tau_0+\tau_1)},
\end{aligned}
\]
and $\text{Var}\left(\theta_j \mid \psi \right) = \tau_1^{-1}$, 
that leads to
\[
\gamma_1(\psi^*) = \frac{m\tau_0^*}{m\tau_0^*+\tau_1^*},
\]
as desired.
\end{proof}

\subsubsection{Proof of Corollary \ref{spectral_normal} for $\gamma_2(\psi^*)$ and $\gamma_3(\psi^*)$}
We need a technical Lemma.
\begin{lemma}\label{C_Var_normal}
Consider the setting of Proposition \ref{prop_normal_model}. Then
\[
C(\psi^*) = 
\begin{bmatrix}
\frac{\tau_1^*}{m\tau_0^*+\tau_1^*} & 0\\
0 & -\frac{\tau_1^*+2m\tau_0^*}{\tau_1^*(m\tau_0^*+\tau_1^*)^2}
\end{bmatrix}, \quad V(\psi^*) =
\begin{bmatrix}
\frac{1}{m\tau_0^*+\tau_1^*} & 0\\
0 & \frac{2\tau_1^*+4m\tau_0^*}{\tau_1^*(m\tau_0^*+\tau_1^*)^2}
\end{bmatrix}\,,
\]
with $C(\psi^*)$ and $V(\psi^*)$ as in \eqref{def_C_V}.
\end{lemma}
\begin{proof}
Recall that, in the context of Proposition \ref{prop_normal_model}, we define $T_1(\theta_j) = \theta_j$ and $T_2(\theta_j)=(\theta_j-\mu^*)^2$.
By \eqref{posterior_theta} we have
\[
\begin{aligned}
&E[T_1(\theta_j) \mid Y_j, \psi] 
=  \frac{m\tau_0}{m\tau_0+\tau_1}\bar{Y}_j+\frac{\tau_1}{m\tau_0+\tau_1}\mu,\\
& E[T_2(\theta_j) \mid Y_j, \psi] 
=
  \frac{1}{m\tau_0+\tau_1}+ \left(\frac{m\tau_0}{m\tau_0+\tau_1}\bar{Y}_j+\frac{\tau_1}{m\tau_0+\tau_1}\mu-\mu^* \right)^2.
\end{aligned}
\]
Therefore we can compute $C(\psi^*)$ as
\[
\begin{aligned}
& E_{Y_j} \left[\partial_\mu M_1(\psi^* \mid Y_j) \right] = \frac{\tau_1^*}{m\tau_0^*+\tau_1^*},\\
& E_{Y_j} \left[\partial_\mu M_2(\psi^* \mid Y_j) \right] 
=
 E_{Y_j} \left[\frac{2\tau_1^*}{m\tau_0^*+\tau_1^*} \left(\frac{m\tau_0^*}{m\tau_0^*+\tau_1^*}\bar{Y}_j-\frac{m\tau_0^*}{m\tau_0^*+\tau_1^*}\mu^* \right)\right] = 0,\\
&E_{Y_j} \left[\partial_{\tau_1} M_1(\psi^* \mid Y_j) \right] 
=
 E_{Y_j} \left[-\frac{m\tau_0^*}{(m\tau_0^*+\tau_1^*)^2}\bar{Y}_j+\frac{m\tau_0^*}{(m\tau_0^*+\tau_1^*)^2}\mu^* \right] = 0,\\
&\begin{aligned}
E_{Y_j} \left[\partial_{\tau_1} M_2(\psi^* \mid Y_j) \right] =& -\frac{1}{(m\tau_0^*+\tau_1^*)^2}
+\\&
E_{Y_j} 
\left[2\left(-\frac{m\tau_0^*}{(m\tau_0^*+\tau_1^*)^2}\bar{Y}_j+\frac{m\tau_0^*}{(m\tau_0^*+\tau_1^*)^2}\mu^* \right)\left(\frac{m\tau_0^*}{m\tau_0^*+\tau_1^*}\bar{Y}_j-\frac{m\tau_0^*}{m\tau_0^*+\tau_1^*}\mu^* \right) \right] 
\\=& 
-\frac{1}{(m\tau_0^*+\tau_1^*)^2}-2\frac{(m\tau_0^*)^2}{(m\tau_0^*+\tau_1^*)^3}E_{Y_j} 
\left[ (\bar{Y}_j-\mu^*)^2\right]
\\=& 
-\frac{1}{(m\tau_0^*+\tau_1^*)^2}-2\frac{m\tau_0^*}{\tau_1^*(m\tau_0^*+\tau_1^*)^2},
\end{aligned}
\end{aligned}
\]
by \eqref{distribution_mean}.

We now consider $V\left(\psi^* \right)$. 
Given $X \sim N(\mu, \sigma^2)$, we have
\[
\text{Cov}(X, X^2) = 2\mu\sigma^2, \quad \text{Var}(X^2) = 2\sigma^4+4\mu^2\sigma^2\,,
\]
which can be easily derived by computing the first four moments of $X$ using Lemma \ref{normal_moments}, which are
$E[X] = \mu$, 
$E[X^2] = \mu^2+\sigma^2$, $E[X^3] =3\mu\sigma^2+\mu^3$ and 
$E[X^4]= 3\sigma^4+6\mu^2\sigma^2+\mu^4$.
By \eqref{posterior_theta} we have
\[
\begin{aligned}
&\text{Var}(\theta_j \mid Y_j, \psi^*) = \frac{1}{m\tau_0^*+\tau_1^*},\\
& \text{Cov}(\theta_j, (\theta_j-\mu^*)^2 \mid Y_j, \psi^*) = \text{Cov}(\theta_j-\mu^*, (\theta_j-\mu^*)^2 \mid Y_j, \psi^*) = 2\frac{m_j-\mu^*}{m\tau_0^*+\tau_1^*},\\
&\begin{aligned}
 \text{Var}( (\theta_j-\mu^*)^2 \mid Y_j, \psi^*) &= \frac{2}{(m\tau_0^*+\tau_1^*)^2}+\frac{4}{m\tau_0^*+\tau_1^*}\left(m_j-\mu^* \right)^2\\
& = \frac{2}{(m\tau_0^*+\tau_1^*)^2}+\frac{4}{m\tau_0^*+\tau_1^*}\left(\frac{m\tau_0^*}{m\tau_0^*+\tau_1^*}(\bar{Y_j}-\mu^*)+\mu^* \right)^2.
\end{aligned}
\end{aligned}
\]
Therefore, we conclude
\[
\begin{aligned}
&E_{Y_j} \left[ \text{Cov}(\theta_j, \theta^2_j \mid Y_j, \psi^*) \right] = 0
\end{aligned}
\]
and
\[
\begin{aligned}
E_{Y_j}\left[\text{Var}( \theta^2_j \mid Y_j, \psi^*)\right] &= \frac{2}{(m\tau_0^*+\tau_1^*)^2}+\frac{4}{m\tau_0^*+\tau_1^*}E_{Y_j}\left[\left(\frac{m\tau_0^*}{m\tau_0^*+\tau_1^*}\right)^2(\bar{Y_j}-\mu^*)^2  \right]\\
& = \frac{2}{(m\tau_0^*+\tau_1^*)^2}+\frac{4m^2(\tau_0^*)^2}{(m\tau_0^*+\tau_1^*)^3}E_{Y_j}\left[(\bar{Y_j}-\mu^*)^2\right]\\
&= \frac{2}{(m\tau_0^*+\tau_1^*)^2}+\frac{4m\tau_0^*}{\tau^*_1(m\tau_0^*+\tau_1^*)^2},
\end{aligned}
\]
as desired.
\end{proof}
\begin{lemma}\label{limiting_sigma_normal}
Consider the same assumptions of Proposition \ref{prop_normal_model}. Then
\[
\lTV \L(\d\tilde{\bm{T}}, \d\tilde{\psi} \mid Y_{1:J})-N\left(\bm{0}, \Sigma \right) \rTV \to 0,
\]
as $J \to \infty$, in $Q_{\psi^*}^{(\infty)}$-probability, where $(\tilde{\bm{T}}, \tilde{\psi})$ are derived by \eqref{eq:psi_T_normal} with transformations \eqref{tilde_psi} and \eqref{tilde_T} and where
\begin{equation}\label{eq:Sigma_def_normal}
\Sigma =
\begin{bmatrix}
2\frac{\tau_1^*+2m\tau_0^*}{m^2(\tau_0^*)^2\tau_1^*} & 0 & -2\frac{\tau_1^*(\tau_1^*+2m\tau_0^*)}{m^2(\tau_0^*)^2} & 0\\
0 & \frac{1}{m\tau_0^*} & 0 & \frac{1}{m\tau_0^*}\\
-2\frac{\tau_1^*(\tau_1^*+2m\tau_0^*)}{m^2(\tau_0^*)^2} & 0 & 2\frac{(\tau_1^*)^2(\tau_1^*+m\tau_0^*)^2}{m^2(\tau_0^*)^2} & 0\\
0 & \frac{1}{m\tau_0^*} & 0 & \frac{m\tau_0^*+\tau_1^*}{m\tau_0^*\tau_1^*}
\end{bmatrix}
\end{equation}
\end{lemma}
\begin{proof}
The result follows by an argument similar to the proof of Proposition \ref{limiting_sigma}, where
\[
\Sigma =
\begin{bmatrix}
V(\psi^*)+C(\psi^*)\mathcal{I}^{-1}(\psi^*)C^\top(\psi^*) & \quad C(\psi^*)\mathcal{I}^{-1}(\psi^*)\\
&\\
\mathcal{I}^{-1}(\psi^*)C^\top(\psi^*) & \mathcal{I}^{-1}(\psi^*)
\end{bmatrix}
\]
The entries of $\Sigma$ can be computed through Lemmas \ref{fisher_normal_hier} and \ref{C_Var_normal}.
\end{proof}

\begin{proof}[Proof of Corollary \ref{spectral_normal} for $\gamma_2(\psi^*)$ and $\gamma_3(\psi^*)$]
Recall that $P_2$ is the transition kernel of the Gibbs sampler that alternates updates from $\L\left(\d\mu,\d\bm{\theta}\mid \tau_1, Y_{1:J} \right)$
and  $\L\left(\d\tau_1\mid \bm{\theta}, \mu,Y_{1:J} \right)$. Through the same reasoning of Lemma \ref{sufficient_lemma}, the mixing times of $P_2$ are the same of the Gibbs sampler targeting $\L\left(\d\mu, \d\tau_1, \d\bm{T} \mid Y_{1:J}\right)$ by alternating updates from $\L\left(\d\mu, \d\bm{T} \mid \tau_1, Y_{1:J}\right)$ and $\L\left(\d\tau_1 \mid \mu, \bm{T}, Y_{1:J}\right)$. Indeed
\[
\L\left( \d\tau_1 \mid \mu, \bm{\theta}, Y_{1:J}\right) = \L\left(\d\tau_1 \mid \mu, \bm{T}(\bm{\theta}), Y_{1:J}\right).
\]
Therefore, by Corollary \ref{mixing_gap} $\gamma_2(\psi^*)$ is the spectral gap of the Gibbs sampler alternating updates from $\tilde{\L}\left(\d\tilde{\mu},\d\tilde{\bm{T}}_1, \d\tilde{\bm{T}}_2 \mid \tilde{\tau}_1\right)$ and $\tilde{\L}\left(\d\tilde{\tau}_1 \mid \tilde{\mu}, \tilde{\bm{T}}_1,\tilde{\bm{T}}_2\right)$, where $\tilde{\L}(\cdot)$ is the law identified in Lemma \ref{limiting_sigma_normal}. By inspection of the matrix \eqref{eq:Sigma_def_normal}, $\left(\tilde{\mu}, \tilde{\bm{T}}_1\right)$ is independent from $\tilde{\tau}_1$ and $\tilde{\bm{T}}_2$ according to $\tilde{\L}$, so that $\left(\tilde{\mu}, \tilde{\bm{T}}_1\right)$ is sampled independently from everything else at each iteration. Therefore by the same arguments of the proof of Corollary \ref{spectral_radius} we have
\[
\gamma_2(\psi^*) = 1-\frac{\Sigma_{24}^2}{\Sigma_{22}\Sigma_{44}} = \left(\frac{m\tau_0^*}{m\tau_0^*+\tau_1^*}\right)^2.
\]
Instead, recall that $P_3$ is the transition kernel of the Gibbs sampler that alternates updates from $\L\left(\d\bm{\theta}\mid \tau_1, Y_{1:J} \right)$,
$\L\left(\d\mu\mid \bm{\theta}, \tau_1,Y_{1:J} \right)$ and
$\L\left(\d\tau_1\mid \bm{\theta}, \mu,Y_{1:J} \right)$. Reasoning as before, by Corollary \ref{mixing_gap} $\gamma_3(\psi^*)$ is the spectral gap of the Gibbs sampler alternating updates from $\tilde{\L}\left(\d\tilde{\bm{T}}\mid \tilde{\mu}, \tilde{\tau}_1 \right)$, $\tilde{\L}\left(\d\tilde{\mu}\mid \tilde{\tau}_1, \tilde{\bm{T}} \right)$ and $\tilde{\L}\left(\d\tilde{\tau}_1 \mid \tilde{\mu}, \tilde{\bm{T}}\right)$, where $\tilde{\L}(\cdot)$ is the law identified in Lemma \ref{limiting_sigma_normal}. By inspection of the matrix \eqref{eq:Sigma_def_normal}, the pair $(\tilde{\mu}, \tilde{\bm{T}}_1)$ is independent from $(\tilde{\tau}_1, \tilde{\bm{T}}_2)$, according to $\tilde{\L}$. By standard properties of the Gibbs samplers (e.g. Lemma $2$ in \cite{PZ20}), the spectral gap is given by the minimum of the spectral gaps of the Gibbs samplers associated to the two pairs, i.e.
\[
\gamma_3(\psi^*) = \min \left\{1-\frac{\Sigma_{24}^2}{\Sigma_{22}\Sigma_{44}}, 1-\frac{\Sigma_{13}^2}{\Sigma_{11}\Sigma_{33}} \right\} = \left(\frac{m\tau_0^*}{m\tau_0^*+\tau_1^*}\right)^2.
\]
Notice that the result of Lemma \ref{positiveDefinite_Gibbs} holds even if $P_3$ has three blocks: indeed, by inspection of the matrix \eqref{eq:Sigma_def_normal}, $\tilde{\mu}$ and $\tilde{\tau}_1$ are independent according to $\tilde{\L}$, so that the updates $\tilde{\L}\left(\d\tilde{\mu}\mid \tilde{\tau}_1, \tilde{\bm{T}} \right)$ and $\tilde{\L}\left(\d\tilde{\tau}_1 \mid \tilde{\mu}, \tilde{\bm{T}}\right)$ can be equivalently seen as a single one.
\end{proof}

\subsection{Proof of Lemma \ref{finite_B4B6}}
Since it will be useful in the following, we denote
\[
c(\mu,\tau)= \min_{r\in\{0,\dots,m\}} 
g(y_r \mid \mu, \tau)\,,
\]
with $g(y_r\mid \mu, \tau)$ defined in \eqref{marginal_finite}.
Notice that by construction, see e.g.\ \eqref{eq:ass_discrete}, we have $0 < c(\mu,\tau) \leq 1$.
Also, $g(y_r\mid \mu, \tau)$ is continuous w.r.t.\ $(\mu, \tau)$ since it is defined in \eqref{marginal_finite} as the integral of a bounded function, $\theta\mapsto f(y \mid \theta)$, with respect to the normal kernel which is continuous w.r.t.\ $(\mu, \tau)$.
It follows that also $c(\mu,\tau)$ is continuous, since it is the minimum of a finite number of continuous functions. 
Define 
\begin{equation}\label{eq:def}
c:=\inf_{(\mu,\tau)\in B}c(\mu,\tau)>0
\end{equation}
where $B$ is the largest of the three balls --  namely $B_{\delta_4}$, $B_{\delta_5}$ and $B_{\delta_6}$ -- centered at $\psi^*=(\mu^*,\tau^*)$ defined in (B4), (B5) and (B6), respectively.
The positivity of $c$ follows from the continuity of $c(\mu,\tau)$ and the compactness of $B$.

Recall that $T(\theta_j) = \left(\theta_j, \theta_j^2 \right)$. Thus we need three lemmas.
\begin{lemma}\label{B4_finite}
Consider the setting of Lemma \ref{finite_B4B6}. Then assumption $(B4)$ is satisfied.
\end{lemma}
\begin{proof}
First of all, consider $V(\psi^*)$, as defined in \eqref{def_C_V}. For every $y = 0, \dots, m$, we have that the posterior distribution of $\theta_j$ admits a density with respect to the Lebesgue measure of the form
\[
p(\theta_j \mid y, \mu, \tau) \propto f(y_r \mid \theta_j)N(\theta_j \mid \mu, \tau), 
\]
which implies that
\[
\text{Var}(\theta_j \mid y, \psi^*) >0, \quad \text{Var}(\theta_j^2 \mid y, \psi^*) >0, \quad |\text{Corr}(\theta_j, \theta_j^2 \mid y, \psi^*)| < 1.
\]
Consequently $V(\psi^*)$ is a sum of positive definite matrices and is therefore non singular.

Secondly, let $s, p = 1, 2$. Then by Bayes' Theorem it follows
\[
M_s^{(p)}(y_r \mid \mu, \tau) = \frac{\int_{\R} \theta^{sp}f(y_r \mid \theta)N(\theta \mid \mu, \tau^{-1}) \, \d \theta}{\int_{\R} f(y_r \mid \theta)N(\theta \mid \mu, \tau^{-1}) \, \d \theta}, \quad r = 0, \dots, m.
\]
Therefore
\[
\begin{aligned}
|\partial_\mu M_1^{(p)}(y_r \mid \mu, \tau)| \leq& \left \lvert \frac{\int_{\R} \theta^pf(y_r \mid \theta)\partial_\mu N(\theta \mid \mu, \tau^{-1}) \, \d \theta}{\int_{\R} f(y_r \mid \theta)N(\theta \mid \mu, \tau^{-1}) \, \d \theta}\right \rvert +\\
&\left \lvert \frac{\left(\int_{\R} \theta^pf(y_r \mid \theta) N(\theta \mid \mu, \tau^{-1}) \, \d \theta \right)\left( \int_{\R} f(y_r \mid \theta) \partial_\mu N(\theta \mid \mu, \tau^{-1}) \, \d \theta\right)}{\left( \int_{\R} f(y_r \mid \theta)N(\theta \mid \mu, \tau^{-1}) \, \d \theta \right)^2}\right \rvert.
\end{aligned}
\]
By definition of $c$ we have
\[
\begin{aligned}
|\partial_\mu M_1^{(p)}(y_r \mid \mu, \tau)| \leq& \frac{1}{c}\int_{\R} |\theta|^p\left \lvert \partial_\mu N(\theta \mid \mu, \tau^{-1}) \right \rvert \, \d \theta +\\
& \frac{1}{c^2}\left(\int_{\R} |\theta|^p N(\theta \mid \mu, \tau^{-1}) \, \d \theta \right)\left( \int_{\R} |\theta|^p \left \lvert \partial_\mu N(\theta \mid \mu, \tau^{-1}) \right\rvert \, \d \theta\right) \\
=& \frac{\tau}{c}\int_{\R} |(\theta-\mu)\theta^p| N(\theta \mid \mu, \tau^{-1})\, \d \theta +\\
&\frac{\tau}{c^2}\left(\int_{\R} |(\theta-\mu)\theta|^p N(\theta \mid \mu, \tau^{-1}) \, \d \theta \right)\left( \int_{\R} |\theta|^pf \left \lvert  N(\theta \mid \mu, \tau^{-1}) \right\rvert \, \d \theta\right).
\end{aligned}
\]
The right hand side does not depend on the data, so that
\[
E_{Y_j}\left[|\partial_\mu M_1^{(p)}(y_r \mid \mu, \tau)|\right] \leq m\frac{\tau}{c}E[|(\theta_j-\mu)\theta_j^p| \mid \mu, \tau]+m\frac{\tau}{c^2}E[|(\theta_j-\mu)\theta_j^p| \mid \mu, \tau]E[|\theta_j|^p \mid \mu, \tau].
\]
By the specification of model \eqref{finite_model}, the prior absolute moments are all finite and continuous function of $\mu$ and $\tau$: therefore the right hand side is uniformly bounded for every bounded neighborhood of $(\mu^*, \tau^*)$. Using a similar argument for all the other quantities involved, it is easy to see that assumption $(B4)$ holds for every $\delta_4 < \tau^*$.
\end{proof}
\begin{lemma}\label{B5_finite}
Consider the setting of Lemma \ref{finite_B4B6}. Then assumption $(B5)$ is satisfied with $k = 5$.
\end{lemma}
\begin{proof}
Consider the random vector $X = (X_1, X_2) = (\sum_{j = 1}^5\theta_j, \sum_{j = 1}^5\theta_j^2)$. First of all we prove that $X$ admits a density function with respect to the Lebesgue measure on $\R^2$, conditional to $(\mu, \tau)$. By Lemma \ref{auxiliary_characteristic} and conditional independence of $\theta_j$ we have
\[
\left \lvert E \left[e^{i(t_1X_1+t_2X_2)} \mid  \mu, \tau_1 \right] \right \rvert \leq \frac{e^{-5\frac{\sigma^2}{2}\frac{(2\mu t_2+t_1)^2}{1+4t_2^2\sigma^4}}}{\left(1+4t_2^2\sigma^4 \right)^{5/4}},
\]
where we denote $\sigma^2 = \tau^{-1}$, so that we can write
\begin{equation}\label{cf_categorical}
\begin{aligned}
\int_{\R^2}|\varphi_X(t \mid \mu, \tau)| \, \d t &= \int_{\R^2} \left \lvert E \left[e^{i(t_1X_1+t_2\sum_{j = 1}^3X_2)} \mid Y, \mu, \tau_1 \right] \right \rvert \, \d t_1 \d t_2  \\
&\leq \int_{\R} \frac{1}{\left(1+4t_2^2\sigma^4 \right)^{5/4}} \left(\int_{\R}e^{-5\frac{\sigma^2}{2}\frac{(2\mu t_2+t_1)^2}{1+4t_2^2\sigma^4}} \, \d t_1 \right) \, \d t_2\\
& = \sqrt{\frac{2\pi}{5\sigma^2}} \int_\R\frac{1}{\left(1+4t_2^2\sigma^4\right)^{3/4}} \, \d t_2 < \infty.
\end{aligned}
\end{equation}
Therefore, by the Inversion Formula we have that $X$ admits a density $p(x \mid \mu, \tau)$ with respect to the Lebesgue measure on $\R^2$. Thus, by Bayes' Theorem we can write
\[
p(x \mid Y_{1:5}, \mu, \tau) = \frac{f(Y_{1:5}\mid x, \mu, \tau)p(x \mid \mu, \tau)}{\int_{\R^2}f(Y_{1:5}\mid x, \mu, \tau)p(x \mid \mu, \tau) \, \d x},
\]
where $f(Y_{1:5}\mid x, \mu, \tau) = \int \prod_{j = 1}^5f(Y_j \mid \theta_j)\L(\d \theta_{1:5} \mid x, \mu, \tau)$. It is easy to see that $f(Y_{1:5}\mid x, \mu, \tau) \leq 1$ and
\[
\int_{\R^2}f(Y_{1:5}\mid x, \mu, \tau)p(x \mid \mu, \tau) \, \d x = \prod_{j = 1}^5g(Y_j \mid \mu, \tau) \geq c^5,
\]
for every $(\mu, \tau) \in B_{\delta_5}$, with $\delta_5$ to be fixed. We can therefore conclude that
\[
p(x \mid Y_{1:5}, \mu, \tau) \leq \frac{p(x \mid  \mu, \tau)}{c^5}.
\]
We can now apply the Plancherel identity to get
\[
\int_{\R^2} \left \lvert \varphi^{(5)}(t \mid Y, \mu, \tau) \right \rvert^2 \, \d t = \int_{\R^2} p^2(x_1, x_2 \mid Y, \mu, \tau) \, \d x \leq \frac{1}{c^{10}}\int_{\R^2} p^2(x_1, x_2 \mid \mu, \tau) \, \d x.
\]
Applying again the Plancherel identity we obtain
\[
\int_{\R^2} \left \lvert \varphi^{(5)}(t \mid Y, \mu, \tau) \right \rvert^2 \, \d t \leq \frac{1}{c^{10}}\int_{\R^2} \left \lvert \varphi_X(t \mid \mu, \tau) \right \rvert^2 \, \d t \leq \frac{1}{c^{10}}\int_{\R^2} \left \lvert \varphi_X(t \mid \mu, \tau) \right \rvert \, \d t < \infty,
\]
by \eqref{cf_categorical} for every $\tau > 0$. Therefore assumption $(B5)$ follows with $\delta_5 < \tau^*$.
\end{proof}
\begin{lemma}\label{B6_finite}
Consider the setting of Lemma \ref{finite_B4B6}. Then assumption $(B6)$ is satisfied with $k' = 5$.
\end{lemma}
\begin{proof}
As shown in the proof of Lemma \ref{B5_finite}, the vector $(\sum_{j = 1}^5\theta_j, \sum_{j = 1}^5\theta_j^2)$ admits a density with respect to the Lebesgue measure on $\R^2$, conditional to $Y$ and $(\mu^*, \tau^*)$. Therefore, by Lemma $4$ in Chapter $15$ of \cite{F70}, $|\varphi^{(5)}(t \mid Y, \mu^*, \tau^*)| < 1$ for every $t= (t_1, t_2)$. Moreover, by Riemann-Lebesgue Lemma we have
\[
|\varphi^{(5)}(t \mid Y, \mu^*, \tau^*)| \quad \to \quad 0,
\]
as $|t| \to \infty$. We conclude
\[
\underset{|t| \geq \epsilon}{\sup} \, \left \lvert \varphi^{(5)}(t \mid Y, \mu^*, \tau^*) \right \rvert <1.
\]
Let $\delta_6 > 0$ to be chosen later and $(\mu, \tau) \in B_{\delta_6}$. Then by Taylor formula we get
\begin{equation}\label{taylor_characteristic}
|\varphi^{(5)}(t \mid Y, \mu, \tau)|^2 = |\varphi^{(5)}(t \mid Y, \mu^*, \tau^*)|^2+(\mu^*-\mu)\partial_{\mu}|\varphi^{(5)}(t \mid Y, \bar{\mu}, \bar{\tau})|^2 + (\tau^*-\tau)\partial_{\tau}|\varphi^{(5)}(t \mid Y, \bar{\mu}, \bar{\tau})|^2,
\end{equation}
where $(\bar{\mu}, \bar{\tau}) \in B_{\delta_6}$. Notice that
\[
\begin{aligned}
|\varphi^{(5)}(t \mid Y, \mu, \tau)|^2  =&  \left( \int_{\R^3}\cos \left(t_1\sum_{j = 1}^5\theta_j +t_2\sum_{j = 1}^5\theta_j^2\right) \left \{\prod_{j = 1}^5\frac{f(Y_j \mid \theta_j)N(\theta_j \mid \mu, \tau^{-1})}{\int_\R f(Y_j \mid \psi_j)N(\psi_j \mid \mu, \tau^{-1}) \d \psi_j}\right\} \, \d \theta_{1:5} \right)^2\\
&+  \left( \int_{\R^5}\sin\left(t_1\sum_{j = 1}^5\theta_j +t_2\sum_{j = 1}^5\theta_j^2\right)   \left \{\prod_{j = 1}^5\frac{f(Y_j \mid \theta_j)N(\theta_j \mid \mu, \tau^{-1})}{\int_\R f(Y_j \mid \psi_j)N(\psi_j \mid \mu, \tau^{-1}) \d \psi_j}\right\} \, \d \theta_{1:5}\right)^2,
\end{aligned}
\]
which implies
\[
\begin{aligned}
\left \lvert  \partial_{\mu}|\varphi^{(5)}(t \mid Y, \mu, \tau)|^2 \right \rvert \leq&  2\left \lvert \int_{\R^5}\cos \left(t_1\sum_{j = 1}^5\theta_j +t_2\sum_{j = 1}^5\theta_j^2\right) \partial_\mu \left \{\prod_{j = 1}^5\frac{f(Y_j \mid \theta_j)N(\theta_j \mid \mu, \tau^{-1})}{\int_\R f(Y_j \mid \psi_j)N(\psi_j \mid \mu, \tau^{-1}) \d \psi_j}\right\} \, \d \theta_{1:5}\right \rvert\\
&+  2\left \lvert \int_{\R^5}\sin\left(t_1\sum_{j = 1}^5\theta_j +t_2\sum_{j = 1}^5\theta_j^2\right)   \partial_\mu \left \{\prod_{j = 1}^5\frac{f(Y_j \mid \theta_j)N(\theta_j \mid \mu, \tau^{-1})}{\int_\R f(Y_j \mid \psi_j)N(\psi_j \mid \mu, \tau^{-1}) \d \psi_j}\right\} \, \d \theta_{1:5}\right \rvert
\end{aligned}
\]
and therefore
\begin{equation}\label{bound_derivative}
\begin{aligned}
\left \lvert  \partial_{\mu}|\varphi^{(5)}(t \mid Y, \mu, \tau)|^2 \right \rvert &\leq 4 \int_{\R^5}\left \lvert   \partial_\mu \left \{\prod_{j = 1}^5\frac{f(Y_j \mid \theta_j)N(\theta_j \mid \mu, \tau^{-1})}{\int_\R f(Y_j \mid \psi_j)N(\psi_j \mid \mu, \tau^{-1}) \d \psi_j}\right\}\right \rvert \, \d \theta_{1:5}\\
&= 4\sum_{j = 1}^5\int_\R\left \lvert \partial_\mu \left \{\frac{f(Y_j \mid \theta_j)N(\theta_j \mid \mu, \tau^{-1})}{\int_\R f(Y_j \mid \psi_j)N(\psi_j \mid \mu, \tau^{-1}) \d \psi_j}\right\} \right \rvert \, \d \theta_j.
\end{aligned}
\end{equation}
Moreover, for every $r = 0, \dots, m$, we have
\[
\begin{aligned}
&\left \lvert \partial_\mu \left \{\frac{f(y_r \mid \theta)N(\theta \mid \mu, \tau^{-1})}{\int_\R f(y_r \mid \psi)N(\psi \mid \mu, \tau^{-1}) \d \psi}\right\} \right \rvert \leq \left \lvert  \left \{\frac{f(y_r \mid \theta)\partial_\mu N(\theta \mid \mu, \tau^{-1})}{\int_\R f(y_r \mid \psi)N(\psi \mid \mu, \tau^{-1}) \d \psi}\right\} \right \rvert\\
& + \left \lvert  \left \{\frac{f(y_r \mid \theta)\partial_\mu N(\theta \mid \mu, \tau^{-1})\left(\int_\R f(y_r \mid \psi)\partial_\mu N(\psi \mid \mu, \tau^{-1}) \d \psi\right)}{\left(\int_\R f(y_r \mid \psi)N(\psi \mid \mu, \tau^{-1}) \d \psi\right)^2}\right\} \right \rvert\\
& \leq \frac{|\partial_\mu N(\theta \mid \mu, \tau^{-1})|}{c}+\frac{1}{c^2}|\partial_\mu N(\theta \mid \mu, \tau^{-1})|\left(\int_\R |\partial_\mu N(\psi \mid \mu, \tau^{-1})| \d \psi\right)\\
& = 2\tau \frac{|\theta-\mu| N(\theta \mid \mu, \tau)}{c} +\frac{4\tau^2}{c^2}|\theta-\mu| N(\theta \mid \mu, \tau^{-1})\left(\int_\R|\psi-\mu|  N(\psi \mid \mu, \tau^{-1}) \d \psi\right).
\end{aligned}
\]
Therefore, by \eqref{bound_derivative} there exists $C(\delta_6) < \infty$ which does not depend on $\mu$ and $\tau$ such that
\[
\begin{aligned}
\left \lvert  \partial_{\mu}|\varphi^{(5)}(t \mid Y, \mu, \tau)|^2 \right \rvert \leq& 40\tau \frac{\int_\R |\theta-\mu| N(\theta \mid \mu, \tau^{-1}) \, \d \theta}{c}+80\tau^2\left( \frac{\int_\R |\theta-\mu| N(\theta \mid \mu, \tau^{-1}) \, \d \theta}{c}\right)^2\\
&\leq C(\delta_6),
\end{aligned}
\]
for every $(\mu, \tau) \in B_{\delta_6}$ Notice that $C(\delta_6) $ becomes smaller as $\delta_6$ decreases. Similarly holds for $\partial_{\tau}|\varphi^{(3)}(t \mid Y, \mu, \tau)|^2$, so that by \eqref{taylor_characteristic} we have
\[
\begin{aligned}
|\varphi^{(5)}(t \mid Y, \mu, \tau)|^2 &\leq |\varphi^{(5)}(t \mid Y, \mu^*, \tau^*)|^2+|\mu^*-\mu|C(\delta_6) + |\tau^*-\tau|C(\delta_6)\\
&\leq |\varphi^{(5)}(t \mid Y \mu^*, \tau^*)|^2+2\delta_6 C(\delta_6).
\end{aligned}
\]
Since $\underset{|t| \geq \epsilon}{\sup} \, |\varphi^{(5)}(t \mid Y, \mu^*, \tau^*)|^2 < 1$, by choosing $\delta_6$ small enough we have
\[
\underset{(\mu, \tau) \in B_{\delta_6}}{\sup} \, \, \underset{|t| \geq \epsilon}{\sup} \, |\varphi^{(5)}(t \mid Y, \mu, \tau)|^2 \leq  \underset{|t| \geq \epsilon}{\sup} \, |\varphi^{(5)}(t \mid Y, \mu^*, \tau^*)|^2+2\delta_6 C(\delta_6) < 1,
\]
and $(B6)$ is satisfied.
\end{proof}
\begin{proof}[Proof of Lemma \ref{finite_B4B6}]
Assumption (B4) is satisfied by Lemma \ref{B4_finite}, assumption (B5) by Lemma \ref{B5_finite} and assumption (B6) by Lemma \ref{B6_finite}.
\end{proof}
\subsection{Proof of Proposition \ref{prop_logit_model}}
\begin{proof}
Requirements $(B1)-(B3)$ of Theorem \ref{theorem_one_level_nested} are satisfied by assumption, while $(B4)-(B6)$ hold by Lemma \ref{finite_B4B6}.
\end{proof}
\subsection{Proof of Corollary \ref{spectral_logit}}
\begin{proof}
The result is a direct consequence of Corollary \ref{spectral_radius}.
\end{proof}
\subsection{Statement and proof of Lemma \ref{fisher_finite}}
Let 
\begin{equation}\label{logit_link}
f(y \mid \theta) = \binom{m}{y}\frac{e^{y\theta}}{(1+e^\theta)^m}, 
\end{equation}
where $y = 0, \dots, m$. It means that for each group, conditional to $\theta$, $m$ independent Bernoulli trials are performed, with probability of success given by $e^{\theta}/(1+e^{\theta})$. The following Section is devoted to the proof of the following lemma.
\begin{lemma}\label{fisher_finite}
Consider the setting of Proposition \ref{prop_logit_model} with likelihood \eqref{logit_link}. The Fisher Information Matrix $I(\mu,\tau)$ is non-singular if and only if $m \geq 2$, for every $(\mu, \tau)$.
\end{lemma}
First of all we need few preliminary results.
\begin{lemma}\label{derivative_binary}
Consider the setting of Proposition \ref{prop_logit_model} with likelihood \eqref{logit_link} and fix $(\mu, \tau)$. Let $h(y \mid \mu, \tau) = \log g(y \mid \mu, \tau)$, with $g(\cdot)$ as in \eqref{marginal_finite}. Then it holds
\[
E_Y\left[\frac{\partial}{\partial \mu}h(Y \mid \mu, \tau) \right] = E_Y\left[\frac{\partial}{\partial \tau}h(Y \mid \mu, \tau) \right] = 0
\]
and
\[
E_Y \left[\left(\frac{\partial}{\partial \mu}h(Y \mid \mu, \tau) \right)^2\right] < \infty, \quad E_Y \left[\left(\frac{\partial}{\partial \tau_1}h(Y \mid \mu, \tau) \right)^2\right] < \infty.
\]
Moreover, for every $y = 0, \dots, m$ we have
\[
\frac{\partial}{\partial \mu} g(y \mid, \mu, \tau) =  \binom{m}{y}\int \frac{e^{y\theta}\left[y+ye^\theta-me^\theta \right]}{(1+e^\theta)^{m+1}}\sqrt{\frac{\tau}{2\pi}}e^{-\frac{\tau}{2}(\theta-\mu)^2} \, \d \theta
\]
and
\[
\frac{\partial}{\partial \tau} g(y \mid, \mu, \tau) =  -\binom{m}{y}\frac{1}{2\tau}\int (\theta-\mu)\frac{e^{y\theta}\left[y+ye^\theta-me^\theta \right]}{(1+e^\theta)^{m+1}}\sqrt{\frac{\tau}{2\pi}}e^{-\frac{\tau}{2}(\theta-\mu)^2} \, \d \theta.
\]
\end{lemma}
\begin{proof}
Through Dominated Convergence Theorem it is easy to verify that
\[
\frac{\partial}{\partial \mu}g(y \mid \mu, \tau) = \binom{m}{y}\int \frac{e^{y\theta}}{(1+e^{\theta})^m}\frac{\partial}{\partial \mu}\left\{\sqrt{\frac{\tau}{2\pi}}e^{-\frac{\tau}{2}(\theta-\mu)^2} \right\} \, \d \theta
\]
and
\[
\frac{\partial}{\partial \tau}g(y \mid \mu, \tau) = \binom{m}{y}\int \frac{e^{y\theta}}{(1+e^{\theta})^m}\frac{\partial}{\partial \tau}\left\{\sqrt{\frac{\tau}{2\pi}}e^{-\frac{\tau}{2}(\theta-\mu)^2} \right\} \, \d \theta,
\]
that is integrals and derivatives can be exchanged. Therefore
\[
\frac{\partial}{\partial \mu} h(y \mid, \mu, \tau) = E\left[ \theta -\mu \mid y, \mu, \tau\right], \quad \frac{\partial}{\partial \mu} h(y \mid, \mu, \tau) = \frac{1}{2\tau}-\frac{1}{2}E\left[ (\theta -\mu)^2 \mid y, \mu, \tau\right]
\]
and the statements on $h(y \mid \mu, \tau)$ easily follow. Moreover
\[
\begin{aligned}
\frac{\partial}{\partial \mu}g(y \mid \mu, \tau) &= \binom{m}{y}\int \frac{e^{y\theta}}{(1+e^{\theta})^m}(\theta-\mu)\sqrt{\frac{\tau}{2\pi}}e^{-\frac{\tau}{2}(\theta-\mu)^2} \, \d \theta \\
&= \binom{m}{y}\int \frac{e^{y\theta}\left[y+ye^\theta-me^\theta \right]}{(1+e^\theta)^{m+1}}\sqrt{\frac{\tau}{2\pi}}e^{-\frac{\tau}{2}(\theta-\mu)^2} \, \d \theta
\end{aligned}
\]
integrating by parts. Similarly
\[
\begin{aligned}
\frac{\partial}{\partial \tau}g(y \mid \mu, \tau) &= \binom{m}{y}\frac{1}{2\tau}\int \frac{e^{y\theta}}{(1+e^{\theta})^m}\sqrt{\frac{\tau}{2\pi}}e^{-\frac{\tau}{2}(\theta-\mu)^2} \, \d \theta \\
&- \binom{m}{y}\frac{1}{2}\int \frac{e^{y\theta}}{(1+e^{\theta})^m}(\theta-\mu)^2\sqrt{\frac{\tau}{2\pi}}e^{-\frac{\tau}{2}(\theta-\mu)^2} \, \d \theta\\
& = -\binom{m}{y}\frac{1}{2\tau}\int (\theta-\mu)\frac{e^{y\theta}\left[y+ye^\theta-me^\theta \right]}{(1+e^\theta)^{m+1}}\sqrt{\frac{\tau}{2\pi}}e^{-\frac{\tau}{2}(\theta-\mu)^2} \, \d \theta.
\end{aligned}
\]
\end{proof}

\begin{lemma}\label{mean_binary}
Consider the setting of Proposition \ref{prop_logit_model} with likelihood \eqref{logit_link} and let $y, y' \in \{0, 1, \dots, m\}$ be such that $y < y'$ and $m \geq 1$. Then
\[
E\left[\theta \mid y, \mu, \tau \right] < E\left[\theta \mid y', \mu, \tau \right]
\]
for every $(\mu, \tau_1)$.
\end{lemma}
\begin{proof}
Fix $(\mu, \tau)$. Consider the function
\[
r(x) = \frac{\int \theta\frac{e^{x\theta}}{(1+e^\theta)^{m}}\sqrt{\frac{\tau_1}{2\pi}}e^{-\frac{\tau}{2}(\theta-\mu)^2} \, \d \theta}{\int\frac{e^{x\theta}}{(1+e^\theta)^{m}}\sqrt{\frac{\tau}{2\pi}}e^{-\frac{\tau}{2}(\theta-\mu)^2} \, \d \theta}.
\]
with $x \in (0, m)$. Notice that
\[
r(y) = E\left[\theta \mid y, \mu, \tau \right] \quad \text{and} \quad r(y') = E\left[\theta \mid y', \mu, \tau \right].
\]
Notice that
\[
\begin{aligned}
\frac{\d}{\d x} r(x) &= \frac{\int \theta^2\frac{e^{x\theta}}{(1+e^\theta)^{m}}\sqrt{\frac{\tau}{2\pi}}e^{-\frac{\tau}{2}(\theta-\mu)^2} \, \d \theta}{\int\frac{e^{x\theta}}{(1+e^\theta)^{m}}\sqrt{\frac{\tau}{2\pi}}e^{-\frac{\tau}{2}(\theta-\mu)^2} \, \d \theta} - \left[\frac{\int \theta\frac{e^{x\theta}}{(1+e^\theta)^{m}}\sqrt{\frac{\tau}{2\pi}}e^{-\frac{\tau}{2}(\theta-\mu)^2} \, \d \theta}{\int\frac{e^{x\theta}}{(1+e^\theta)^{m}}\sqrt{\frac{\tau}{2\pi}}e^{-\frac{\tau}{2}(\theta-\mu)^2} \, \d \theta} \right]^2 > 0
\end{aligned}
\]
for every $x \in (0, m)$ by Jensen inequality. Therefore $r(x)$ is strictly increasing and $r(y) < r(y')$.
\end{proof}
\begin{lemma}\label{fisher_information_binary}
Consider the setting of Proposition \ref{prop_logit_model} with likelihood \eqref{logit_link}. Then the Fisher Information Matrix $I(\mu,\tau)$ is non-singular in $(\mu, \tau)$ if and only if there exists $\alpha = \alpha(\mu, \tau) \neq 0$ such that
\[
\frac{\partial}{\partial \mu}g(y \mid \mu, \tau) = \alpha \frac{\partial}{\partial \tau}g(y \mid \mu, \tau)
\]
for every $y = 0, \dots, m$.
\end{lemma}
\begin{proof}
Fix a pair $(\mu, \tau)$. By Lemma \ref{derivative_binary} the matrix $I(\mu,\tau)$ is well-defined. The determinant is given by
\[
E_Y \left[\left(\frac{\partial}{\partial \mu}h(Y \mid \mu, \tau) \right)^2\right]E_Y \left[\left(\frac{\partial}{\partial \tau}h(Y \mid \mu, \tau) \right)^2\right] -E^2 \left[\left(\frac{\partial}{\partial \mu}h(Y \mid \mu, \tau) \right)\left(\frac{\partial}{\partial \tau}h(Y \mid \mu, \tau) \right)\right].
\]
By Cauchy--Schwartz inequality, the above formula is always non-negative and it is equal to $0$ if and only if $\frac{\partial}{\partial \mu}h(Y \mid \mu, \tau)$ and $\frac{\partial}{\partial \tau}h(Y \mid \mu, \tau)$ are linearly dependent, that is
\begin{equation}\label{loglik_equality}
\frac{\partial}{\partial \mu}h(y \mid \mu, \tau) = \alpha \frac{\partial}{\partial \tau}h(y \mid \mu, \tau)+\beta
\end{equation}
for every $y \in \{0, 1, \dots, m\}$ and for constants $\alpha$ and $\beta$. By Lemma \ref{derivative_binary} it is immediate to prove $\beta = 0$. Moreover, by Lemma \ref{mean_binary}, we deduce that $\alpha \neq 0$. Multiplying by $g(y \mid \mu, \tau)$ on both sides of \eqref{loglik_equality} we get the final result.
\end{proof}
\begin{proof}[Proof of Lemma \ref{fisher_finite}]
Fix $(\mu, \tau)$ and let $m = 1$. Define
\[
\alpha := \frac{\frac{\partial}{\partial \mu}g(0 \mid \mu, \tau)}{\frac{\partial}{\partial \tau}g(0 \mid \mu, \tau)}.
\]
Notice that $\alpha$ is well defined, since $\frac{\partial}{\partial \tau}g(0 \mid \mu, \tau) \neq 0$ for every $(\mu, \tau)$. Then by construction
\[
\frac{\partial}{\partial \mu}g(0 \mid \mu, \tau) = \alpha \frac{\partial}{\partial \tau}g(0 \mid \mu, \tau)
\]
and
\[
\frac{\partial}{\partial \mu}g(1 \mid \mu, \tau) =-\frac{\partial}{\partial \mu}g(0 \mid \mu, \tau) = - \alpha \frac{\partial}{\partial \tau}g(0 \mid \mu, \tau) = \alpha \frac{\partial}{\partial \tau}g(1 \mid \mu, \tau),
\]
so that the Fisher Information matrix is singular by Lemma \ref{fisher_information_binary}.

Let $m \geq 2$ and fix $(\mu, \tau)$. Assume by contradiction that $I(\mu,\tau)$ is singular. By Lemma \ref{fisher_information_binary} we have that there exists $\alpha \neq 0$ such that
\[
\frac{\partial}{\partial \mu}g(y \mid \mu, \tau) = \alpha \frac{\partial}{\partial \tau}g(y \mid \mu, \tau)
\]
for every $y \in \{0, 1, \dots, m\}$. By the second part of Lemma \ref{derivative_binary} for $y = 0$ and $y = m$ it implies
\[
-m\int \frac{e^{\theta}}{(1+e^\theta)^{m+1}}\sqrt{\frac{\tau_1}{2\pi}}e^{-\frac{\tau}{2}(\theta-\mu)^2} \, \d \theta = \alpha\frac{m}{2\tau}\int (\theta-\mu)\frac{e^{\theta}}{(1+e^\theta)^{m+1}}\sqrt{\frac{\tau}{2\pi}}e^{-\frac{\tau}{2}(\theta-\mu)^2} \, \d \theta
\]
and
\[
m\int \frac{e^{m\theta}}{(1+e^\theta)^{m+1}}\sqrt{\frac{\tau}{2\pi}}e^{-\frac{\tau}{2}(\theta-\mu)^2} \, \d \theta = -\alpha\frac{m}{2\tau}\int (\theta-\mu)\frac{e^{m\theta}}{(1+e^\theta)^{m+1}}\sqrt{\frac{\tau}{2\pi}}e^{-\frac{\tau}{2}(\theta-\mu)^2} \, \d \theta.
\]
Since $\alpha \neq 0$, we conclude
\[
\frac{\int (\theta-\mu)\frac{e^{m\theta}}{(1+e^\theta)^{m+1}}\sqrt{\frac{\tau}{2\pi}}e^{-\frac{\tau}{2}(\theta-\mu)^2} \, \d \theta}{\int\frac{e^{m\theta}}{(1+e^\theta)^{m+1}}\sqrt{\frac{\tau}{2\pi}}e^{-\frac{\tau}{2}(\theta-\mu)^2} \, \d \theta} = \frac{\int (\theta-\mu)\frac{e^{\theta}}{(1+e^\theta)^{m+1}}\sqrt{\frac{\tau}{2\pi}}e^{-\frac{\tau}{2}(\theta-\mu)^2} \, \d \theta}{\int\frac{e^{\theta}}{(1+e^\theta)^{m+1}}\sqrt{\frac{\tau}{2\pi}}e^{-\frac{\tau}{2}(\theta-\mu)^2} \, \d \theta},
\]
that means
\[
E[\theta \mid m, \mu, \tau] = E[\theta \mid 1, \mu, \tau].
\]
Since $m > 1$, the above equality directly contradicts Lemma \ref{mean_binary}. Therefore the Fisher Information matrix is non singular.
\end{proof}

\subsection{Proof of Proposition \ref{prop_extended_normal_model}}
Define a one-to-one transformation of $\psi = (\mu, \tau_1, \tau_0)$ as
\begin{equation}\label{psi_normal}
\tilde{\psi} = \sqrt{J}\left(\psi-\psi^* \right)-\Delta_J, \quad \Delta_J = \frac{1}{\sqrt{J}}\sum_{j = 1}^J\Fisher^{-1}(\psi^*) \nabla \log g(Y_j \mid \psi^*),
\end{equation}
with $g(\cdot)$ as in \eqref{marginal_normal} and $\Fisher(\psi^*)$ as in \eqref{eq:fisher_normal_hier}.
\begin{lemma}\label{asymptotic_distribution_psi_extend}
Consider the assumptions of Proposition \ref{prop_extended_normal_model}. Then it holds
\[
\left \lvert \left \lvert \L(\d\tilde{\psi} \mid Y_{1:J})-N\left(\bm{0}, \Fisher^{-1}(\psi^*) \right) \right \rvert \right \rvert_{TV} \to 0,
\]
as $J \to \infty$ in $Q_{\psi^*}^{(\infty)}$-probability, with $\Fisher(\psi^*)$ non singular matrix as in \eqref{eq:fisher_normal_hier}.
\end{lemma}
\begin{proof}
The result follows by Theorem \ref{BvM}. Indeed, the map $\psi \to g(y \mid \psi)$ clearly satisfies identifiability and smoothness requirements. Moreover, by Lemma \ref{fisher_normal_hier} we have
\[
\text{det}\left(\mathcal{I}(\psi^*)\right) = \frac{m^3(m-1)\tau_0^*}{4\tau_1^*(\tau_1^*+m\tau_0^*)^3},
\]
that is strictly positive for every $\psi^*$, with $m \geq 2$. As regards the testing conditions, analogously to Lemma \ref{tests_normal} define
\[
\begin{aligned}
\Psi =& \Psi_1 \times \Psi_2 \times \Psi_3  = \left[\mu^* - 1, \mu^* + 1\right] \times \left[\frac{\tau_1^*}{2}, 2\tau_1^* \right] \times \left[\frac{\tau_0^*}{2}, 2\tau_0^* \right]
\end{aligned}
\]
compact neighborhood of $\psi^*$ and
\[
u_J(Y_{1:J}) = 1- \mathbbm{1}_{g_1(Y_{1:J}) \leq c_1} \, \mathbbm{1}_{g_2(Y_{1:J}) \leq c_2}\,\mathbbm{1}_{g_3(Y_{1:J}) \leq c_3},
\]
where $(c_1, c_2, c_3)$ are positive constants to be fixed and
\[
\begin{aligned}
g_1(Y_{1:J}) =& \left\lvert \bar{Y} - \mu^*\right \rvert, \quad g_2(Y_{1:J}) = \left \lvert \frac{1}{J}\sum_{j = 1}^J\left(\bar{Y}_j-\bar{Y}\right)^2-\frac{1}{\tau_1^*}-\frac{1}{m\tau_0^*} \right \rvert,\\
& g_3(Y_{1:J}) = \left \lvert \frac{1}{J}\sum_{j = 1}^J\left(Y_{j,1}-\hat{Y}_1\right)\left(Y_{j,2}-\hat{Y}_2\right)-\frac{1}{\tau_1^*} \right \rvert,
\end{aligned}
\]
where
\[
\bar{Y} = \frac{1}{J}\sum_{j = 1}^J\bar{Y}_j, \quad \hat{Y}_i = \frac{1}{J}\sum_{j = 1}^JY_{j,i}.
\]
By definition of $g(\cdot)$ in \eqref{marginal_normal}, by the Law of Large numbers we have
\[
\begin{aligned}
\int u_J(y_{1:J})& \, \prod_{j = 1}^Jg(\d y_j \mid \psi^*)\\
& \leq P\left(g_1(Y_{1:J})  > c_1 \right)+P\left(g_2(Y_{1:J})  > c_2 \right)+P\left(g_3(Y_{1:J})  > c_3 \right) \to 0,
\end{aligned}
\]
as $J \to \infty$ for every strictly positive constants $(c_1, c_2, c_3)$. Moreover, notice that
\[
\begin{aligned}
\sup_{\psi \not \in \Psi} \,& \int [1-u_J(y_{1:J})] \, \prod_{j = 1}^Jg(\d y_j \mid \psi)\\
& \leq \sup_{\tau_1 \not\in \Psi_2}P\left(g_3(Y_{1:J})  \leq c_3 \right) +\sup_{\tau_1 \in \Psi_2, \, \tau_0 \not\in \Psi_3} \,P\left(g_2(Y_{1:J})  \leq c_2 \right)+\sup_{\mu \not\in \Psi_1, \, \tau_0 \in \Psi_3, \, \tau_1 \in \Psi_2 }\,P\left(g_1(Y_{1:J})  > c_1 \right).
\end{aligned}
\]
With the same reasoning of the proof of Lemma \ref{tests_normal}, we can find $(c_1, c_2, c_3)$ such that the three suprema goes to $0$ as $J \to \infty$.
\end{proof}
We need another technical Lemma.
\begin{lemma}\label{moments_original_scale}
Consider the setting of Proposition \ref{prop_extended_normal_model}. Then we have
\[
\begin{aligned}
&E \left[(\theta_j-\mu)^2 \mid Y, \psi \right] = \frac{1}{m\tau_0+\tau_1}+ \left(\frac{m\tau_0}{m\tau_0+\tau_1} \right)^2(\bar{Y}_j-\mu)^2,\\
&E \left[(\theta_j-\bar{Y}_j)^2   \mid Y, \psi\right] = \frac{1}{m\tau_0+\tau_1}+ \left(\frac{\tau_1}{m\tau_0+\tau_1} \right)^2(\bar{Y}_j-\mu)^2
\end{aligned}
\]
and
\[
\begin{aligned}
&\text{Var} \left((\theta_j-\mu)^2  \mid Y, \psi \right) = \frac{2}{(m \tau_0+\tau_1)^2}+4\frac{m^2\tau_0^2}{(m\tau_0+\tau_1)^3}(\bar{Y}_j-\mu)^2,\\
& \text{Var} \left((\theta_j-\bar{Y}_j)^2 \mid Y, \psi \right) = \frac{2}{(m \tau_0+\tau_1)^2}+4\frac{\tau_1^2}{(m\tau_0+\tau_1)^3}(\bar{Y}_j-\mu)^2
\end{aligned}
\]
and
\[
\text{Cov}\left((\theta_j-\mu)^2, (\theta_j-\bar{Y}_j)^2  \mid Y, \psi \right) = \frac{2}{(m \tau_0+\tau_1)^2}-4\frac{m\tau_0\tau_1}{(m\tau_0+\tau_1)^3}(\bar{Y}_j-\mu)^2.
\]
\end{lemma}
\begin{proof}
Notice that by \eqref{posterior_theta} we have
\[
(\theta_j-\mu) \mid Y_j, \psi \sim N \left(\frac{m \tau_0}{m\tau_0+\tau_1}(\bar{Y}_j-\mu) ,(m\tau_0+\tau_1)^{-1}\right)
\]
and
\[
(\theta_j-\bar{Y}_j)\mid Y_j, \psi \sim N \left(\frac{\tau_1}{m\tau_0+\tau_1}(\mu-\bar{Y}_j) ,(m\tau_0+\tau_1)^{-1}\right).
\]
Therefore we have
\[
\begin{aligned}
E \left[(\theta_j-\mu)^2  \mid Y, \psi \right]  = \frac{1}{m\tau_0+\tau_1}+ \left(\frac{m\tau_0}{m\tau_0+\tau_1} \right)^2(\bar{Y}_j-\mu)^2,
\end{aligned}
\]
and similarly for the other case. If $X \sim N(\mu, \sigma^2)$, by Lemma \ref{normal_moments} we have $E[X^4] = 3\sigma^4+6\mu^2\sigma^2+\mu^4$.
In our case, considering $\sigma = \left(m \tau_0+\tau_1 \right)^{-1/2}$ and $\mu = \frac{m \tau_0}{m\tau_0+\tau_1}(\bar{Y}_j-\mu)$, we have
\[
E \left[(\theta_j-\mu)^4  \mid Y, \psi \right] = \frac{3}{(m \tau_0+\tau_1)^2}+6\frac{m^2\tau_0^2}{(m\tau_0+\tau_1)^3}(\bar{Y}_j-\mu)^2+\left(\frac{m\tau_0}{m\tau_0+\tau_1} \right)^4(\bar{Y}_j-\mu)^4
\]
and
\[
E^2 \left[(\theta_j-\mu)^2  \mid Y, \psi \right] = \frac{1}{(m \tau_0+\tau_1)^2}+2\frac{m^2\tau_0^2}{(m\tau_0+\tau_1)^3}(\bar{Y}_j-\mu)^2+\left(\frac{m\tau_0}{m\tau_0+\tau_1} \right)^4(\bar{Y}_j-\mu)^4.
\]
Therefore
\[
\text{Var} \left((\theta_j-\mu)^2  \mid Y, \psi \right) = \frac{2}{(m \tau_0+\tau_1)^2}+4\frac{m^2\tau_0^2}{(m\tau_0+\tau_1)^3}(\bar{Y}_j-\mu)^2,
\]
and similarly for the other one. Finally, again by Lemma \ref{normal_moments}, if $Z \sim N(0,1)$ we have $E[\left(\sigma Z+\mu_1\right)^2\left(\sigma Z+\mu_2\right)^2] =  3 \sigma^4+\sigma^2(\mu_1^2+4\mu_1\mu_2+\mu_2^2)+\mu_1^2\mu_2^2$.
In our case, considering $\sigma = \left(m \tau_0+\tau_1 \right)^{-1/2}$, $\mu_1 = \frac{m \tau_0}{m\tau_0+\tau_1}(\bar{Y}_j-\mu)$ and $\mu_2 = \frac{\tau_1}{m\tau_0+\tau_1}(\mu-\bar{Y}_j)$, we have
\[
\begin{aligned}
E\left[(\theta_j-\mu)^2(\theta_j-\bar{Y}_j)^2  \mid Y, \psi \right] =& \frac{3}{(m \tau_0+\tau_1)^2} + \frac{m^2\tau_0^2}{(m\tau_0+\tau_1)^3}(\bar{Y}_j-\mu)^2+\frac{\tau_1^2}{(m\tau_0+\tau_1)^3}(\bar{Y}_j-\mu)^2\\
&-4\frac{m\tau_0\tau_1}{(m\tau_0+\tau_1)^3}(\bar{Y}_j-\mu)^2+\frac{m^2\tau_0^2\tau_1^2}{(m\tau_0+\tau_1)^4}(\bar{Y}_j-\mu)^4
\end{aligned}
\]
and
\[
\begin{aligned}
&E\left[(\theta_j-\mu)^2  \mid Y, \psi\right]E \left[(\theta_j-\bar{Y}_j)^2  \mid Y, \psi \right] =\\
&\frac{1}{(m \tau_0+\tau_1)^2} + \frac{m^2\tau_0^2}{(m\tau_0+\tau_1)^3}(\bar{Y}_j-\mu)^2+\frac{\tau_1^2}{(m\tau_0+\tau_1)^3}(\bar{Y}_j-\mu)^2+\frac{m^2\tau_0^2\tau_1^2}{(m\tau_0+\tau_1)^4}(\bar{Y}_j-\mu)^4.
\end{aligned}
\]
Therefore
\[
\text{Cov}\left((\theta_j-\mu)^2, (\theta_j-\bar{Y}_j)^2  \mid Y, \psi \right) = \frac{2}{(m \tau_0+\tau_1)^2}-4\frac{m\tau_0\tau_1}{(m\tau_0+\tau_1)^3}(\bar{Y}_j-\mu)^2,
\]
as desired. 
\end{proof}
Define
\begin{equation}\label{defin_C_V_normal}
C(\psi) = \begin{bmatrix}
0 & \frac{1}{(m\tau_0+\tau_1)^2} & \frac{m}{(m\tau_0+\tau_1)^2}\\
0 & \frac{1}{(m\tau_0+\tau_1)^2} & \frac{m}{(m\tau_0+\tau_1)^2}
\end{bmatrix}, \quad V(\psi) = \begin{bmatrix}
\frac{2}{(m \tau_0+\tau_1)^2}+4\frac{m\tau_0(\tau_1)^{-1}}{(m\tau_0+\tau_1)^2} & -\frac{2}{(m \tau_0+\tau_1)^2}\\
 -\frac{2}{(m \tau_0+\tau_1)^2} & \frac{2}{(m \tau_0+\tau_1)^2}+4\frac{\tau_1 (m\tau_0)^{-1}}{(m\tau_0+\tau_1)^2}
\end{bmatrix}.
\end{equation}
Now we define a linear rescaling of $\bm{T} = \left(\sum_{j = 1}^J(\theta_j -\bar{Y}_j)^2, \sum_{j = 1}^J(\theta_j -\mu)^2 \right) $ as
\begin{equation}\label{tilde_T_normal}
\tilde{\bm{T}} = \frac{1}{\sqrt{J}}\sum_{j = 1}^J
\begin{bmatrix}
(\theta_j -\bar{Y}_j)^2 - \frac{1}{m\tau_0^*+\tau_1^*}-\left(\frac{\tau_1^*}{m\tau_0^*+\tau_1^*} \right)^2\left(\bar{Y}_j-\mu^* \right)^2\\
(\theta_j -\mu)^2 - \frac{1}{m\tau_0^*+\tau_1^*}-\left(\frac{m\tau_0^*}{m\tau_0^*+\tau_1^*} \right)^2\left(\bar{Y}_j-\mu^* \right)^2
\end{bmatrix}
- C(\psi^*)\Delta_J,
\end{equation}
with $\Delta_J$ as in \eqref{psi_normal}. The next lemma shows the asymptotic distribution of $\tilde{\bm{T}}$  using the weak topology.
\begin{lemma}\label{weak_convergence_normal}
Define $\tilde{\psi}$ and $\tilde{\bm{T}}$ as in \eqref{psi_normal} and \eqref{tilde_T_normal}, respectively. For every $\tilde{\psi}\in\R^D$ it holds
\[
\left\| \L(\d\tilde{\bm{T}} \mid Y_{1:J}, \tilde{\psi})-N\left(C(\psi^*)\tilde{\psi}, V(\psi^*) \right) \right\|_{W} \to 0,
\]
$Q_{\psi^*}^{(\infty)}$-almost surely as $J \to \infty$.
\end{lemma}
\begin{proof}
The result follows by arguments similar to the proof of Lemma \ref{weak_asymptotic_distribution_T}. First of all notice that $C(\psi)$ defined in \eqref{defin_C_V_normal} is such that
\[
C(\psi) =
\begin{bmatrix}
E_{Y_j} \left[\partial_\mu E[(\theta_j-\bar{Y}_j)^2 \mid Y_j, \psi \right] & E_{Y_j} \left[\partial_{\tau_1} E[(\theta_j-\bar{Y}_j)^2 \mid Y_j, \psi \right] & E_{Y_j} \left[\partial_{\tau_0} E[(\theta_j-\bar{Y}_j)^2 \mid Y_j, \psi \right]\\
E_{Y_j} \left[\partial_\mu E[(\theta_j-\mu)^2 \mid Y_j, \psi \right] & E_{Y_j} \left[\partial_{\tau_1} E[(\theta_j-\mu)^2 \mid Y_j, \psi \right] & E_{Y_j} \left[\partial_{\tau_0} E[(\theta_j-\mu)^2 \mid Y_j, \psi \right]
\end{bmatrix},
\]
since by Lemma \ref{moments_original_scale} we have
\[
\begin{aligned}
&E_{Y_j} \left[\partial_\mu E[(\theta_j-\bar{Y}_j)^2 \mid Y_j, \psi \right] = E_{Y_j} \left[\partial_\mu E[(\theta_j-\mu)^2 \mid Y_j, \psi \right] = 0,\\
&E_{Y_j} \left[\partial_{\tau_0} E[(\theta_j-\bar{Y}_j)^2 \mid Y_j, \psi \right]  = E_{Y_j} \left[\partial_{\tau_0} E[(\theta_j-\mu)^2 \mid Y_j, \psi \right] = \frac{m}{(m\tau_0+\tau_1)^2},\\
&E_{Y_j} \left[\partial_{\tau_1} E[(\theta_j-\bar{Y}_j)^2 \mid Y_j, \psi \right]  = E_{Y_j} \left[\partial_{\tau_1} E[(\theta_j-\mu^*)^2 \mid Y_j, \psi \right] = \frac{1}{(m\tau_0+\tau_1)^2}.
\end{aligned}
\]
By the same reasoning in the proofs of \eqref{first_to_prove} and \eqref{second_to_prove} we get
\[
\begin{aligned}
& E_{Y_j} \left[\tilde{T} \mid Y_{1:J}, \psi^*+\frac{\tilde{\psi}+\Delta_J}{\sqrt{J} } \right]  \quad \to \quad C(\psi^*)\tilde{\psi}
\end{aligned}
\]
and
\[
\left \lvert\text{Cov} \left(\tilde{\bm{T}}\mid Y_{1:J}, \psi^*+\frac{\tilde{\psi}+\Delta_J}{\sqrt{J}} \right)-\text{Cov} \left(\tilde{\bm{T}}\mid Y_{1:J}, \psi^* \right) \right\rvert \quad \to \quad 0,
\]
$Q_{\psi^*}^{(\infty)}$-almost surely as $J \to \infty$. Then by \eqref{distribution_mean}, Lemma \ref{moments_original_scale} and the Law of Large Numbers we have
\[
\begin{aligned}
\Var \left( \frac{1}{\sqrt{J}}\sum_{j = 1}^J(\theta_j-\bar{Y}_j)^2 \mid Y_{1:J}, \psi^*  \right) &= \frac{2}{(m \tau^*_0+\tau^*_1)^2}+4\frac{(\tau_1^*)^2}{(m\tau^*_0+\tau^*_1)^3}\frac{1}{J}\sum_{j = 1}^J(\bar{Y}_j-\mu^*)^2\\
& \to \frac{2}{(m \tau^*_0+\tau^*_1)^2}+4\frac{(m\tau_0^*)^{-1}\tau_0^*}{(m\tau^*_0+\tau^*_1)^2}
\end{aligned}
\]
and
\[
\begin{aligned}
\Var \left( \frac{1}{\sqrt{J}}\sum_{j = 1}^J(\theta_j-\mu^*)^2  \mid Y_{1:J}, \psi^* \right) &= \frac{2}{(m \tau^*_0+\tau^*_1)^2}+4\frac{(m\tau_0^*)^2}{(m\tau^*_0+\tau^*_1)^3}\frac{1}{J}\sum_{j = 1}^J(\bar{Y}_j-\mu^*)^2\\
& \to \frac{2}{(m \tau^*_0+\tau^*_1)^2}+4\frac{m\tau_0^*(\tau_1^*)^{-1}}{(m\tau^*_0+\tau^*_1)^2}
\end{aligned}
\]
and
\[
\begin{aligned}
\text{Cov} &\left( \frac{1}{\sqrt{J}}\sum_{j = 1}^J(\theta_j-\bar{Y}_j)^2  ,\frac{1}{\sqrt{J}}\sum_{j = 1}^J(\theta_j-\mu^*)^2  \mid Y_{1:J}, \psi^* \right) = \frac{2}{(m \tau_0+\tau_1)^2}-4\frac{m\tau_0\tau_1}{(m\tau_0+\tau_1)^3}\frac{1}{J}\sum_{j = 1}^J(\bar{Y}_j-\mu)^2\\
& \to  -\frac{2}{(m \tau^*_0+\tau^*_1)^2},
\end{aligned}
\]
$Q_{\psi^*}^{(\infty)}$-almost surely as $J \to \infty$. Finally, by the Law of Large Numbers and calculations similar to Lemma \ref{moments_original_scale}, we have
\[
E\left[(\theta_j-\bar{Y}_j)^{12}  \mid Y_J, \psi\right] < \infty, \quad E\left[(\theta_j-\mu)^{12}  \mid Y_J, \psi\right] < \infty
\]
for every $\psi$. Therefore, with the same arguments in the proof of \eqref{third_to_prove} we conclude that
\[
\frac{1}{J^{3/2}}\sum_{j = 1}^JE \left[(\theta_j-\bar{Y}_j)^{12}  \mid Y_j, \psi ^*+\frac{\tilde{\psi}+\Delta_J}{\sqrt{J}}\right] \to 0, \quad \frac{1}{J^{3/2}}\sum_{j = 1}^JE \left[(\theta_j-\mu^*)^{12}  \mid Y_j, \psi ^*+\frac{\tilde{\psi}+\Delta_J}{\sqrt{J}}\right] \to 0,
\]
$Q_{\psi^*}^{(\infty)}$-almost surely, as $J \to \infty$. The result then follows by Lyapunov version of Central Limit Theorem.
\end{proof}
We need another technical Lemma.
\begin{lemma}\label{single_characteristic}
Consider the assumptions of Proposition \ref{prop_extended_normal_model}. Then it holds
\[
\left \lvert E \left[e^{it_1(\theta_j-\mu)^2+it_2(\theta_j-\bar{Y}_j)^2} \mid Y_j, \psi \right] \right \rvert \leq \frac{e^{-\frac{2\sigma^2\left[\nu_j(t_1+t_2)-(t_1\mu+t_2\bar{Y}_j) \right]^2}{1+4\sigma^4(t_1+t_2)^2}}}{\left[1+4(t_1+t_2)^2\sigma^4\right]^{1/4}},
\]
with $(t_1, t_2) \in \R_2$ and
\[
\nu_j = \frac{m\tau_0}{m\tau_0+\tau_1}\mu+\frac{\tau_1}{m\tau_0+\tau_1}\bar{Y}_j, \quad \sigma^2 = \frac{1}{m\tau_0+\tau_1}.
\]
\end{lemma}
\begin{proof}
By simple computations we get
\[
t_1(\theta_j-\mu)^2+t_2(\theta_j-\bar{Y}_j)^2 = (t_1+t_2)\theta_j^2-2\theta_j(t_1\mu+t_2\bar{Y}_j)+t_1\mu^2+t_2\bar{Y}_j^2.
\]
Therefore
\[
\left \lvert E \left[e^{it_1(\theta_j-\mu)^2+it_2(\theta_j-\bar{Y}_j)^2} \right] \right \rvert \leq \left \lvert E \left[e^{i\left((t_1+t_2)\theta_j^2-2\theta_j(\mu+\bar{Y}_j) \right)} \right] \right \rvert.
\]
Then we can apply Lemma \ref{auxiliary_characteristic}, with
\[
a = t_1+t_2, \quad b = -2(t_1\mu+t_2\bar{Y}_j), \quad \nu = \frac{m\tau_0}{m\tau_0+\tau_1}\mu+\frac{\tau_1}{m\tau_0+\tau_1}\bar{Y}_j, \quad \sigma^2 = \frac{1}{m\tau_0+\tau_1}.
\]
\end{proof}
Consistently with the previous Sections, we denote
\[
\varphi(t \mid Y_j, \psi) = E\left[e^{it_1(\theta_j - \bar{Y}_j)^2+it_2(\theta_j - \mu)^2} \mid Y_j, \psi\right], \quad \tilde{\varphi}(t \mid Y_{1:J}, \psi)=
\E\left[e^{it^\top\tilde{\bT}}\mid Y_{1:J},\psi\right] 
\]
for every $\psi$ and $t = (t_1, t_2) \in \R^2$. The next lemma proves the same convergence of Lemma \ref{weak_convergence_normal} using the total variation distance. 
\begin{lemma}\label{convergence_normal}
Define $\tilde{\psi}$ and $\tilde{\bm{T}}$ as in \eqref{psi_normal} and \eqref{tilde_T_normal}, respectively. For every $\tilde{\psi}\in\R^D$ it holds
\[
\left\| \L(\d\tilde{\bm{T}} \mid Y_{1:J}, \tilde{\psi})-N\left(C(\psi^*)\tilde{\psi}, V(\psi^*) \right) \right\|_{TV} \to 0,
\]
$Q_{\psi^*}^{(\infty)}$-almost surely as $J \to \infty$.
\end{lemma}
\begin{proof}
Since the result holds under the weak metric by Lemma \ref{weak_convergence_normal}, with the same reasoning of Lemma \ref{lemma:conditions_TV_conv} it suffices to prove
\[
\lim_{A \to \infty} \lim_{B \to \infty}\limsup_{J\to\infty}\int_{\left( (t_1+t_2)^2 \leq A, t_1^2 \leq B\right)^c} \left \lvert \tilde{\varphi}(t \mid Y_{1:J}, \psi^{(J)})\right \rvert \, \d t= 0
\]
$Q_{\psi^*}^{(\infty)}$-almost surely as $J \to \infty$, where
\[
\psi^{(J)} = \psi^*+\frac{\tilde{\psi}+\Delta_J}{\sqrt{J}}
\]
Analogously, denote also
\[
\mu^{(J)} = \mu^*+\frac{\tilde{\mu}+\Delta_{J,1}}{\sqrt{J}}, \quad \tau_1^{(J)} = \tau_1^*+\frac{\tilde{\tau}_1+\Delta_{J,2}}{\sqrt{J}}, \quad \tau_0^{(J)} = \tau_0^*+\frac{\tilde{\tau}_0+\Delta_{J,3}}{\sqrt{J}}.
\]
As in \eqref{eq:varphi_tilda_and_not} we have
\[
\left \lvert \tilde{\varphi}(t \mid Y_{1:J},\psi) \right \rvert = \left \lvert \prod_{j = 1}^J\varphi\left(\frac{t}{\sqrt{J}} \mid Y_j, \psi\right) \right \rvert.
\]
Therefore, with the change of variables $u = t_1+t_2$ and $v = t_1$, we have
\[
\begin{aligned}
\int_{\left( (t_1+t_2)^2 \leq A, t_1^2 \leq B\right)^c} &\left \lvert \tilde{\varphi}(t \mid Y_{1:J}, \psi^{(J)})\right \rvert \, \d t \\
&= \int_{\left( u^2 \leq A, v^2 \leq B\right)^c}\prod_{j = 1}^J \left \lvert \varphi\left(\frac{(v, u-v)}{\sqrt{J}} \mid Y_j, \psi^{(J)}\right) \right \rvert \, \d u\d v
\end{aligned}
\] 
Moreover it is easy to see that
\[
\left\{(u, v) \mid u^2 \leq A \text{ and } v^2 \leq B \right\}^c \subset \left\{(u, v) \mid u^2 > A \right\} \cup \left\{(u, v) \mid u^2 \leq A \text{ and } v^2 > B \right\},
\] 
so that 
\begin{equation}\label{decomposition_char}
\begin{aligned}
 \int_{\left( u^2 \leq A, v^2 \leq B\right)^c}&\prod_{j = 1}^J \left \lvert \varphi\left(\frac{(v, u-v)}{\sqrt{J}} \mid Y_j, \psi^{(J)}\right) \right \rvert \, \d u\d v \leq  \int_{ u^2 > A}\prod_{j = 1}^J \left \lvert \varphi\left(\frac{(v, u-v)}{\sqrt{J}} \mid Y_j, \psi^{(J)}\right) \right \rvert \, \d u\d v\\
 &+\int_{ (u^2 \leq A, v^2 > B)}\prod_{j = 1}^J \left \lvert \varphi\left(\frac{(v, u-v)}{\sqrt{J}} \mid Y_j, \psi^{(J)}\right) \right \rvert \, \d u\d v.
\end{aligned}
\end{equation}
For every $\psi$, by Lemma \ref{single_characteristic} with
\[
\nu_j = \frac{m\tau_0}{m\tau_0+\tau_1}\mu+\frac{\tau_1}{m\tau_0+\tau_1}\bar{Y}_j, \quad \sigma^2 = \frac{1}{m\tau_0+\tau_1}
\] 
we have
\[
\prod_{j = 1}^J \left \lvert \varphi\left(\frac{(v, u-v)}{\sqrt{J}} \mid Y_j, \psi\right) \right \rvert \leq \frac{e^{-\frac{2\sigma^2\frac{1}{J}\sum_{j = 1}^J\left[u(\nu_j-\bar{Y}_j)-v(\mu-\bar{Y}_j) \right]^2}{1+4\sigma^4u^2}}}{\left[1+4u^2\sigma^4\right]^{J/4}}.
\]
Notice that
\[
\begin{aligned}
\frac{1}{J}&\sum_{j = 1}^J\left[u(\nu_j-\bar{Y}_j)-v(\mu-\bar{Y}_j) \right]^2 = \\
& = v^2\left[\frac{1}{J}\sum_{j = 1}^J(\mu-\bar{Y}_j)^2 \right]-2uv\left[\frac{1}{J}\sum_{j = 1}^J(\nu_j-\bar{Y}_j)(\mu-\bar{Y}_j)\right]+u^2\left[\frac{1}{J}\sum_{j = 1}^J(\nu_j-\bar{Y}_j)^2 \right]\\
&
\begin{aligned}
= \left[\frac{1}{J}\sum_{j = 1}^J(\mu-\bar{Y}_j)^2 \right]&\left[v-u\frac{\frac{1}{J}\sum_{j = 1}^J(\nu_j-\bar{Y}_j)(\mu-\bar{Y}_j)}{\frac{1}{J}\sum_{j = 1}^J(\mu-\bar{Y}_j)^2} \right]^2\\
&+u^2\left[ \frac{1}{J}\sum_{j = 1}^J(\nu_j-\bar{Y}_j)^2-\frac{\left\{\frac{1}{J}\sum_{j = 1}^J(\nu_j-\bar{Y}_j)(\mu-\bar{Y}_j)\right\}^2}{\frac{1}{J}\sum_{j = 1}^J(\mu-\bar{Y}_j)^2}\right].
\end{aligned}
\end{aligned}
\]
As regards the first element in \eqref{decomposition_char}, by integrating with respect to $v$ we get
\[
\begin{aligned}
\int_{u^2 > A}&\prod_{j = 1}^J \left \lvert \varphi\left(\frac{(v, u-v)}{\sqrt{J}} \mid Y_j, \psi^{(J)} \right) \right \rvert \, \d u\d v \leq \int_{u^2 > A}\frac{e^{-\frac{2\sigma_J^2\frac{1}{J}\sum_{j = 1}^J\left[u(\nu_j-\bar{Y}_j)-v(\mu^{(J)}-\bar{Y}_j) \right]^2}{1+4\sigma_J^4u^2}}}{\left[1+4u^2\sigma_J^4\right]^{J/4}} \, \d u\d v\\
&\leq \sqrt{\frac{\pi}{2\sigma_J^2\frac{1}{J}\sum_{j = 1}^J(\mu^{(J)}-\bar{Y}_j)^2}}\int_A^\infty\frac{e^{-\frac{2\sigma_J^2}{1+4\sigma_J^4u^2}u^2\left[ \frac{1}{J}\sum_{j = 1}^J(\nu_j-\bar{Y}_j)^2-\frac{\left\{\frac{1}{J}\sum_{j = 1}^J(\nu_j-\bar{Y}_j)(\mu_J-\bar{Y}_j)\right\}^2}{\frac{1}{J}\sum_{j = 1}^J(\mu^{(J)}-\bar{Y}_j)^2}\right]}}{{\left[1+4u^2\sigma_J^4\right]^{J/4-1/2}}} \, \d u,
\end{aligned}
\]
where
\[
\sigma_J^2 = \frac{1}{m\tau_0^{(J)}+\tau_1^{(J)}}, \quad \nu_j = \frac{m\tau_0^{(J)}}{m\tau_0^{(J)}+\tau_1^{(J)}}\mu^{(J)}+\frac{\tau_1^{(J)}}{m\tau_0^{(J)}+\tau_1^{(J)}}\bar{Y}_j.
\]
By the Law of Large Numbers we have
\[
\lim \inf \, \frac{1}{J}\sum_{j = 1}^J(\mu^{(J)}-\bar{Y}_j)^2 = \lim \inf \, \frac{1}{J}\sum_{j = 1}^J(\mu^*-\bar{Y}_j)^2 = c_1 > 0
\]
$Q_{\psi^*}^{(\infty)}$-almost surely and similarly
\[
\lim \inf \, \left\{\frac{1}{J}\sum_{j = 1}^J(\nu_j-\bar{Y}_j)^2-\frac{\left\{\frac{1}{J}\sum_{j = 1}^J(\nu_j-\bar{Y}_j)(\mu^{(J)}-\bar{Y}_j)\right\}^2}{\frac{1}{J}\sum_{j = 1}^J(\mu^{(J)}-\bar{Y}_j)^2}\right\} = c_2 > 0,
\]
by Cauchy-Schwartz inequality, $Q_{\psi^*}^{(\infty)}$-almost surely. Moreover, by Lemma \ref{convergence_delta}
\[
\sigma^2_J \in \left(\frac{1}{2}\frac{1}{m\tau_0^*+\tau_1^*}, \frac{2}{m\tau_0^*+\tau_1^*} \right) = (\sigma_1^2, \sigma_2^2)
\]
$Q_{\psi^*}^{(\infty)}$-almost surely, for $J$ high enough. Therefore
\[
\begin{aligned}
\lim_{A \to \infty} \lim_{B \to \infty} \limsup_{J\to\infty} \,& \int_{u^2 > A}\prod_{j = 1}^J \left \lvert \varphi\left(\frac{(v, u-v)}{\sqrt{J}} \mid Y_j, \psi^{(J)} \right) \right \rvert \, \d u\d v \\
&\leq \lim_{A \to \infty} \,\sqrt{\frac{\pi}{2\sigma_1^2c_1}}\int_A^\infty\frac{e^{-\frac{2c_2\sigma_1^2}{1+4\sigma_2^4u^2}u^2}}{{\left[1+4u^2\sigma_1^4\right]^{J/4-1/2}}} \, \d u = 0
\end{aligned}
\]
$Q_{\psi^*}^{(\infty)}$-almost surely. As regards the second addend in \eqref{decomposition_char} we get
\[
\begin{aligned}
\limsup_{J\to\infty} \, \int_{(u^2 \leq A, v^2 > B)}&\prod_{j = 1}^J \left \lvert \varphi\left(\frac{(v, u-v)}{\sqrt{J}} \mid Y_j, \psi^{(J)} \right) \right \rvert \, \d u\d v \\
&\leq \int_{(u^2 \leq A, v^2 > B)} e^{-\frac{2\sigma_1^2}{1+\sigma_2^4A^2}\left[v-u\frac{\frac{1}{J}\sum_{j = 1}^J(\nu_j-\bar{Y}_j)(\mu^{(J)}-\bar{Y}_j)}{\frac{1}{J}\sum_{j = 1}^J(\mu^{(J)}-\bar{Y}_j)^2} \right]^2} \, \d u \d v,
\end{aligned}
\]
$Q_{\psi^*}^{(\infty)}$-almost surely. Fix $A > 0$ and notice that for every $u$ we have
\[
\lim_{B \to \infty}\int_B^\infty e^{-\frac{2\sigma_1^2}{1+\sigma_2^4A^2}\left[v-u\frac{\frac{1}{J}\sum_{j = 1}^J(\nu_j-\bar{Y}_j)(\mu^{(J)}-\bar{Y}_j)}{\frac{1}{J}\sum_{j = 1}^J(\mu^{(J)}-\bar{Y}_j)^2} \right]^2} \, \d v = 0.
\]
Moreover
\[
\int_{u^2 \leq A} e^{-\frac{2\sigma_1^2}{1+\sigma_2^4A^2}\left[v-u\frac{\frac{1}{J}\sum_{j = 1}^J(\nu_j-\bar{Y}_j)(\mu-\bar{Y}_j)}{\frac{1}{J}\sum_{j = 1}^J(\mu-\bar{Y}_j)^2} \right]^2} \, \d u \d v < \infty,
\]
so that, by Dominated Convergence Theorem we get
\[
\lim_{B \to \infty} \int_{(u^2 \leq A, v^2 > B)} e^{-\frac{2\sigma_1^2}{1+\sigma_2^4A^2}\left[v-u\frac{\frac{1}{J}\sum_{j = 1}^J(\nu_j-\bar{Y}_j)(\mu^{(J)}-\bar{Y}_j)}{\frac{1}{J}\sum_{j = 1}^J(\mu^{(J)}-\bar{Y}_j)^2} \right]^2} \, \d u \d v = 0,
\]
for every $A > 0$ and the result follows.
\end{proof}
\begin{proof}[Proof of Proposition \ref{prop_extended_normal_model}]
The result follows by arguments similar to the proof of Theorem \ref{theorem_one_level_nested}, that we briefly summarize. Since by construction
\[
\L\left(\d\psi \mid \bm{\theta}, Y_{1:J} \right) = \L\left(\d\psi \mid \bm{T}, Y_{1:J} \right)
\]
a direct analogue of Lemma \ref{sufficient_lemma} holds. Moreover, by Lemmas \ref{asymptotic_distribution_psi_extend} and \ref{convergence_normal}, we can use Lemma \ref{lemma_marginal_conditional} to prove that $\L\left(\d\tilde{\bm{T}}, \d\tilde{\psi} \mid Y_{1:J}\right)$, as in \eqref{psi_normal}, converges to a Gaussian vector with non singular covariance matrix. Finally, Lemma \ref{positiveDefinite_Gibbs} holds for $P$, being a two-block Gibbs sampler. Therefore the Gibbs sampler on the limit Gaussian target has a strictly positive spectral gap: thus the result follows by Corollary \ref{mixingCorollary}.
\end{proof}

\subsection{Proof of Corollary \ref{extended_spectral_normal}}
Let $\phi = (\tau_1, \tau_0)$ and define
\[
\mathcal{I}(\phi^*) =
\begin{bmatrix}
 \frac{m^2(\tau_0^*)^2}{2(\tau_1^*)^2(\tau_1^*+m\tau_0^*)^2}  & \frac{m}{2(\tau_1^*+m\tau_0^*)^2}  \\
 \frac{m}{2(\tau_1^*+m\tau_0^*)^2} & \frac{m-1}{2(\tau_0^*)^2}+\frac{(\tau_1^*)^2}{2(\tau_0^*)^2(\tau_1^*+m\tau_0^*)^2}
\end{bmatrix}, \quad C(\phi^*) = 
\begin{bmatrix}
\frac{1}{(m\tau_0^*+\tau_1^*)^2} & \frac{m}{(m\tau_0^*+\tau_1^*)^2}\\
\frac{1}{(m\tau_0^*+\tau_1^*)^2} & \frac{m}{(m\tau_0^*+\tau_1^*)^2}
\end{bmatrix}
\]
and 
\[
V(\phi^*) = 
\begin{bmatrix}
\frac{2}{(m \tau^*_0+\tau^*_1)^2}+4\frac{m\tau_0^*(\tau_1^*)^{-1}}{(m\tau^*_0+\tau^*_1)^2} & -\frac{2}{(m \tau^*_0+\tau^*_1)^2}\\
 -\frac{2}{(m \tau^*_0+\tau^*_1)^2} & \frac{2}{(m \tau^*_0+\tau^*_1)^2}+4\frac{\tau_1^* (m\tau_0^*)^{-1}}{(m\tau^*_0+\tau^*_1)^2}
\end{bmatrix}.
\]
We have a preliminary Lemma.
\begin{lemma}\label{rate_convergence_normal}
Consider the setting of Proposition \ref{prop_extended_normal_model}.  Then we have
\[
\gamma(\psi^*) = \min \left\{\frac{1}{1+\lambda_i} \, ; \, \lambda_i \text{ eigenvalue of } V^{-1}\left( \phi^* \right)C(\phi^*)\mathcal{I}^{-1}(\phi^*)C^\top(\phi^*) \right\}.
\]
\end{lemma}
\begin{proof}
With the same reasoning of Corollary \ref{spectral_radius}, $\gamma(\psi^*)$ is the spectral gap on the limiting Gaussian distribution of $\left(\tilde{\psi}, \tilde{\bm{T}} \right)$, given by by Lemmas \ref{asymptotic_distribution_psi_extend} and \ref{convergence_normal}. By inspecting $\Fisher(\psi^*)$ in \eqref{eq:fisher_normal_hier} and $C(\psi^*)$ in \eqref{defin_C_V_normal},  we have that $\tilde{\mu}$ is asymptotically independent from everything else, therefore it suffices to study the Gibbs sampler that alternates updates of $(\tilde{\tau}_1, \tilde{\tau}_0)$ and $\tilde{\bm{T}}$. Then the result follows by the same arguments of Corollary \ref{spectral_radius}.
\end{proof}
\begin{proof}[Proof of Corollary \ref{extended_spectral_normal}]
By Lemma \ref{rate_convergence_normal} we have to study the eigenvalues of 
\begin{equation}\label{eigenvalues}
V^{-1}\left( \phi^* \right)C(\phi^*)\mathcal{I}^{-1}(\phi^*)C^\top(\phi^*).
\end{equation}
Notice that
\[
\mathcal{I}(\phi^*) = \frac{1}{(m\tau_0^*+\tau_1^*)^2}
\begin{bmatrix}
\frac{m^2(\tau_0^*)^2}{2(\tau_1^*)^2}  & \frac{m}{2}\\
\frac{m}{2} & \frac{(m-1)(m\tau_0^*+\tau_1^*)^2+(\tau_1^*)^2}{2(\tau_0^*)^2}
\end{bmatrix}, \quad
C(\phi^*) = \frac{1}{(m\tau_0^*+\tau_1^*)^2}
\begin{bmatrix}
1 & m\\
1& m
\end{bmatrix}
\]
and
\[
V(\phi^*) = \frac{1}{(m\tau_0^*+\tau_1^*)^2}
\begin{bmatrix}
2+4\frac{m\tau_0^*}{\tau_1^*} & -2\\
-2 & 2+4\frac{\tau_1^*}{m\tau_0^*}
\end{bmatrix}
\]
Notice that
\[
\begin{aligned}
\left((m\tau_0^*+\tau_1^*)^2V(\phi^*) \right)^{-1}& = \frac{m\tau_0^*\tau_1^*}{8(m\tau_0^*+\tau_1^*)^2}
\begin{bmatrix}
2+4\frac{\tau_1^*}{m\tau_0^*} & 2\\
2 & 2+4\frac{m\tau_0^*}{\tau_1^*}
\end{bmatrix}\\
& = \frac{1}{4(m\tau_0^*+\tau_1^*)^2}
\begin{bmatrix}
m\tau_0^*\tau_1^*+2(\tau_1^*)^2 & m\tau_0^*\tau_1^*\\
 m\tau_0^*\tau_1^* & m\tau_0^*\tau_1^*+2(m\tau_0^*)^2
\end{bmatrix}
\end{aligned}
\]
and
\[
\begin{aligned}
\left((m\tau_0^*+\tau_1^*)^2\mathcal{I}(\phi^*) \right)^{-1}& = \frac{2(\tau_1^*)^2}{m^2(m-1)(m\tau_0^*+\tau_1^*)^2}
\begin{bmatrix}
\frac{(m-1)(m\tau_0^*+\tau_1^*)^2+(\tau_1^*)^2}{(\tau_0^*)^2} & -m\\
-m & \frac{(m\tau_0^*)^2}{(\tau_1^*)^2} 
\end{bmatrix}
\end{aligned}
\]
Therefore
\[
\begin{aligned}
 \frac{m^2(m-1)(m\tau_0^*+\tau_1^*)^4}{2(\tau_1^*)^2}&C(\phi^*)\mathcal{I}^{-1}(\phi^*)C^\top(\phi^*) =
\begin{bmatrix}
-m^2+\frac{(m-1)(m\tau_0^*+\tau_1^*)^2+(\tau_1^*)^2}{(\tau_0^*)^2} & \frac{m^3(\tau_0^*)^2}{(\tau_1^*)^2}-m \\
-m^2+\frac{(m-1)(m\tau_0^*+\tau_1^*)^2+(\tau_1^*)^2}{(\tau_0^*)^2} & \frac{m^3(\tau_0^*)^2}{(\tau_1^*)^2}-m 
\end{bmatrix}
\begin{bmatrix}
1 & 1\\
m & m
\end{bmatrix}\\
& = \left(\frac{m^4(\tau_0^*)^2}{(\tau_1^*)^2}-2m^2+\frac{(m-1)(m\tau_0^*+\tau_1^*)^2+(\tau_1^*)^2}{(\tau_0^*)^2} \right)
\begin{bmatrix}
1 & 1\\
1 & 1
\end{bmatrix}\\
& = \left(\frac{m^4(\tau_0^*)^4-2m^2(\tau_0^*)^2(\tau_1^*)^2+(m-1)(\tau_1^*)^2(m\tau_0^*+\tau_1^*)^2+(\tau_1^*)^4}{(\tau_0^*)^2(\tau_1^*)^2} \right)
\begin{bmatrix}
1 & 1\\
1 & 1
\end{bmatrix}
\end{aligned}
\]
and
\[
\begin{aligned}
V^{-1}\left( \phi^* \right)C(\phi^*)\mathcal{I}^{-1}(\phi^*)C^\top(\phi^*) =& \left(\frac{m^4(\tau_0^*)^4-2m^2(\tau_0^*)^2(\tau_1^*)^2+(m-1)(\tau_1^*)^2(m\tau_0^*+\tau_1^*)^2+(\tau_1^*)^4}{2m^2(m-1)(\tau_0^*)^2(m\tau_0^*+\tau_1^*)^4} \right)\\
& \begin{bmatrix}
2m\tau_0^*\tau_1^*+2(\tau_1^*)^2 & 2m\tau_0^*\tau_1^*+2(\tau_1^*)^2\\
2m\tau_0^*\tau_1^*+2(m\tau_0^*)^2 & 2m\tau_0^*\tau_1^*+2(m\tau_0^*)^2
\end{bmatrix}
\end{aligned}
\]
Notice that the matrix on the right hand side admits $0$ as an eigenvalue, so that the highest eigenvalue in absolute value is given by its trace, that is
\[
4m\tau_0^*\tau_1^*+2(\tau_1^*)^2+2(m\tau_0^*)^2 = 2(m\tau_0^*+\tau_1^*)^2,
\]
so that the highest eigenvalue of \eqref{eigenvalues} is given by
\[
\frac{m^4(\tau_0^*)^4-2m^2(\tau_0^*)^2(\tau_1^*)^2+(m-1)(\tau_1^*)^2(m\tau_0^*+\tau_1^*)^2+(\tau_1^*)^4}{m^2(m-1)(\tau_0^*)^2(m\tau_0^*+\tau_1^*)^2}.
\]
The result follows by noticing
\[
\begin{aligned}
m^4(\tau_0^*)^4&-2m^2(\tau_0^*)^2(\tau_1^*)^2+(\tau_1^*)^4 = \left[m^2(\tau_0^*)^2-(\tau_1^*)^2 \right]^2\\
& = (m\tau_0^*-\tau_1^*)^2(m\tau_0^*+\tau_1^*)^2.
\end{aligned}
\]
\end{proof}

\subsection{Proof of Lemma \ref{sufficient_lemma_GP}}
\begin{proof}
The proof follows the same lines of Lemma \ref{sufficient_lemma}, that we briefly summarize. Since 
\begin{equation}\label{suff_GP}
\L\left(\d\theta, \d \tau_\beta \mid \betag,  Y^{(n)}) = \L(\d\theta, \d \tau_\beta \mid \bm{T}(\betag),  Y^{(n)}\right)
\end{equation}
holds by definition of $\bm{T}$, reasoning  as in \eqref{key_equality} we can conclude
\[
\begin{aligned}
\L\biggl(\d\bm{T}^{(t)}, \d\theta^{(t)}, \d \tau_\beta^{(t)} \mid & \bm{T}^{(t-1)}, \theta^{(t-1)}, \tau_\beta^{(t-1)} \biggr)\\
& = \hat{\pi}_n\left(\d\bm{T}^{(t)} \mid \theta^{(t-1)}, \tau_\beta^{(t-1)} \right)\hat{\pi}_n\left(\d\theta^{(t)}, \d \tau_\beta^{(t)}\mid \bm{T}^{(t)} \right),
\end{aligned}
\]
which proves that the transition kernel of the induced chain $\left(\bT^{(t)}, \theta^{(t)},  \tau_\beta^{(t)} \right)_{t \geq 1}$ coincides with $\hat{P}_n$. The second part of the Lemma follows by the same reasoning used in \eqref{equality_suff}.
\end{proof}

\subsection{Proof of Corollary \ref{corollary_GP}}
\begin{proof}
By Lemma \ref{sufficient_lemma_GP} we have
\[
t^{(n)}_{mix}(\epsilon, M) = \sup_{\nu \in \sN \left(\hat{\pi}_n, M \right)} \hat{t}^{(n)}_{mix}(\epsilon, \nu).
\]
The result then follows by Corollary \ref{mixingCorollary}, whose conditions hold by assumption.
\end{proof}

\subsection{Proof of Corollary \ref{corollary_GP_ext}}
\begin{proof}
It is easy to show that an analogue of Lemma \ref{sufficient_lemma_GP} holds, with $\psi = (\theta, \tau_\beta, \tau_\epsilon)$ and $\bm{T} = \left(T_{\theta}, T_{\tau_\beta}, T_{\tau_\epsilon}  \right)$. Thus the result follows with the same reasoning of Corollary \ref{corollary_GP}.
\end{proof}

\subsection{Proof of Theorem \ref{theorem_feasible}}
Denote with $\tilde{\mu}_J$ the push-forward measure of $\mu_J$ according to transformations \eqref{tilde_psi} and \eqref{tilde_T}. The next theorem shows that the rescaled version of $\mu_J$ is a warm start for the limiting distribution in Proposition \ref{limiting_sigma}.
\begin{lemma}\label{lemma_feasible}
Let $\mu_J \in \mathcal{P}\left(\R^{lJ+D} \right)$ be as in \eqref{def_feasible}. Then under assumptions $(B1)-(B3)$ there exists a positive constant $M = M(c)$ such that
\[
Q_{\psi^*}^{(J)}\biggl(\tilde{\mu}_J \in \sN\left(N(\bm{0}, \Sigma) ,M\right) \biggr) \quad \to \quad 1,
\]
as $J \to \infty$, with $\Sigma$ as in Proposition \ref{limiting_sigma}.
\end{lemma}
\begin{proof}
According to transformations \eqref{tilde_psi}, we have
\[
\tilde{\mu}_J^{(-1)} = \text{Unif} \left(\sqrt{J}\left( \hat{\psi}_J-\psi^*\right)-\Delta_J, c \right).
\]
Denote with $B_r(\x)$ the closed ball of radius $r > 0$ and center $\x \in \mathbb{R}^D$. By Theorem $5.39$ in \cite{V00} it holds
\begin{equation}\label{efficient}
Q_{\psi^*}^{(J)} \left(\left(\sqrt{J}\left( \hat{\psi}_J-\psi^*\right)-\Delta_J\right) \in B_1(\textbf{0})\right) \quad \to \quad 1,
\end{equation}
as $J \to \infty$. Define now
\begin{equation}\label{M_definition}
M = \max_{\x \in B_{c+1}(\textbf{0}) } \, \frac{\text{Vol}\left(B_{c+1}(\textbf{0})\right)}{N(\x \mid \textbf{0}, \Sigma_D)},
\end{equation}
where Vol$(A)$ is the volume of set $A$ and $N(\textbf{0}, \Sigma_D)$ is the marginal distribution of $N(\textbf{0}, \Sigma)$ over the last $D$ components. It is easy to see that $M < \infty$ and it does not depend on $J$. Therefore, by \eqref{efficient}, we conclude
\[
\begin{aligned}
Q_{\psi^*}^{(J)}\biggl(\tilde{\mu}_J \in \sN\left(N(\bm{0}, \Sigma) ,M\right) \biggr) &\leq Q_{\psi^*}^{(J)} \left(\max_{\x \in B_{c+1}(\textbf{0}) } \, \frac{\d \tilde{\mu}_J^{(-1)}}{\d N(\textbf{0},\Sigma_D) }(\x) \leq M \right)\\
& \leq Q_{\psi^*}^{(J)} \left(\left(\sqrt{J}\left( \hat{\psi}_J-\psi^*\right)-\Delta_J\right) \in B_1(\textbf{0})\right) \quad \to \quad 1,
\end{aligned}
\]
as $J \to \infty$.
\end{proof}
\begin{proof}[Proof of Theorem \ref{theorem_feasible}]
Let $\mu_J \in \mathcal{P}\left(\R^{lJ+D} \right)$ be as in \eqref{def_feasible}. Thus, by Lemma \ref{lemma_feasible} the event $\left\{\tilde{\mu}_J \in \sN\left(\tilde{\pi} ,M\right)\right\}$ with $M$ as in \eqref{M_definition} holds with probability converging to $1$, with respect to the law $Q_{\psi^*}^{(J)}$. 
Then, by Lemma \ref{constructive_lemma}, there exists $\tilde{\nu}_J \in \sN(\tilde{\pi}_J, M)$ such that
\[
\lTV \tilde{\nu}_J - \tilde{\mu}_J \rTV \leq M\lTV \tilde{\pi}_J - \tilde{\pi} \rTV.
\]
Therefore, by the above facts, the triangle inequality and Lemma \ref{sufficient_lemma} we have
\[
\begin{aligned}
\lTV \mu_JP^t_J - \pi_J \rTV &= \lTV \tilde{\mu}_J\tilde{P}^t_J - \tilde{\pi}_J \rTV \\
&\leq \lTV \tilde{\mu}_J\tilde{P}^t_J - \tilde{\nu}_J\tilde{P}^t_J \rTV+\lTV \tilde{\nu}_J\tilde{P}^t_J - \tilde{\pi}_J \rTV\\
&\leq \lTV \tilde{\mu}_J - \tilde{\nu}_J \rTV+\lTV \tilde{\nu}_J\tilde{P}^t_J - \tilde{\pi}_J \rTV\\
& \leq M\lTV \tilde{\pi}_J - \tilde{\pi} \rTV+\sup_{\tilde{\nu}_J\in\sN(\tilde{\pi}_J, M)}\lTV \tilde{\nu}_J\tilde{P}^t_J - \tilde{\pi}_J \rTV\\
& = M\lTV \tilde{\pi}_J - \tilde{\pi} \rTV+\sup_{\nu_J\in\sN(\pi_J, M)}\lTV \nu_JP^t_J - \pi_J \rTV.
\end{aligned}
\]
Thus the result follows by Theorem \ref{theorem_one_level_nested}.
\end{proof}

\end{appendices}

\end{document}